\documentclass[10pt]{phdthesis}

\usepackage{palatino}

\usepackage{fancyhdr}
\usepackage{amssymb}
\usepackage{epsfig}
\usepackage{subcaption}
\usepackage{graphics}
\usepackage{float}
\usepackage{here}
\usepackage{rotating}
\usepackage{multirow}
\usepackage{comment}
\usepackage{epigraph}
\usepackage{array,multirow,nicefrac}
\usepackage{appendix}
\usepackage{acronym}
\usepackage[chapter]{algorithm}
\usepackage{algorithmic}
\RequirePackage{filecontents}
\usepackage{amsthm}
\usepackage{forest}
\usepackage{pgf}
\usepackage{tikz}
\usetikzlibrary{arrows,automata}

\usepackage[utf8]{inputenc}

\usepackage{nameref}
\usepackage{amsmath}
\usepackage{amsthm}
\usepackage{amsfonts}
\usepackage{array}
\usepackage{fancybox}
\usepackage{multirow}
\usepackage{appendix}
\usepackage{epsfig}
\usepackage{cases}
\usepackage{forest}
\usepackage{subcaption}
\usepackage{xcolor}

\usepackage{lscape}
\usepackage{tabularx,booktabs}
\usepackage{wrapfig}

\usepackage{pstricks,pst-node}
\usepackage{fancybox}
\usepackage{psfrag}

\usepackage{changepage}                 

\newtheorem{thm}{Theorem}[section]
\newtheorem{prop}[thm]{Proposition}
\theoremstyle{definition}
\newtheorem{defn}[thm]{Definition}

\usepackage{float,afterpage,dblfloatfix}
\newenvironment{Abstract}
{\begin{center}\textbf{Abstract}%
 \end{center} \small \it \begin{quote}}
{\end{quote}}

\def\addsymbol #1: #2#3{$#1$ \> \parbox{4.6in}{#2 \dotfill \pageref{#3}}\\}

\usepackage{hyperref}
\usepackage[numbers,sort&compress]{natbib}

\usepackage{hyphenat}

\newcommand{\dd}{{\rm d}}

\newcommand{\Lag}{\mathcal{L}}

\newcommand{\mS}{\mathcal{S}}
\newcommand{\mT}{\mathcal{T}}
\newcommand{\mH}{\mathcal{H}}

\newcommand{\GB}{\mathcal{G}}

\newcommand{\be}{\begin{equation}}
\newcommand{\ee}{\end{equation}}
\newcommand{\bea}{\begin{eqnarray}}
\newcommand{\eea}{\end{eqnarray}}

\newcommand{\nablab}{\mathring{\nabla}}
\newcommand{\Rb}{\mathring{R}}
\newcommand{\Gb}{\mathring{G}}

\renewcommand{\bf}[1]{{\textbf{#1}}}

\newcommand{\mU}{\ensuremath\mathcal{U}}

\newcommand{\mpl}{M_{\rm Pl}}

\newcommand{\gt}{\tilde{g}}
\newcommand{\Rt}{\tilde{R}}

\begin{document}
\acrodef{LVQ}{Learning Vector Quantization}
\acrodef{OLVQ1}{optimized learning-rate LVQ}

\pagestyle{fancyplain}
\pagenumbering{roman}

\thispagestyle{empty}

\begin{center}

\pagecolor{black}

\vspace*{1cm}

{\color{white}{\LARGE \textbf{New effective theories of gravitation and their phenomenological consequences}}}

\vspace*{2cm}

\newlength{\offsetpage}
\setlength{\offsetpage}{3.5cm}
\newenvironment{widepage}{\begin{adjustwidth}{-\offsetpage}{-\offsetpage}%
    \addtolength{\textwidth}{2\offsetpage}}%
{\end{adjustwidth}}

\begin{widepage}

\includegraphics[width=1\textwidth]{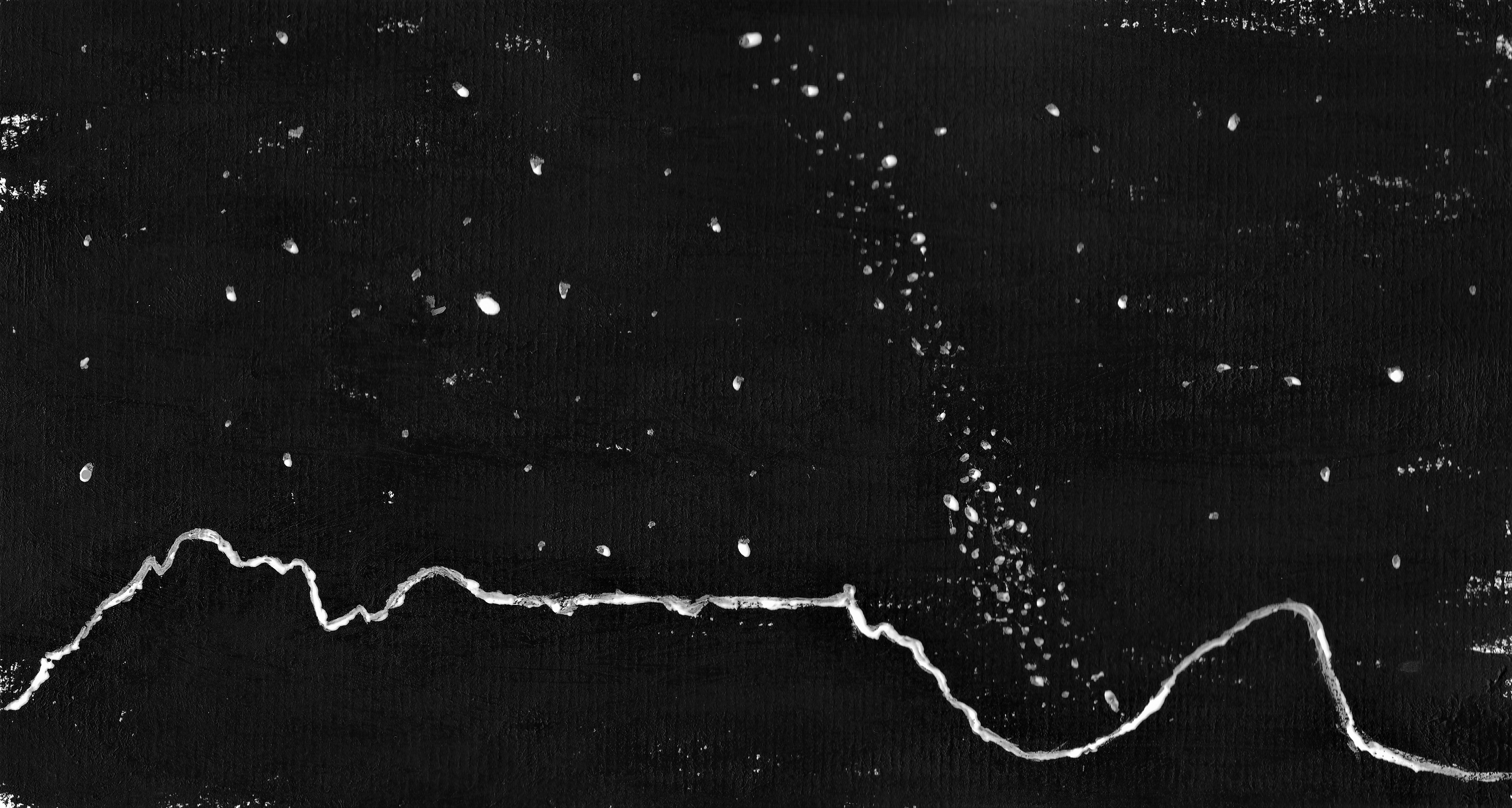}

\end{widepage}

\vspace*{2cm}

{\color{white}{\large \bf{Francisco Jos\'e Maldonado Torralba}}}

\end{center}

\clearpage

\pagecolor{white}

\thispagestyle{empty}
\noindent {\bf Book cover:} Night sky over Table Mountain, F.J. Maldonado Torralba.

\vfill

\noindent 

\cleardoublepage

\thispagestyle{empty}

\parbox[t][0.99\textheight][s]{1\textwidth}{
\centering
\includegraphics[scale=0.16]{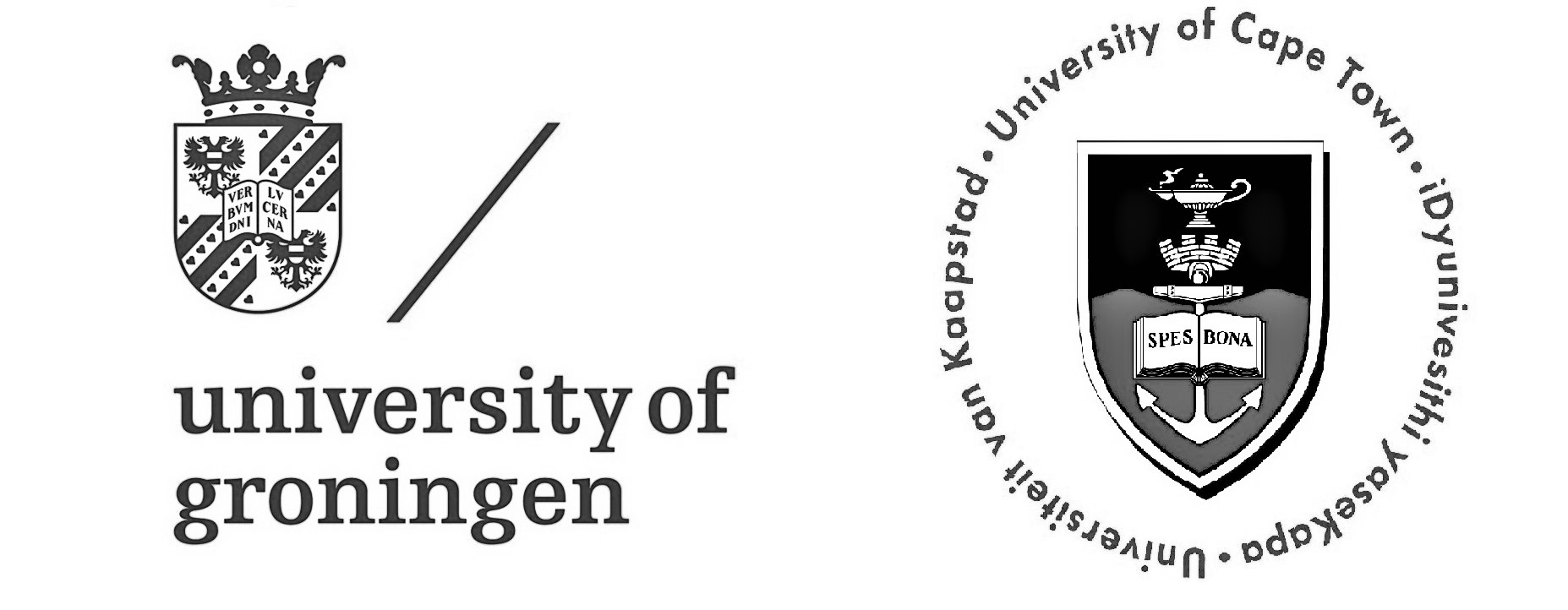}

\bigskip

\huge

\textbf{New effective theories of\\ gravitation and their\\ phenomenological consequences}

\bigskip

\Large

\textbf{PhD thesis}

\large

\bigskip
to obtain the degree of PhD at the\\
University of Groningen\\
on the authority of the\\
Rector Magnificus Prof.\ C.~Wijmenga,\\
and in accordance with\\
the decision by the College of Deans\\
\bigskip
and\\
\bigskip
to obtain the degree of PhD at the\\
University of Cape Town\\
on the authority of the\\
Vice-Chancellor Prof. M. Phakeng\\
and in accordance with\\
the decision by the Doctoral Degrees Board\\
\bigskip
This thesis will be defended in public on\\
\medskip
Tuesday 17 November 2020 at 11.00 hours\\
\medskip
by\\
\medskip
\textbf{Francisco Jos\'e Maldonado Torralba}\\
\medskip
born on 12 January 1993\\
in Sevilla, Spain
}

\newpage

\thispagestyle{empty}

\begin{flushleft}

\noindent

\textbf{Supervisors}\\
$\,$Prof.\ A. Mazumdar\\
$\,$Prof.\ A. de la Cruz Dombriz\\
\vspace{1cm}
\textbf{Assessment Committee}\\ 
$\,$Prof.\ L. Heisenberg\\
$\,$Prof.\ C. Kiefer\\
$\,$Prof.\ D. Roest\\
$\,$Prof.\ P. Dunsby

\vfill

\noindent
Copyright \textcopyright $\,$2020 Francisco Jos\'e Maldonado Torralba

\end{flushleft}


\cleardoublepage
\begin{Abstract}

The objective of this Thesis is to explore Poincaré Gauge theories of gravity and expose some contributions to this field, which are detailed below. Moreover, a novel ultraviolet non-local extension of this theory shall be provided, and it will be shown that it can be ghost- and singularity-free at the linear level.

First, we introduce some fundamentals of differential geometry, base of any gravitational theory. We then establish that the affine structure and the metric of the spacetime are not generally related, and that there is no physical reason to impose a certain affine connection to the gravitational theory. We review the importance of gauge symmetries in Physics and construct the quadratic Lagrangian of Poincar\'e Gauge gravity by requiring that the gravitational theory must be invariant under local Poincar\'e transformations. We study the stability of the quadratic Poincar\'e Gauge Lagrangian, and prove that only the two scalar degrees of freedom (one scalar and one pseudo-scalar) can propagate without introducing pathologies. We provide extensive details on the scalar, pseudo-scalar, and bi-scalar theories. Moreover, we suggest how to extend the quadratic Poincar\'e Gauge Lagrangian so that more modes can propagate safely.

We then proceed to explore some interesting phenomenology of Poincar\'e Gauge theories. Herein, we calculate how fermionic particles move in spacetimes endowed with a non-symmetric connection at first order in the WKB approximation. Afterwards, we use this result in a particular black-hole solution of Poincar\'e Gauge gravity, showing that measurable differences between the trajectories of a fermion and a boson can be observed. Motivated by this fact, we studied the singularity theorems in theories with torsion, to see if this non-geodesical behaviour can lead to the avoidance of singularities. Nevertheless, we prove that this is not possible provided that the conditions for the appearance of black holes of any co-dimension are met. In order to see which kind Black Hole solutions we can expect in Poincar\'e Gauge theories, we study Birkhoff and no-hair theorems under physically relevant conditions. 

Finally, we propose an ultraviolet extension of Poincar\'e Gauge theories by introducing non-local (infinite derivatives) terms into the action, which can ameliorate the singular behaviour at large energies. We find solutions of this theory at the linear level, and prove that such solutions are ghost- and singularity-free. We also find new features that are not present in metric Infinite Derivative Gravity. 
\end{Abstract}
\lhead[]{\fancyplain{}{\rightmark}}
\chead[\fancyplain{}{}]{\fancyplain{}{}}
\rhead[\fancyplain{}{\leftmark}]{\fancyplain{}{}}
\tableofcontents
\addcontentsline{toc}{chapter}{Publications}
\chapter*{Publication list}
\pagestyle{empty}

\vspace{-0.6cm} 

This is the list of publications I have written during the course of my PhD. The names of the authors in each article are in alphabetical order.

\paragraph{P1}\label{P1} J.~A.~R.~Cembranos, J.~Gigante Valcarcel and F.~J.~MALDONADO TORRALBA\\
\phantom{aaaa}\emph{Singularities and n-dimensional black holes in torsion theories}\\
\phantom{aaaa}JCAP {\bf 1704} 021 (2017)\\
\phantom{aaaa}arXiv:1609.07814

\paragraph{P2}\label{P2} Á.~de la Cruz-Dombriz and F.~J.~MALDONADO TORRALBA\\
\phantom{aaaa}\emph{Birkhoff's theorem for stable torsion theories}\\
\phantom{aaaa}JCAP {\bf 1903} 002 (2019)\\
\phantom{aaaa}arXiv:1811.11021

\paragraph{P3}\label{P3} J.~A.~R.~Cembranos, J.~Gigante Valcarcel and F.~J.~MALDONADO TORRALBA\\
\phantom{aaaa}\emph{Fermion dynamics in torsion theories}\\
\phantom{aaaa}JCAP {\bf 1904} 039 (2019)\\
\phantom{aaaa}arXiv:1805.09577

\paragraph{P4}\label{P4} Á.~de la Cruz-Dombriz, F.~J.~MALDONADO TORRALBA and A.~Mazumdar\\
\phantom{aaaa}\emph{Nonsingular and ghost-free infinite derivative gravity with torsion}\\
\phantom{aaaa}Phys.\ Rev.\ D {\bf 99} no.10, 104021 (2019)\\
\phantom{aaaa}arXiv:1812.04037

\paragraph{P5}\label{P5} J.~A.~R.~Cembranos, J.~Gigante Valcarcel and F.~J.~MALDONADO TORRALBA\\
\phantom{aaaa}\emph{Non-Geodesic Incompleteness in Poincaré Gauge Gravity}\\
\phantom{aaaa}Entropy {\bf 21} no.3, 280 (2019)\\
\phantom{aaaa}arXiv:1901.09899

\paragraph{P6}\label{P6} J.~Beltr\'an Jiménez and F.~J.~MALDONADO TORRALBA\\
\phantom{aaaa}\emph{Revisiting the Stability of Quadratic Poincar\'e Gauge Gravity}\\
\phantom{aaaa}Eur. Phys. J. C 80 7, 611 (2020)\\
\phantom{aaaa}arXiv:1910.07506

\paragraph{P7}\label{P7} Á.~de la Cruz-Dombriz, F.~J.~MALDONADO TORRALBA and A.~Mazumdar\\
\phantom{aaaa}\emph{Ghost-free higher-order theories of gravity with torsion}\\
\phantom{aaaa}Submitted, 2020\\
\phantom{aaaa}arXiv:1911.08846

\cleardoublepage
\addcontentsline{toc}{chapter}{Acknowledgements}
\chapter*{Acknowledgements}
\pagestyle{empty}

\vspace{-1cm} 

\begin{flushright}
Francisco Jos\'e Maldonado Torralba\\ Sevilla\\ \today\\
\end{flushright}

\vspace*{0.5cm}

First of all, I would like to thank both my supervisors, Dr. \'Alvaro de la Cruz Dombriz and Prof. Anupam Mazumdar, for giving me the opportunity to do this Dual PhD program and work at the Cosmology and Gravity group of the University of Cape Town and the Van Swinderen Institute at the University of Groningen. Their help and support have played an important role in the development of this Thesis. I am also grateful to them for encouraging me to travel and present this work at different international conferences and seminars in Spain, Norway, Netherlands, France, Czech Republic, and South Africa.\\

I would like to thank the people from both the University of Cape Town and the University of Groningen. They have provided me with a very comfortable and inspirational place of work, and the discussions of research, projects, and life in general, has influenced the outcome of this work. I would like to express my profound gratitude to my collaborators Dr. Jorge Gigante Valcarcel, Prof. José Alberto Ruiz Cembranos, and Dr. Jose Beltr\'an Jim\'enez. I have learnt a lot from our discussions, no matter the subject, and it has always been a pleasure to work with you.\\

I would like to thank also the financial support of National Reasearch Foundation of South Africa Grants No.120390, Reference: BSFP190416431035, and No.120396, Reference: CSRP190405427545, and No 101775, Reference: SFH150727131568. I would like to acknowledge the financial support from the NASSP Programme - UCT node. Also, the PhD research was funded by the Netherlands Organization for Scientific Research (NWO) grant number 680-91-119. Moreover, during the PhD I had also the opportunity to perform research visits at various institutions. For that, I would like to thank the financial support from the Erasmus+ KA107 Alliance4Universities programme to do a research stay at the Universidad Aut\'onoma de Madrid, hosted by Prof. Juan Garc\'ia-Bellido Capdevila. I would like to acknowledge the financial support from the Erasmus+ programme to do a research visit at Radboud University, hosted by Dr. David Nichols. I would like to thank the financial support of the Norwegian Centre for International Cooperation in Education to do a research visit at the University of Oslo, hosted by Prof. David Mota. Finally, I would like to acknowledge the financial support of the Universidad de Salamanca, to do a research visit hosted by Dr. Jose Beltr\'an Jim\'enez.\\

I want to also thank the examiners for the comments and suggestions, which have led to an improvement of the Thesis.\\

Of course, there are many people who have inspired, supported, and helped me during this PhD period whom I would like to acknowledge. \\

First of all, as it cannot be otherwise, I would like to thank my parents because it is only through their constant support, comprehension and love, that I am the person I am today. \emph{Gracias, porque me lo hab\'eis dado todo}. Also, I have been lucky enough to share all my concious life with a sister that enlightens everything in her path,  including this thesis. \\

Moreover, I am glad to be surrounded by a wonderful family of aunts, uncles and cousins, which have always been there giving me the best \emph{Gracias, sab\'eis lo importantes que sois para m\'i}.\\

I must also thank Alberto, who is like another member of my family. You have played a very important role in most of my cheriest and unforgettable memories. Thank you for everything brother.\\

A PhD is quite an ardous journey, and sometimes it can overcome you, but it has given me the opportunity to meet very special people along the way.\\

Idoia, probably nothing that I can say would make justice to what has meant knowing you. Nonetheless, let me use this lines to thank you for your constant presence, despite the distance, your essential support, and in general for the way you make my life better. I could not wish for a better companion.\\

Miguel, I could not have imagined I would have met someone like you during this time. Everywhere in the world, you have made me feel like I was at home, and I am deeply grateful for that.\\

Alberto Valenciano, I want to thank you because you have made my stay in Cape Town a lot funnier with your sense of humour. \\

Finally, I would like to thank the people at the first floor of the Mathematics building at UCT, with whom I have shared laughs and created awesome memories.

\cleardoublepage

\pagestyle{empty}

\vspace*{4cm} 

\begin{flushright}
\emph{In the beginning when the world was young\\ there were a great many thoughts\\ but no such thing as a truth}\\
$\,$\\
Sherwood Anderson\\
\emph{Winesburg, Ohio}\\ 
New York: B.W. Huebsch (1919)
\end{flushright}

\cleardoublepage

\addcontentsline{toc}{chapter}{Acronyms and conventions}
\chapter*{Acronyms and conventions}
\pagestyle{empty}

\vspace{-1.5cm} 

\section*{List of Acronyms}

\begin{acronym}
\acro{BH}{Black Hole}
\acro{FLRW}{Friedmann-Lema\^itre-Robertson-Walker}
\acro{IDG}{Infinite Derivative Gravity}
\acro{IR}{Infrared}
\acro{GR}{General Relativity}
\acro{PG}{Poincar\'e Gauge}
\acro{PGT}{Poincar\'e Gauge Theory}
\acro{STEGR}{Symmetric Teleparallel Gravity}
\acro{SM}{Standard Model}
\acro{TEGR}{Teleparallel Gravity}
\acro{UV}{Ultraviolet}
\acro{WKB}{Wentzel-Kramers-Brillouin}
\end{acronym}

\section*{Conventions and Notations}

In this Thesis we shall consider the \emph{mostly plus} metric signature $\left(-++\,+\right)$. Furthermore, unless specified, we shall work in natural units $c=\hbar=G=1$. Sometimes these constants shall be written explicitly for clarity purposes. We will use the index 0 to refer to the temporal component and the rest of the indices 1,2,3 for the spatial ones. Greek indices will denote the spacetime coordinates and Latin indices will indicate the tangent space coordinates. Moreover, the Einstein summation convention shall apply.

For the expressions containing symmetric or antisymmetric terms we shall use the usual parentheses and brackets in the indices, that are defined as follows:
\begin{equation}
A_{\left(\mu\nu\right)}:=\frac{1}{2}\left(A_{\mu\nu}+A_{\nu\mu}\right)\quad{\rm and}\quad A_{\left[\mu\nu\right]}:=\frac{1}{2}\left(A_{\mu\nu}-A_{\nu\mu}\right).
\nonumber
\end{equation}

The conventions for the affine connection $\Gamma$ and curvature of the of the spacetime are given in the following. The covariant derivative $\nabla$ of a tensor shall be computed as
\begin{eqnarray}
\nabla_{\rho}A^{\mu_{1}...\mu_{k}}\,_{\nu_{1}...\nu_{l}}&=&\partial_{\rho}A^{\mu_{1}...\mu_{k}}\,_{\nu_{1}...\nu_{l}}+\sum_{i=1}^{k}\Gamma^{\mu_{i}}\,_{\rho d}A^{\mu_{1}...d...\mu_{k}}\,_{\nu_{1}...\nu_{l}}
\nonumber
\\
&&-\sum_{i=1}^{l}\Gamma^{d}\,_{\rho\nu_{i}}A^{\mu_{1}...\mu_{k}}\,_{\nu_{1}...d...\nu_{l}}.
\nonumber
\end{eqnarray}
Moreover, the D'Alambertian operator will be defined as $\Box=g^{\mu\nu}\nabla_{\mu}\nabla_{\nu}$. 

The expression in coordinates of a general affine connection, $\Gamma^{\rho}\,_{\mu\nu}$, can be decomposed into three terms as follows
\begin{equation}
\Gamma^{\rho}\,_{\mu\nu}=\mathring{\Gamma}^{\rho}\,_{\mu\nu}+K^{\rho}\,_{\mu\nu}+L^{\rho}\,_{\mu\nu},
\nonumber
\end{equation}
where
\begin{itemize}

\item $\mathring{\Gamma}^{\rho}\,_{\mu\nu}$ are the Christoffel symbols of the \emph{Levi-Civita connection}, related with the metric tensor $g_{\mu\nu}$ as
\begin{equation}
\mathring{\Gamma}^{\rho}\,_{\mu\nu}=\frac{1}{2}g^{\rho\sigma}\left(\partial_{\mu}g_{\nu\sigma}+\partial_{\nu}g_{\sigma\mu}-\partial_{\sigma}g_{\mu\nu}\right).
\nonumber
\end{equation}

\item $K^{\rho}\,_{\mu\nu}$ is the \emph{contorsion tensor}, which is defined as
\begin{equation}
K^\rho{}_{\mu\nu}=\frac12 \Big(T^\rho{}_{\mu\nu}+T_{\mu}{}^\rho{}_\nu+T_\nu{}^\rho{}_\mu\Big),
\nonumber
\end{equation}
where $T^\rho{}_{\mu\nu}$ is the antisymmetric part of the connection, known as the \emph{torsion tensor}:
\begin{equation}
T^{\rho}\,_{\mu\nu}=\Gamma^{\rho}\,_{\mu\nu}-\Gamma^{\rho}\,_{\nu\mu}.
\nonumber
\end{equation}

\item $L^{\rho}\,_{\mu\nu}$ is the \emph{disformation tensor}, defined as follows
\begin{equation}
L^\rho{}_{\mu\nu}=\frac12 \Big(M^\rho{}_{\mu\nu}-M_{\mu}{}^\rho{}_\nu-M_\nu{}^\rho{}_\mu\Big),
\nonumber
\end{equation}
where $M^\rho{}_{\mu\nu}$ is \emph{non-metricity} of the connection, given as
\begin{equation}
M_{\rho\mu\nu}=\nabla_{\rho}g_{\mu\nu}.
\nonumber
\end{equation}

\end{itemize}

The expressions of the curvature tensors in terms of the affine connection shall follow Wald's convention, namely:

\begin{itemize}
\item {\emph{Riemann tensor}}
\begin{equation}
R_{\mu\nu\rho}\,^{\sigma}=\partial_{\nu}\Gamma^{\sigma}\,_{\mu\rho}-\partial_{\mu}\Gamma^{\sigma}\,_{\nu\rho}+\Gamma^{\alpha}\,_{\mu\rho}\Gamma^{\sigma}\,_{\alpha\nu}-\Gamma^{\alpha}\,_{\nu\rho}\Gamma^{\sigma}\,_{\alpha\mu}.
\nonumber
\end{equation}
\item {\emph{Ricci tensor}}
\begin{equation}
R_{\mu\nu}=R_{\mu\rho\nu}\,^{\rho}.
\nonumber
\end{equation}
\item {\emph{Scalar curvature}} or {\emph{Ricci scalar}}
\begin{equation}
R=g^{\mu\nu}R_{\mu\nu}.
\nonumber
\end{equation}
\end{itemize}

The metric that has only diagonal components different from zero, given by $(-1,1,1,1)$, is known as the \emph{Minkowski metric}, and is usually denoted $\eta_{\mu\nu}$.


\cleardoublepage

\pagestyle{headings}
\newcommand{\publ}{}

\pagestyle{fancyplain}
\renewcommand{\sectionmark}[1]{\markright{\it \thesection.\ #1}}
\renewcommand{\chaptermark}[1]{\markboth{
       \it \thechapter.\ #1}{}}
\lhead[\thepage]{\fancyplain{}{\rightmark}}
\chead[\fancyplain{}{}]{\fancyplain{}{}}
\rhead[\fancyplain{}{\leftmark}]{\fancyplain{}{\thepage}}
\lfoot[]{}
\cfoot[]{}
\rfoot[]{}

\pagenumbering{arabic}

\renewcommand{\publ}{}


\chapter{Introduction}

\PARstart{A}ll natural sciences share a common feature, their truth is derived from empirical observation. That is why it is always exciting to find phenomena that we cannot explain with our current models of nature. When this happens, two lines of thought can be considered. Either there is something that we have not observed yet that is affecting that strange measure, or the current theory is wrong, since it is no longer validated by experimentation.

A great example of this fact occurred in the middle of the nineteenth century, as a consequence of the study of the motion of the known planets of the Solar System made by Urbain Le Verrier \cite{le1858theorie}. During the development of that study, he realised a strange behaviour in the motion of Uranus, which could be explained by the presence of an unknown planet. He predicted its mass and position and sent it to the German astronomer Johann Galle \cite{Leverrier:1910}, who observed the planet which we now denote as Neptune the same evening the letter from Le Verrier arrived \cite{galle1846account}.

During his study of the Solar System, Le Verrier also measured an anomaly in the orbit of Mercury: its perihelium precesses 38'' per century \cite{leverrier1859compte}. This observation could not be described using Newton's law of gravitation with the known planets. Inspired by the success of the Neptune discovery, Le Verrier proposed a new planet, Vulcan, that would be placed between Mercury and the Sun, and which would be able to explain the precession. On publication of this research, Lescarbault, an amateur astronomer, announced that he had already observed such a planet transiting the Sun. The discovery was supported by many members of the scientific community, and other astronomers also reported sightings of this object, so Vulcan became the new planet of the Solar System.

Nevertheless, many of the observations that were used to prove the existence of the planet turned out to be false or mistaken. Also, the astronomer Simon Newcomb confirmed the precession of the perihelium of Mercury measured by Le Verrier, and found a slightly larger value, 43'' per century \cite{newcomb1882discussion}. Moreover, in this and subsequent works he gave strong arguments to discard all the proposed hypotheses of additional matter between Mercury and the Sun \cite{newcomb1895elements}. This was the time to open Pandora's box by allowing modifications of Newton's law.

This gave rise to multiple theories that claimed to correctly predict the anomaly. Probably the most famous was Asaph Hall's proposal \cite{hall1894suggestion}, which consisted in modifying the inverse square law in Newton's Equation as $F_{g}=Cr^{-\alpha}$, where $\alpha$ is a constant that could be tuned to explain the precession, obtaining a value of $\alpha=2.00000016$. Unfortunately, this small change made that the motion of the other inner planets and the Moon could not be explained satisfactorily. 

Other physicists, such as Tisserand, Weber, Z\"olner, L\'evy, and Ritz, tried to applied laws inspired in Electromagnetism to obtain the correct value for Mercury's precession \cite{tisserand1890mouvement,zollner1872natur,levy1890application,warburton1946advance}. Unfortunately, they either gave incorrect values or they could be refuted by other physical observations.

It was not until 1915, with the development of the General Theory of Relativity by Albert Einstein \cite{einstein1915allgemeinen,Einstein:1915ca,Einstein:1916vd}, that a succesful explanation compatible with the current experimental data was found \cite{Einstein:1915bz}. This is quite a remarkable and beautiful theory, as we shall explore in Section \ref{2.1}. For a nice review on the history of the Mercury problem and the development of General Relativity we refer the reader to \cite{ROSEVEARE1979165}.

Einstein's General Theory of Relativity (GR) is based on the fact that the effects of a homogeneous gravitational field are indistinguishable from uniformly accelerated motion, which is known as the \emph{Weak Equivalence Principle}. Or, which is equivalent, that any gravitational field can be canceled out locally by inertial forces. As an example, if we were inside a plane that starts to free fall into the ground we would feel no gravitational effect at all. This means that the gravitational fields need to have the same structure as inertial forces. The way that Einstein thought to take this into account was to propose that the spacetime that we live in is actually curved, and that the effect of gravity would be a consequence of this curvature. An enlightening thought experiment (or \emph{gedanken}) is to imagine two planes going from different points in the Equator to the South Pole. Therefore it is clear that at some point they would see each other getting closer, as if an attractive force was acting between the two. Indeed, as we know, this is an effect due exclusively to the curvature of the Earth. Still, it feels like a real force and if one of the two pilots does not accelerate in another direction a fatal accident would occur. We shall give more details on GR and the structure of gravitational theories in Section \ref{2.1}.\\

An analogous situation to the Neptune and Mercury problems is occurring at this moment in the field of gravitational physics, where there are at least two cosmological and astrophysical phenomena that cannot be explained within the conceptual formalism of GR and the matter content of the Standard Model of Particles (SM). On the one hand, we have the observation of the accelerated expansion of the Universe by the Supernova Cosmology Project and the High-Z Supernova Search Team \cite{Riess:1998cb} (also confirmed by later measures such as Baryon Acoustic Oscillations \cite{Eisenstein:2005su} and the Cosmic Microwave Background \cite{Boughn:2003yz}). On the other hand, the rotational curves of galaxies that have been measured do not fit the predictions of GR with baryonic matter \cite{Sahni:2004ai}. 

The usual way of solving this problem is to assume that GR is the correct theory to describe gravitation, and that we need to add new forms of energy and matter to be in agreement with experiment. In the case of the accelerated expansion, an exotic form of energy is introduced to the field equations in the form of a cosmological constant $\Lambda$, that produces a repulsive gravitational {\it force} capable of explaining the current measurements \cite{Peebles:2002gy}. Elseways, in the case of the rotational curves, a new type of non-baryonic matter is introduced. This new form of matter, usually known as cold dark matter, has the property that it only interacts weakly and gravitationally, and its velocity is much lower than the speed of light. This is why our current cosmological model is indeed known as $\Lambda$ Cold Dark Matter model ($\Lambda$CDM) \cite{weinberg2008cosmology}.

Nevertheless, the previous approach suffers from some important shortcomings. First of all, the expected theoretical value of the cosmological constant exceeds the observations by 120 orders of magnitude \cite{Weinberg:1988cp,Martin:2012bt}, which is by far the worst prediction in the history of Physics. With respect to the introduction of Dark Matter, it is not clear yet which may be suitable and detectable candidates, since all the attempts so far have just found constraints on the possible mass and other properties of the proposed particles. There are no direct and conclusive evidences of these weakly interacting particles, we just know there are some proposals such that their effect is compatible with the current measures \cite{Bertone:2004pz,Ackermann:2013yva,Khachatryan:2014rra,Buckley:2017ijx}. Although, since they are indirect measures, we do not know if these effects are due to the Dark Matter or other astrophysical phenomena. 

Moreover, even if we consider that $\Lambda$CDM may be a good description of the large-scale Universe, there exists a tension between local and late-time measurements of the Hubble parameter, which accounts for the rate of expansion of the Universe \cite{Riess:2020sih}.

And last but not least, the singularities present in GR indicate the limited range of validity of the theory, which is a purely classical theory that does not take into account any quantum effect \cite{Senovilla:2018aav}. 

Due to these disadvantages, other approaches have been proposed, as in the case when modifications of Newton's gravity were considered. They are based on modifying the GR action, commonly known as the \emph{Einstein-Hilbert action}, which is given by
\begin{equation}
S_{{\rm {GR}}}=\int {\rm d}^{4}x\sqrt{-g}\left(\frac{1}{16\pi G}\mathring{R}+\mathcal{L}_{M}\right),
\end{equation}
where $\mathring{R}$ is the Ricci scalar in terms of the Levi-Civita connection, $g$ is the determinant of the metric tensor, $G$ is the gravitational constant, and $\mathcal{L}_{M}$ accounts for the Lagrangian of the matter fields present in the system.\\
Then, one can think of a straightforward modification consisting on changing the scalar of curvature by an arbitrary function of it, $f(\mathring{R})$\footnote{From now on we will refer to these theories as $f(R)$ theories, in order to follow the usual convention in the literature.}. Indeed, these are the well-known and widely studied $f(R)$ theories of gravity \cite{Sotiriou:2008rp,DeFelice:2010aj,Capozziello:2011et}. One can show that, with the correct choice of the function, one can obtain the accelerated expansion of the Universe without dark energy \cite{delaCruzDombriz:2006fj,Amendola:2006we,Cognola:2007zu,Hu:2007nk}. Nevertheless, there is not any global function that can fit all the current data without introducing new sources of matter and energy. In fact, it can be proven that $f(R)$ theories, independently of the chosen function, are equivalent to GR with an extra scalar field \cite{Barrow:1988xh,Nunez:2004ji}. This extra degree of freedom in the theory is the one allowing to predict the cosmological acceleration. 

As a matter of fact, all the Lorentz invariant four-dimensional local extensions of the Einstein-Hilbert action introduce new degrees of freedom into the theory. This is due to the Lovelock theorem, which proves that from a local gravitational action which contains only second derivatives of a single four-dimensional spacetime metric, the only possible equations of motion are the well-known Einstein field equations \cite{Lovelock:1971yv}. This indicates that if we modify GR we need to violate one or more of the assumptions in Lovelock theorem.\\
Accordingly, the theories that break the Lovelock assumptions by considering more fields apart from the metric, can be classified depending on the nature of the extra fields that they add to the theory, \emph{i.e.} scalar, vector, or tensor fields \cite{Capozziello:2007ec,Capozziello:2011et,Clifton:2011jh,Berti:2015itd,baojiu2019modified,Heisenberg:2018vsk}:
\begin{itemize}

\item \emph{Scalar-tensor theories}: these are some of the most studied and best established modified theories in the literature. In 1974 Horndeski introduced in his famous article the most general Lorentz and diffeomorphism invariant scalar-tensor theory with second order equations of motion \cite{Horndeski:1974wa}. The latter condition is considered in order to avoid instabilities, but actually one can consider having higher derivatives in the equations of motion without incurring in a pathological behaviour. These are known as \emph{beyond Horndeski} theories \cite{Zumalacarregui:2013pma,Gleyzes:2014dya,Gleyzes:2014qga,Horndeski:2016bku}. One can even go beyond these theories and allow the propagation of 3 stable degrees of freedom, obtaining the so-called \emph{DHOST} theories \cite{Langlois:2015cwa,Achour:2016rkg}.\\
Paradigmatic examples of scalar-tensor theories, which are particular cases of the already mentioned, include $f(R)$ theories \cite{Sotiriou:2008rp}, generalised Brans-Dicke theories \cite{Bergmann:1968ve}, Galileons \cite{Nicolis:2008in}, and the \emph{Fab Four} \cite{Charmousis:2011bf}.

\item \emph{Vector-tensor theories}: in this case one has to differentiate between massive and massless vector fields. On the one hand, for the massless case, only one non-minimal coupling to the curvature is allowed in order to maintain second order equations of motion \cite{Horndeski:1976gi,Jimenez:2013qsa}. On the other hand, when the field is massive, the most general theory with second order field equations becomes more complex and can be found in \cite{Heisenberg:2014rta,Jimenez:2016isa}. Again in this case, the condition of having second order equations of motions can be relaxed if we still make sure that no extra degrees of freedom propagate, which would induce instabilities. Following these reasoning one arrives at the \emph{beyond generalised Proca} theories \cite{Heisenberg:2016eld}.\\
Some other examples of vector-tensor theories are Einstein-aether \cite{Jacobson:2008aj} and Horava-Lifschitz gravity \cite{Sotiriou:2010wn}.

\item \emph{Tensor-tensor theories}: in this kind of theories is more complicated to perform a stability analysis to construct the most general stable action. We will just outline that the most relevant modified theories that fall under this classification are massive gravity \cite{deRham:2010ik} and bimetric gravity \cite{rosen1978bimetric}.

\end{itemize}

Of course, there will be theories that propagate different kinds of degrees of freedom at the same time. This is the case for example of Moffat's scalar-tensor-vector gravity theory \cite{Moffat:2005si}, which propagates a scalar and a vector. Moreover, it is the case of Poincar\'e Gauge Gravity, which introduces two scalars, two vectors, and two tensor fields. This particular theory is the object of most of this thesis research, and we shall introduce it and motivate it in Section \ref{2.2}. Also, we shall analyse the stability of its propagating modes in Section \ref{2.3}, and study some of its interesting phenomenology in Chapter \ref{3}.\\

One can wonder about what happens if we break the locality assumption of the Lovelock's theorem, which we know it will lead to modifications of Einstein's theory. A non-local Lagrangian can be constructed using non-polynomial differential operators, such as
\begin{equation}
\mathcal{L}=\mathcal{L}\left(...,\frac{1}{\Box}\pi,\rm{ln}\left(\frac{\Box}{M_{S}^{2}}\right),\rm{e}^{\frac{\Box}{M_{S}^{2}}},...\right),
\end{equation}
where $\pi$ can be any kind of tensorial field (which clearly includes scalars and vectors), $\Box$ is the d'Alembertian operator, $M_S$ is the mass scale at which non-local effects manifest, and the non-polynomial operators contain infinite-order covariant derivatives, which is not the case when considering polynomial operators. The fact that the action contains infinite derivatives implies that the theory is non-local, meaning that a measure at a certain point can be affected by what it is occurring at other points of the spacetime at the same time. This will be shown in Section \ref{4.1}. 

The authors in \cite{Biswas:2005qr} started to use these kinds of non-polynomial functions of the d'Alembertian to construct an Ultra-Violet extension of GR. Moreover, they showed that the non-locality can potentially ameliorate the singularities present in GR and in local modifications of gravity. Indeed, within this theory, established in its general form in \cite{Biswas:2011ar}, exact non-singular bouncing solutions and black holes that are regular at the linear level have been found \cite{Biswas:2010zk,Biswas:2012bp,Koshelev:2018rau,Frolov:2015bta,Buoninfante:2019swn}. We will give more detail in Section \ref{4.1}.

Inspired by this approach, in Section \ref{4.2} we shall propose a new theory that is an ultraviolet completion of Poincar\'e Gauge Gravity. We will also prove that one can find ghost and singularity free solutions of this theory in Section \ref{4.3}.\\

Independently of how we modify Einstein's theory, there are several important aspects that we need to take into account. First of all, we need to make sure that the theory does not exhibit instabilities, which will render it as unphysical. Also, it needs to be compatible with the current experimental measures \cite{Berti:2015itd}. For instance, the detection of the gravitational wave GW170817 and its electromagnetic counterpart from a binary neutron star system \cite{TheLIGOScientific:2017qsa,GBM:2017lvd}, allowed us to ruled out many modified gravity theories \cite{Ezquiaga:2017ekz,Baker:2017hug,Creminelli:2017sry,Sakstein:2017xjx}. Lastly, from a theoretical perspective, we have to explain where the extra fields that we are introducing, and the breaking of some of the Lovelock's theorem assumptions, come from. Otherwise, we will just be parametrising our ignorance. 

Indeed, in Poincar\'e Gauge Gravity the extra degrees of freedom appear naturally when considering a gravitational gauge theory of the Poincar\'e group. On the other hand, the infinite derivative functions that we use to make a non-local extension of Poincar\'e Gauge theories are inspired in string theory models, and can effectively take into account quantum effects. 

In the following section we shall give an outline of the content of the thesis.

\section{Scope of the thesis}

The objective of this thesis is to review the Poincar\'e Gauge theory of gravity and expose some novel results we have obtained in this field. Moreover, a novel ultraviolet non-local extension of this theory shall be provided, and it will be shown that it can be ghost and singularity free at the linear level. For this purpose, the thesis has been structured as follows.

\paragraph{Chapter 2} We first shall explain the foundations of any gravitational theory and show how there is no physical relation between the metric and affine structure of the space-time. Then, we shall introduce the Poincar\'e Gauge theory of gravity and motivate its use. We shall analyse its stability in a general background, obtaining that only two scalar degrees of freedom can propagate safely. Moreover, we will comment on how it is possible to extend the Lagrangian to overcome the instabilities of the vector sector.

\paragraph{Chapter 3} We will show how fermionic particles follow non-geodesical trajectories in theories with a non-symmetric connection, and work out an explicit example. Also, we shall study the possible avoidance of singularities for fermionic and bosonic particles in Poincar\'e Gauge Gravity. Finally, we shall explore what kind of black-hole solutions we can expect in Poincar\'e Gauge Gravity by exploring the Birkhoff and no-hair theorems in different physically relevant scenarios.

\paragraph{Chapter 4} We shall propose a novel non-local extension of Poincar\'e Gauge Gravity based on the introduction of infinite derivatives in the gravitational action. We shall show how this theory can be made ghost and singularity free at the linear limit.

\paragraph{Chapter 5} We summarise the main results and discuss the possible outlook.

\paragraph{Appendices} In Appendix \ref{ap:1} we give the components for the acceleration of an electron moving in a particular solution of Poincar\'e Gauge Gravity. In Appendix \ref{ap:2} we expand the differents terms that appear in the action of the non-local theory. 

In Appendix \ref{ap:3} we show the explicit form of the functions that compose the linearised action of the same theory. 

In Appendix \ref{ap:4} we calculate the local limit of the infinite derivative theory and find the conditions to recover Poincar\'e Gauge Gravity.

\renewcommand{\publ}{}


\chapter{Poincar\'e Gauge Theories of Gravity}

\label{2}

\PARstart{B}oth its solid mathematical structure and experimental confirmation renders the theory of General Relativity (GR) one of the most successful theories in Physics~\cite{Wald:1984rg,Will:2005va}. As a matter of fact, some phenomena that were predicted by the theory over a hundred years ago, such as gravitational waves~\cite{Abbott:2016blz}, have been measured for the first time in our days. Nevertheless, as we commented in the Introduction, GR suffers from some important shortcomings that need to be addressed. One of them is that the introduction of fermionic matter in the energy-momentum tensor appearing on GR field equations may be cumbersome, since new formalisms would be required~\cite{Hehl:1994ue}. \\
This issue can be solved by introducing a gauge approach facilitating a better understanding of gravitational theories. This was done by Sciama and Kibble in~\cite{sciama1962analogy} and~\cite{Kibble:1961ba} respectively, where the idea of a Poincar\'e gauge (PG) formalism for gravitational theories was first introduced. Following this description one finds that the space-time connection must be metric compatible, albeit not necessarily symmetric. Therefore, a non-vanishing torsion field $T_{\,\,\nu\rho}^{\mu}$ emerges as a consequence of the non-symmetric character of the connection. For an extensive review of the torsion gravitational theories {\it c.f.}~\cite{Gauge,Gauge2}. \\
An interesting fact about these theories is that they appear naturally as gauge theories of the Poincar\'e Group, rendering their formalism analogous to the one used in the Standard Model of Particles, and hence making them good candidates to explore the quantisation of gravity.\\

This chapter is divided as follows. In Section \ref{2.1} we introduce the basic theoretical structure of any gravitational theory, making emphasis on the affine connection. Then in Section \ref{2.2} we will obtain the theory that appears when gauging with respect to the Poincar\'e group. This theory is usually known as Poincar\'e Gauge Gravity. Finally in Section \ref{2.3} we shall study the stability of this theory using a background independent approach, based on the publication {\bf{\nameref{P6}}}.

\clearpage

\section{Non-Riemannian spacetimes}
\label{2.1}
As we already mentioned in the Introduction, the equivalence principle is one of the theoretical keys to understand the way modern gravitational theories are formulated. The fact that the gravitational effects can be removed for a certain observer just by a change of coordinates reminds of a well-known property of differential geometry: for every point of a geodesic curve one can construct a set of coordinates, the so-called {\it{normal coordinates}}, for which the components of the connection in that basis and at that point are zero, hence removing locally the effect of curvature. Therefore, the surroundings of every point ``look like" $\mathbb{R}^{\rm{n}}$ but globally the system possesses quite different properties.\\
That is why Gravity can be described resorting to a differential geometry approach by considering a manifold, to be referred to as the spacetime, where the free-falling observers are assumed to follow geodesics and the gravitational effects are encoded in the global properties of the manifold. Let us establish these ideas more specifically by reminding a few concepts of differential geometry \cite{Wald:1984rg,Hawking:1973uf,sanchez2003introduccion}.

\begin{defn}[Manifold]
\label{def:manifold}
A $C^r$ {\it n-dimensional manifold} $\mathcal{M}$ is a set $\mathcal{M}$ together with a $C^r$ {\it atlas} $\left\{ \mathcal{U}_{\alpha},\psi_{\alpha}\right\} $, {\it i.e.} a collection of charts $\left(\mathcal{U}_{\alpha},\psi_{\alpha}\right)$ where the $\mathcal{U}_{\alpha}$ are subsets of $\mathcal{M}$ and the $\psi_{\alpha}$ are one-to-one maps of the corresponding $\mathcal{U}_{\alpha}$ to open sets in $\mathbb{R}^{\rm{n}}$ such that:\\
(1) The $\mathcal{U}_{\alpha}$ covers $\mathcal{M}$, which means that $\mathcal{M}=\underset{\alpha}{\cup}\;\mathcal{U}_{\alpha}$\\
(2) If $\mathcal{U}_{\alpha}\cap\mathcal{U}_{\beta}$ is non-empty, then the map
\[
\psi_{\alpha}\circ\psi_{\beta}^{-1}:\;\psi_{\beta}\left(\mathcal{U}_{\alpha}\cap\mathcal{U}_{\beta}\right)\longrightarrow\psi_{\alpha}\left(\mathcal{U}_{\alpha}\cap\mathcal{U}_{\beta}\right)
\]
\phantom{(2)} is a $C^r$ map of an open subset of $\mathbb{R}^{\rm{n}}$ to an open subset of $\mathbb{R}^{\rm{n}}$ (see Figure \ref{fig:1}).\\

\end{defn}

\begin{figure}
\centering
\includegraphics[width=0.88\linewidth]{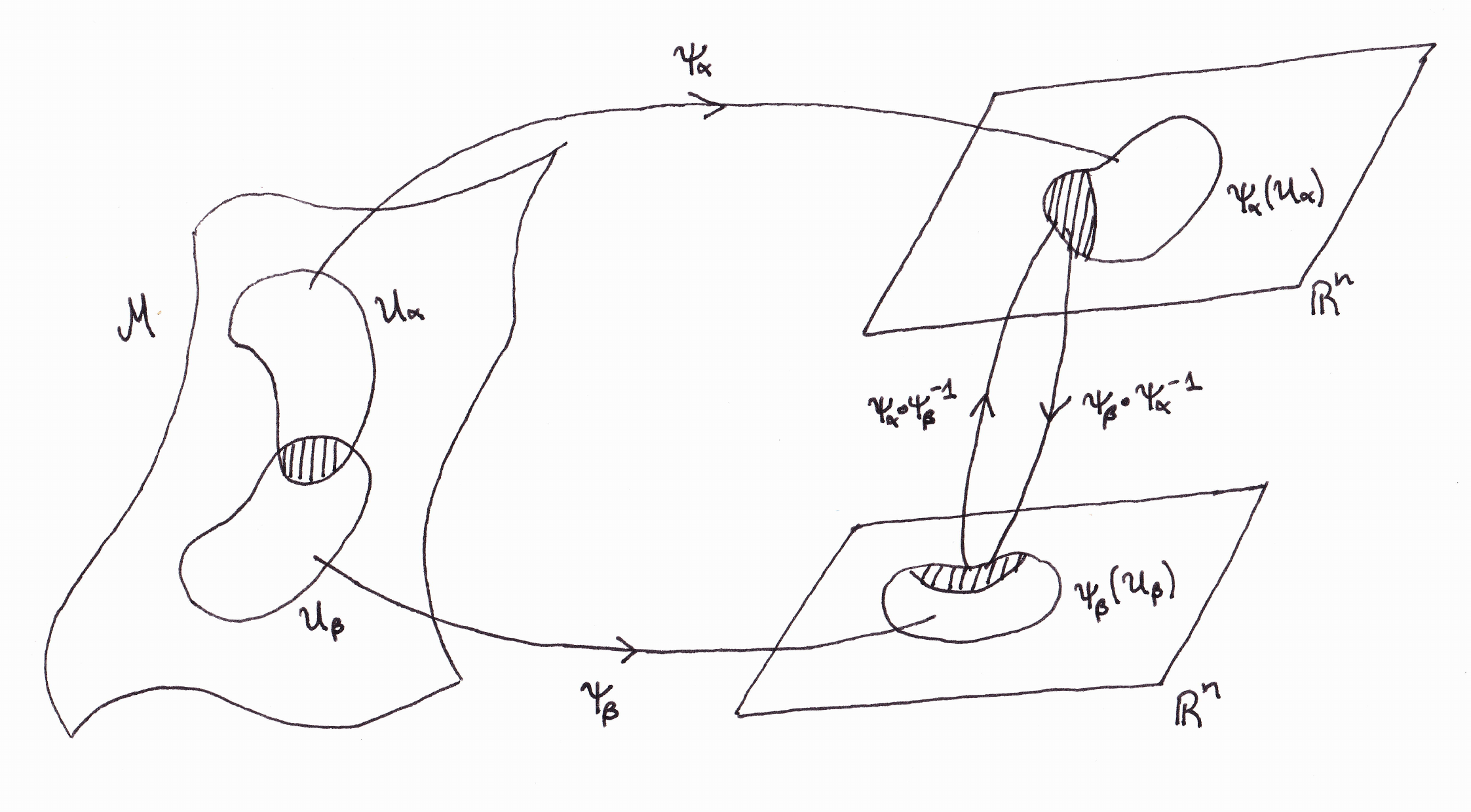}
\caption{Relation between two intersecting charts in a manifold, according to Definition \ref{def:manifold}.}
\label{fig:1}
\end{figure}

If all the possible charts compatible with the condition (2) are included, the atlas $\left\{ \mathcal{U}_{\alpha},\psi_{\alpha}\right\} $ is known as {\emph{maximal}}. From now on we will assume that that is the case. Moreover, an atlas $\left\{ \mathcal{U}_{\alpha},\psi_{\alpha}\right\} $ is known as {\emph{locally finite}} if every point $p\in\mathcal{M}$ has an open neighbourhood which only intersects a finite number of the sets $\mathcal{U}_{\alpha}$.

Nevertheless, a manifold is still a very general structure, and we need to impose more conditions in order to represent a physical system: 
\begin{enumerate}

\item We shall requiere that the manifold satisfies the \emph{Hausdorff} separation axiom: if $p,q$ are two distinct points in $\mathcal{M}$, then there exists disjoint open sets $\mathcal{U}$, $\mathcal{V}$ in $\mathcal{M}$ such that $p\in\mathcal{U}$ and $q\in\mathcal{V}$.

\item It must be \emph{paracompact}, meaning that for every atlas $\left\{ \mathcal{U}_{\alpha},\psi_{\alpha}\right\} $ there exists a locally finite atlas $\left\{ \mathcal{V}_{\beta},\varphi_{\beta}\right\} $ with each $\mathcal{V}_{\beta}$ contained in some $\mathcal{U}_{\alpha}$.

\item We will impose that it is \emph{connected}, \emph{i.e.} the manifold cannot be divided into two disjoints open sets.

\end{enumerate}

These three conditions are required because, if the manifold meets them, they imply that the manifold has a countable basis, which means that there is a countable collection of open sets such that any open set can be expressed as the union of members of this collection \cite{kobayashi1963foundations}.\\

Moreover, in every physical system we need to define $i)$ a vector structure and $ii)$ a way of measuring distances, so in the following we shall explain how these concepts are introduced in differential geometry.\\
We will define a tangent vector to a point as an equivalence class of curves that pass through that point. To concretise this definition let us consider a manifold $\mathcal{M}$ of dimension $n$ and fix a point $p\in \mathcal{M}$. Let
\[
\mathcal{C}_{p}=\left\{ \gamma:\left]-\epsilon_{\gamma},\epsilon_{\gamma}\right[\longrightarrow\mathcal{M}\,;\;\epsilon_{\gamma}>0,\gamma\left(0\right)=p,\gamma\,{\rm differentiable}\right\} 
\]
be the set of differentiable curves contained in $\mathcal{M}$ that pass through $p$. We shall establish a class of equivalence in $\mathcal{C}_{p}$ as follows: two curves $\gamma,\rho\in\mathcal{C}_{p}$ will be {\emph{equivalent}} $\gamma\sim\rho$ if, for some coordinate neighbourhood $\left(\mathcal{U},\psi=\left(q_{1},...,q_{n}\right)\right)$ of $p$, they verify that $\left.\frac{d}{dt}\right|_{t=0}\left(\psi\circ\gamma\right)\left(t\right)=\left.\frac{d}{dt}\right|_{t=0}\left(\psi\circ\rho\right)\left(t\right)$ (let us note that $\psi\left(\gamma\left(0\right)\right)=\psi\left(\rho\left(0\right)\right)=\psi\left(p\right)$). This means that they are equivalent if the tangent vector in $\mathbb{R}^{n}$ of those curves coincides\footnote{It is important to stress that this definition is completely independent of the coordinate neighbourhood that is chosen.} (see Figure \ref{fig:2}). \\

\begin{figure}
\centering
\includegraphics[width=0.88\linewidth]{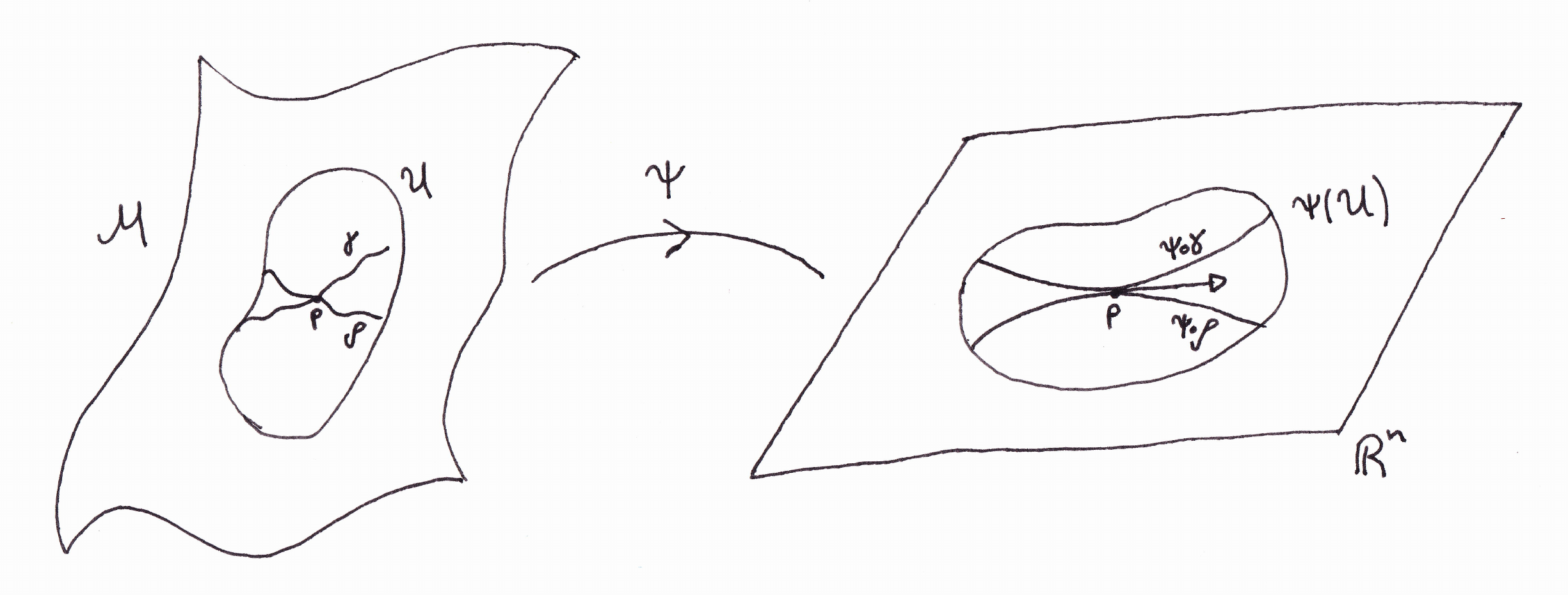}
\caption{Graphic example showing two curves that have the same tangent vector in $\mathbb{R}^{n}$.}
\label{fig:2}
\end{figure}

We shall define a tangent vector as
\begin{defn}[Tangent vector, by equivalent classes]
We will call tangent vector to $\mathcal{M}$ in $p$ to each of the equivalent classes defined by $\sim$ in $\mathcal{C}_{p}$.
\end{defn}
There is also a more abstract, although equivalent, definition of a tangent vector, that may be more appealing to physicists. Namely
\begin{defn}[Tangent vector, by coordinates]
A tangent vector to $\mathcal{M}$ in $p$ is a map that to every coordinate neighbourhood $\left(\mathcal{U},\psi=\left(q_{1},...,q_{n}\right)\right)$ of $p$ it assigns an element $\left(a^{1},...,a^{n}\right)\in\mathbb{R}^{n}$, in such a way that given another coordinate neighbourhood $\left(\tilde{\mathcal{U}},\tilde{\psi}=\left(\tilde{q}_{1},...,\tilde{q}_{n}\right)\right)$ the new assigned element $\left(\tilde{a}^{1},...,\tilde{a}^{n}\right)\in\mathbb{R}^{n}$ verifies
\begin{equation}
\label{vectortransformation}
\tilde{a}^{i}=\sum_{j=1}^{n}\frac{\partial\tilde{q}^{i}}{\partial q^{j}}\left(p\right)a^{j}\quad\forall i\in\left\{ 1,...,n\right\} ,
\end{equation}
that is known as the vector transformation law.
\end{defn}
This definition is the formalisation of the physical idea that a vector is a $n$ component object such that it assigns to every coordinate system an element in $\mathbb{R}^{n}$ that ``transforms as a vector''.\\
Consequently, we will define the vector structure as 
\begin{defn}[Tangent space] We will define the tangent space to $\mathcal{M}$ in $p$, $V_{p}$ , as the set of all tangent vectors to $\mathcal{M}$ in $p$.
\label{tangentspace}
\end{defn}
The dimension of the tangent space $V_{p}$ is the same one as the manifold $\mathcal{M}$. Having this vector structure allows us to introduce the concept of tensors, that will represent the physical quantities in a gravitational theory. \\
First, let us recall that for every real vector space $V\left(\mathbb{R}\right)$ one can define the {\emph{dual vector space}} $V^{\ast}\left(\mathbb{R}\right)$ as
\[
V^{\ast}\left(\mathbb{R}\right):=\left\{ v^{\ast}:V\longrightarrow\mathbb{R}\;/\;v^{\ast}\:{\rm linear}\right\} ,
\]
where the elements of the dual $v^{\ast}\in V^{\ast}\left(\mathbb{R}\right)$ are known as {\emph{linear forms}} or {\emph{dual vectors}}.\\ 
In particular, for every tangent space $V_{p}$ there is a dual tangent space $V^{\ast}_{p}$, which will be of the same dimension.

We now have the necessary concepts to introduce tensors. 
\begin{defn}[Tensor]
Let $V\left(\mathbb{R}\right)$ be a real vector space of finite dimension. A tensor $T$ of type $\left(k,l\right)$ over $V\left(\mathbb{R}\right)$ is a map
\[
T:\left(V^{\ast}\right)^{k}\times V^{l}\longrightarrow\mathbb{R},
\]
that is multilinear, {\emph{i.e.}} linear in each of its $k+l$ variables.
\end{defn}
One important operation that one can perform with tensors is the so-called {\emph{outer product}}: given a tensor $T$ of type $\left(k,l\right)$ and another tensor $T'$ of type $\left(k',l'\right)$, one can construct a new tensor of type $\left(k+k',l+l'\right)$, the outer product $T\otimes T'$, which is described by the following rule. Let us have $k+k'$ dual vectors $\left\{ v^{1^{\ast}},...,v^{\left(k+k'\right)^{\ast}}\right\} $ and $l+l'$ vectors $\left\{ w_{1},...,w_{l+l'}\right\} $. Then we shall define $T\otimes T'$ acting on these vectors to be the product of $T\left(v^{1^{\ast}},...,v^{k^{\ast}},w_{1},...,w_{l}\right)$ and $T'\left(v^{\left(k+1\right)^{\ast}},...,v^{\left(k+k'\right)^{\ast}},w_{l+1},...,w_{l+l'}\right)$.\\
One can show that every tensor of type $\left(k,l\right)$ can be expressed as a sum of the outer product of simple tensors, namely
\begin{equation}
\label{tensorcoordinates}
T=\sum_{\mu_{1},...,\nu_{l}=1}^{n}T^{\mu_{1}...\mu_{k}}\,_{\nu_{1}...\nu_{l}}v_{\mu_{1}}\otimes...\otimes v^{\nu_{l}}.
\end{equation}
The basis expansion coefficients, $T^{\mu_{1}...\mu_{k}}\,_{\nu_{1}...\nu_{l}}$, are known as the components of the tensor $T$ with respect to the basis $\left\{v_{\mu}\right\}$ of the vector space. From this Section onwards we will work directly with the components of the tensors under certain basis. \\

At this stage, we still need an element endowed to the manifold that allow us to define distances. In $\mathbb{R}^{n}$ distances are measured by the {\emph{scalar product}}, which is a scalar linear map acting on two vectors and meeting certain properties. Does this ring a bell? Indeed, this concept can be generalised to a tensor of type $\left(0,2\right)$, known as the {\emph{metric tensor}}, as follows
\begin{defn}[Metric]
A metric tensor $g$ at a point $p\in\mathcal{M}$ is a $\left(0,2\right)$ tensor that meets the following properties:
\begin{itemize}
\item It is symmetric, meaning that for all $v_{1},v_{2}\in V_{p}$ we have $g\left(v_{1},v_{2}\right)=g\left(v_{2},v_{1}\right)$.
\item It is non-degenerate, which implies that the only case in which we have $g\left(v,v_{1}\right)=0$ for all $v\in V_{p}$ is when $v_{1}=0$.
\end{itemize}
\end{defn}
In a coordinate system it can be expanded as
\begin{equation}
{\rm d}s^{2}=g_{\mu\nu}{\rm d}x^{\mu}{\rm d}x^{\nu},
\end{equation} 
where the Einstein summation convention applies (as from now on), and the outer product sign has been ommited.\\
From the Sylvester theorem \cite{sylvester1852xix} we have that for any metric $g$ one can always find an {\emph{orthonormal basis}} $\left\{ v^{1},...,v^{n}\right\} $ of $V_{p}$, such that $g\left(v_{\mu},v_{\nu}\right)=0$ if $\mu\neq\nu$ and $g\left(v_{\mu},v_{\mu}\right)=\pm 1$. The number of $+$ and $-$ signs occurring is independent of the chosen orthonormal basis (which is not unique), and it is known as the {\emph{signature}} of the metric. If the signature of the metric is $+...+$ it is called {\emph{Riemannian}}, while if the signature is $-+...+$ is known as {\emph{Lorentzian}}. \\

At this time, we are ready to define what we understand by a physical spacetime
\begin{defn}[Spacetime]
\label{spacetime}
A spacetime manifold is a pair $\left(\mathcal{M},g\right)$, in which $\mathcal{M}$ is a connected Hausdorff $C^{\infty}$ n-dimensional manifold, and $g$ a Lorentzian metric on $\mathcal{M}$. 
\end{defn}
As we pointed out earlier, gravitational effects are a consequence of the global properties of the spacetime. More specifically, we shall identify the observers that are only affected by gravity with those following geodesics of the spacetime, {\emph{i.e.}} the trajectories that maximise the length $L$ of a curve $\gamma$ between two points $p=\gamma\left(a\right),q=\gamma\left(b\right)$, where the length is calculated by integrating the tangent vector $\gamma'\left(t\right)$ along the curve $\gamma$
\begin{equation}
\label{length}
L=\int_{a}^{b}\left(\left|g\left(\gamma'\left(t\right),\gamma'\left(t\right)\right)\right|\right)^{\frac{1}{2}}{\rm d}t.
\end{equation}
Therefore, in light of \eqref{length}, one can see that the global aspects of the spacetime are encoded in the metric tensor.\\
The reader might wonder about the fact that we have not still talked about one of the most essential aspects of any physical theory. Indeed, since we want to describe the dynamics of Nature we need to know how to perform variations in the manifold, {\emph{i.e.}} how does one define derivatives.\\
For that, we need to provide an {\emph{affine structure}} to the manifold, which will allow us to differenciate {\emph{vector fields}} (and consequently tensor fields as well). These are maps assigning a vector to every point in the manifold, and the set of all the vector fields on $\mathcal{M}$ is denoted as $\mathfrak{X}(Q)$. With this in mind we define
\begin{defn}[Affine connection]
An {\emph{affine connection}} on $\mathcal{M}$ is a map $\nabla$,
\begin{eqnarray}
\nabla:\mathfrak{X}\left(\mathcal{M}\right)\times\mathfrak{X}\left(\mathcal{M}\right)&\longrightarrow&\mathfrak{X}\left(\mathcal{M}\right)
\nonumber
\\
\left(X,Y\right)&\longmapsto&\nabla_{X}Y
\nonumber
\end{eqnarray}
that meets the following conditions
\begin{enumerate}
\item It is $\mathbb{R}$-linear with respect to the second variable, that is,
\[
\nabla_{X}\left(aY+b\overline{Y}\right)=a\nabla_{X}Y+b\nabla_{X}\overline{Y},\quad\forall a,b\in\mathbb{R},\quad\forall X,Y,\overline{Y}\in\mathfrak{X}\left(\mathcal{M}\right).
\]
\item It verifies the Leibniz rule with respect to the second variable
\[
\nabla_{X}\left(fY\right)=X\left(f\right)Y+f\nabla_{X}Y,\quad\forall f\in C^{\infty}\left(\mathcal{M}\right),\quad\forall X,Y\in\mathfrak{X}\left(\mathcal{M}\right).
\]
\item It is $\mathbb{R}$-linear with respect to the first variable,
\[
\nabla_{aX+b\overline{X}}\left(Y\right)=a\nabla_{X}Y+b\nabla_{\overline{X}}Y,\quad\forall a,b\in\mathbb{R}\quad\forall X,\overline{X},Y\in\mathfrak{X}\left(\mathcal{M}\right).
\]
\item It is $C^{\infty}\left(\mathcal{M}\right)$-linear with respect to the first variable,
\[
\nabla_{fX}\left(Y\right)=f\nabla_{X}Y,\quad\forall f\in C^{\infty}\left(\mathcal{M}\right)\quad\forall X,Y\in\mathfrak{X}\left(\mathcal{M}\right).
\]
\end{enumerate}
The pair $\left(\mathcal{M},\nabla\right)$ is known as {\emph{affine manifold}}.
\end{defn}   

It is important to see how we can express the connection in a certain coordinate basis. Let $\left(\mathcal{M},\nabla\right)$ be an affine manifold and $\left(\mathcal{U},q^{1},...,q^{n}\right)$ a coordinate neighbourhood of $\mathcal{M}$. It is known from standard differential geometry that $\left(\partial_{1}\equiv\frac{\partial}{\partial q^{1}},...,\partial_{n}\equiv\frac{\partial}{\partial q^{n}}\right)$ is a basis of the vector fields over $\mathcal{M}$ \cite{o1983semi}. Since $\nabla_{\partial_{\mu}}\partial_{\nu}$ is also a vector field, we can express it at every point as a linear combination of coordinates fields $\left(\partial_{1},...,\partial_{n}\right)$. Consequently, there are $n^{3}$ differentiable functions on $\mathcal{U}$ such that
\begin{equation}
\label{Christoffel}
\nabla_{\partial_{\mu}}\partial_{\nu}=\sum_{\rho=1}^{n}\Gamma^{\rho}\,_{\mu\nu}\partial_{\rho}.
\end{equation}
The functions $\Gamma^{\rho}\,_{\mu\nu}$, with $\mu,\nu,\rho\in\left\{ 1,...,n\right\} $, given in \eqref{Christoffel} are denoted as the {\emph{Christoffel symbols}} of $\nabla$ in the coordinates $\left(\partial_{1},...,\partial_{n}\right)$.\\
This coordinate expression of the connection allows us to define a derivative on the tensor fields (which of course includes vector fields) that transforms properly, meaning that the derivative of a tensor shall be another tensor. This is known as the {\emph{covariant derivative}}:
\begin{defn}[Covariant derivative, in coordinates]
The covariant derivative $\nabla_{\rho}$ of a tensor field of type $\left(k,l\right)$ $T^{\mu_{1}...\mu_{k}}\,_{\nu_{1}...\nu_{l}}$ is another tensor of type $\left(k,l+1\right)$, $\nabla_{\rho}T^{\mu_{1}...\mu_{k}}\,_{\nu_{1}...\nu_{l}}$, that can be written in coordinates as
\begin{eqnarray}
\label{ChristoffelCoord}
\nabla_{\rho}T^{\mu_{1}...\mu_{k}}\,_{\nu_{1}...\nu_{l}}&=&\partial_{\rho}T^{\mu_{1}...\mu_{k}}\,_{\nu_{1}...\nu_{l}}+\sum_{i=1}^{k}\Gamma^{\mu_{i}}\,_{\rho d}T^{\mu_{1}...d...\mu_{k}}\,_{\nu_{1}...\nu_{l}}
\nonumber
\\
&&-\sum_{i=1}^{l}\Gamma^{d}\,_{\rho\nu_{i}}T^{\mu_{1}...\mu_{k}}\,_{\nu_{1}...d...\nu_{l}}.
\end{eqnarray}
\end{defn}

At this step we need to point out a very crucial fact. {\bf{The affine structure of a spacetime is not unique}}. Let us explain this in more detail.\\
A gravitational theory is a physical theory that relates the energy and matter content of a system with the global structure of the spacetime describing such a system. From the Definition \ref{spacetime} and the subsequent discussion we know that such a structure is given by the metric tensor. Hence, the field equations of the gravitational theory will be {\bf{dynamical}} equations that have the energy and matter content as input quantities and the metric tensor and the affine connection\footnote{The latter appearing due to the dynamical character of the field equations.} as the unknowns that we want to solve.\\
The metric tensor is the one that defines the structure of spacetime and the connection is going to tell us how to take derivatives, so the latter will of course affect the field equations. In general, these two quantities are independent, but there is a special choice of connection that relates them, known as the {\emph{Levi-Civita}} connection. The existence of such a connection is sometimes referred in mathematics literature as the ``miracle'' of Lorentzian geometry, since it proves that every pair $\left(\mathcal{M},g\right)$ can be understood as an affine manifold.  
\begin{thm}
Let $\left(\mathcal{M},g\right)$ be a n-dimensional Lorentzian manifold. Then there exists an unique connection $\mathring{\nabla}$, with Christoffel symbols $\mathring{\Gamma}^{\rho}\,_{\mu\nu}$, that verifies the following properties in all coordinate systems:
\begin{enumerate}
\item It is symmetric, that is
\begin{equation}
\label{LCsymmetry}
\mathring{\Gamma}^{\rho}\,_{\mu\nu}=\mathring{\Gamma}^{\rho}\,_{\nu\mu}\quad \mu,\nu,\rho\in\left\{ 1,...,n\right\}.
\end{equation}
\item It is metric compatible, which means that
\begin{equation}
\label{LCmetricity}
\mathring{\nabla}_{\mu}g_{\nu\rho}=0\quad \mu,\nu,\rho\in\left\{ 1,...,n\right\}.
\end{equation}
\end{enumerate}
\end{thm}
Moreover, based on the properties \eqref{LCsymmetry} and \eqref{LCmetricity}, one can define two tensors that account for ``the lack of symmetry'' \eqref{LCsymmetry} and ``the lack of metricity'' \eqref{LCmetricity} of the connection.
\begin{defn}[Torsion tensor]
Let $\nabla$ be a connection with Christoffel symbols $\Gamma^{\rho}\,_{\mu\nu}$. Then, the torsion tensor $T^{\rho}\,_{\mu\nu}$ is defined as the antisymmetric part of the connection, namely
\begin{equation}
\label{TORSION}
T^{\rho}\,_{\mu\nu}=\Gamma^{\rho}\,_{\mu\nu}-\Gamma^{\rho}\,_{\nu\mu}.
\end{equation}
\end{defn}
\begin{defn}[Non-metricity tensor]
Let $\nabla$ be a connection with Christoffel symbols $\Gamma^{\rho}\,_{\mu\nu}$. Then, the non-metricity tensor $M_{\mu\nu\rho}$ is defined as
\begin{equation}
\label{NONMETRICITY}
M_{\mu\nu\rho}=\nabla_{\mu}g_{\nu\rho}.
\end{equation}
\end{defn}
Hence one can introduce the following
\begin{defn}[Levi-Civita connection]
The connection $\mathring{\nabla}$ with Christoffel symbols $\mathring{\Gamma}^{\rho}\,_{\mu\nu}$, which has null torsion and non-metricity, is known as the Levi-Civita connection.
\end{defn}
This connection has very interesting properties. First of all, it is uniquely related to the metric tensor as
\begin{equation}
\label{LCwithg}
\mathring{\Gamma}^{\rho}\,_{\mu\nu}=\frac{1}{2}g^{\rho\sigma}\left(\partial_{\mu}g_{\nu\sigma}+\partial_{\nu}g_{\sigma\mu}-\partial_{\sigma}g_{\mu\nu}\right),
\end{equation}
Also, from the properties of the Levi-Civita connection, \eqref{LCsymmetry} and \eqref{LCmetricity}, one can prove (for a detailed proof see \cite{o1983semi}) the following
\begin{thm}
Let $p$ and $q$ be two points in the spacetime $\left(\mathcal{M},g\right)$, and let $\gamma$ be a curve that joints this two points, with a tangent vector $v^{\mu}$ that is parallely transported along itself in terms of the Levi-Civita connection, i.e.
\begin{equation}
\label{geodesiceq}
v^{\mu}\mathring{\nabla}_{\mu}v^{\nu}=0.
\end{equation}
Then, $\gamma$ is also the curve that extremise the length between the two points, which we have defined earlier as a geodesic.
\end{thm}

As we have seen, from a mathematical point of view it makes sense to stick to the Levi-Civita connection, due to its properties. Nevertheless, {\bf{there is not any physical reason to assume that this is the affine structure preferred by Nature}}. To illustrate this, we will give in the following three subsections a very enlightening example. We shall briefly sketch three gravitational theories, with different affine structures, and prove that they are equivalent, in the sense that the field equations are the same. These theories are GR, Teleparallel Gravity (TEGR), and Symmetric Teleparallel Gravity (STEGR), which are sometimes referred to as the {\emph{Geometrical Trinity of Gravity}} \cite{BeltranJimenez:2019tjy}.

\subsection{General Relativity}

The theory of GR, first introduced by Albert Einstein in 1916 \cite{Einstein:1916vd}, is the currently accepted gravitational theory. This theory is formulated in terms of the usual {\emph{curvature tensors}} of the spacetime. They are defined in terms of the connection of the spacetime as follows\footnote{Throughout the thesis we shall use Wald's convention \cite{Wald:1984rg}.}:
\begin{itemize}
\item {\emph{Riemann tensor}}
\begin{equation}
\label{Riemanntensor}
R_{\mu\nu\rho}\,^{\sigma}=\partial_{\nu}\Gamma^{\sigma}\,_{\mu\rho}-\partial_{\mu}\Gamma^{\sigma}\,_{\nu\rho}+\Gamma^{\alpha}\,_{\mu\rho}\Gamma^{\sigma}\,_{\alpha\nu}-\Gamma^{\alpha}\,_{\nu\rho}\Gamma^{\sigma}\,_{\alpha\mu}.
\end{equation}
\item {\emph{Ricci tensor}}
\begin{equation}
\label{Riccitensor}
R_{\mu\nu}=R_{\mu\rho\nu}\,^{\rho}.
\end{equation}
\item {\emph{Scalar curvature}} or {\emph{Ricci scalar}}
\begin{equation}
\label{Ricciscalar}
R=g^{\mu\nu}R_{\mu\nu}.
\end{equation}
\end{itemize}
In the case of GR {\bf{the affine structure is the Levi-Civita one}}. Then the curvature tensors shall be denoted as $\mathring{R}_{\mu\nu\rho}\,^{\sigma}$, $\mathring{R}_{\mu\nu}$, and $\mathring{R}$.\\

As every physical theory, GR can be constructed from an action, which is a functional that gives the field equations when extremising with respect to the independent variables. More specifically, the GR action can be written as
\begin{equation}
\label{GRaction}
S_{{\rm {GR}}}=\int {\rm d}^{4}x\sqrt{-g}\left(\frac{1}{16\pi G}\mathring{R}+\mathcal{L}_{M}\right),
\end{equation}
where $g$ is the determinant of the metric tensor, $G$ is the gravitational constant, and $\mathcal{L}_{M}$ accounts for the energy and matter content of the system.\\
Then, we shall obtain the field equations in the following by performing variations with respect to the metric tensor $g_{\mu\nu}$ and finding the extremising condition, that is $\frac{\delta S_{{\rm {GR}}}}{\delta g^{\mu\nu}}=0$. Let us study each of the terms separately. Firstly, for the curvature part we have
\begin{eqnarray}
&&\delta\left(\sqrt{-g}\mathring{R}\right)=\delta\left(\sqrt{-g}g^{\mu\nu}\mathring{R}_{\mu\nu}\right)=\sqrt{-g}\left(\delta\mathring{R}_{\mu\nu}\right)g^{\mu\nu}+\sqrt{-g}\mathring{R}_{\mu\nu}\delta g^{\mu\nu}
\nonumber
\\
&&+\mathring{R}\delta\left(\sqrt{-g}\right).
\end{eqnarray}
It is a known result that\footnote{See Chapter 7 of \cite{Wald:1984rg} for details.}
\begin{equation}
g^{\mu\nu}\delta\mathring{R}_{\mu\nu}=\mathring{\nabla}^{\mu}\left[\mathring{\nabla}^{\nu}\left(\delta g_{\mu\nu}\right)-g^{\rho\sigma}\mathring{\nabla}_{\mu}\left(\delta g_{\rho\sigma}\right)\right].
\end{equation}
Moreover, it is easy to check that
\begin{equation}
\delta(\sqrt{-g})=\frac{1}{2}\sqrt{-g}g^{\mu\nu}\delta g_{\mu\nu}=-\frac{1}{2}\sqrt{-g}g_{\mu\nu}\delta g^{\mu\nu}.
\end{equation}
Therefore we have
\begin{eqnarray}
\label{variationsGR}
\delta S_{{\rm {GR}}}&=&\frac{1}{16\pi G}\left\{ \int {\rm d}^{4}x\sqrt{-g}\mathring{\nabla}^{\mu}\left[\mathring{\nabla}^{\nu}\left(\delta g_{\mu\nu}\right)-g^{\rho\sigma}\mathring{\nabla}_{\mu}\left(\delta g_{\rho\sigma}\right)\right]\right.
\nonumber
\\
&&\left.+\int {\rm d}^{4}x\sqrt{-g}\left(\mathring{R}_{\mu\nu}-\frac{1}{2}g_{\mu\nu}\mathring{R}\right)\delta g^{\mu\nu}\right\}
\nonumber
\\
&&+\int {\rm d}^{4}x\sqrt{-g}\left(\frac{\delta\mathcal{L}_{M}}{\delta g^{\mu\nu}}-\frac{1}{2}g_{\mu\nu}\mathcal{L}_{M}\right)\delta g^{\mu\nu}.
\end{eqnarray}
The first term of \eqref{variationsGR} is clearly the integral of a divergence. By the Stokes theorem we know that this integral just contributes a boundary term.  When the variations of the metric $\delta g_{\mu\nu}$ and its derivatives vanish in the boundary, as we shall require, that integral also vanish. Hence, the variation with respect to the metric tensor is
\begin{equation}
\frac{\delta S_{{\rm {GR}}}}{\delta g^{\mu\nu}}=\frac{1}{16\pi G}\int {\rm d}^{4}x\sqrt{-g}\left[\mathring{R}_{\mu\nu}-\frac{1}{2}g_{\mu\nu}\mathring{R}+16\pi G\left(\frac{\delta\mathcal{L}_{M}}{\delta g^{\mu\nu}}-\frac{1}{2}g_{\mu\nu}\mathcal{L}_{M}\right)\right].
\end{equation}
Then, the extremising condition $\frac{\delta S_{{\rm {GR}}}}{\delta g^{\mu\nu}}=0$ gives the so-called Einstein field equations 
\begin{equation}
\mathring{R}_{\mu\nu}-\frac{1}{2}g_{\mu\nu}\mathring{R}=8\pi G T_{\mu\nu},
\end{equation}
where 
\begin{equation}
\label{EnergyMomentum}
T_{\mu\nu}:=\frac{-2}{\sqrt{-g}}\frac{\delta\left(\sqrt{-g}\mathcal{L}_{M}\right)}{\delta g^{\mu\nu}}=-2\frac{\delta\mathcal{L}_{M}}{\delta g^{\mu\nu}}+g_{\mu\nu}\mathcal{L}_{M}
\end{equation}
is known as the {\emph{energy-momentum tensor}}.\\

The reader might be wondering why the action \eqref{GRaction} is chosen, and not another curvature invariant, such as for instance $\mathring{R}_{\mu\nu}\mathring{R}^{\mu\nu}$. It is because this choice leads to the simplest theory {\bf{that is endowed with the Levi-Civita connection}}, which is able to explain the basic features of classical gravity. The addition of other curvature invariants in the action may result in higher-order field equations, which, as we shall explain in the Section \ref{2.3}, usually lead to unstable solutions.\\ 
In the following subsections we will show that we can find theories that have a different affine structure, and have the same field equations of GR.

\subsection{Teleparallel Gravity}

As we shall see below, Teleparallel Gravity (TEGR) is an equivalent theory to GR, first proposed just one year after Einstein's article by the usually unrecognised G. Hessenberg \cite{hessenberg1917vektorielle} (for more details on this theory see \cite{DeAndrade:2000sf,Aldrovandi:2013wha}). In the following years the concept of teleparallelism was studied and structured by Cartan, Weitzenb\"ock, and Einstein \cite{cartan1922generalisation,Weitzenbock:1923efa,cartan1979letters}. \\
{\bf{The affine structure of TEGR is the Weitzenb\"ock connection}} $\hat{\Gamma}$. This is the unique connection that has null curvature and null non-metricity, while having non-zero torsion. As every non-symmetric and metric compatible connection, it can be related with the Levi-Civita connection as
\begin{equation}
\label{LCWeitzenbock}
\hat{\Gamma}^{\alpha}\,_{\mu\nu}=\mathring{\Gamma}^{\alpha}\,_{\mu\nu}+K^{\alpha}\,_{\mu\nu},
\end{equation}
where
\begin{equation}
\label{contorsion}
K^\alpha{}_{\mu\nu}=\frac12 \Big(T^\alpha{}_{\mu\nu}+T_{\mu}{}^\alpha{}_\nu+T_\nu{}^\alpha{}_\mu\Big)
\end{equation}
is the {\emph{contortion tensor}}.\\

The gravitational part of the action of this theory is 
\begin{equation}
\label{TEGRaction}
S_{{\rm {TEGR}}}=\frac{1}{16\pi G}\int {\rm d}^{4}x\sqrt{-g}\,\mathbb{T},
\end{equation}
where we have omitted the matter part and $\mathbb{T}$ is denoted as the {\emph{torsion scalar}}, which is defined as 
\begin{equation}
\mathbb{T}\equiv\frac{1}{4}T_{\mu\nu\rho}T^{\mu\nu\rho}+\frac{1}{2}T_{\mu\nu\rho}T^{\nu\mu\rho}-T^{\mu}\,_{\mu\rho}T_{\nu}\,^{\nu\rho}.
\end{equation}
The equivalence between TEGR and GR can be proved by making use of the definition of the torsion tensor in \eqref{TORSION}, the relation \eqref{LCWeitzenbock} and having in mind that $\hat{R}_{\mu\nu\rho}\,^{\sigma}=0$. Taking these three relations into account in the action \eqref{TEGRaction} gives us the following result
\begin{equation}
S_{{\rm {TEGR}}}=S_{{\rm {GR}}}+2\int {\rm d}^{4}x\sqrt{-g}\,\mathring{\nabla}_{\rho}T_{\nu}\,^{\nu\rho}.
\end{equation}
Hence, since the actions differ by the integral of a total derivative only, the field equations of TEGR and GR will be the same \cite{Aldrovandi:2013wha}.

\subsection{Symmetric Teleparallel Gravity}

Symmetric Teleparallel Gravity (STEGR), also dubbed Coincident General Relativity, was introduced by Nester and Yo in 1999 \cite{Nester:1998mp}. It was recently revisited and given more insight by Beltr\'an, Heisenberg and Koivisto in \cite{BeltranJimenez:2017tkd}.\\
{\bf{The affine structure of STEGR is the one that has null curvature and torsion}}, $\tilde{\Gamma}$. As every symmetric and non-metric connection, it can be related with the Levi-Civita connection as
\begin{equation}
\label{LCnonmetric}
\tilde{\Gamma}^{\alpha}\,_{\mu\nu}=\mathring{\Gamma}^{\alpha}\,_{\mu\nu}+L^{\alpha}\,_{\mu\nu},
\end{equation}
where
\begin{equation}
\label{disformation}
L^\alpha{}_{\mu\nu}=\frac12 \Big(M^\alpha{}_{\mu\nu}-M_{\mu}{}^\alpha{}_\nu-M_\nu{}^\alpha{}_\mu\Big)
\end{equation}
is the {\emph{disformation tensor}}.\\

The gravitational part of the action of STEGR is given by
\begin{equation}
S_{{\rm {STEGR}}}=\frac{1}{16\pi G}\int {\rm d}^{4}x\sqrt{-g}\,\mathbb{Q},
\end{equation}
where $\mathbb{Q}$ is known as the {\emph{non-metric scalar}}, and it is written as
\begin{equation}
\mathbb{Q}\equiv-\frac{1}{4}M_{\mu\nu\rho}M^{\mu\nu\rho}+\frac{1}{2}M_{\mu\nu\rho}M^{\nu\mu\rho}+\frac{1}{4}M_{\mu\nu}\,^{\nu}M^{\mu\rho}\,_{\rho}-\frac{1}{2}M_{\mu\nu}\,^{\nu}M_{\rho}\,^{\rho\mu}.
\end{equation}
Using expressions \eqref{LCnonmetric} and \eqref{disformation}, while taking into account that $\tilde{R}_{\mu\nu\rho}\,^{\sigma}=0$, we can find the following relation
\begin{equation}
\mathbb{Q}=\mathring{R}+\mathring{\nabla}_{\mu}\left(M^{\mu\rho}\,_{\rho}-M_{\rho}\,^{\rho\mu}\right).
\end{equation} 
Therefore STEGR and GR actions will differ by the integral of a total derivative only, so these theories are equivalent. Moreover, the authors in \cite{BeltranJimenez:2017tkd} found that under a certain coordinate basis, known as the {\emph{coincident gauge}}, the connection can be trivialised, {\emph{i.e.}} $\left.L^{\alpha}\,_{\mu\nu}\right|_{{\rm c.g.}}=-\mathring{\Gamma}^{\alpha}\,_{\mu\nu}\rightarrow\tilde{\Gamma}^{\alpha}\,_{\mu\nu}=0$. With this choice one can show that STEGR action would be equivalent to GR without the boundary term (see the first term of Eq.\eqref{variationsGR}). Hence, in this scenario the variational principle can be performed without assuming any conditions on the boundary. \\

Finally, despite the general belief, the two teleparallel theories that we have introduced, TEGR and STEGR, are not the unique theories with a connection different from Levi-Civita that are equivalent to GR. In fact, as the authors in \cite{Jimenez:2019ghw} found the most general theory that one can construct with quadratic terms in both torsion and non-metricity that is equivalent to GR, from which TEGR and STEGR are special cases.\\
The fact that we can find equivalent theories with different affine structure supports the statement that there is not any physical reason to assume that Levi-Civita is the preferred affine structure of the spacetime.

\section{Poincar\'e Gauge Gravity}
\label{2.2}
One of the greatest successes of twentieth century physics has been the ability to describe the laws of Nature in terms of symmetries \cite{Mills:1989wj}. The first person that worked in this aspect was the renowned mathematician Emmy Noether \cite{dick1981emmy}. In 1918 she published an article containing what we now know as the {\emph{Noether's theorem}} \cite{Noether:1918zz}. The theorem states that {\emph{for every symmetry of nature there is a corresponding conservation law, and for every conservation law there is a symmetry}}\footnote{Its original form is more technical and complicated, therefore we have simplified here its formulation while preserving enough generality for our purposes.}. This statement has profound consequences in our understanding of the Universe. For example, everyone can recall from high school the famous sentence about the energy of a system: ``energy is not created or lost, it is only transformed from one form to another''. Thanks to the Noether's theorem this is not a mantra anymore, there is a reason why the energy is a conserved quantity: it is because the laws of nature are invariant under time translations, {\emph{i.e.}} the laws governing the universe are the same now as at any other time. 

The next step was done by Hermann Weyl also in 1918, while trying to obtain Electromagnetism as the manifestation of a local symmetry \cite{Weyl:1918pdp}. More specifically, he wanted to relate the conservation of electric charge with a local invariance with respect to the change of scale, or as he called it, {\emph{gauge invariance}}\footnote{Although this expression was initially referred to a scale invariance, now it is used for any requirement of a local symmetry.}.\\ 
After Einstein found some flaws in Weyl's paper, the idea was abandoned until 1927, when Fritz London realised that the symmetry associated with electric charge conservation was a {\emph{phase}} invariance, {\emph{i.e.}} the invariance under a local arbitrary change in the complex phase of the wavefunction.\\
Thirty years later, in 1954, Yang and Mills applied this local symmetry principle to the invariance under isotopic spin rotation \cite{Yang:1954ek}, opening the door for describing the fundamental interactions by their internal symmetries. 

All of the above can be summarised in the so-called {\emph{gauge principle}}, which is represented in Fig. \ref{gaugeprinciple}. First, as we introduced, there is the Noether's theorem, which states that for every conservation law there is an associated symmetry and vice versa; second, there is the fact that requiring a local symmetry leads to an underlying so-called {\emph{gauge field theory}}; and finally, we find that the gauge field theory determined in this way necessarily includes interactions between the gauge field and the conserved quantity with which we started. 

Thus we have that for every conservation law there is a complete theory of a gauge field for which the given conserved quantity is the source. The only restriction is that the conservation law be associated with a continuous symmetry (this would exclude, for example, parity, which is associated with a discrete reflection symmetry). The resulting theory has just one free parameter, the interaction strength. Nevertheless, one can increase the number of free parameters by considering more than one symmetry and/or more than one field. For instance, one could impose gauge invariance under special unitary group of degree $n$, $\mathrm{SU}\left(n\right)$, in a field $A$ and under the unitary group of degree $m$, $\mathrm{U}\left(m\right)$, in a field $B$, hence obtaining a gauge theory locally invariant under $\mathrm{SU}\left(n\right)_{A}\times\mathrm{U}\left(m\right)_{B}$. This theory will have two field strengths, one for $A$ and $B$ respectively. Actually, in what we currently believe is a reliable picture of fundamental subatomic particles and its interactions, there are two separate gauge theories: the Glashow-Weinberg-Salam theory for electromagnetic and weak interactions \cite{Weinberg:1979pi}, the colour gauge theory for strong interactions \cite{Gross:1973id}. These two theories, together with the spectrum of elementary particles associated with them, make up what is now referred to as the {\emph{Standard Model of Particles}} \cite{Tanabashi:2018oca}, which is a gauge field theory invariant under $\mathrm{SU}(3)\times\mathrm{SU}(2)\times\mathrm{U}(1)$.\\

\begin{figure}
\begin{center}
\begin{tikzpicture}[->,>=stealth',shorten >=1pt,auto,node distance=5cm,
                    semithick]

  \node[state] (A)                  [text width=1.9cm,style={align=center}]  {Conserved quantity};
  \node[state]         (B) [above right of=A,text width=1.9cm,style={align=center}] {Symmetry};
  \node[state]         (C) [below right of=B,text width=1.9cm,style={align=center}] {Gauge field};

  \path (A) edge              node [sloped,anchor=center, above, text width=2.0cm,style={align=center}]{Noether's theorem} (B)
            edge              node {$\,$} (C)
        (B) edge              node [sloped,anchor=center, above, text width=2.0cm,style={align=center}]{Local symmetry} (C)
	      edge		node {$\,$} (A)
	(C) edge              node {Interaction} (A)
	      edge		node {$\,$} (B);
\end{tikzpicture}
\end{center}
\caption{Sketch of the reasoning behind the gauge principle based on \cite{Mills:1989wj}.}
\label{gaugeprinciple}
\end{figure}
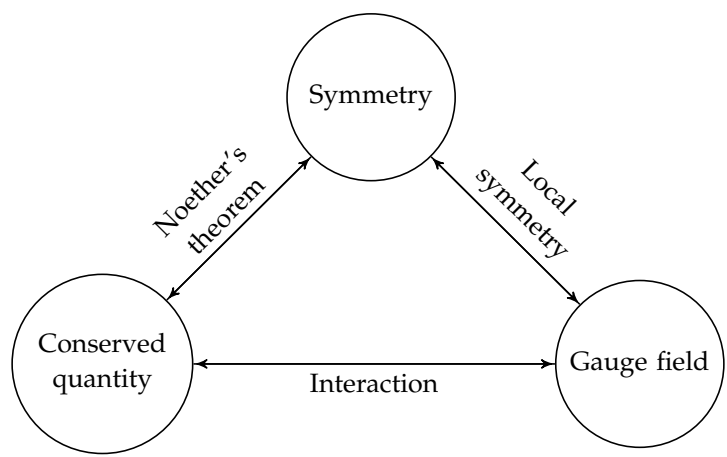

\subsection{Gauge theory of translations}

Since the gauge principle has been so successful in the description of the subatomic interactions, one can also wonder if it could be useful for describing gravity. Indeed, GR can be formulated from a gauge field perspective. If one thinks about it, it is very intuitive, because one of the main principles in any physical theory is that ``the equations of physics are invariant when we make coordinate displacements'' \cite{Feynman:1996kb}. This clearly suggests that the group of spacetime translations in 4 dimensions $T\left(4\right)$ is an ideal candidate for applying the gauge principle. Let us elaborate on this idea.

The generators of the gauge transformations need to be defined in a vector space at every given point. In the previous section, we already introduced this vectorial structure as the tangent space (see Definition \ref{tangentspace}), that we shall explore in more detail in the following.\\ 
Let $p$ be a point of the 4-dimensional spacetime $\left(\mathcal{M},g\right)$, and let $V_{p}$ be the 4-dimensional vector space at that point. As we saw, a coordinate neighbourhood of $p$, $\left(\mathcal{U},x^{\mu}\right)$, induces a natural basis on $V_{p}$ (and on $V^{\ast}_{p}$), $\left\{ \partial_{\mu}\equiv\frac{\partial}{\partial x^{\mu}}\right\} $ (and $\left\{{\rm d}x^{\mu}\right\} $ respectively). Moreover, since a spacetime is a Lorentzian manifold, the tangent space at any point shall be isomorphic to the Minkowski space, which means that we can always find a basis of that particular tangent space, $\left\{ h_{a}\right\}$ (and $\left\{ h^{a}\right\}$ for the dual), for which the metric $g$ expressed in this basis will have the values of the Minkowski metric\footnote{See Acronyms and Conventions Section at the beginning of the manuscript for details.} $\eta$ \cite{Aldrovandi:1996ke}. Please note that we have chosen to use greek indices for the natural basis and the latin indeces for the ``proper'' basis of the tangent space, as it is customary.\\
Given the natural basis of the tangent space, we can always express the other basis as a linear combination of it, in particular
\begin{equation}
h_{a}=h_{a}\,^{\mu}\partial_{\mu},\quad\quad h^{a}=h^{a}\,_{\mu}dx^{\mu},
\end{equation}  
with
\begin{equation}
h^{b}\left(h_{a}\right)=h^{b}\,_{\mu}h_{a}\,^{\nu}dx^{\mu}\left(\partial_{\nu}\right)=\delta_{a}^{b}.
\end{equation}
The fact that the metric expressed in this basis is the Minkowski one means that
\begin{equation}
g\left(h_{a},h_{b}\right)=h_{a}\,^{\mu}h_{b}\,^{\nu}g\left(\partial_{\mu},\partial_{\nu}\right)=h_{a}\,^{\mu}h_{b}\,^{\nu}g_{\mu\nu}=\eta_{ab}.
\end{equation}
Consequently, we have that
\begin{equation}
g_{\mu\nu}=h^{a}\,_{\mu}h^{b}\,_{\nu}\eta_{ab}.
\end{equation}
The coefficients $h^{a}\,_{\mu}$ of the expansion of the basis $\left\{ h_{a}\right\}$ in terms of the natural basis are known as {\emph{tetrads}} or {\emph{virbein}}, and they relate the proper coordinates of each of the tangent spaces at any point with the natural coordinates induced by the spacetime ones. \\
Therefore, since the gauge transformations are defined locally, they will apply on the local basis, which in this case we have denoted as $\left\{ h_{a}\right\}$.\\

Before going any further, we will briefly review the mathematical procedures of the gauge principle, as explained in \cite{Aldrovandi:2013wha,Hehl:2019csx}. First of all, let us consider an action for a certain matter field $\psi$
\begin{equation}
\label{actiongauge}
S=\int{\rm d}^{4}x\mathcal{L}\left(\psi,\partial_{a}\psi\right),
\end{equation}   
and the transformation of this field under a m-parameter global symmetry group $G$
\begin{equation}
\psi\left(x\right)\rightarrow\psi'\left(x\right)=\psi\left(x\right)+\delta\psi\left(x\right),\quad\delta\psi\left(x\right)=\varepsilon^{B}\left(x\right)T_{B}\psi\left(x\right),
\end{equation}
where $\varepsilon^{B}\left(x\right)$, with $B=\left\{1,...,m\right\}$, are the m parameters of the group, which remain constant in $x$ since it is a global symmetry. The $T_{B}$ are known as the transformation generators, which satisfy the following commutation relation
\begin{equation}
\left[T_{B},T_{C}\right]=f^{A}\,_{BC}T_{A},
\end{equation}
where the $f^{A}\,_{BC}$ are the structure constants of the group's Lie algebra.\\
We will assume that the action is invariant under $G$, that is $\delta S=0$. Hence, according to the Noether's theorem there would be the following conservation law 
\begin{equation}
\partial_{i}J_{A}^{i}=0,\quad J_{A}^{i}:=T_{A}\psi\frac{\partial\mathcal{L}}{\partial\left(\partial_{i}\psi\right)},
\end{equation}
where $J_{A}^{i}$ is known as the {\emph{Noether current}}.\\
In order to apply the gauge principle, we shall impose that the transformation is local, which means that we relax the condition on the parameters of the group $\varepsilon^{B}\left(x\right)$, and allow them to take different values along $x$. Nevertheless, this would imply that the action \eqref{actiongauge} is no longer invariant under this local transformation. In order to achieve invariance again, we need to add a compensating {\emph{gauge field}} $A^{B}\,_{a}$ via the {\emph{minimal coupling prescription}}
\begin{equation}
\mathcal{L}\left(\psi,\partial_{a}\psi\right)\longrightarrow\mathcal{L}\left(\psi,D_{a}\psi\right),
\end{equation}
where the partial derivative has been replaced by a covariant one $D_{a}$ (does this ring a bell?), that is defined as follows  
\begin{equation}
D_{a}\psi\left(x\right)=\partial_{a}\psi\left(x\right)-A^{B}\,_{a}T_{B}\psi\left(x\right).
\end{equation}
Therefore the invariance is recovered because the field transforms as
\begin{equation}
\delta A^{C}\,_{a}=-\left[\partial_{a}\varepsilon^{C}\left(x\right)+f^{C}\,_{BD}A^{B}\,_{a}\varepsilon^{D}\left(x\right)\right].
\end{equation}
The gauge field can be promoted to a true dynamical variable of the system by adding its corresponding kinetic term $\mathcal{K}$ to the Lagrangian density. Of course, such a term needs to be gauge invariant, so that the whole action remains so. This is assured by constructing the kinetic term using the {\emph{gauge field strength}} as
\begin{equation}
\label{fstrength}
F^{A}\,_{ij}=\partial_{i}A^{A}\,_{j}-\partial_{j}A^{A}\,_{i}+f^{A}\,_{BC}A^{B}\,_{i}A^{C}\,_{j}.
\end{equation}
Then, the subsequent gauge action after we have applied the gauge principle would be
\begin{equation}
S_{{\rm gauge}}=\int{\rm d}^{4}x\left[\mathcal{L}\left(\psi,D_{a}\psi\right)+\mathcal{K}\left(F^{A}\,_{ij}\right)\right].
\end{equation}

Now we are ready to explore what happens when the gauge procedure is applied to the group of spacetime translations $T\left(4\right)$. We shall explain this based on \cite{DeAndrade:2000sf,Aldrovandi:2013wha}.\\
As it is known, the infinitesimal change under a local spacetime translation is given in the proper coordinates of the tangent space as
\begin{equation}
\delta\psi\left(x\right)=\varepsilon^{a}\left(x\right)P_{a}\psi\left(x\right),
\end{equation} 
with $P_{a}\equiv\frac{\partial}{\partial x^{a}}$ being the translation generators, which have the following commutation relations
\begin{equation}
\left[P_{a},P_{b}\right]=0.
\end{equation}
This local transformation induces a gauge field $B^{a}\,_{\mu}$, such as the covariant derivative $h_{\mu}$ is given by
\begin{eqnarray}
h_{\mu}\psi\left(x\right)&=&\partial_{\mu}\psi\left(x\right)+B^{a}\,_{\mu}P_{a}\psi\left(x\right)=\left(\partial_{\mu}x^{a}\right)\partial_{a}\psi\left(x\right)+B^{a}\,_{\mu}\partial_{a}\psi\left(x\right)
\nonumber
\\
&=&h^{a}\,_{\mu}\partial_{a}\psi\left(x\right),
\end{eqnarray}
where $h^{a}\,_{\mu}$ is a tetrad field defined as
\begin{equation}
h^{a}\,_{\mu}=\partial_{\mu}x^{a}+B^{a}\,_{\mu}.
\end{equation}
Since the structure constants of the group of translations are zero, the field strength of $B^{a}\,_{\mu}$ is given by
\begin{equation}
F^{a}\,_{\mu\nu}=\partial_{\mu}B^{a}\,_{\nu}-\partial_{\nu}B^{a}\,_{\mu}=\partial_{\mu}h^{a}\,_{\nu}-\partial_{\nu}h^{a}\,_{\mu}.
\end{equation}
Also, given a tetrad field, one can construct the following connection
\begin{equation}
\hat{\Gamma}^{\rho}\,_{\mu\nu}=h_{a}\,^{\rho}\partial_{\nu}h^{a}\,_{\mu},
\end{equation}
which is actually the Weitzenb\"ock connection that we introduced in the previous section. Then, it is easy to check that the field strength of translations is just the torsion of this connection written in the spacetime coordinates
\begin{equation}
T^{\rho}\,_{\mu\nu}=h_{a}\,^{\rho}F^{a}\,_{\nu\mu}.
\end{equation}
Therefore, as it is usual in gauge theories, we shall construct the action of the theory with quadratic terms in the field strength, which in this case is the torsion tensor of the Weitzenb\"ock connection, obtaining
\begin{equation}
\label{actionT4}
S_{T\left(4\right)}=\int {\rm d}^{4}x\sqrt{-g}\left(a_{1}T_{\mu\nu\rho}T^{\mu\nu\rho}+a_{2}T_{\mu\nu\rho}T^{\nu\mu\rho}+a_{3}T^{\mu}\,_{\mu\rho}T_{\nu}\,^{\nu\rho}\right).
\end{equation} 
Finally, should at this point the action \eqref{actionT4} be required to be invariant under local Lorentz transformacions, the coefficients $a_i$ would need to get fixed to $a_{1}=\frac{1}{4}$, $a_{2}=\frac{1}{2}$, and $a_{3}=-1$. Hence we find that
\begin{equation}
S_{T\left(4\right)}\left(a_{1}=\frac{1}{4};a_{2}=\frac{1}{2};a_{3}=-1\right)=S_{{\rm TEGR}}.
\end{equation}
As we know, TEGR is an equivalent theory to GR, and we have just proved that it can be obtained as a gauge theory. 

\subsection{Gauge theory of the Poincar\'e group}

\label{GaugePG}

The important question now is whether the group of spacetime translations is the adequate group to gauge in order to obtain the gravitational theory. To answer this question we shall rely on experiment, in particular on the Colella-Overhauser-Werner (COW) experiment \cite{Colella:1975dq}, and its more precise reproductions \cite{Kasevich:1991zz,Asenbaum:2017rwf,Overstreet:2017gdp}.  \\
This kind of experiments consist of a neutron (which is a half-spin particle) beam that is split into two beams which travel in different gravitational potentials. Later on the two beams are reunited and an interferometric picture is observed due to their relative phase shift. Therefore, such an inference pattern proves that there is an interaction between the internal spin of particles and the gravitational field. This suggests that the test particle for gravity should not be the ``Newton's apple'', but instead a particle with mass $m$ and spin $s$ should be used.\\
On the other hand, from Wigner's work \cite{Wigner:1939cj}, we know that a quantum system can be identified by its mass and spin in Minkowski spacetime, which is invariant under global Poincar\'e transformations. Therefore the Poincar\'e group $T\left(4\right)\times SO\left(1,3\right)$, which is formed by the homogeneous Lorentz group $SO\left(1,3\right)$ and the spacetime translations $T\left(4\right)$, seems the natural choice to apply the gauge principle. This was thought by Sciama \cite{sciama1962analogy} and Kibble \cite{Kibble:1961ba}, and later on formalised by Hayashi \cite{Hayashi:1968hc} and Hehl {\emph{et al.}} \cite{Hehl:1976kj}. We shall apply the gauge procedure on the Poincar\'e group in the following.\\

The infinitesimal change under a global Poincar\'e transformation is given in the proper coordinates of the tangent space as
\begin{equation}
\delta\psi\left(x\right)=\varepsilon^{a}\partial_{a}\psi\left(x\right)+\varepsilon^{ab}S_{ab}\psi\left(x\right),
\label{poincare:transformation}
\end{equation}
where $\varepsilon^{ab}$ are the six parameters of the Lorentz group, and $S_{ab}$ its generators, which along with the generators of translations $\partial_{a}$ follow the known commutation relations of the Poincar\'e group, namely
\begin{eqnarray}
&&\left[S_{ab},S_{cd}\right]=\frac{1}{2}\left(\eta_{ac}S_{bd}+\eta_{db}S_{ac}-\eta_{bc}S_{ad}-\eta_{ad}S_{bc}\right),
\nonumber
\\
&&\left[S_{ab},\partial_{c}\right]=\frac{1}{2}\left(\eta_{ac}\partial_{b}-\eta_{bc}\partial_{a}\right),
\\
&&\left[\partial_{a},\partial_{b}\right]=0.
\nonumber
\end{eqnarray}
Now, in order to apply the gauge principle we need to consider how this transformation \eqref{poincare:transformation} behaves locally, which means that the parameters of the transformation can vary along the spacetime. As we have explained in the previous subsection, an action which is invariant under the global transformation might not be under the local one. Consequently, as we also have done for the $T(4)$ case, we shall construct a covariant derivative that will have the induced gauge fields associated to the translations $B^{a}\,_{\mu}$ and the Lorentz group $\omega^{ab}$, which will compensate the fact that these transformations are not the same at every point. In this case, the gauge covariant derivative $D_{\mu}$ is given by \cite{Blagojevic:2003cg}
\begin{equation}
D_{\mu}\psi\left(x\right)=\partial_{\mu}\psi\left(x\right)+B^{a}\,_{\mu}\partial_{a}\psi\left(x\right)-\omega^{ab}\,_{\mu}S_{ab}\psi\left(x\right)=h^{a}\,_{\mu}D_{a}\psi\left(x\right),
\end{equation}
where
\begin{equation}
D_{a}:=\partial_{a}+\omega^{bc}\,_{a}(x)S_{cb},
\end{equation}
which is actually the generator of the {\bf{local}} rotation-free translations in the local Poincar\'e group.\\
The field strengths of the gauge potentials shall be denoted $F^{a}\,_{\mu\nu}$ for the translations and $H^{ab}\,_{\mu\nu}$ for the Lorentz transformations respectively. Both of them can be obtained from the commutator of the gauge covariant derivative as follows\footnote{One can easily check that expression \eqref{fstrength2} is compatible with the definition \eqref{fstrength} for one gauge field.}
\begin{equation}
\label{fstrength2}
\left[D_{\mu},D_{\nu}\right]\psi\left(x\right)=F^{a}\,_{\mu\nu}D_{a}\psi\left(x\right)+H^{ab}\,_{\mu\nu}S_{ab}\psi\left(x\right),
\end{equation}
where
\begin{equation}
F^{a}\,_{\mu\nu}=\partial_{\mu}h^{a}\,_{\nu}-\partial_{\nu}h^{a}\,_{\mu}+\omega^{ab}\,_{\mu}h_{b\nu}-\omega^{ab}\,_{\nu}h_{b\mu},
\end{equation}
and
\begin{equation}
H^{ab}\,_{\mu\nu}=\partial_{\mu}\omega^{ab}\,_{\nu}-\partial_{\nu}\omega^{ab}\,_{\mu}+\omega^{ac}\,_{\nu}\omega^{b}\,_{c\mu}-\omega^{ac}\,_{\mu}\omega^{b}\,_{c\nu}.
\end{equation}
At this point, we can interpret what the field strenghts are in terms of the spacetime quantities. Let us notice that one can build a connection in terms of $h^{a}\,_{\mu}$ and $\omega^{ab}\,_{\mu}$ as
\begin{equation}
\label{connectionpg}
\Gamma^{\rho}\,_{\mu\nu}=h_{a}\,^{\rho}\partial_{\mu}h^{a}\,_{\nu}+h_{a}\,^{\rho}h_{b\nu}\omega^{ab}\,_{\mu}.
\end{equation}
Then we can see the field strength of the translations $F$ as the torsion tensor of this connection in spacetime coordinates, and the field strength of the Lorentz transformations $H$ as the Riemmann tensor of this connection expressed in spacetime coordinates. Namely
\begin{equation}
T^{\rho}\,_{\mu\nu}=h_{a}\,^{\rho}F^{a}\,_{\mu\nu},
\end{equation}
and
\begin{equation}
R_{\mu\nu\rho}\,^{\sigma}=h_{a}\,^{\rho}h_{b\rho}H^{ab}\,_{\mu\nu}.
\end{equation}
It is important to stress that in this case the connection in \eqref{connectionpg} is not fixed by the tetrad structure, therefore it is independent of the metric of the spacetime.

Finally, to build the action of the Poincar\'e gauge theory we shall use the invariants of the field strengths up to quadratic order
\begin{eqnarray}
\label{actionPGT}
S_{{\rm PG}}=&&\int {\rm d}^{4}x\sqrt{-g}\left(a_{0}R+a_{1}T_{\mu\nu\rho}T^{\mu\nu\rho}+a_{2}T_{\mu\nu\rho}T^{\nu\rho\mu}+a_{3}T_{\mu}T^{\mu}+b_{1}R^{2}\right.
\nonumber
\\
&&+b_{2}R_{\mu\nu\rho\sigma}R^{\mu\nu\rho\sigma}+b_{3}R_{\mu\nu\rho\sigma}R^{\rho\sigma\mu\nu}+b_{4}R_{\mu\nu\rho\sigma}R^{\mu\rho\nu\sigma}+b_{5}R_{\mu\nu}R^{\mu\nu}
\nonumber
\\
&&+\left.b_{6}R_{\mu\nu}R^{\nu\mu}\right),
\end{eqnarray}
where the $a_i$ and $b_i$ are the constants parameters of the theory.\\
This theory has clearly more degrees of freedom than GR, due to the quadratic curvature and torsion terms. To analyse those introduced by the torsion tensor, it is customary to decompose the torsion in three terms \cite{Shapiro:2001rz}
\begin{equation}
\label{decomposition}
\begin{cases}
\textrm{Trace vector:}\,\,T_{\mu}=T^{\nu}\,_{\mu\nu},\\
\,\\
\textrm{Axial vector:}\,\,S_{\mu}=\varepsilon_{\mu\nu\rho\sigma}T^{\nu\rho\sigma},\\
\,\\
\textrm{Tensor}\,\,q^{\rho}\,_{\mu\nu},\,\,{\rm such\,that}\;\;q_{\,\,\mu\nu}^{\nu}=0\,\,\textrm{and}\,\,\varepsilon_{\mu\nu\rho\sigma}q^{\nu\rho\sigma}=0,
\end{cases}
\end{equation}
such that
\begin{equation}
\label{decomposition2}
T^{\rho}\,_{\mu\nu}=\frac{1}{3}\left(T_{\mu}\delta_{\nu}^{\rho}-T_{\nu}\delta_{\mu}^{\rho}\right)+\frac{1}{6}\varepsilon^{\rho}\,_{\mu\nu\sigma}S^{\sigma}+q^{\rho}\,_{\mu\nu}.
\end{equation}
These three pieces are irreducible under the Lorentz group as real representations and correspond to $(\frac12,\frac12)$, $(\frac12,\frac12)$ and $(\frac32,\frac12)\oplus(\frac12,\frac32)$ respectively \cite{Hayashi:1979wj}. This decomposition turns out to be very useful, thanks to the fact that the three terms in Eq.\eqref{decomposition} propagate different dynamical off-shell degrees of freedom. Hence, it is better to study them separately, as we shall corroborate in some sections of this thesis, compared to all the torsion contribution at the same time.\\

Poincar\'e Gauge gravity has been widely studied since the 1970s, and we shall present some of its interesting phenomenology in the third chapter of this thesis. Nonetheless, one needs to first unveil the field content and analyse its stable/unstable nature, since this is one of the important questions for the viability of the theories
with a crucial impact on the phenomenology. We will study this aspect in the next section.

\section{Stability of Poincar\'e Gauge gravity}
\label{2.3}
As we have already mentioned at the beginning of Subsection \ref{GaugePG}, the fields that are present in a theory that is locally invariant under Poincar\'e transformations can be uniquely classified in terms of two quantities, their masses and spins. For example, GR (and its equivalent formulations) can be seen as a massless spin-2 particle \cite{Heisenberg:2018vsk}, commonly known as {\emph{graviton}}.\\
In the case of the PG gravitational theory \eqref{actionPGT}, apart from the usual graviton, there are two massive spin-2, two massive spin-1 and two spin-0 fields \cite{Hayashi:1980qp}. Already in \cite{Sezgin:1979zf,Hayashi:1980qp,Sezgin:1981xs} it was shown that, in a Minkowski spacetime, all of these fields cannot propagate simultaneously without incurring in some pathological behaviour. In particular, it was proven that the absence of ghosts and tachyons instabilities, which we shall explain in the next subsection, restrains the spectrum to contain at most three propagating components, along with other restrictions on the parameters of the theory. Later on, the authors in \cite{Yo:1999ex,Yo:2001sy} performed a more complete Hamiltonian analysis of PG theories (see also a more recent analysis in \cite{Blagojevic:2018dpz}), where they find that the introduction of non-linearities would impose further constraints. Moreover, they showed that the only modes that could propagate were two scalars with different parity. We will arrive at the same conclusion following a different path in the next subsections. But first, it is important to study what it means to have an unstable theory, which we will explain in the following subsection.

\subsection{Instabilities}

Here we shall explain what we understand by instabilities and what causes them. Roughly speaking, an unstable theory is one where a perturbation in one of the variables produces an unboundly increase in its absolute value (see Figure \ref{fig:stable} for an illustrative example). 

\begin{figure}
\centering
\includegraphics[width=0.88\linewidth]{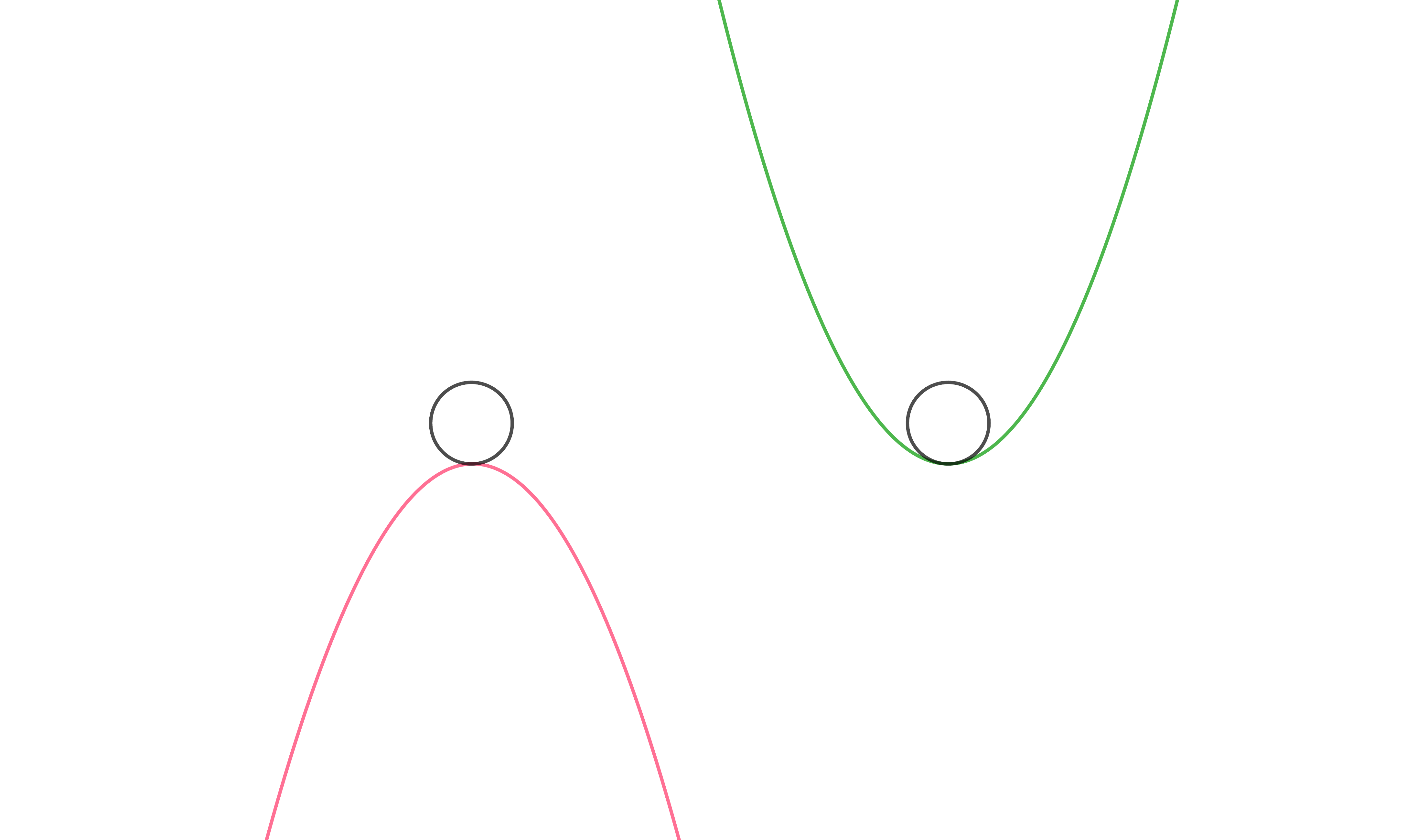}
\caption{Graphic example showing an unstable (left) and a stable (right) configuration. It is clear that even the slightest perturbation in the left one will make the height of the ball to decrease unboundly.}
\label{fig:stable}
\end{figure}

There are different situations where that kind of behaviour occurs. In the following we shall enumerate and explain the most important types of instabilities.

\subsubsection{Ghosts}

A field is known as a {\emph{ghost}} if its kinetic term has the wrong sign, hence the particle associated with this field would have negative kinetic energy (for a recent review on this subject see \cite{Sbisa:2014pzo}). This instability is related with the momentum of the particles, in the sense that this is the variable that would experience an unboundly increase in its absolute value. We can see this intuitively with a very simple example. Let us imagine a collision between two protons in a theory that introduces a ghost. Then, since we have a particle with negative kinetic energy that can be part of the products of this collision, there is no limit to the momentum of the non-ghost particles that result from the collision. This is because we can always compensate the excess in energy by increasing the absolute value of the kinetic energy of the produced ghost, so that the conservation of energy holds. More specifically, this means that the volume available in the momentum space is infinite. \\
In order to give a more rigorous insight, let us consider the following Lagrangian density for a ghost scalar field $\psi$ and a non-ghost scalar field $\phi$\footnote{We have chosen scalar fields for simplicity, but the results are generalisable to vector and tensor fields.} in Minkowski spacetime
\begin{equation}
\label{ghost:example}
\mathcal{L}=\frac{1}{2}\partial_{\mu}\phi\partial^{\mu}\phi+\frac{1}{2}m_{\phi}^{2}\phi^{2}-\frac{1}{2}\partial_{\mu}\psi\partial^{\mu}\psi-\frac{1}{2}m_{\psi}^{2}\psi^{2}-V_{\mathrm{int}}(\phi,\psi),
\end{equation}
where the interaction potential $V_{\mathrm{int}}$ does not contain derivative interaction terms, is analytic in both $\psi$ and $\phi$, and the configuration $\psi=\phi=0$ is a local minimum of the potential.\\
Performing a Legendre transformation with respect to $\dot{\psi}$ and $\dot{\phi}$ we obtain the Hamiltonian density, which is a measure of the total energy of the system
\begin{equation}
\mathcal{H}=-\frac{1}{2}\left[\dot{\phi}^{2}+(\vec{\nabla}\phi)^{2}\right]-\frac{1}{2}m_{\phi}^{2}\phi^{2}+\frac{1}{2}\left[\dot{\psi}^{2}+(\vec{\nabla}\psi)^{2}\right]+\frac{1}{2}m_{\psi}^{2}\psi^{2}+V_{\mathrm{int}}(\phi,\psi).
\end{equation}
For each field we can make a Fourier decomposition as
\begin{equation}
\label{fourier}
\phi(\vec{x},t)=\int_{\mathbb{R}^{3}}\frac{{\rm d}^{3}p}{(2\pi)^{3}}\phi_{\vec{p}}(t){\rm e}^{i\vec{p}\cdot\vec{x}},
\end{equation}
where $\vec{p}$ is the 3-momentum of the modes $\phi_{\vec{p}}$.\\
Now, it is clear that even if $V_{\text {int }} \neq 0$ it remains true that configuration $\phi(\vec{x}, t)=\psi(\vec{x}, t)=0$ is a solution of the equations of motion, since those values extremise the Lagrangian density. Therefore, if we set the system in the state $\phi\left(\vec{x}, t_{0}\right)=\psi\left(\vec{x}, t_{0}\right)=0$ at an initial time $t_{0},$ the fields $\phi$ and $\psi$ remain in the ``vacuum'' configuration forever. In order to check what happens if we slightly perturb the vacuum configuration, we must take into account that the configuration introduced in \eqref{fourier}, where the ghost and the ordinary field are plane waves of vanishing total momentum, which do not have zero energy due to the presence of the interaction terms. However, since derivative interactions are absent, choosing the amplitudes of the plane waves to be small enough would enable us to construct configurations with energy as close to zero as we want, without constraining the wavevector of each plane wave whose magnitude can be arbitrarily big. Therefore, for every value of the energy $E \geq 0$ there exists an infinite number of excited configurations and the volume of momentum space available for each (ordinary/ghost) sector is infinite. For entropy reasons, the decay towards these excited states is extremely favoured, and we conclude that the system is unstable against small oscillations.\\

Nevertheless, in most cases it is more difficult to identify the kinetic term of a certain field than in the previous example \eqref{ghost:example}, specially if higher derivatives are involved in the action. The ghosts that appearing as a consequence of having higher derivatives in the field equations are known as Ostrogradski instabilities, since they are predicted by his famous stability theorem \cite{Ostrogradsky:1850fid} (see \cite{Woodard:2006nt} for a recent study). It mainly states that there is a linear instability in the Hamiltonians associated with Lagrangians which depend upon more than one time derivative in such a way that the dependence cannot be eliminated by partial integration.\\

Another relevant ghost that can appear in modified gravity theories is the Boulware-Deser ghost \cite{Boulware:1973my}, which is a scalar ghost that may be present when considering massive spin-2 fields, like in massive gravity \cite{Hassan:2011hr}.

\subsubsection{Tachyons}

Tachyon instability is related with the field having an imaginary mass, which means that we have the wrong sign for the mass term \cite{Heisenberg:2018vsk}. Indeed we can see this by studying the Lagrangian density of a tachyonic free scalar field $\phi$ in Minkowski spacetime
\begin{equation}
\mathcal{L}=\frac{1}{2}\partial_{\mu}\phi\partial^{\mu}\phi-\frac{1}{2}m^{2}\phi^{2}.
\end{equation}
Then the field equation would be
\begin{equation}
\left(\Box+m^{2}\right)\phi=0,
\end{equation}
where $\Box=\eta_{\mu\nu}\partial^{\nu}\partial^{\nu}$. Clearly one needs an imaginary mass in order to recover the correct sign of the mass in the Klein-Gordon equation of free scalar fields.\\
For a general potential the condition of having tachyonic stabilities is that the second derivative of the potential is negative, which means that the potential is at a local maximum instead of a local minimum, hence producing an unstable situation analogous to the red graph of example in Figure \ref{fig:stable}. 

\subsubsection{Laplacian instabilities}

Laplacian instabilities can be present in certain configurations when performing small perturbations around the background, which, depending on the parameters of the considered theory, may grow unboundly. In order to understand this kind of instability let us consider the perturbations of a scalar field $\delta\phi$ around an arbitrary background configuration $\bar{\phi}$, up to quadratic order\footnote{Again for simplicity we have considered the scalar case, but these results are easily generalisable to vector and tensor perturbations.} \cite{Heisenberg:2018vsk}
\begin{equation}
\delta\mathcal{L}=\frac{1}{2}\mathcal{F}\left(-\delta\dot{\phi}^{2}+c_{s}^{2}\nabla\delta\phi^{2}\right)+\frac{1}{2}m^{2}\delta\phi^{2},
\end{equation}
where $\mathcal{F}$, $m$ and $c_{s}^{2}$ depend on the specific background configuration. Then, as known from the differential equations literature, the field equations that these perturbations follow are unstable if $c_{s}^{2}<0$. Then, to avoid Laplacian instabilities we shall impose that $c_{s}^{2}\ge 0$, meaning that the so-called {\emph{scalar propagation speed}} $c_{s}$ is a real number.

\subsubsection{Strong coupling}

This instability is slightly different from the previous ones. The problem appears when some of the extra degrees of freedom present in the theory do not propagate in certain backgrounds. We can explain why this is an issue by taking into account the illustrative example in Figure \ref{fig:stable}. In a stable situation, if we make a slight perturbation from the equilibrium of the system we shall be able to recover the initial state. Therefore the aforementioned backgrounds would be problematic because if we make a perturbation over that background, those degrees of freedom are going to propagate. Hence, it would be impossible to recover the initial situation since in the background that degree of freedom does not propagate. \\
This is a known problem in Ho\v{r}ava gravity \cite{Charmousis:2009tc}, massive gravity \cite{Deffayet:2005ys}, and $f(T)$ theories \cite{Jimenez:2020ofm}. Also, it could be a potential problem in Infinite Derivative Gravity, as we shall see in Chapter \ref{4}.
\,\\

Let us note that all the mentioned instability issues are purely classical, and that there can be also problems associated to quantum corrections of the considered theory. The main problem is due to the fact that the coefficients of the theory, which are tuned in order to explain the experimental data, can be affected by an strong renormalisation under these quantum corrections. If the coefficients get detuned within the scale of validity of the theory, then the theory is render as quantum unstable. A clear example such instability is GR with a cosmological constant \cite{Weinberg:1988cp,Martin:2012bt}. Also, paradigmatic examples of quantum corrections studies include scalar-tensor theories \cite{Luty:2003vm,Nicolis:2004qq,deRham:2012ew,Koyama:2013paa,Brouzakis:2013lla,Heisenberg:2014raa,Goon:2016ihr,Heisenberg:2020cyi}, massive gravity \cite{deRham:2013qqa}, and generalised Proca theories \cite{Heisenberg:2020jtr}.\\  

Now we are ready to study the stability of PG gravity. We shall summarise the results of {\bf{\nameref{P6}}}, where the stable modes of propagation of the action \eqref{actionPGT} were found. We shall unveil the presence of pathological terms in a background-independent approach just by looking at the interactions of the different torsion components.\\ 
In order to avoid ghosts already for the graviton when the torsion is set to zero, we will impose the recovery of the Gauss-Bonnet term in the limit of vanishing torsion. In $d=4$ dimensions this allows one to use the topological nature of the Gauss-Bonnet term to remove one of the parameters. More explicitly, we have
\begin{eqnarray}
\label{eq:PGactionzeroT}
\mathcal{L}_{\rm PG}\big\vert_{T=0}=a_{0}\mathring{R}+\left(b_{2}+b_{3}+\frac{b_{4}}{2}\right)\mathring{R}_{\mu\nu\rho\sigma}\mathring{R}^{\mu\nu\rho\sigma}+\left(b_{5}+b_{6}\right)\mathring{R}_{\mu\nu}\mathring{R}^{\mu\nu}+b_{1}\mathring{R}^{2},
\end{eqnarray}
so the Gauss-Bonnet term for the quadratic sector is recovered upon requiring
\begin{equation}
\label{con1}
b_{5}=-4b_{1}-b_{6},\;\;\;b_{4}=2(b_{1}-b_{2}-b_{3}),
\end{equation}
that we will assume throughout the rest of this chapter unless otherwise stated. The parameter $b_1$ plays the role of the coupling constant for this Gauss-Bonnet term. In $d=4$ dimensions this parameter is irrelevant, but it is important for $d>4$.\\
In the following subsection we shall show that imposing stability in the torsion vector modes reduces drastically the parameter space of PG gravity.

\subsection{Ghosts in the vector sector}

In $d=4$ dimensions, a vector field $A_{\mu}$ has four components: one temporal $A_{0}$, and three spatial $A_{i}$, with $i=1,2,3$. However, they cannot propagate at the same time without introducing a ghost degree of freedom (d.o.f.). In particular for any theory describing a massive vector, like the ones present in the PG action, we must require the following conditions in order to avoid ghosts \cite{Heisenberg:2014rta,Heisenberg:2018vsk}:
\begin{itemize}
\item {\bf{The equations of motion must be of second order}}. As we explained in the previous subsection, this is because the Ostrogradski theorem predicts ghosts for higher-order equations of motion.

\item {\bf{The temporal component of the vector field $A_0$ should not be dynamical}}. This is required because if this degree of freedom propagates, its kinetic term would be of opposite sign of the one of the spatial components, hence being a ghost. Therefore, the massive vector under this ghost-free condition would only propagate three degrees of freedom, which is exactly the ones that the massive spin-1 representation of the Lorentz group can propagate. 
\end{itemize}
Following these prescriptions, in this subsection we shall constrain the parameter space of PG gravity by imposing stability in the two massive spin-1 fields that are part of the particle spectrum of this theory.\\
In order to do so, we look at the vector sector containing the trace $T_\mu$ and the axial component $S_\mu$ of the torsion, while neglecting the pure tensor part $q^{\rho}\,_{\mu\nu}$ for the moment. Plugging the decompositions \eqref{LCWeitzenbock} and \eqref{decomposition2} into the PG Lagrangian \eqref{actionPGT} we obtain
\bea
\Lag_{\rm v}&=&-\frac29\big(\kappa-\beta\big)\mT_{\mu\nu}\mT^{\mu\nu}+\frac{1}{72}\big(\kappa-2\beta\big)\mS_{\mu\nu}\mS^{\mu\nu}+\frac12m_T^2T^2+\frac12m_S^2S^2+\frac{\beta}{81} S^2 T^2\nonumber\\
&&+\frac{4\beta-9b_2}{81}\Big[(S_\mu T^\mu)^2+3S^\mu S^\nu \nablab_\mu T_\nu\Big]+\frac{\beta}{54} S^2\nablab_\mu T^\mu+\frac{\beta-3b_2}{9}S^\mu T^\nu\nablab_\mu S_\nu \nonumber\\
&&+\frac{\beta-3b_2}{12} (\nablab_\mu S^\mu)^2+\frac{\beta}{36}\Big(2\Gb^{\mu\nu}S_\mu S_\nu+\Rb S^2\Big),
\label{eq:PGaction2}
\eea
where $\mT_{\mu\nu}=2\partial_{[\mu} T_{\nu]}$ and $\mS_{\mu\nu}=2\partial_{[\mu} S_{\nu]}$ are the field strengths of the trace and axial vectors respectively and we have defined
\bea
\kappa&=&4b_1+b_6\, ,\\
\beta&=&b_1+b_2-b_3\, ,\\
m_T^2&=&-\frac23\big(2a_0-2a_1+a_2-3a_3\big)\, ,\\
m_S^2&=&\frac{1}{12}\big(a_0-4a_1-4a_2\big).
\eea
In order to arrive at the final expression \eqref{eq:PGaction2} we have used the Bianchi identities to eliminate terms containing $\Rb_{\mu\nu\rho\sigma}\epsilon^{\alpha\nu\rho\sigma}$ and express $\Rb_{\mu\nu\rho\sigma}\Rb^{\mu\rho\nu\sigma}=\frac12\Rb_{\mu\nu\rho\sigma}\Rb^{\mu\nu\rho\sigma}$. We have also dropped the Gauss-Bonnet invariant of the Levi-Civita connection and the total derivative $\varepsilon_{\mu\nu\alpha\beta} \mS^{\mu\nu}\mT^{\alpha\beta}$. Moreover, we have made a few integrations by parts and used the commutator of covariant derivatives. Let us point out that the parameter $b_1$ does not play any role and can be freely fixed since it simply corresponds to the irrelevant Gauss-Bonnet coupling constant.

The Lagrangian \eqref{eq:PGaction2} has some interesting characteristics. Indeed, if we look at the pure trace sector $T_\mu$, we see that it does not contain non-minimal couplings. This is an accidental property in four dimensions. To show this fact more explicitly, we shall give the Lagrangian for the pure trace sector in an arbitrary dimension $d\ge 4$
\bea
&&\Lag^{d}_{T}=-\frac{d-2}{(d-1)^2}\left(\frac{d-2}{2}\kappa-\beta\right)\mT_{\mu\nu}\mT^{\mu\nu}+\frac12m_T^2(d)T^2\nonumber\\
&&+b_1\frac{(d-4)(d-3)(d-2)}{(d-1)^3}\left[\Big(T^4-4T^2\nablab_\mu T^\mu\Big)+4\frac{d-1}{d-2}\Gb_{\mu\nu} T^\mu T^\nu\right],
\label{eq:Td}
\eea
with
\be
m_T^2(d)=\frac{2}{1-d}\Big[(d-2)a_0-2a_1+a_2+(1-d)a_3\Big].
\ee
Indeed, all the interactions trivialise\footnote{Notice that $b_1$ is the coupling constant of the Gauss-Bonnet term also for arbitrary dimension $d$, so the trace interactions only contribute if the Gauss-Bonnet is also present, which is dynamical for $d>4$.} in $d=4$ dimensions. It is remarkable however that in \eqref{eq:Td} the non-gauge-invariant derivative interaction $T^2\nablab_\mu T^\mu$ is of the vector-Galileon type, and the non-minimal coupling is only to the Einstein tensor, which is precisely one of the very few ghost-free couplings to the curvature for a vector field (see \emph{e.g.} \cite{Jimenez:2016isa}).  The obtained result agrees with the findings in \cite{Jimenez:2015fva,Jimenez:2016opp} where a general connection determined by a vector field that generates both torsion and non-metricity was considered. 

Let us now return our attention to the full vector Lagrangian \eqref{eq:PGaction2}. Unlike the torsion trace, the axial component $S_\mu$ shows very worrisome terms that appear in the three following ways:
\begin{itemize}
\item  The perhaps most evidently pathological term is $(\nablab_\mu S^\mu)^2$ that introduces a ghostly d.o.f. associated to the temporal component $S_0$, because it clearly makes this temporal component propagate. We shall get rid of it by imposing $\beta=3b_2$. This constraint has already been found in the literature in order to guarantee a stable spectrum on Minkowski.

\item The non-minimal couplings to the curvature are also known to lead to ghostly d.o.f.'s \cite{Himmetoglu:2008zp,Jimenez:2008sq,ArmendarizPicon:2009ai,Himmetoglu:2009qi}. The presence of these instabilities shows in the metric field equations where again the temporal component of the vector will enter with second derivatives, hence revealing its problematic dynamics. As mentioned above, an exception is the coupling to the Einstein tensor that avoids generating second time derivatives of the temporal component due to its divergenceless property. For this reason we have explicitly separated the non-minimal coupling to the Einstein tensor in \eqref{eq:PGaction2}. It is therefore clear that we need to impose the additional constraint $\beta=0$ to guarantee the absence ghosts, which, in combination with the above condition $\beta=3b_2$, results in $\beta=b_2=0$.

\item Furthermore, there are other interactions in \eqref{eq:PGaction2} with a generically pathological character schematically given by $S^2\nabla T$ and $ST\nabla S$. Although these may look like safe vector Galileon-like interactions, actually the fact that they contain both sectors makes them dangerous. This can be better understood by introducing St\"uckelberg fields and taking an appropriate decoupling limit, so we effectively have $T_\mu\rightarrow\partial_\mu T$ and $S_\mu\rightarrow\partial_\mu S$ with $T$ and $S$ the scalar and pseudo-scalar St\"uckelbergs. The interactions in this limit become of the form $(\partial S)^2\partial^2 T$ and $\partial T\partial S\partial^2T$ that, unlike the pure Galileon interactions, generically give rise to higher-order equations of motion and, therefore, Ostrogradski instabilities. Nevertheless, we can see that the avoidance of this pathological behaviour does not introduce new constraints on the parameters, since the coefficients in front of them in \eqref{eq:PGaction2} are already zero if we take into account the two previous stability considerations. 

\end{itemize}

The extra constraint $\beta=0$ conforms the crucial obstruction for stable PGTs. This new constraint genuinely originates from the quadratic curvature interactions in the PGT Lagrangian. Such interactions in the Lagrangian induce the non-minimal couplings between the axial sector and the graviton, as well as the problematic non-gauge-invariant derivative interactions. Also, {\bf{this constraint cannot be obtained from a perturbative analysis on a Minkowski background}} because, in that case, these interactions will only enter at cubic and higher orders so that the linear analysis is completely oblivious to it.

We can see that the stability conditions not only remove the obvious pathological interactions mentioned before, but they actually eliminate all the interactions and only leave the free quadratic part
\bea
\Lag_{\rm v}\big\vert_{b_2,\beta=0}&=&-\frac{2}{9}\kappa\mT_{\mu\nu}\mT^{\mu\nu}+\frac12m^2_T T^2+\frac{1}{72}\kappa\mS_{\mu\nu}\mS^{\mu\nu}+\frac12m^2_S S^2
\label{eq:PGaction4}
\eea
where we see that the kinetic terms for $T_\mu$ and $S_\mu$  have the same normalisation but with opposite signs, hence leading to the unavoidable presence of a ghost. Therefore, the only stable possibility is to exactly cancel both kinetic terms. Consequently, the entire vector sector becomes non-dynamical.\\

Now that we have shown that the vector sector must trivialise in stable PGTs, we can return to the full torsion scenario by including the pure tensor sector $q^\rho{}_{\mu\nu}$. Instead of using the general decomposition \eqref{decomposition2}, it is more convenient to work with the torsion directly for our purpose here. We can perform the post-Riemannian decomposition for the theories with a stable vector sector to obtain
\begin{equation}
\label{Stablenonprop}
\mathcal{L}_{\rm stable}=a_{0}\mathring{R}+b_{1}\GB+a_{1}T_{\mu\nu\rho}T^{\mu\nu\rho}+a_{2}T_{\mu\nu\rho}T^{\nu\rho\mu}+a_{3}T_\mu T^\mu.
\end{equation}
The first term is just the usual Einstein-Hilbert Lagrangian, modulated by $a_0$, while the second term corresponds to the topological Gauss-Bonnet invariant for a connection with torsion, so we can safely drop it in four dimensions and, consequently, the first two terms in the above expression simply describe GR. The rest of the expression clearly shows the non-dynamical nature of the full torsion so that having a stable vector sector also eliminates the dynamics for the tensor component, therefore making the full connection an auxiliary field. We can then integrate the connection out and, similarly to the Einstein-Cartan theories, the resulting effect will be the generation of interactions for fermions that couple to the axial part of the connection. From an Effective Field Theory perspective, the effect will simply be a shift in the corresponding parameters of those interactions with no observable physical effect whatsoever.\\

\subsubsection{Explicit cosmological example}

At this point we think it is interesting to work out a specific example, since that will help us show how the ghosts appear and rederive the same conclusions in a concrete simplified situation. The study of particular situations is important to guarantee the absence of hidden constraints that could secretly render the theory stable even if the Lagrangian contains dangerous-looking operators, like the ones present in the axial vector sector. In this respect, we need to bear in mind that worrisome terms can be generated from perfectly healthy interactions via field redefinitions (see \emph{e.g.} the related discussion in \cite{Jimenez:2019hpl}), or that the coupling between the propagating d.o.f. could ameliorate the ghostly behaviour presumed by the presence of higher-order derivatives \cite{Motohashi:2016ftl,Klein:2016aiq,deRham:2016wji}. Therefore, we must make sure that the terms arising in the quadratic PGTs do not correspond to some obscure formulation of well-behaved theories. There is no obvious reason to expect any such mechanism at work for PGTs and in fact we shall demonstrate that if we do not impose the constraints obtained above the temporal component of the axial vector propagates, even in a very simple setup.

In order to prove the dynamical nature of the ghost mode $S_0$, we will consider a homogeneous vector sector, meaning that it only depends on time, in a cosmological background described by the flat FLRW metric\footnote{In fact, we could have sticked to a Minkowski background. We have prefered however to use a general cosmological background to not trivialise any interaction in \eqref{eq:PGaction2} and to explicitly show the irrelevant role of the curvature for our analysis.}
\be
\dd s^2=a^2(t)\big(-\dd t^2+\dd\vec{x}^2\big).
\label{eq:FLRW}
\ee
The tensor sector is kept trivial so that we only have to care about the vector components. It is straightforward to see from \eqref{eq:PGaction2} that $T_0$ is always an auxiliary field, since it does not exhibit any dinamics. To calculate its equation of motion we shall assume that the spatial part of the vectors is aligned with the $z$ direction, that is $T_{\mu}=\left(T_{0},0,0,T_{z}\right)$ and $S_{\mu}=\left(S_{0},0,0,S_{z}\right)$. Having this in mind we find that $\frac{\delta S}{\delta T_{0}}=0$ implies
\bea
&&\Big[-27m_T^2a^2+2(\beta-3b_2)S_0+\frac23\beta S_z^2\Big]T_0+\frac23(9b_2-4\beta)S_0\vec{S}\cdot\vec{T}\nonumber\\
&&+6(3b_2-\beta)HS_0^2-2\beta H\vec{S}^2+\frac{3}{2}(3b_2-2\beta)(S_0^2)'+\frac\beta2(\vec{S}^2)'=0,
\eea
where the prime represents the derivative with respect to time and $H=\frac{a'\left(t\right)}{a\left(t\right)}$ is the so-called Hubble parameter. We can then solve for $T_0$ in the previous expression and integrate it out from the action. After performing a few integrations by parts, we can compute the corresponding Hessian from the resulting Lagrangian, which is defined as follows
\be
\mathcal{H}_{ij}=\frac{\delta \mS_{\rm B}}{\delta \dot{X}_i\delta\dot{X}_j}\,.
\ee
The Hessian of a system is quite important since it allows us to determine the presence of constraints. In particular, if the determinant of the Hessian matrix is different from zero there will not be additional constraints between the variables that have been chosen to calculate the Hessian. In our case, the variables will be $\vec{X}=(S_0,T_z,S_z)$, and the resulting Hessian is
\bea
\mathcal{H}_{ij}=
\left(
\begin{array}{ccc}
  \lambda_1& \tilde{\lambda} &  0 \\
  \tilde{\lambda}& \lambda_2  &   0\\
 0 &0   &   \frac89(\kappa-\beta)
\end{array}
\right)\,,
\eea
where we have defined
\bea
\lambda_1&=&\frac{\beta-3b_2}{6}+\frac{(3b_2-2\beta)^2 S_0^2}{81 m_T^2 a^2 + 6 (3 b_2 - \beta) S_0^2 - 2 \beta S_z^2},\nonumber\\
\lambda_2&=&\frac{1}{18}\left[\beta-\kappa+\frac{81 m_T^2 a^2 + 6 (3 b_2 - \beta) S_0^2 }{81 m_T^2 a^2 + 6 (3 b_2 - \beta) S_0^2 - 2 \beta S_z^2}\right],\\
\tilde{\lambda}&=&\frac13 \frac{(3b_2-2\beta)\beta }{81 m_T^2 a^2 + 6 (3 b_2 - \beta) S_0^2 - 2 \beta S_z^2}S_0S_z\,.\nonumber
\eea
In order to ensure the presence of constraints, so that $S_0$ is not an independent propagating d.o.f., we need to solve the equation $\det \mathcal{H}_{ij}=0$ for arbitrary values of the fields. By solving this equation we recover the conditions $\beta=b_2=0$ and the Hessian reduces to 
\bea
\mathcal{H}_{ij}=
\left(
\begin{array}{ccc}
  0&0 &  0 \\
  0& -\frac{1}{18}\kappa &   0\\
 0 &0   &   \frac89\kappa
\end{array}
\right)\nonumber
\eea
that is trivially degenerate and ensures a non-propagating $S_0$. Moreover, we also see the ghostly nature of either $T_\mu$ or $S_\mu$ since the non-vanishing eigenvalues have opposite signs. These results indeed confirm the conclusions reached above from the study of the vector sector stability.

\subsection{Constructing stable Poincar\'e Gauge theories}

The precedent subsection has been devoted to showing the presence of ghosts in general quadratic PGTs. Although this is a drawback for generic theories, we will now show how to avoid the presence of the discussed instabilities by following different routes. In particular, we will show specific classes of ghost-free theories and how to stabilise the vector sector in the general PGT by adding suitable operators of the same dimensionality as those already present in the quadratic PGTs.

\subsubsection{$R^2$ theories}
Here we shall study the stability when we restrict the quadratic curvature sector to be exactly the square of the Ricci scalar of the full connection, \emph{i.e.} $R^2$. This theory will evidently have the $\Rb^{2}-$limit at vanishing torsion, but it avoids the ghostly interactions that originate from the other Riemann contractions as we show in the following. Therefore, we set the parameters of the PGT \eqref{actionPGT} to $b_2=b_3=b_4=b_5=b_6=0$ and $b_1\neq0$, so we will consider the particular PG Lagrangian
\begin{equation}
\mathcal{L}=a_{0}{R}+a_{1}T_{\mu\nu\rho}T^{\mu\nu\rho}
+a_{2}T_{\mu\nu\rho}T^{\nu\rho\mu}+a_{3}T_\mu T^\mu+b_{1}{R}^{2}.
\label{eq:PGactionR2}
\end{equation}
The matter content of this Lagrangian is the graviton plus a scalar field (which is the $0^{+}$ mode of the PG action). Its non-pathological behaviour was already found in \cite{Hecht:1996np,Yo:1999ex} by analysing its well-posedness and Hamiltonian structure. Our approach here will confirm these results by a different procedure and will give further insights. The idea is to rewrite the Lagrangian \eqref{eq:PGactionR2} in a way where we can see explicitly the additional scalar. As usual, we start by performing a Legendre transformation in order to recast the Lagrangian above in the more convenient form
\be
\Lag=a_{0}\varphi+b_1\varphi^2+\chi(R-\varphi) +\frac12m_T^2 T^2+\frac12m_S^2S^2,
\ee
where we have introduced the non-dynamical fields $\chi$ and $\varphi$ and we have neglected the pure tensor sector $q^\alpha{}_{\mu\nu}$ for the moment, although we will come back to its relevance later. Using the field equation for $\chi$ we can recover the original Lagrangian, while the equation for $\varphi$ yields
\be
\varphi=\frac{\chi-a_0}{2b_1}
\ee
that gives $\varphi$ as a function of $\chi$. We can now use the post-Riemannian expansion of the Ricci scalar, given by
\be
R=\Rb+\frac{1}{24}S^2-\frac{2}{3}T^2+2\nablab_\mu T^\mu,
\label{eq:Ricciscalar}
\ee
in order to express the Lagrangian in the following suitable form
\bea
\Lag=\chi\left(\Rb+\frac{1}{24}S^2-\frac{2}{3}T^2+2\nablab_\mu T^\mu\right)-\frac{\big(\chi-a_0\big)^2}{4b_1}+\frac12m_T^2T^2+\frac12 m_S^2S^2.
\label{eq:PGactionR22}
\eea
The equation for the axial part imposes $S_\mu=0$, while the trace part yields
\be
T_\mu=\frac{2\partial_\mu\chi}{m_T^2-\frac43 \chi}
\ee
which indeed shows that $T_\mu$ can only propagate a scalar\footnote{An analogous result was obtained in \cite{Ozkan:2015iva} by considering $f(R)$ theories where the Ricci scalar is replaced by $R\rightarrow R+A^2+\beta \nablab_\mu A^\mu$ with $A_\mu$ a vector field and in \cite{Jimenez:2015fva} within the context of geometries with vector distortion.} since the trace vector can be expressed as $T_\mu=\partial_\mu\tilde{\chi}$ with
\be
\tilde{\chi}=-\frac32\log\Big\vert3m_T^2-4\chi\Big\vert.
\ee
The theory is then equivalently described by the action
\be
\mS=\int\dd^4x\sqrt{-g}\left[\chi\Rb-\frac{2(\partial\chi)^2}{m_T^2-\frac{4}{3}\chi}-\frac{\big(\chi-a_0\big)^2}{4b_1}\right]
\label{eq:BDPGT}
\ee
which reduces to a simple scalar-tensor theory of a generalised Brans-Dicke type with a field-dependent Brans-Dicke parameter
\be
\omega_{\rm BD}(\chi)=\frac{2\chi}{m_T^2-\frac43\chi}\,.
\ee
The previous reasoning can be extended to arbitrary $f(R)$ extensions of PGTs, the only difference with respect to \eqref{eq:BDPGT} being the specific form of the potential for $\chi$. An interesting feature of the resulting Lagrangian is the singular character of the massless limit $m_T^2\rightarrow0$ that gives $\omega_{\rm BD}(m_T^2\rightarrow 0)=-3/2$, exactly the value that makes the scalar field non-dynamical. This is also the case for the Palatini formulation of $f(R)$ theories where the scalar is non-dynamical (see {\emph{e.g.}} \cite{Olmo:2011uz} and references therein). For any other value of the mass, the scalar field is fully dynamical. We can see this in detail by performing the following conformal transformation $\gt_{\mu\nu}=\frac{2\chi}{\mpl^2}g_{\mu\nu}$, that brings the action \eqref{eq:BDPGT} into the Einstein frame 
 \begin{align}
\mS=\int\dd^4x\sqrt{-\gt}&\left[a_0\Rt-\frac{3m_T^2a_0}{2\chi^2(m_T^2-\frac43\chi)}(\partial\chi)^2-\frac{a_0}{4  b_1}\left(1-\frac{a_0}{\chi}\right)^2\right].
\label{eq:R2Einstein}
 \end{align}
In this frame it becomes apparent that the scalar $\chi$ loses its kinetic term for $m_T^2=0$. This feature can be related to the breaking of a certain conformal symmetry by the mass term. If we perform a conformal transformation of the metric together with a projective transformation of the torsion\footnote{The torsion transformation is $T^\alpha{}_{\mu\nu}\rightarrow T^\alpha{}_{\mu\nu}-2\delta^\alpha_{[\mu}\partial_{\nu]}\Omega$ that gives the transformation for the vector trace quoted in the main text, while the axial and pure tensor pieces remain invariant. See e.g. \cite{Obukhov:1982zn,HelayelNeto:1999tm,Shapiro:2001rz} for interesting discussions on conformal transformations involving torsion.} given by
\be
 g_{\mu\nu}\rightarrow {\rm e}^{2\Omega}g_{\mu\nu},\quad  T_\mu\rightarrow T_\mu+3\partial_\mu\Omega,
 \ee
with $\Omega$ being an arbitrary function, we have that the Ricci scalar transforms as $R\rightarrow e^{-2\Omega}R$. Consequently, we have that the only term in the Lagrangian \eqref{eq:PGactionR22} that is not invariant under the above transformations, supplemented with $\chi\rightarrow e^{-2\Omega}\chi$, is the mass term\footnote{Actually, the potential for $\chi$ also breaks the conformal invariance, but since it does not affect the dynamical nature of $\chi$ we can neglect it for this discussion.}. Thus, for $m_T^2=0$, the fact that the torsion is given in terms of the gradient of $\chi$  together with the discussed symmetry allows to completely remove the kinetic terms for $\chi$ by means of a conformal transformation. The mass however breaks this symmetry and, consequently, we recover the dynamical scalar described by \eqref{eq:R2Einstein}. Furthermore, the mass $m_T^2$ also determines the region of ghost freedom for the theory. If $m_T^2>0$ we have an upper bound for the scalar field that must satisfy $\chi<\frac34m_T^2$ in order to avoid the region where it becomes a ghost. On the other hand, if $m_T^2<0$, the scalar field is confined to the region $\chi>\frac34m_T^2$. For the potential to be bounded from below we only need to have $b_1>0$. These conditions have been summarised in Table \ref{table}.
 
It is worth noticing that the absence of ghosts in the $R^2$-theories is due to the removal of the Maxwell kinetic terms for the vector sector, hence avoiding its propagation. By inspection of the Ricci scalar \eqref{eq:Ricciscalar} we see that only the trace $T_\mu$ enters with derivatives and only through the divergence $\nablab_\mu T^\mu$. As it is well-known this is precisely the dual of the usual Maxwell-like kinetic term for the dual 3-form field so the theory can be associated to a massive 3-form which propagates one dof\footnote{See e.g. \cite{Germani:2009iq, Koivisto:2009sd,Koivisto:2009ew} for some cosmological applications of 3-forms.}. This dof can be identified with the scalar that we have found. Just like the $U(1)$ gauge symmetry of the Maxwell terms is crucial for the stability of vector theories, the derivative term $\nablab_\mu T^\mu$ has the symmetry $T^\mu\rightarrow T^\mu+\epsilon^{\mu\nu\rho\sigma}\partial_{\nu} \theta_{\rho\sigma}$ for an arbitrary $\theta_{\rho\sigma}$ that plays a crucial role for guaranteeing the stability of the theories. Of course, this symmetry is inherited from the gauge symmetry of the dual 3-form.

Let us finally highlight that the inclusion of the tensor sector $q^\rho{}_{\mu\nu}$ does not change the final result because one can check that, similarly to the axial part, it only enters as an auxiliary field whose equation of motion imposes $q^\alpha{}_{\mu\nu}=0$. To see this more clearly, we can give the full post-Riemannian expansion of the Ricci scalar including the tensor piece
\be
R=\Rb+\frac{1}{24}S^2-\frac{2}{3}T^2+2\nablab_\mu T^\mu+\frac12q_{\mu\nu\rho}q^{\mu\nu\rho}.
\ee
It is clear then that the contribution of the tensor part to the Lagrangian \eqref{eq:PGactionR22} gives rise to the equation of motion $\chi q_{\mu\nu\rho}=0$ which, for $\chi\neq0$, trivialises the tensor component. The same will apply to theories described by an arbitrary function $f(R)$ so one can safely neglect the tensor sector for those theories as well.

\subsubsection{Holst square theories}
We have just seen how to obtain a non-trivial quadratic PG theory that propagates an extra scalar, and how this can be ultimately related to the absence of Maxwell-like terms for the vector sector. We can then ask whether there is some non-trivial healthy theory described by \eqref{eq:PGaction2} where the scalar is associated to the axial vector rather than to the trace. The answer is indeed affirmative, and in order to prove such a result we simply need to impose the vanishing of the Maxwell kinetic terms that results in the following conditions
\be
\kappa=0\quad {\text{and}}\quad \beta=0.
\ee
Imposing these conditions, performing a few integrations by parts and dropping the Gauss-Bonnet term, the Lagrangian then reads
\bea
\Lag_{\rm Holst}&=&a_0\Rb+\frac12m_T^2 T^2+\frac12m_S^2S^2
\nonumber
\\
&&+\alpha\left[(\nablab_\mu S^\mu)^2-\frac43 S_\mu T^\mu\nablab_\nu S^\nu+\frac49(S_\mu T^\mu)^2\right], 
\label{eq:Holst1}
\eea
with $\alpha\equiv-\frac{b_2}{4}$. It is clear that we now obtain the same structure as in the $R^2$ case but now for the axial part. This is not an accidental property, and it can be derived from the relation of the resulting Lagrangian with the Holst term\footnote{Although this term is commonly known as the Holst term, due to the research article of Soren Holst in 1995 \cite{Holst:1995pc}, in the context of torsion gravity it was first introduced by R. Hojman {\emph{et. al.}} in 1980 \cite{Hojman:1980kv}.} \cite{Hojman:1980kv,Holst:1995pc} that is given by $\mathcal{H}\equiv\epsilon^{\mu\nu\rho\sigma} R_{\mu\nu\rho\sigma}$ and whose post-Riemannian expansion is
\be
\mathcal{H}=\frac23 S_\mu T^\mu-\nablab_\mu S^\mu
\label{eq:HoldstPR}
\ee
where we have used that $\epsilon^{\mu\nu\rho\sigma} \Rb_{\mu\nu\rho\sigma}=0$ by virtue of the Bianchi identities. Thus, it is obvious that the Lagrangian can be expressed as
\be
\Lag_{\rm Holst}=a_0\Rb+\frac12m_T^2 T^2+\frac12m_S^2S^2
+\alpha \mathcal{H}^2 .
\label{eq:Holst2}
\ee
This particular PG theory was identified in \cite{Hecht:1996np} as an example of a theory with dynamical torsion described by a scalar with a well-posed initial value problem. We will understand the nature of this scalar by following an analogous aproach to the $R^2$ theories. For that purpose, we first introduce an auxiliary field $\phi$ to rewrite \eqref{eq:Holst2} as
\bea
\Lag_{\rm Holst}=a_0\Rb+\frac12m_T^2 T^2+\frac12m_S^2S^2-\alpha\phi^2+2\alpha\phi\epsilon^{\mu\nu\rho\sigma} R_{\mu\nu\rho\sigma}.
\label{eq:Holst3}
\eea
We see that the resulting equivalent Lagrangian corresponds to the addition of a Holst term where the Barbero-Immirzi parameter acts as a pseudo-scalar field. As we shall show now, this pseudo-scalar is dynamical and corresponds to the $0^-$ mode in the PG Lagrangian identified in \cite{Hecht:1996np}. The massless theory with $m_ T^2=m_S^2=0$ and without the $\phi^2$ potential has been considered in extensions of GR inspired by Loop Quantum Gravity \cite{Taveras:2008yf,Calcagni:2009xz}. At this moment, we can introduce the post-Riemannian expansion \eqref{eq:HoldstPR} into the Lagrangian, obtaining
\bea
\Lag_{\rm Holst}=a_0\Rb+\frac12m_T^2 T^2+\frac12m_S^2S^2-\alpha\phi^2+2\alpha\phi\left(\frac23 S_\mu T^\mu-\nablab_\mu S^\mu\right).
\label{eq:Holst4}
\eea
The correspondent equations for $S^\mu$ and $T^\mu$ are 
\bea
m_S^2S_\mu+\frac{4\alpha\phi}{3}T_\mu+2\alpha\partial_\mu\phi=0,\label{eq:S}\\
m_T^2T_\mu+\frac{4\alpha\phi}{3}S_\mu=0,
\eea
respectively. For $m_T^2\neq0$\footnote{The singular value $m_T^2=0$ leads to uninteresting theories where all the dynamics is lost so we will not consider it any further here. The same conclusion was reached in \cite{Hecht:1996np}.} we can algebraically solve these equations as
\bea
T_\mu&=&-\frac{4\alpha\phi}{3m_T^2}S_\mu,\\
S_\mu&=&-\frac{2\alpha\partial_\mu\phi}{m_S^2-\left(\frac{4\alpha\phi}{3m_T}\right)^2},\label{eq:solS}
\eea
that we can plug into the Lagrangian to finally obtain
\be
\Lag_{\rm Holst}=a_0\Rb-\frac{2\alpha^2}{m_S^2-\left(\frac{4\alpha\phi}{3m_T}\right)^2}(\partial\phi)^2-\alpha\phi^2.
\label{eq:Holst4}
\ee
This equivalent formulation of the theory where all the auxiliary fields have been integrated out explicitly shows the presence of a propagating pseudo-scalar field. The parity invariance of the original Lagrangian translates into a $\mathbb{Z}_2$ symmetry in the pseudo-scalar sector. The obtained result is also valid for theories described by an arbitrary function of the Holst term, where the effect of considering different functions leads to different potentials for the pseudo-scalar $\phi$. \\
Moreover, we can see how including the pure tensor part $q^\rho{}_{\mu\nu}$ into the picture does not change the conclusions because it contributes to the Holst term as
\be
\mathcal{H}=\frac23 S_\mu T^\mu-\nablab_\mu S^\mu+\frac12\epsilon_{\alpha\beta\mu\nu}q_\lambda{}^{\alpha\beta}q^{\lambda\mu\nu}.
\ee
This shows that $q_{\rho\mu\nu}$ only enters as an auxiliary field whose equation of motion trivialises it, as it occurs for the $R^2$ case.

At this moment, let us point out how the appearance of a (pseudo-)scalar could have been expected by using the relation of the Holst term with the Nieh-Yan topological invariant $\mathcal{N}$, that is given by
\be
\mathcal{N}\equiv\epsilon^{\mu\nu\rho\sigma}\Big(R_{\mu\nu\rho\sigma}-\frac12 T^\alpha{}_{\mu\nu}T_{\alpha\rho\sigma}\Big).
\ee
In a Riemann-Cartan spacetime it is easy to show that this term is nothing but the total derivative $\mathcal{N}=-\nablab_\mu S^\mu$. The remarkable property of this invariant is that it is linear in the curvature so its square must belong to the class of parity preserving quadratic PG theories, even though $\mathcal{N}$ itself breaks parity. Then, as it happens with other invariants like the Gauss-Bonnet one, including a general non-linear dependence on the invariant is expected to give rise to dynamical scalar modes. In standard Riemannian geometries for example, the inclusion of an arbitrary function of the Gauss-Bonnet invariant results in a highly non-trivial scalar field with Horndeski interactions \cite{Kobayashi:2011nu}.

The stability constraints on the parameters can now be obtained very easily. From \eqref{eq:Holst4} we can realise that $\alpha$ must be positive to avoid having an unbounded potential from below. On the other hand, the condition to prevent $\phi$ from being a ghost depends on the signs of $m_S^2$ and  $m_T^2$,  which are not defined by any stability condition so far. We can distinguish the following possibilities:
\begin{itemize}
\item $m_S^2>0$: We then need to have $1-\left(\frac{4\alpha\phi}{3m_Tm_S}\right)^2>0$. For $m_T^2<0$ this is always satisfied, while for $m_T^2>0$ there is an upper bound for the value of the field given by $\vert\phi\vert<\vert\frac{3m_Sm_T}{4\alpha}\vert$.

\item $m_S^2<0$: The ghost-freedom condition is now $1-\left(\frac{4\alpha\phi}{3m_Tm_S}\right)^2<0$, which can never be fulfilled if $m_T^2>0$. If $m_T^2<0$ we instead have the lower bound $\vert\phi\vert>\vert\frac{3m_Sm_T}{4\alpha}\vert$.
\end{itemize}

For a better visualisation we have summarised these ghost-free conditions in Table \ref{table}. We can gain a better intuition on the dynamics of the pseudo-scalar by canonically normalising it. For that purpose we introduce a new field $\hat{\phi}$ defined by
\be
\hat{\phi}=\frac{2\alpha}{\sqrt{m_S^2}}\int\frac{\dd\phi}{\sqrt{1-\left(\frac{4\alpha\phi}{3m_Tm_S}\right)^2}}.
\label{eq:canfield}
\ee
Again, depending on the sign of $m_S^2$ we have two cases
\begin{itemize}

\item For $m_S^2>0$ we obtain
\be
\phi(\hat{\phi})=\frac{3m_Tm_S}{4\alpha}\sin\left(\frac{2\hat{\phi}}{3m_T}\right),
\label{eq:phican1}
\ee
in terms of which the Lagrangian for the pseudo-scalar reads
\be
\Lag_{\hat{\phi}}\vert_{m_S^2>0}=-\frac12(\partial\hat{\phi})^2-V(\hat{\phi}),
\ee
with $V(\hat{\phi})=\alpha\phi^2(\hat{\phi})$. We can see how the shape of the potential crucially depends on the sign of $m_T^2$. If $m_T^2>0$ we have the following oscillatory potential
\be
V(\hat{\phi})=\frac{9 m_T^2m_S^2}{16\alpha}\sin^2\left(\frac{2\hat{\phi}}{3 m_T}\right),\quad\quad m_S^2>0, \;m_T^2>0,
\ee
which has a discrete symmetry $\hat{\phi}\rightarrow\hat{\phi}+\frac32nm_T\pi$ with $n\in\mathbb{Z}$ arising from the original upper bound of $\phi$. Notice that the field redefinition \eqref{eq:phican1} guarantees the ghost-free condition $\vert\phi\vert\leq\vert\frac{3m_Tm_S}{4\alpha}\vert$.\\
For $m_T^2<0$ the potential takes instead the form
\be
V(\hat{\phi})=\frac{9\vert m_T^2\vert m_S^2}{16\alpha}\sinh^2\left(\frac{2\hat{\phi}}{3\vert m_T\vert} \right),\quad\quad m_S^2>0, \;m_T^2<0.
\ee

\item On the other hand, for $m_S^2<0$, we neccesarily need to have $m_T^2<0$ to avoid ghosts, and the integral \eqref{eq:canfield} gives
\be
\phi=\pm\frac{3\vert m_T m_S\vert}{4\alpha}\cosh\left(\frac{4\hat{\phi}}{3\vert m_T\vert}\right),
\ee
where we have fixed the integration constant so that the origin of $\hat{\phi}$ corresponds to the lower bound for $\vert\phi\vert$. The Lagrangian for the canonically normalised field is given by
\be
\Lag_{\hat{\phi}}\vert_{m_S^2<0}=-\frac12(\partial\hat{\phi})^2-\frac{9 m_T^2m_S^2}{16\alpha}\cosh^2\left(\frac{2\hat{\phi}}{3 \vert m_T\vert}\right),\quad\quad m_S^2<0,\; m_T^2<0.
\ee
\end{itemize}
In all cases, it is straightforward to analyse the corresponding solutions by simply looking at the shape of the corresponding potential. In particular, we see that the small-field regime gives an approximate quadratic potential so, provided the mass is sufficiently large\footnote{By large we of course mean relative to the Hubble parameter in the late time universe so that the field can undergo multiple oscillations around the minimum in a Hubble time. This typically requires masses around $m\sim10^{-22}$ eV so they actually represent ultra-light particles from a particle physics perspective.}, the coherent oscillations of the pseudo-scalar can give rise to dark matter \cite{Turner:1983he,PhysRevLett.64.1084,Cembranos:2015oya,Hui:2016ltb} as the misalignment mechanism for axions \cite{Marsh:2015xka} or the Fuzzy Dark Matter models \cite{Hu:2000ke}. A similar mechanism was explored in \cite{Cembranos:2008gj} within pure $R^2$ gravity. On the other hand, it is also possible to generate large-field inflationary scenarios or dark energy models if the field slowly rolls down the potential at field values sufficiently far from the minimum.

One important difference with respect to the $R^2$ theories discussed earlier is that here we have obtained the Lagrangian for the pseudo-scalar already in the Einstein frame, while this was only achieved after performing a conformal transformation to disentangle the scalar field from the Einstein-Hilbert term for the $R^2$ theories. Hence, while the scalar couples directly to matter in the Einstein frame through a conformal metric for the $R^2$ theories, the pseudo-scalar field of the Holst square theories does not. This could be useful for dark matter and/or dark energy models because they could easily evade local gravity constraints. Actually, the obtained effective potential for the pseudo-scalar field allows for both accelerating cosmologies (that could be used for dark energy or inflation) and dark matter dominated universes. 

Nevertheless, we have to take into account that Dirac fermions do couple to the axial part of the connection (see {\emph{e.g.}} \cite{Hammond:2002rm,Shapiro:2001rz}). A consequence of this kind of coupling coupling is that we would expect to have the dual of the hypermomentum $\Delta_\mu=\delta\mS/\delta S^\mu$ entering on the r.h.s. of \eqref{eq:S}. This means that the solutions for $S_\mu$ and $T_\mu$ in \eqref{eq:solS} should include $\Delta_\mu$ so that the final Lagrangian \eqref{eq:Holst4} features the coupling between the pseudo-scalar $\phi$ and Dirac fermions. Since $\Delta_\mu$ in the equations can be simply generated by the replacement $2\alpha\partial_\mu\phi\rightarrow 2\alpha\partial_\mu\phi+\Delta_\mu$ in \eqref{eq:S}, the explicit computation of the interactions including the axial coupling to the fermions can be easily obtained by making the corresponding replacement in \eqref{eq:Holst4}, namely
\be
\Lag_{\rm Holst}=a_0\Rb-\frac{(2\alpha\partial_\mu\phi+\Delta_\mu)^2}{m_S^2-\left(\frac{4\alpha\phi}{3m_T}\right)^2}-\alpha\phi^2.
\label{eq:Holst5}
\ee
We then obtain the usual four-point fermion interactions given by $\Delta^2$ that are also generated in {\emph{e.g.}} Einstein-Cartan gravity plus a derivative coupling of the pseudo-scalar to the axial current $\Delta_\mu$ carried by the fermions. Interestingly, this derivative coupling can yield to an effective mass for the fermion\footnote{Let us recall that the axial current for a fermion $\psi$ has the form $\Delta_\mu\propto \bar{\psi}\gamma_5\gamma_\mu\psi$ so the derivative coupling indeed generates an effective mass.} that depends on the evolution of the pseudo-scalar. Therefore, it is worth noting the possibility that this scenario offers for a natural framework to have dark energy and/or dark matter interacting with neutrinos that could result in some interesting phenomenologies for their cosmological evolution. On the other hand, these couplings could also give rise to natural reheating mechanisms within inflationary models.

\renewcommand{\arraystretch}{1.5}
\begin{table}\centering
\begin{tabular}{ccc||c||c||c||c||c||c||c||c||c||ccc}
\hline 
 & Scalar $\chi$& \multicolumn{11}{c}{} & \multicolumn{2}{c}{Pseudo-scalar $\phi$}\tabularnewline
\cline{2-2} \cline{14-15} 
 & $b_{1}>0$ & \multicolumn{11}{c}{} & $m_{S}^{2}>0$ & $m_{S}^{2}<0$\tabularnewline
\cline{2-15} 
$m_{T}^{2}>0$ & $\chi<\frac{3}{4}m_{T}^{2}$ & \multicolumn{11}{c}{} & $|\phi|<\left|\frac{3m_{S}m_{T}}{4\alpha}\right|$ & {\color{red}Ghost}\tabularnewline
$m_{T}^{2}<0$ & $\chi>\frac{3}{4}m_{T}^{2}$ & \multicolumn{11}{c}{} & {\color{green}Healthy} & $|\phi|>\left|\frac{3m_{S}m_{T}}{4\alpha}\right|$\tabularnewline
\hline 
\end{tabular}
\caption{This table summarises the conditions to avoid ghosts for the scalar and the pseudo-scalar fields.}
\label{table}
\end{table}

\renewcommand{\arraystretch}{1}

\subsubsection{The general healthy bi-scalar theory}

For completeness, we shall analyse the theory that propagates simultaneously both the scalar and pseudo-scalar fields obtained above. It should be clear that the corresponding theory will be described by the Lagrangian
\be
\Lag=a_{0}R+\frac12m_T^2T^2+\frac12m_S^2S^2+b_1R^2+\alpha\mathcal{H}^2.
\ee
The matter content of this Lagrangian is indeed the graviton plus the $0^+$ and $0^-$ modes present in the PG action. We will proceed analogously to the previous cases, that is by introducing auxiliary fields, but we shall overlook the unnecessary details. The transformed Lagrangian in the post-Riemannian expansion can then be written as
\bea
\Lag&=&\mU(\chi,\phi)+\chi\Rb+\frac12M_T^2(\chi)T^2+\frac12M_S^2(\chi)S^2+\frac43\alpha\phi S_\mu T^\mu
\nonumber
\\
&&-2T^\mu\partial_\mu\chi+2\alpha S^\mu\partial_\mu\phi\,,
\label{eq:bisint}
\eea
where we have defined
\bea
\mU(\chi,\phi)=-\frac{\big(\chi-a_0\big)^2}{4b_1}-\alpha\phi^2,\quad M_T^2=m_T^2-\frac43\chi\quad\text{and}\quad 
M_S^2=m_S^2+\frac{1}{12}\chi.
\label{biscalarfunctions}
\eea
We can rewrite the Lagrangian \eqref{eq:bisint} in a more compact and useful way by using matrices as
\be
\Lag=\mU(\chi,\phi)+\chi\Rb+\frac12\vec{Z}^t\hat{M}\vec{Z}+\vec{Z}^t\cdot\vec{\Phi}
\label{eq:bisint2}
\ee
with $\vec{Z}^t=(T_\mu,S_\mu)$, $\vec{\Phi}^t=(-2\partial_\mu\chi,2\alpha\partial_\mu\phi)$ and
\bea
\hat{M}=
\left(
\begin{array}{cc}
  M_T^2(\chi)&\frac43\alpha\phi  \\
  \frac43\alpha\phi&M_S^2(\chi)  \\
\end{array}
\right).
\eea
The equations for $S^\mu$ and $T^\mu$ can then be expressed as
\be
\hat{M}\vec{Z}=-\vec{\Phi}\;\Longrightarrow\;\vec{Z}=-\hat{M}^{-1}\vec{\Phi},
\ee
where the inverse of $\hat{M}$ given by
\bea
\hat{M}^{-1}=\frac{1}{M_S^2(\chi)M_T^2(\chi)-\left(\frac43\alpha\phi\right)^2}
\left(
\begin{array}{cc}
  M_S^2(\chi)&-\frac43\alpha\phi  \\
  -\frac43\alpha\phi&M_T^2(\chi)  \\
\end{array}
\right).
\eea
By plugging this solution into the Lagrangian \eqref{eq:bisint2} we finally obtain
\be
\Lag=\mU(\chi,\phi)+\chi\Rb-\frac12\vec{\Phi}^t\hat{M}^{-1}\vec{\Phi}.
\ee
It is quite clear that, as we have already pointed out, the theory describes two propagating scalars. At this stage, we can undo the above compact form of the Lagrangian to make everything more explicit
\be
\Lag=\chi\Rb+6\frac{3M_S^2(\chi)(\partial\chi)^2+3\alpha^2M_T^2(\chi)(\partial\phi)^2-8\alpha^2\phi\partial_\mu\phi\partial^\mu\chi}{(4\alpha\phi)^2-9M_S^2(\chi)M_T^2(\chi)}+\mU(\chi,\phi).
\label{eq:biscalar3}
\ee
It is easy to see that, as expected, this Lagrangian reduces to \eqref{eq:BDPGT} for $\phi=0$ and to \eqref{eq:Holst4} for $\chi=0$ (except for the Einstein-Hilbert term that should be added). Of course, the general discussions for the $R^2$ and Holst square theories also apply to the present case. We observe in \eqref{eq:biscalar3} that the scalar $\chi$ exhibits a non-minimal coupling that can be removed by means of the same conformal transformation as before $\gt_{\mu\nu}=\frac{\chi}{a_0}g_{\mu\nu}$. After performing this transformation to the Einstein frame the Lagrangian \eqref{eq:biscalar3} reads
\bea
\Lag&=&a_0\tilde{R}-\left[1-\frac{12M_S^2(\chi)}{(4\alpha\phi)^2-9M_S^2(\chi)M_T^2(\chi)}\right](\partial\chi)^2\nonumber\\
&&+\frac{6a_0}{\chi}\frac{3\alpha^2M_T^2(\chi)(\partial\phi)^2-8\alpha^2\phi\partial_\mu\phi\partial^\mu\chi}{(4\alpha\phi)^2-9M_S^2(\chi)M_T^2(\chi)}+\left(\frac{a_0}{\chi}\right)^2\mU(\chi,\phi).
\label{eq:biscalar2}
\eea
Once again, the conformal transformation will couple $\chi$ directly to matter through the conformal metric, while the pseudo-scalar $\phi$ couples only to the axial fermionic current given by the dual of the corresponding hypermomentum. The same reasoning used to obtain \eqref{eq:Holst5} applies here, so this axial coupling eventually generates couplings achievable via the replacement $2\alpha\partial_\mu\phi\rightarrow 2\alpha\partial_\mu\phi+\Delta_\mu$ in \eqref{eq:biscalar2}. Notice that additional couplings between $\chi$ and fermions will be generated by this mechanism. The resulting Lagrangian \eqref{eq:biscalar2} resembles a two dimensional non-linear sigma model \cite{GellMann:1960np}, with the following target space metric
\bea
h_{ij}(\chi,\phi)=\frac{2\mpl^2}{\chi}
\left(
\begin{array}{cc}
\frac{3}{4}+ \frac{1}{\chi}(\hat{M}^{-1})_{11} &\alpha(\hat{M}^{-1})_{12}  \\
  \alpha(\hat{M}^{-1})_{12}&\alpha^2(\hat{M}^{-1})_{22}  \\
\end{array}
\right).
\eea
Nevertheless, this resemblance is only formal at this point due to the pseudo-scalar nature of $\phi$. The ghost-free conditions are obtained by imposing the positivity of the eigenvalues of this metric, whose expressions are more involved in this case because of the couplings between both scalars. A much simpler condition can be obtained by computing the determinant
\be
\det h_{ij}=\frac{4a_{0}^2\alpha^2}{\chi^3}\frac{3M_T^2(\chi)+4\chi}{M_S^2(\chi)M_T^2(\chi)-\left(\frac43\alpha\phi\right)^2}>0,
\ee
which clearly is a necessary condition to guarantee ghost-freedom, although it is not sufficient. Moreover, having $\det h_{ab}=0$ will determine the degenerate cases where the phase space is reduced. This happens trivially for $\alpha=0$, that corresponds to the pure $R^2$ theory. The pure Holst square limit is more complicated to obtain because the conformal transformation becomes singular for $\chi=0$. We shall not explore further the general bi-scalar theory, although it should be clear that such theories will contain a much richer structure due to its enlarged phase space.\\

We will end our discussion by explicitly showing how our results can be extended to theories described by a general function of $R$ and $\mathcal{H}$. For that, let us then consider the following Lagrangian
\be
\Lag=F(R,\mH,T,S,q),
\ee
where $F$ is some arbitrary scalar function. In addition, for the sake of generality, we have allowed an arbitrary dependence on the torsion invariants as well. The Lagrangian can be recast as
\be
\Lag=F(\tilde{\chi},\tilde{\phi},T,S,q)+\chi \big(R-\tilde{\chi}\big)+\phi\big(\mH-\tilde{\phi}\big)
\ee
where we have introduced a set of auxiliary fields, following the same reasoning as in all the previous cases. The equations for $\tilde{\chi}$ and $\tilde{\phi}$ allow to express these fields in terms of the rest of fields. Therefore, we can write
\bea
\Lag&=&\mU(\chi,\phi,T,S,q)+\chi \left(\Rb+\frac{1}{24}S^2-\frac{2}{3}T^2+2\nablab_\mu T^\mu+\frac12q_{\mu\nu\rho}q^{\mu\nu\rho}\right)\nonumber\\
&&+\phi\left(\frac23 S_\mu T^\mu-\nablab_\mu S^\mu+\frac12\epsilon_{\alpha\beta\mu\nu}q_\lambda{}^{\alpha\beta}q^{\lambda\mu\nu}\right),
\label{eq:biscalargen}
\eea
where the potential $\mU$ already includes the effects of integrating out $\tilde{\chi}$ and $\tilde{\phi}$. Again, we see that the pure tensor sector only enters as an auxiliary field so we can also integrate it out to finally express the Lagrangian \eqref{eq:biscalargen} as
\be
\Lag=\tilde{\mU}(\chi,\phi,T,S)+\chi\Rb-2T^\mu\partial_\mu\chi+S^\mu\partial_\mu\phi
\label{eq:biscalargen2}
\ee
where $\tilde{\mU}$ contains all the terms without derivatives. This Lagrangian resembles \eqref{eq:bisint} with the only difference that the non-derivative terms are different. We can then proceed analogously by integrating out the vector sector $T_\mu$ and $S_\mu$ by solving their equations of motion
\bea
\frac{\partial\tilde{\mU}}{\partial T^\mu}-2\partial_\mu\chi=0\, ,\\
\frac{\partial\tilde{\mU}}{\partial S^\mu}+\partial_\mu\phi=0\, ,
\eea
that will give $T_\mu=T_\mu(\chi,\phi,\partial\chi,\partial\phi)$ and $S_\mu=S_\mu(\chi,\phi,\partial\chi,\partial\phi)$. By plugging these solutions back in the Lagrangian \eqref{eq:biscalargen2} we finally arrive at the explicit bi-scalar theory, but now with more involved interactions that will depend on the specific function $F$ describing the Lagrangian. If we include couplings to fermions, we can use the same trick as before to take such inclusion into account.

\subsubsection{Adding dimension 4 operators}

We have studied how to constrain the parameters in order to remove the ghosts of the quadratic PGTs. We shall at this point discuss how to avoid the ghosts by extending the Lagrangian in a suitable form. For this purpose, it is worth noting that the constructed quadratic theory \eqref{actionPGT} contains up to dimension 4 torsion terms that come from the curvature squared terms. It would then seem natural to include all the operators up to that dimensionality. For instance, since the Riemann squared terms generate quartic interactions for the torsion, there seems not to be a reason why they should not be included from the construction of the theory, apart from following the usual Yang-Mills approach. If we do allow for all the operators up to dimension four, there are many additional torsion terms that one could add. Particularly, we can include the operators $\mT_{\mu\nu} \mT^{\mu\nu}$ and $\mS_{\mu\nu} \mS^{\mu\nu}$ modulated by arbitrary coefficients. With the addition of these terms, it is trivial to see that the unavoidable ghostly nature of the vector sector concluded above by removing dangerous non-minimal couplings is resolved. Moreover, since these are just standard Maxwell terms, they will tackle the ghosts issue without introducing new potentially pathological interactions for the vector sector and affecting the pure tensor sector.

Once the presence of arbitrary dimension 4 operators is allowed, we can also include other phenomenologically interesting interactions. In particular, we can add non-minimal couplings that do not spoil the stabilisation achieved by including the already mentioned Maxwell terms. For instance, we can introduce interactions that mix the curvature and the torsion. Generically, these interactions will be pathological. However, there is a class of operators that gives rise to non-pathological non-minimal couplings for the vector sector. That is the case of $G_{\mu\nu} T^\mu T^\nu$, which generates the following couplings in the post-Riemannian expansion
\be
\Lag\supset \Gb _{\mu\nu} T^\mu T^\nu-T^2\nablab_\mu T^\mu+\frac13 T^4-\frac{1}{144}S^2 T^2-\frac{1}{72} (S_\mu T^\mu)^2.
\ee
Such a class of Lagrangians includes the non-minimal coupling to the Einstein tensor and a vector-Galileon term for the vector trace. Nevertheless, one would need to take into account that, if the tensor piece is included, some other worrisome terms will also enter which could potentially jeopardise the stability of the vector sector. 

\section{Chapter conclusions and outlook}

In this section we shall expose the main results of this chapter and outline the possible applications. Within this chapter we have first introduced the mathematical foundations of any gravitational theory, and showed with explicit examples that there is no physical reason to assume a priori that the Levi-Civita connection is the affine structure of the spacetime.

In the second section we have explained how one can construct physical theories by imposing the invariance of the action under local symmetries, which is known as the gauge procedure. Using such a procedure we have constructed the gravitational gauge theory of the translation group, which is TEGR, and the gauge theory of the Poincaré group, PG gravity.
 
Finally, in section \ref{2.3}, we have studied the stability of the latter theory, showing that only the two scalar modes present in the general PG Lagrangian can propagate safely. We then give details of the possible PG theories that can be considered using those two scalars. We also comment that by introducing dimension 4 torsion terms in the Lagrangian could allow the propagation of the vector modes without introducing any pathological behaviour.\\

Based on the previous results it will be of interest to study the cosmological and astrophysical solutions of the stable PG theories. As a matter of fact, in the next chapter we will study, among other things, the possible black-hole solutions of such stable scenarios. Moreover, another application of the findings of the current chapter can be seen in chapter \ref{4}, where the local limit of the proposed non-local theory contains suitable 4 dimensional torsion terms in the Lagrangian, hence allowing the stable propagation of the vector modes

\renewcommand{\publ}{}

\chapter{Phenomenology of Poincar\'e Gauge Theories}

\label{3}

\PARstart{S}ince the inception of Poincar\'e Gauge gravity, the different attempts to extend the properties and theorems of GR have been quite an active field. Paradigmatic examples include the study of singularities~\cite{Stewart:1973ux,Trautman:1973wy,Cembranos:2016xqx,Cembranos:2019mcb,delaCruz-Dombriz:2018aal}, the Birkhoff theorem \cite{Neville:1979fk,Rauch:1981tva,delaCruz-Dombriz:2018vzn}, existence of exact solutions \cite{Bakler:1984cq,Obukhov:1987tz,Blagojevic:2015zma,Cembranos:2016gdt,Obukhov:2019fti,Ziaie:2019dmq}, cosmological models \cite{Kerlick:1975tr,Yo:2006qs,Shie:2008ms,Chen:2009at,Baekler:2010fr,Ho:2015ulu}, the motion of particles~\cite{Hehl:2013qga,Cembranos:2018ipn} and, as we have already studied, the analysis of their stability~\cite{Sezgin:1979zf,Sezgin:1981xs,Chern:1992um,Yo:1999ex,Yo:2001sy,Vasilev:2017twr,Jimenez:2019qjc}. In this chapter we will present our results in different aspects of the aforementioned phenomenology, and shall be structured as follows.\\

In the previous chapter we mentioned that the axial vector of the torsion couples to the internal spin of fermions, {\emph{i.e.}} half-spin particles. This clearly induces a non-geodesical behaviour, which we shall calculate in Section \ref{3.1}, based on the work {\bf{\nameref{P3}}}. The fact that the fermions do not follow geodesics makes us think that they could escape somehow from the spacetime singularities, since the classical singularity theorems are formulated in terms of null and timelike geodesics. In Section \ref{3.2}, based on results presented in {\bf{\nameref{P1}}} and {\bf{\nameref{P5}}},  we shall show how such a scenario is not possible for spacetimes with a black hole regions of any dimension. Consequently, in Section \ref{3.3} we will study what kind of black-hole solutions we can expect in PG gravity by exploring the Birkhoff and no-hair theorems in different scenarios. This latter section is based on {\bf{\nameref{P2}}}. 

\clearpage

\section{Fermion dynamics}
\label{3.1}

Due to the coupling of the axial vector part of the torsion $S^{\mu}$ with the internal spin of fermions, it is clear that these particles would move along timelike curves that are not geodesics. While there is consensus on this fact, there is still an ongoing debate on which is the actual trajectory that they follow. Here we shall outline the most relevant ones (for a comprehensive review {\emph{cf.}} \cite{Hehl:2013qga}):
\begin{itemize}
\item In 1971, Ponomariev~\cite{Ponomarev:1971zz} proposed that the test particles will move along autoparallels (curves in which the velocity is parallel transported along itself with the total connection). Although there was no reason given, surprisingly this has been a recurrent proposal in the subsequent literature~\cite{Kleinert:1998cz, Mao:2006bb}. 
\item Hehl~\cite{HEHL1971225}, also in 1971, obtained the equation of motion using the energy-momentum conservation law, in the single-point approximation, {\emph{i.e.}} only taking into account first order terms. He also pointed out that torsion could be measured by using half-spin particles.
\item In 1981, Audretsch~\cite{Audretsch:1981xn} analysed the movement of a Dirac electron in a spacetime with torsion. He employed the WKB approximation, and obtained the same results as Rumpf had obtained two years earlier via an unconventional quantum mechanical approach~\cite{bergmann2012cosmology}. It was with this article that the coupling between spin and torsion was understood. 
\item In 1991, Nomura, Shirafuji and Hayashi~\cite{nomura1991spinning} computed the equations of motion by the application of the Mathisson-Papapetrou method to expand the energy-momentum conservation law. They obtained the equations at first order, which are the ones that Hehl had already calculated, but also made the second order approximation, finding the same spin precession as Audretsch.
\end{itemize}

In the following, we shall focus on Audretsch's approach, since it is the only procedure that takes into account the quantum mechanical nature of fermions.

\subsection{WKB approximation}

In this subsection we will outline the work of Audretsch in \cite{Audretsch:1981xn}, where the precession of spin and the trajectories of fermionic particles in theories with torsion were calculated. In order to do so, we shall start with the Dirac field equation of a fermionic field minimally coupled to torsion
\begin{equation}
\label{eq:7}
i\hbar\left(\gamma^{\mu}\mathring{\nabla}_{\mu}\Psi+\frac{1}{4}K_{\left[\mu\nu\rho\right]}\gamma^{\mu}\gamma^{\nu}\gamma^{\rho}\Psi\right)-m\Psi=0,
\end{equation}
where $K_{\mu\nu\rho}$ is the contortion tensor that we introduced in the previous chapter, and the $\gamma^{\alpha}$ are the modified gamma matrices, related to the standard ones via the vierbein
\begin{equation}
\gamma^{\alpha}=e^{\alpha}\,_{a}\gamma^{a},
\end{equation}
and $\Psi$ is a general spinor state. It is clearly observed that the contribution of the torsion to the Dirac equation is proportional to the totally antisymmetric part of the torsion tensor, {\emph{i.e.}} the axial vector $S^{\mu}$. Consequently, using the contortion expression \eqref{contorsion} and the torsion decomposition \eqref{decomposition2}, we can rewrite the Dirac equation \eqref{eq:7} as
\begin{equation}
\label{diracaxial}
i\hbar\left(\gamma^{\mu}\mathring{\nabla}_{\mu}\Psi+\frac{1}{24}\varepsilon_{\mu\nu\rho\sigma}S^{\sigma}\gamma^{\mu}\gamma^{\nu}\gamma^{\rho}\Psi\right)-m\Psi=0.
\end{equation}
This implies that a torsion field with vanishing antisymmetric component, which is commonly known as an \emph{inert torsion} field, will not couple to the fermions. \\
Given the fact that there is no analytical solution to Equation~\eqref{diracaxial}, we need to make approximations in order to solve it. As it is usual in Quantum Mechanics, we can make use of the WKB expansion to obtain simpler versions of this equation. Following this procedure, we can expand the general spinor as
\begin{equation}
\label{eq:8}
\Psi\left(x\right)={\rm e}^{i\frac{G\left(x\right)}{\hbar}}(-i\hbar)^{n}a_{n}\left(x\right),
\end{equation}
where we have used the Einstein sum convention (with $n$ going from zero to infinity). Moreover, we have assumed that $G\left(x\right)$ is real and $a_{n}\left(x\right)$ are spinors.\\
As any approximation, it has a limited range of validity. In this case, we can use it as long as $\mathring{R}^{-1}\gg \lambda_{B}$, where $\lambda_{B}$ is the de Broglie wavelength of the particle. This specific inequality expresses the fact that we cannot apply the mentioned approximation in presence of strong gravitational fields and that we cannot consider highly relativistic particles.

By plugging the WKB expansion \eqref{eq:8} into the Dirac equation \eqref{diracaxial} we can obtain the expressions for the zero and first order in $\hbar$, namely
\begin{equation}
\label{eq:9}
\left(\gamma^{\mu}\mathring{\nabla}_{\mu}G+m\right)a_{0}\left(x\right)=0,
\end{equation}
and
\begin{equation}
\label{eq:10}
\left(\gamma^{\mu}\mathring{\nabla}_{\mu}G+m\right)a_{1}\left(x\right)=-\gamma^{\mu}\mathring{\nabla}_{\mu}a_{0}-\frac{1}{24}\varepsilon_{\mu\nu\rho\sigma}S^{\sigma}\gamma^{\mu}\gamma^{\nu}\gamma^{\rho}a_{0}.
\end{equation}
We shall assume that the four-momentum $p^{\mu}$ of the fermions is orthogonal to the surfaces of constant $G\left(x\right)$, and introduce it as
\begin{equation}
\label{eq:11}
p_{\mu}=-\mathring{\nabla}_{\mu}G=-\partial_{\mu}G.
\end{equation}
At this point, the reader might be thinking that the choice in \eqref{eq:11} is an arbitrary decision. In fact, it is not, as we shall prove in the following. It is clear that at the lowest order of $\hbar$, where the internal spin of particles does not play any role, the fermions should move following geodesics. Then, the Dirac equation at order zero \eqref{eq:9} needs to be the equation of a geodesic \eqref{geodesiceq}. Indeed, if we make use of \eqref{eq:11} in \eqref{eq:9} we find
\begin{equation}
\left(\gamma^{\mu}p_{\mu}-m\right)a_{0}\left(x\right)=0,
\end{equation}
which clearly implies that
\begin{equation}
\label{HamiltonJacobi}
p_{\mu}p^{\mu}=m^{2},
\end{equation} 
commonly known as the \emph{Hamilton-Jacobi equation}. Acordingly, the four-velocity $u^{\mu}$ of the fermion at this order will be given by
\begin{equation}
u_{\mu}=\frac{1}{m}p_{\mu}=-\frac{1}{m}\partial_{\mu}G,
\end{equation}
and upon the use of Equation \eqref{HamiltonJacobi} we obtain the property
\begin{equation}
u_{\mu}u^{\mu}=1.
\end{equation}
Then, taking into account the two previous expressions we can arrive at
\begin{eqnarray}
u^{\mu}\mathring{\nabla}_{\mu}u_{\nu}&=&-\frac{1}{m}u^{\mu}\mathring{\nabla}_{\mu}\mathring{\nabla}_{\nu}G=-\frac{1}{m}u^{\mu}\mathring{\nabla}_{\nu}\mathring{\nabla}_{\mu}G=u^{\mu}\mathring{\nabla}_{\nu}u_{\mu}
\nonumber
\\
&=&\mathring{\nabla}_{\nu}\left(u^{\mu}u_{\mu}\right)-u^{\mu}\mathring{\nabla}_{\nu}u_{\mu}=-u^{\mu}\mathring{\nabla}_{\nu}u_{\mu}.
\end{eqnarray}
Finally, since the equation is equal to something and its negative it means the necessarily we have that
\begin{equation}
u^{\mu}\mathring{\nabla}_{\mu}u_{\nu}=0.
\end{equation}
which it is precisely the geodesic equation.

Now that we have clarified the behaviour of the lowest order it is time to explore the first order in $\hbar$, where the coupling of the axial vector of the torsion with the fermion plays a crucial role. For the explicit calculations we shall refer the reader to~\cite{Audretsch:1981xn}. Here we will just give the definitions and obtain the main results. In order to see how the internal spin of the fermion precesses along its trajectory we have considered the spin density tensor to be defined as
\begin{equation}
\label{eq:12}
S^{\mu\nu}=\frac{\overline{\Psi}\sigma^{\mu\nu}\Psi}{\overline{\Psi}\Psi},
\end{equation}
where the $\sigma^{\mu\nu}$ are the modified spin matrices, given by
\begin{equation}
\sigma^{\alpha\beta}=\frac{i}{2}\left[\gamma^{\alpha},\,\gamma^{\beta}\right].
\end{equation}
Then, we can obtain the spin vector from this density
\begin{equation}
\label{eq:13}
s^{\mu}=\frac{1}{2}\varepsilon^{\mu\nu\alpha\beta}u_{\nu}S_{\alpha\beta},
\end{equation}
Using the WKB expansion, we can write the lowest order of the spin vector as
\begin{equation}
\label{eq:14}
s_{0}^{\mu}=\overline{b}_{0}\gamma^{5}\gamma^{\mu}b_{0},
\end{equation}
where $b_{0}$ is the $a_{0}$ spinor but normalised.

With the previous definitions, we can then compute the evolution of the spin vector
\begin{equation}
\label{eq:15}
u^{\nu}\mathring{\nabla}_{\nu}s_{0}^{\mu}=\frac{1}{2}\varepsilon^{\mu\nu\rho\sigma}S_{\sigma}s_{0\,\rho}u_{\nu}.
\end{equation}
On the other hand, the calculation of the acceleration of the particle, {\emph{i.e.}} the deviation from geodesical movement, comes from the splitting of the Dirac current via the Gordon decomposition and from the identification of the velocity with the normalised convection current. Then, it can be shown that the non-geodesical behaviour is governed by the following expression for the four-acceleration
\begin{equation}
\label{eq:16}
a_{\mu}=v^{\varepsilon}\mathring{\nabla}_{\varepsilon}v_{\mu}=\frac{\hbar}{4m_{e}}\widetilde{R}_{\mu\nu\alpha\beta}\overline{b}_{0}\sigma^{\alpha\beta}b_{0}v^{\nu},
\end{equation}
where $\widetilde{R}_{\mu\nu\alpha\beta}$ refers to the intrinsic part of the Riemann tensor associated with the totally antisymmetric component of the torsion tensor:
\begin{equation}
\label{eq:17}
\widetilde{\Gamma}^{\lambda}\,_{\mu\nu}=\mathring{\Gamma}^{\lambda}\,_{\mu\nu}-\frac{1}{2}\varepsilon^{\lambda}\,_{\mu\nu\sigma}S^{\sigma}.
\end{equation}
Unlike most of the literature presented at the beginning of the section, the expression \eqref{eq:16} does not have an explicit contortion term coupled to the spin density tensor, hence all the torsion information is encrypted into the mentioned part of the Riemann tensor. Finally, it is worth noting that the standard case of GR is naturally recovered for inert torsion, as expected. \\

\subsection{Explicit workout example}
\label{3.1.2}

In this subsection we shall provide an example of the non-geodesical behaviour of an electron around a Reissner-Nordstr\"om black-hole sourced by torsion instead of the usual electromagnetic charge. Our results will show that having a strong torsion field would make the difference between the electron trajectory and the geodesics evident, despite being modulated by $\hbar$. 

The mentioned Reissner-Nordstr\"{o}m solution comes from the following PG gravity vacuum action~\cite{Cembranos:2016gdt,Cembranos:2017pcs}:
\begin{eqnarray}
S=\frac{1}{16\pi}\int {\rm d}^{4}x\sqrt{-g}\left[-\mathring{R}+\frac{d_{1}}{2}R_{\lambda\rho\mu\nu}R^{\mu\nu\lambda\rho}-\frac{d_{1}}{4}R_{\lambda\rho\mu\nu}R^{\lambda\rho\mu\nu} \right.
\nonumber
\\
\left.-\frac{d_{1}}{2}R_{\lambda\rho\mu\nu}R^{\lambda\mu\rho\nu}+d_{1}R_{\mu\nu}\left(R^{\mu\nu}-R^{\nu\mu}\right)\right],
\end{eqnarray}
with the exact metric of the solution given by
\begin{equation}
{\rm d}s^{2}=f\left(r\right){\rm d}t^{2}-\frac{1}{f\left(r\right)}{\rm d}r^{2}-r^{2}\left({\rm d}\theta^{2}+\sin^{2}\theta {\rm d}\varphi^{2}\right),
\end{equation}
where
\begin{equation}
f\left(r\right)=1-\frac{2m}{r}+\frac{d_{1}\kappa^{2}}{r^{2}},
\end{equation}
where $m$ is the BH mass and $\kappa$ is a scalar charge sourced by torsion. For the rest of this subsection we shall consider $d_{1}=1$, which simplifies the computations while not compromising the generality of the result.\\
In order to calculate the non-geodesical behaviour of the electron we need to have the values of the non-vanishing torsion components, namely
\begin{equation}
\begin{cases}
T_{tr}^{\,\,\,\,\,t}=a(r)=\frac{\dot{f}\left(r\right)}{2f\left(r\right)},\\
\,\\
T_{tr}^{\,\,\,\,\,r}=b(r)=\frac{\dot{f}\left(r\right)}{2},\\
\,\\
T_{t\theta_{i}}^{\,\,\,\,\,\theta_{i}}=c(r)=\frac{f\left(r\right)}{2r},\\
\,\\
T_{r\theta_{i}}^{\,\,\,\,\,\theta_{i}}=g(r)=-\frac{1}{2r},\\
\,\\
T_{t\theta_{i}}^{\,\,\,\,\,\theta_{j}}=e^{a\theta_{j}}e_{\,\,\theta_{i}}^{b}\varepsilon_{ab}d\left(r\right)=e^{a\theta_{j}}e_{\,\,\theta_{i}}^{b}\varepsilon_{ab}\frac{\kappa}{r},\\
\,\\
T_{r\theta_{i}}^{\,\,\,\,\,\theta_{j}}=e^{a\theta_{j}}e_{\,\,\theta_{i}}^{b}\varepsilon_{ab}h\left(r\right)=-e^{a\theta_{j}}e_{\,\,\theta_{i}}^{b}\varepsilon_{ab}\frac{\kappa}{rf\left(r\right)},
\end{cases}
\end{equation}
where $i,j=1,2$ with $i\neq j$, and we have made the identification $\left\{ \theta_{1},\,\theta_{2}\right\} =\left\{ \theta,\,\varphi\right\}$. Moreover, $\varepsilon_{ab}$ is the Levi-Civita symbol for 2 dimensions, and the dot $\dot{\,}$ represents the derivative with respect to the radial coordinate $r$.\\
Now, with the components of the metric and the torsion tensors, we can calculate the modified connection and therefore the Riemann tensor of Equation~\eqref{eq:16}, in order to obtain the acceleration. Moreover, we know that the $b_{0}$ and $\overline{b}_{0}$ are the lowest order in $\hbar$ of the general spinor state $\Psi$. Hence, we can use the fact that the most general form of a positive energy solution of the Dirac equation for $b_{0}$ and $\overline{b}_{0}$ is~\cite{Alsing:2009px}
\begin{equation}
b_{0}=\left(\begin{array}{c}
\cos\left(\frac{\alpha}{2}\right)\\
{\rm e}^{i\beta}\sin\left(\frac{\alpha}{2}\right)\\
0\\
0
\end{array}\right)\,,\,\,\,\,\,\,\,\overline{b}_{0}=\left(\begin{array}{cccc}
\cos\left(\frac{\alpha}{2}\right), & {\rm e}^{-i\beta}\sin\left(\frac{\alpha}{2}\right), & 0, & 0\end{array}\right)\,,
\end{equation}
where the angles $\alpha$ and $\beta$ give the direction of the spin of the particle at the lowest order in $\hbar$
\begin{equation}
\overrightarrow{n}=\left(\begin{array}{ccc}
\sin\left(\alpha\right)\cos\left(\beta\right), & \sin\left(\alpha\right)\sin\left(\beta\right), & \cos\left(\alpha\right)\end{array}\right).
\end{equation}

Before calculating the acceleration, let us use this form of the spinor to calculate the corresponding spin vector. Using Equation~\eqref{eq:14} we have
\begin{eqnarray}
\label{eqn:30}
&&\left(s_0\right)^{\mu}=\left(\begin{array}{c}
0\\
\,\\
-\sin\left(\alpha\right)\cos\left(\beta\right)\sqrt{f\left(r\right)}\\
\,\\
-\frac{\sin\left(\alpha\right)\sin\left(\beta\right)}{r}\\
\,\\
-\frac{\cos\left(\alpha\right)\csc\left(\theta\right)}{r}
\end{array}\right)\,,
\nonumber
\\
&&\,
\\
&&\left(s_0\right)_{\mu}=\left(\begin{array}{cccc}
0, & \frac{\sin\left(\alpha\right)\cos\left(\beta\right)}{\sqrt{f\left(r\right)}}, & r\sin\left(\alpha\right)\sin\left(\beta\right), & r\sin\left(\theta\right)\cos\left(\alpha\right)\end{array}\right).\nonumber
\end{eqnarray} 
We are now ready to calculate the acceleration components for the Reissner-Nordstr\"{o}m solution, which can be found in Appendix \ref{ap:1}.
It is worth mentioning that the only components of the torsion tensor which contribute to the acceleration are those related to the functions $d(r)$ and $h(r)$, which are precisely the ones that contribute to the axial vector of the torsion. This is important as a consistency check, because if we set the $\kappa$ constant to zero, the torsion tensor is inert, since the axial vector is zero, as it is expected.

The expressions for the four-acceleration components are complex and it is difficult to understand their behaviour intuitively. 
In this sense, it is interesting to study two relevant cases that simplify their interpretation:
\begin{itemize}
\item Low values of $\kappa$:\\
If we consider a realistic physical implementation of this solution, in order to avoid naked singularities, we expect small values of the parameter $\xi=\frac{\kappa}{m^{2}}$. Indeed, $\xi$ is the dimensionless parameter controlling the contribution
of the torsion tensor in the four-acceleration. Therefore, we can see that it is a good approximation to consider only up to first order in an expansion of the acceleration in terms of $\xi$. These results can be found in the Appendix~\ref{ap:2}.

\item Asymptotic behaviour:\\
It is interesting to study what occurs in the asymptotic limit $r\rightarrow\infty$, in order to observe which is the leading term of the corresponding limit and compare its strength with other effects on the particle. We calculate the following:
\begin{eqnarray}
\underset{r\rightarrow\infty}{\rm lim}a^{t}&\simeq&\frac{m^{2}\xi\hbar}{2m_{e}r}\left(\sin(\alpha)\sin(\beta)\theta'(s)+\sin(\theta)\cos(\alpha)\varphi'(s)\right),
\label{at}
\\
\label{ar}
\underset{r\rightarrow\infty}{\rm lim}a^{r}&\simeq&\frac{m^{2}\xi\hbar}{2m_{e}r}\left(\sin(\alpha)\sin(\beta)\theta'(s)+\sin(\theta)\cos(\alpha)\varphi'(s)\right),
\\
\underset{r\rightarrow\infty}{\rm lim}a^{\theta}&\simeq&\frac{m\hbar}{2m_{e}r^{3}}\left[-m\xi r'(s)\left(\sin(\alpha)\sin(\beta)+m^{2}\xi\cos(\alpha)\right)\right.
\nonumber
\\
&+&m\xi t'(s)\left(\sin(\alpha)\sin(\beta)+m^{2}\xi\cos(\alpha)\right)
\nonumber
\\
&-&\left.2\sin(\alpha)\cos(\beta)\sin(\theta)\varphi'(s)\right],
\\
\underset{r\rightarrow\infty}{\rm lim}a^{\varphi}&\simeq&\frac{m\hbar\csc(\theta)}{2m_{e}r^{3}}\left[m\xi r'(s)\left(m^{2}\xi\sin(\alpha)\sin(\beta)-\cos(\alpha)\right)\right.
\nonumber
\\
&+&m\xi t'(s)\left(\cos(\alpha)-m^{2}\xi\sin(\alpha)\sin(\beta)\right)
\nonumber
\\
&+&\left.2\sin(\alpha)\cos(\beta)\theta'(s)\right].
\end{eqnarray}

Where we have used the viability condition \eqref{via}, because as we shall see, it is a necessary condition for the semiclassical approximation. \\
We can observe in \eqref{at} and \eqref{ar} that the time and radial components follow a $r^{-1}$ pattern, while the angular components follow a $r^{-3}$ behaviour. Hence, in the first two components the torsion effect goes asymptotically to zero at a lower rate than the strength provided by the conventional gravitational field. Conversely, in the angular ones, it approaches zero at a higher rate.
\end{itemize} 

Going back to the general form of the four-acceleration components, given in Appendix \ref{ap:1}, it is interesting to analyse the two components of the acceleration that are non-zero in GR, $a^{\theta}$ and $a^{\varphi}$, in order to adquire a deeper understanding. They read
\begin{equation}
\label{plana1}
a^{\theta}|_{\kappa=0}=\frac{m\hbar\sin(\theta)}{2m_{e}r^{3}\sqrt{1-\frac{2m}{r}}}\left[\left(s_{0}\right)^{\varphi}r'(s)+2\left(s_{0}\right)^{r}\varphi'(s)\right],
\end{equation}
and
\begin{equation}
\label{plana2}
a^{\varphi}|_{\kappa=0}=\frac{m\hbar\csc(\theta)}{2m_{e}r^{3}\sqrt{1-\frac{2m}{r}}}\left[\left(s_{0}\right)^{\theta}r'(s)+2\left(s_{0}\right)^{r}\theta'(s)\right],
\end{equation}
where we have used the expression of the spin vector at the lowest order of $\hbar$~\eqref{eqn:30} to simplify the equations. As we can see, the two previous expressions are quite alike, and they can be made equal by establishing the substitutions $\sin(\theta)\leftrightarrow \csc(\theta)$, and $\varphi\leftrightarrow \theta$.
For both of the expression \eqref{plana1} and \eqref{plana2} we observe that the spin-gravity coupling acts as a {\it cross-product force} (\emph{e.g.} the magnetic force), in the sense that the acceleration is perpendicular to the direction of the velocity and the spin vector. \\
Now, to measure the torsion contribution in the acceleration we shall compare the acceleration for $\kappa=0$ and for arbitrary values of $\kappa$. In order to do so, we shall define a new dimensionless parameter as the fraction between the acceleration for a finite value of $\kappa$ and the one for $\kappa =0$
\begin{equation}
B^{\mu}(\kappa)=\frac{a^{\mu}}{a^{\mu}|_{\kappa=0}}.
\end{equation}
As we have stated before, the viability condition \eqref{via} implies that 
\begin{equation}
\cos(\alpha)\theta'(s)-\sin(\alpha)\sin(\beta)\sin(\theta)\varphi'(s)=0, 
\end{equation}
so $a^{t}|_{\kappa=0}$ and $a^{r}|_{\kappa=0}$ would vanish identically. This means that the $B^{\mu}$ parameter can only be defined for the angular coordinates.\\
Let us explore two explicit examples: 
\begin{itemize}

\item Example 1: we consider a BH of 24 solar masses and a particle located near the external event horizon in the $\theta=\pi/2$ plane, at a radial distance of $2m+\varepsilon$, where $\varepsilon=m/10$. The position in $\varphi$ is irrelevant because the acceleration does not depend on this coordinate. We assume that the particle has radial velocity equal to 0.8, and that the direction of the spin is in the $\varphi$ direction. The rest of the velocity components are zero except for $v^{t}=(8.8 \kappa +0.3)^{-1/2}$. It is clear from \eqref{plana1} and \eqref{plana2} that we can only calculate the relative acceleration in the $\theta$ direction, $B_{\theta}$.

\item Example 2: in this case the BH mass and the position of the particle are the same. The velocity is in the $\theta$ direction, and has the same modulus as before. Again, the rest of the components are zero except for 
$v^{t}=1.3 (8.8 \kappa +0.3)^{-1/2}$. The spin has only a radial component, therefore the acceleration would be in the $\varphi$ direction. Consequently, we can only calculate $B_{\varphi}$.

\end{itemize}

Both of them are shown in Figure~\ref{fig:1}, where we represent different components of $B^{\mu}$ in function of $\kappa$ for these two cases.

\begin{figure}
\centering
\includegraphics[width=1\linewidth]{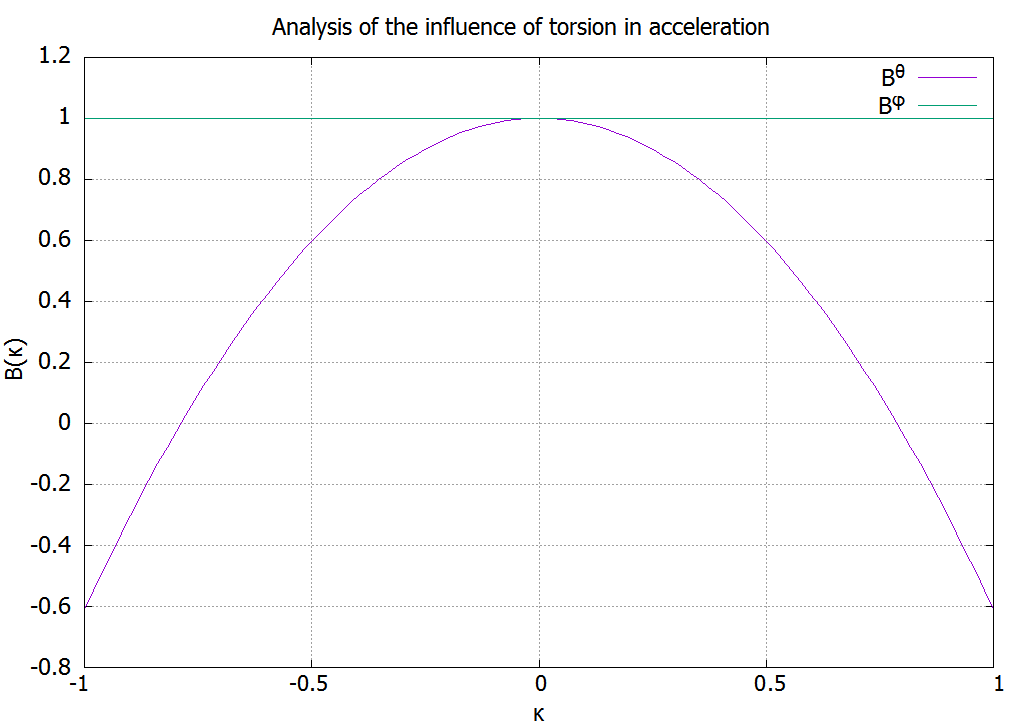}
\label{fig:sfig2}
\caption{We have considered a BH of 24 solar masses and a particle located near the external event horizon in the $\theta=\pi/2$ plane, at a radial distance of $2m+\varepsilon$, where $\varepsilon=m/10$. The position in $\varphi$ is irrelevant because the acceleration does not depend on this coordinate. The $B_{\theta}$ line represents the Example 1, where we assume that the particle has radial velocity equal to 0.8, and that the direction of the spin is in the $\varphi$ direction. The rest of the velocity components are zero except for $v^{t}=(8.8 \kappa +0.3)^{-1/2}$. The $B_{\varphi}$ line represents Example 2, where the velocity is in the $\theta$ direction, and has the same modulus as before. Again, the rest of the components are zero except for 
$v^{t}=1.3 (8.8 \kappa +0.3)^{-1/2}$.}
\label{fig:1}
\end{figure}

It is worthwhile to stress that there is nothing in the form of the metric or in the underlying theory that stops us from taking negative values of $\kappa$, in contrast with the usual electromagnetic version of the solution. 
We can observe that as we take higher absolute values for $\kappa$ we find that the acceleration caused by the spacetime torsion opposite sign with respect to the one produced by the gravitational coupling, reaching significant differences for large $\kappa$. This is expected since we have chosen a large value for the coupling between spin and torsion. \\

At this time, we go one step further and calculate the trajectory of the particle, using Equation~\eqref{eq:16} and having in mind the spinor evolution equation~\eqref{eq:15}, which can be rewritten as
\begin{equation}
v^{\mu}\tilde{\nabla}_{\mu}b_{0}=0.
\end{equation}
For the exact Reissner-Nordstr\"{o}m geometry sourced by torsion, we find some interesting features.
First, in order to maintain the semiclassical approximation and the positive energy associated with the spinor, 
two conditions must be fulfilled. On one hand we have
\begin{equation}
\dot{f} \left(r\right)\ll Lf\left(r\right),
\label{fcondition}
\end{equation}
where $L=3.3\cdot 10^{-8}\,\,{\rm m}^{-1}$, so that the derivative of $f\left(r\right)$ is at least two orders of magnitude below the value of $f\left(r\right)$ in the units we are using.\\
On the other hand, the second condition is
\begin{equation}
\label{via}
\left(\overline{b}_{0}\sigma^{r\beta}b_{0}\right)v_{\beta}=0.
\end{equation}
The first one, \eqref{fcondition}, is a consequence of the method that we are applying: if both curvature and torsion are strong then the interaction is also strong, and the WKB approximation fails. This is a purely metric condition, since it comes from the Levi-Civita part of the Riemann tensor, so it will be the same for all the spherically symmetric solutions. The second one, \eqref{via}, is the radial component of the so-called \emph{Pirani condition} \cite{Pirani:1956tn}.
In order to obtain the trajectory, we have solved the equations \eqref{eq:16} and \eqref{eq:15} numerically for a BH with 24 solar masses and $\kappa=10$, with the electron located outside the external event horizon in the $\theta=\pi/2$ plane. Furthermore, we have assumed that the electron has radial velocity of 0.9 and initial spin aligned in the $\varphi$ direction. The results obtained are given in Figure~\ref{fig:3}.\\

\begin{figure}[htpb]
\begin{subfigure}{1\textwidth}
\centering
\includegraphics[width=0.9\linewidth]{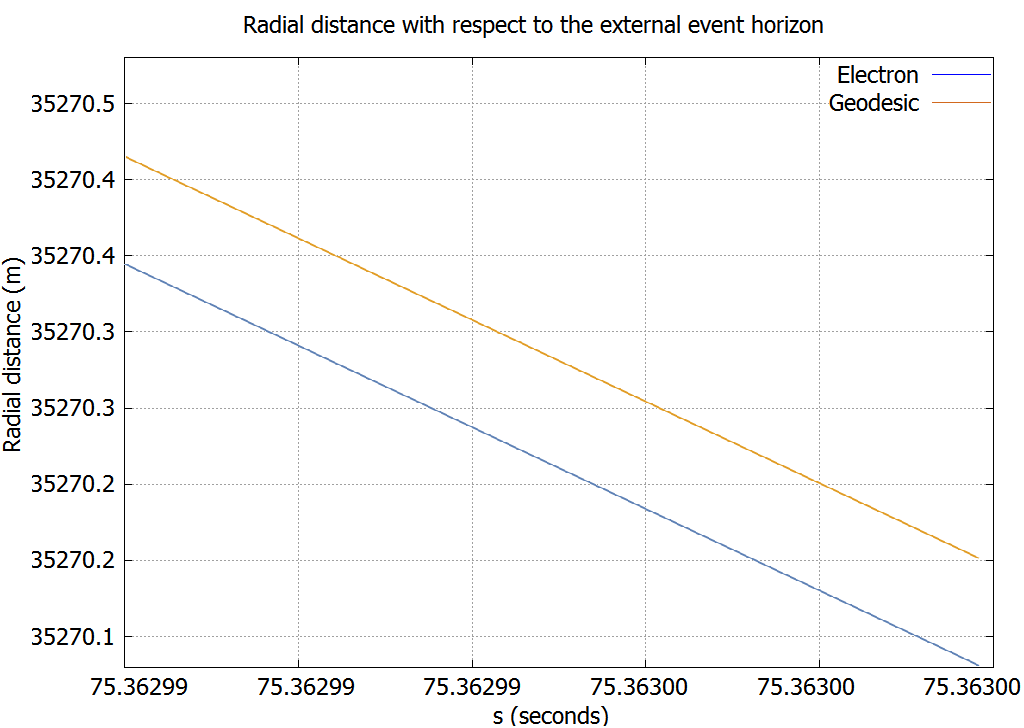}
\caption{Trajectory at 35 km of the event horizon.}
\end{subfigure}\\
\begin{subfigure}{1\textwidth}
\centering
\includegraphics[width=0.9\linewidth]{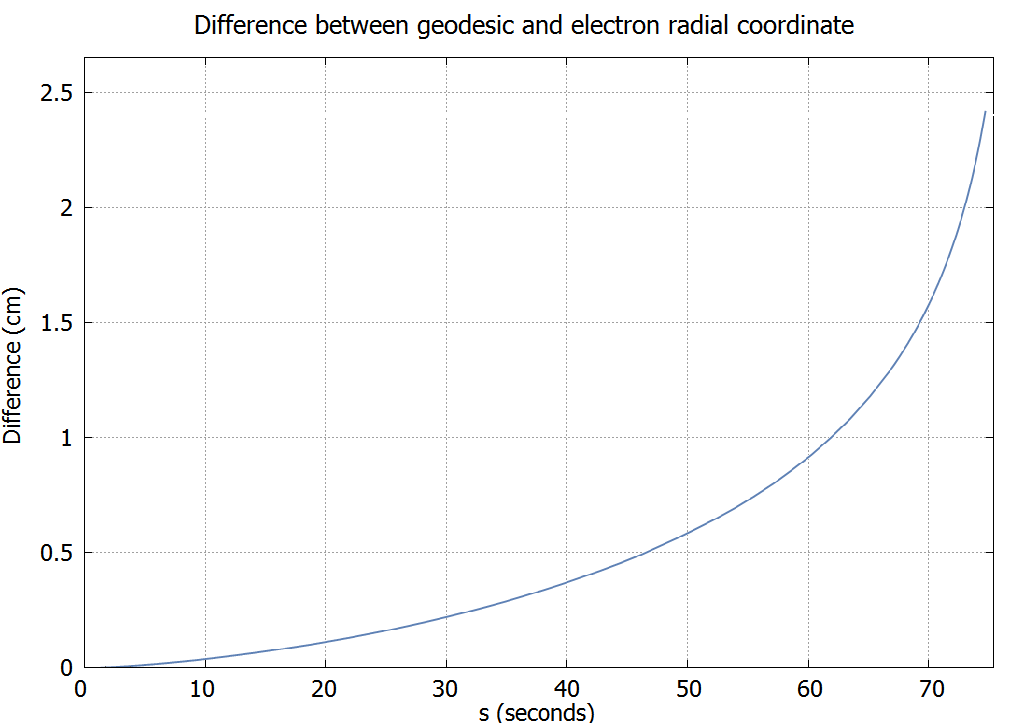}
\caption{Difference of the position of the two particles.}
\end{subfigure}
\caption{For this numerical computation we have considered a BH with 24 solar masses and $\kappa=10$, with the electron located outside the external event horizon in the $\theta=\pi/2$ plane. We have assumed an electron with radial velocity of 0.9 and initial spin aligned in the $\varphi$ direction.}
\label{fig:3}
\end{figure}

It is worthwhile to stress that any difference from geodesical behaviour in the radial coordinate is an exclusive consequence of the torsion-spin coupling, with no presence of geometric terms provided by GR, in virtue of the dependence on $\kappa$ existing in the corresponding acceleration component. Indeed, it is possible to have situations under which the geodesic curves and the trajectories of fermionic particles are distanced due to this effect, even by starting at the same point. Nevertheless, it is not strong enough to avoid their entrance to the BH region, hence they present a singular behaviour.

In the next section we shall study if there exists the possibility for some particles of avoiding the spacetime singularities present in GR in theories endowed with a non-symmetric connection.

\section{Singularities}
\label{3.2}

\emph{Can spin avert singularities?} This is a question that has been under study since Stewart and Hajicek proposed that the introduction of a torsion field, which one of its sources are half-spin particles, would lead to the avoidance of singularities in the spacetime~\cite{Stewart:1973ux}. Before exploring the different answers to this question, let us briefly review the concept of singularity in GR.

In a physical theory, a singularity is usually known as a ``place" where some of the quantities used in the description of a dynamical system diverge. As an example, we can find this situation evaluating the Coulomb potential $V=K\frac{q}{r}$ at the point $r=0$. This kind of behaviour appears because the theory is either invalid in the considered region or we have assumed a simplification. In particular, in the previous example the singularity arises due to the fact that we are considering the charged particle as point-like and omitting the quantum effects.

Following this potential definition, in GR one might expect to observe singularities when certain components of the tensors describing the spacetime geometry diverge. This would mean that curvature is higher than $\frac{1}{l_{p}^{2}}$, where $l_{p}$ is the Planck length, so we need to have into account the quantum effects, which are not considered in this theory. However, there are situations where this behaviour is given as a result of the election of the coordinate system. This is the case of the ``singularity" at $r=2M$ in the Schwarzschild solution. For this reason, another criterion, proposed by Penrose~\cite{Penrose:1964wq}, is used to define a spacetime singularity: geodesic incompleteness. The physical interpretation of this criterion is the existence of free falling observers that either appear or disappear out of nothing. This is clearly ``strange" enough to consider it as a sufficient condition to ensure the occurrence of singularities.

In general, all singularity theorems follow the same pattern, made it explicit by Senovilla in~\cite{Senovilla:2018aav}:
\begin{thm} 
(Pattern singularity ``theorem"). If the spacetime satisfies:\\
1) A condition on the curvature tensor. \\
2) A causality condition. \\
3) An appropriate initial and/or boundary condition. \\
Then, there are null or timelike inextensible incomplete geodesics. 
\end{thm}

In order to answer if the introduction of new degrees of freedom in modified theories can avoid the appearance of singularities, one can study if the standard conditions exposed in this theorem may change in such a theory with respect to GR. In this sense, in the following we aim at showing, based in {\bf{\nameref{P1}}}, that \emph{in a strongly asymptotically predictable spacetime the conditions for having a singular trajectory for any massive particle in theories with torsion are the same as in GR}. 

\subsection{The singularity theorem}

In order to review the singularity theorem in {\bf{\nameref{P1}}}, we need to introduce some additional definitions. We know intuitively that the existence of incomplete null geodesics usually leads to the appearance of BHs, the latter may be understood as regions of the spacetime beyond which an inside observer cannot escape. This applies to all particles following any timelike and null curves, not just geodesics. This is known as the \emph{cosmic censorship conjecture}, which was introduced by Penrose in 1969. It basically states that singularities cannot be \emph{naked}, meaning that they cannot be seen by an outside observer. However, how can this concept be expressed mathematically? The answer lies in the concept of \emph{conformal compactification}, which is defined as~\cite{Frauendiener:2000mk}
\begin{defn}
Let $\left(M,g\right)$ and $\left(\tilde{M},\,\tilde{g}\right)$ be two spacetimes. Then $\left(\tilde{M},\,\tilde{g}\right)$ is said to be a conformal compactification of $M$ if and only if the following properties are met:
\begin{enumerate}
\item $\tilde{M}$ is an open submanifold of $M$ with smooth boundary $\partial\tilde{M}=\mathcal{J}$. This boundary is commonly referred as \emph{conformal infinity}.
\item There exists a smooth scalar field $\Omega$ on $\tilde{M}$, such that $\tilde{g}_{\mu\nu}=\Omega^{2}g_{\mu\nu}$ on $M$, and so that $\Omega=0$ and its gradient $d\Omega\neq 0$ on $\mathcal{J}$.
\end{enumerate}
If additionally, every null geodesic in $\tilde{M}$ acquires a future and a past endpoint on $\mathcal{J}$, the spacetime $\left(\tilde{M},\,\tilde{g}\right)$ is denoted \emph{asymptotically simple}. Moreover, if the Ricci tensor is zero in a neighbourhood of $\mathcal{J}$, the spacetime is said to be \emph{asymptotically empty}.
\end{defn} 
In a conformal compactification, $\mathcal{J}$ is composed of two null hypersurfaces, $\mathcal{J}^{+}$ and $\mathcal{J}^{-}$, known as \emph{future null infinity} and \emph{past null infinity}, respectively.

In order to establish the definition of a BH, we need to introduce two additional concepts, namely \cite{Wald:1984rg}
\begin{defn}
A spacetime $\left(M,g\right)$ is said to be \emph{asymptotically flat} if there is an asymptotically empty spacetime $\left(M',g'\right)$ and a neighbourhood $\mathcal{U}'$ of $\mathcal{J}'$, such that $\mathcal{U}'\cap M'$ is isometric to an open set $\mathcal{U}$ of $M$. 
\end{defn} 
\begin{defn}
Let $\left(M,g\right)$ be an asymptotically flat spacetime with conformal compactification $\left(\tilde{M},\,\tilde{g}\right)$. Then $M$  is called \emph{(future) strongly asymptotically predictable} if there is an open region $\tilde{V}\subset\tilde{M}$, with $\overline{J^{-}\left(\mathcal{J}^{+}\right)\cap M}\subset\tilde{V}$, such that $\tilde{V}$ is globally hyperbolic. 
\end{defn}
This definition does not require the condition of endpoints of the null geodesics, meaning that these types of spacetimes can be singular. Nevertheless, if a spacetime is asymptotically predictable, then the singularities are not naked, \emph{i.e.} they are not visible from $\mathcal{J}^{+}$.

At this time we are ready to present what we understand by a BH:
\begin{defn}
\label{BlackHole}
A strongly asymptotically predictable spacetime $\left(M,g\right)$ is said to contain a BH if $M$ is not contained in $J^{-}\left(\mathcal{J}^{+}\right)$. The BH region, $B$, is defined to be $B=M-J^{-}\left(\mathcal{J}^{+}\right)$ and its boundary, $\partial B$, is known as the \emph{event horizon}. 
\end{defn}
Intuitively, we think that a particle in a so-called \emph{closed trapped surface}\footnote{See \cite{Galloway:2010tx} or {\bf{\nameref{P1}}} for the definition of a closed trapped surface of arbitrary co-dimension. We have not included it in this thesis because it does not contribute to the main results and its introduction may be cumbersome.} cannot escape to $\mathcal{J}^{+}$, meaning that it is part of the BH region of the spacetime. Nevertheless, this is not true in general. In the next proposition, we establish the conditions that ensure the existence of BHs when we have a closed future trapped submanifold of arbitrary co-dimension
\begin{prop}
\label{prop:sin}
Let $\left(M,g\right)$ be a strongly asymptotically predictable spacetime of dimension $n$, and $\Sigma$ a closed future trapped submanifold of arbitrary co-dimension $m$ in $M$. If the \emph{curvature condition}\footnote{\'Idem.} holds along every future directed null geodesic emanating orthogonally from $\Sigma$, then $\Sigma$ cannot intersect $J^{-}\left(\mathcal{J}^{+}\right)$ (i.e. $\Sigma$ is in the BH region $B$ of $M$ \footnote{Analogously, it can be defined a past strongly asymptotically predictable spacetime, and then the proposition would predict the existence of white hole (WH) regions, $B=M-J^{+}\left(\mathcal{J}^{-}\right)$, which are regions where particles cannot enter, only exit.}).
\end{prop} 
\begin{proof}
The proof can be found in {\bf{\nameref{P1}}}.
\end{proof}

The reader might be wondering how the latter proposition is related to the actual singular behaviour of the different kind of particles. \\
From the minimal coupling procedure, it follows that particles without internal spin, which as we have seen are represented by scalar fields, do not feel torsion due to the fact that the covariant derivative of a scalar field is just its partial derivative. Also, since it is impossible to perform the minimally coupling prescription for the Maxwell's field preserving the $U\left(1\right)$ gauge invariance, the Maxwell equations are the same than the ones present in GR. Therefore, they move following null extremal curves (\emph{i.e.} null geodesics), so that the causal structure is determined by the metric structure, just like in GR. This means that the usual test particles follow the geodesic curves provided by the Levi-Civita connection, which allow us to directly generalise the singularity theorems for this kind of trajectories. But, what happens when we consider fermions, which are coupled to the spacetime torsion?\\
All the analysis of the trajectories that follow these kinds of particles, that we have mentioned at the beginning of the previous section, have one thing in common. They all experiment a corrective factor with respect to geodesical behaviour of the form

\begin{equation}
\label{spineq}
a^{\mu}=v^{\rho}\mathring{\nabla}_{\rho}v_{\mu}=C\frac{\hbar}{m}f\left({R}_{\rho\sigma\lambda}\,^{\mu}S^{\rho\sigma}v^{\lambda}+K_{\sigma\rho}\,^{\mu}p^{\rho}v^{\sigma}\right),
\end{equation}
where $C$ is a constant, $f$ is some function, $m$ is the mass of the particle and $S^{\rho\sigma}$ describes the internal spin tensor, that it is related to the spin $s^{\mu}$ of the particle as we studied in the previous section.

It is clear from our analysis so far that massive fermionic particles do not follow timelike geodesics. Nevertheless, independently of how torsion affects these particles,  we know they will follow timelike curves and we assume that locally (in a normal neighbourhood of a point) nothing can be faster than light (null geodesics). Accordingly, it would be interesting to see under which circumstances we have incompleteness of non-geodesical timelike curves. In order to address this, we recall the definition of a n-dimensional BH and WH, that is Def. \ref{BlackHole}. From this definition, we conclude that if these kinds of structures exist in our spacetime, we would have timelike curves (not timelike geodesics exclusively) that do not have endpoints in the conformal infinity, since for the case of BHs the spacetime $M$ is not contained in $J^{-}\left(\mathcal{J}^{+}\right)$, while for WHs, $M$ is not contained in $J^{+}\left(\mathcal{J}^{-}\right)$. Considering these lines, we establish the following theorem
\begin{thm}
\label{thmsing}
Let $\left(M,g\right)$ be a strongly asymptotically predictable spacetime of dimension $n$ and $\Sigma$ a closed future trapped submanifold of arbitrary co-dimension $m$ in $M$. If the curvature condition holds along every future directed null geodesic emanating orthogonally from $\Sigma$, then some timelike curves in $M$ would not have endpoints in the conformal infinity, hence $M$ is a singular spacetime.
\end{thm} 

One might ask if one of the aforementioned incomplete timelike curves may actually represent the trajectory of a spinning particle coupled to the torsion tensor. From Equation (\ref{spineq}), which represents the non-geodesical behaviour, we see that the only possible way that all the trajectories have endpoints in the conformal infinity is having infinite values for the curvature and torsion tensors near the event horizon, which in a physically plausible scenario is not possible. This is why we strongly believe that Theorem \ref{thmsing} is a more physically relevant theorem than the ones based on geodesic incompleteness for the singular behaviour of such particles, since it is strongly related to the actual trajectories of fermions in theories with torsion. \\

With this reasoning we have showed that in strongly asymptotically predictable spacetimes one cannot avoid the occurrence of singularities, even in the presence of torsion, provided that the conditions of Theorem \ref{thmsing} hold. On the opposite case, it is possible to find non-singular configurations assuming some of the conditions are violated \cite{Poplawski:2012ab,Lucat:2015rla}.

\section{Birkhoff theorem}
\label{3.3}

In the previous section we have seen how the existence of BHs is a good criteria for predicting the singular behaviour of all kinds of particles/fields. In this section we shall study what kind of BHs we might expect in PG gravity. Specifically, we will explore if the Birkhoff theorem~\cite{birkhoff1923boundary} is fulfilled. As widely known, such a theorem states that any spherically symmetric solution of the vacuum field equations must be static and asymptotically flat and therefore the only exterior vacuum solution, {\it i.e.}, the spacetime outside a spherical, non rotating, gravitating body, corresponds to the torsionless Schwarzschild spacetime.
This fundamental result in GR \cite{Hawking:1973uf} obviously deserves a deep analysis in every suitable modified gravity theory. Indeed, spherically symmetric vacuum solutions would describe the exterior spacetime around spherically symmetric stars (or BHs) and would help to a better understanding of the measurements coming from weak-field limit tests like the bending of light, the perihelion shift of planets, frame dragging experiments and the Newtonian and post-Newtonian limits of competing gravitational theories, as well as other measurements involving strong-gravity regimes, such as the recently discovered gravitational waves ({\it c.f.} \cite{Abbott:2016blz} and subsequent articles by LIGO/VIRGO collaborations and \cite{Barack:2018yly} for an extensive review of the roadmap of the subject).\\
Moreover, the obtention of vacuum spherically symmetric space-time solutions that are not Schwarzschild-like would provide us with some valuable information about the extra degrees of freedom that we are introducing with the considered modification. Specialised literature has therefore devoted an increasing interest to study the validity of the Birkhoff theorem in different classes of extended theories of gravity, see \emph{e.g.} \cite{Clifton:2006ug,Nzioki:2009av,delaCruzDombriz:2009et,Ferraro:2011ks,Hui:2012qt,Sotiriou:2013qea}.\\

Regarding PG theories, at the beginning of the 1980s some proofs of the Birkhoff theorem were developed for specific models, enriching the PG gravity literature~\cite{Neville:1979fk,Ramaswamy:1979zz,Rauch:1980qj,Bakler:1980is}. Moreover, two weakened versions of the Birkhoff theorem were proposed, either assuming asymptotic flatness of the solutions~\cite{Rauch:1981ua} or considering invariance under spatial reflections in addition to the spherical symmetry~\cite{Rauch:1980qj}.
The most relevant Birkhoff theorem proof was made by Nieh and Rauch in \cite{Rauch:1981tva}, where results from previous literature were summarised. There, authors found two general classes of PG theories in which the theorem holds. Nevertheless, such a remarkable piece of research did not clarify whether these are the only classes of PG theories for which the theorem holds.
Specifically the gravitational Lagrangians of theories considered in \cite{Rauch:1981tva} were either of the form
\begin{equation}
\label{Rauch:1}
\mathcal{L}_{1}=-\lambda R+\alpha R^{2},
\end{equation}
or
\begin{eqnarray}
\label{Rauch:2}
\mathcal{L}_{2}&=&-\lambda R+\frac{1}{12}\left(4\alpha+\beta+3\lambda\right)T_{\alpha\beta\gamma}T^{\alpha\beta\gamma}+\frac{1}{6}\left(-2\alpha+\beta-3\lambda\right)T_{\alpha\beta\gamma}T^{\beta\gamma\alpha}
\nonumber
\\
&+&\frac{1}{3}\left(-\alpha+2\gamma-3\lambda\right)T_{\,\,\beta\alpha}^{\beta}T_{\gamma}^{\,\,\gamma\alpha}\,.
\end{eqnarray}
It is worthwhile to mention that the analysis that we have made in Section \ref{2.3} allows us to prove straightforwardly that these two theories fulfill the Birkhoff theorem. Indeed, the Lagrangian $\mathcal{L}_{1}$ is the same as the one given in \eqref{eq:PGactionR2}, where the parameters have been chosen as $a_0=-\lambda$, $b_1=\alpha$, and $a_1=a_2=a_3=0$. We have shown that this is equivalent to GR plus a non-propagating scalar degree of freedom, since with this choice of parameters the term $m^2_T$, which multiplies the kinetic term in \eqref{eq:R2Einstein}, is zero. Hence, the Birkhoff theorem clearly applies. \\
On the other hand, we can see that the Lagrangian $\mathcal{L}_{2}$ is equal to the one in Eq. \ref{Stablenonprop}, where the coefficients have been taken as $a_0=-\lambda$, $b_1=0$, $a_1=\frac{1}{12}\left(4\alpha+\beta+3\lambda\right)$, $a_2=\frac{1}{6}\left(-2\alpha+\beta-3\lambda\right)$, and $a_3=\frac{1}{3}\left(-\alpha+2\gamma-3\lambda\right)$. As we explained, this theory does not have any observable physical difference with respect to GR, so the Birkhoff theorem is fulfilled.

Summarizing, in these two configurations the Birkhoff theorem holds because only the graviton propagates. In the next subsections we shall explore if there are more complex PG actions where the Birkhoff theorem holds. 

\subsection{Birkhoff theorem in stable configurations}

In Section \ref{2.3} we established that only the two scalar degrees of freedom present in the quadratic Poincar\'e Gauge gravity action \eqref{actionPGT} can propagate without introducing instabilities. The restricted action that describes these two scalars propagating at the same time is given in \eqref{eq:biscalar3} (or \eqref{eq:biscalar2} in the Einstein frame). Before going any further, we need to clarify a few concepts. In modified theories that introduce new degrees of freedom it is very easy to disprove the Birkhoff theorem in its more strict form. This is because the addition of propagating degrees of freedom into the field equations usually breaks the arguments that are used in GR to prove that an spherically symmetric solution in vacuum will be neccesarily static, which is the first statement of the Birkhoff theorem. Moreover, many of these modifications are introduced in order to reproduce the accelerating expansion of the Universe, hence mimicking a \emph{cosmological constant}. This implies that the de Sitter–Schwarzschild metric \cite{1950SRToh..34..160N} will be a solution of the vacuum field equations, which once again goes against the Birkhoff theorem.

Nevertheless, there exists a somewhat ``modification'' of the theorem, known as the \emph{no-hair theorem}\footnote{In fact, there is not a general rigorous proof of this theorem, so strictly speaking it should be denoted as a conjecture, as it is usual in the mathematical literature.}, that is quite relevant in any extension of GR. This theorem states that the external gravitational and electromagnetic fields of a stationary black hole (a black hole that has settled down into its ``final'' state) are determined uniquely by the BH mass $M$, charge $Q$, and intrinsic angular momentum $L$ \cite{Misner:1974qy}. If the extra degrees of freedom of the theory introduces new properties (``hair'') with respect to the BHs of GR, we will say that BHs are hairy in this theory. Actually, even if the BHs of the modified theory remain ``bald'', conceptually there might be physical differences between the definition of the three mentioned observables. A clear example is the BH solution given in Subsection \ref{3.1.2}, where the charge is not of electromagnetic nature, but it is sourced by torsion.

Going back to stable quadratic PG theories, we can easily establish that as long as we consider only the propagation of one of the two scalars, the BH solutions will not have hair. This is due to the result obtained by Sotiriou and Faraoni in \cite{Sotiriou:2011dz}, where they proved that generalised Brans-Dicke theories do not introduce any new observables in the BH solutions\footnote{Given that the solutions are asymptotically flat.}.\\
With respect to the bi-scalar theory, since there is an interaction term between the two scalars it is not straightforward to show if the previous result is going to hold. We shall study it in the following by looking at the field equations, which are obtained by performing variations of the action with respect to the metric and the two scalar modes. We will obtain these equations in the Einstein frame, so that there are no second derivatives of the scalar fields in the variation with respect to the metric $\tilde{g}$. They have the following form:

\begin{itemize}

\item Variations with respect to the metric (or \emph{Einstein Equations}):
\begin{equation}
\frac{\delta S}{\delta \tilde{g}^{\mu\nu}}=0\Longrightarrow\frac{\delta\mathcal{L}}{\delta \tilde{g}^{\mu\nu}}-\frac{1}{2}\tilde{g}_{\mu\nu}\mathcal{L}=0.
\end{equation}
Hence we have
\begin{eqnarray}
\label{biscalar:Einstein}
\tilde{G}_{\mu\nu}&=&\frac{1}{a_{0}}\left(1-\frac{12M_{S}^{2}\left(\chi\right)}{A\left(\chi,\phi\right)}\right)\left(\partial_{\mu}\chi\partial_{\nu}\chi-\frac{1}{2}\tilde{g}_{\mu\nu}\partial_{\rho}\chi\partial^{\rho}\chi\right)
\nonumber
\\
&&-\frac{18\alpha^{2}M_{T}^{2}\left(\chi\right)}{A\left(\chi,\phi\right)\chi}\left(\partial_{\mu}\phi\partial_{\nu}\phi-\frac{1}{2}\tilde{g}_{\mu\nu}\partial_{\rho}\phi\partial^{\rho}\phi\right)
\nonumber
\\
&&+\frac{48\alpha^{2}\phi}{A\left(\chi,\phi\right)\chi}\left(\partial_{\left(\mu\right.}\chi\partial_{\left.\nu\right)}\phi-\frac{1}{2}\tilde{g}_{\mu\nu}\partial_{\rho}\chi\partial^{\rho}\phi\right)
\nonumber
\\
&&+\frac{a_{0}}{2}\tilde{g}_{\mu\nu}\frac{\mathcal{U}\left(\chi,\phi\right)}{\chi^{2}},
\end{eqnarray}
where the explicit expressions of $\mathcal{U}\left(\chi,\phi\right)$, $M_{T}^{2}\left(\chi\right)$ and $M_{S}^{2}\left(\chi\right)$ are given in Eq. \eqref{biscalarfunctions}, and 
\begin{equation}
A\left(\chi,\phi\right)=(4\alpha\phi)^{2}-9M_{S}^{2}(\chi)M_{T}^{2}(\chi).
\end{equation}

\item Variations with respect to the scalar $\chi$:
\begin{equation}
\frac{\delta S}{\delta\chi}=0\Longrightarrow\frac{\partial\mathcal{L}}{\partial\chi}-\partial_{\mu}\left(\frac{\partial\mathcal{L}}{\partial\left(\partial_{\mu}\chi\right)}\right)=0.
\end{equation}
Therefore we obtain
\begin{eqnarray}
\label{biscalar:chieq}
&&\overline{\mathcal{U}}\left(\chi,\phi\right)+\left(A\left(\chi,\phi\right)-12M_{S}^{2}\left(\chi\right)\right)\Box\chi+\frac{\left(12M_{S}^{2}\left(\chi\right)\right)^{2}-\left(4\alpha\phi\right)^{2}}{A\left(\chi,\phi\right)}\partial_{\rho}\chi\partial^{\rho}\chi
\nonumber
\\
&&+\frac{48a_{0}\alpha^{2}\phi}{\chi}\Box\phi+\frac{3a_{0}\alpha^{2}\left[H_{1}\left(\chi\right)-96\alpha^{2}M_{T}^{2}\left(\chi\right)\phi^{2}-384\alpha^{2}\chi\phi^{2}\right]}{A\left(\chi,\phi\right)\chi^{2}}\partial_{\rho}\phi\partial^{\rho}\phi
\nonumber
\\
&&+\frac{768\alpha^{2}M_{S}^{2}\left(\chi\right)\phi}{A\left(\chi,\phi\right)}\partial_{\rho}\chi\partial^{\rho}\phi=0,
\end{eqnarray}
where
\begin{eqnarray}
&&\overline{\mathcal{U}}\left(\chi,\phi\right)=A\left(\chi,\phi\right)\left[\frac{a_{0}^{2}\left(a_{0}^{2}-a_{0}\chi+4b_{1}\alpha\phi^{2}\right)}{2b_{1}\chi^{3}}\right],
\\
&&H_{1}\left(\chi\right)=18M_{S}^{2}\left(\chi\right)M_{T}^{2}\left(\chi\right)\left(3M_{T}^{2}\left(\chi\right)-8\chi\right)+\frac{9}{2}\chi M_{T}^{2}\left(\chi\right).
\end{eqnarray}

\item Variations with respect to the pseudo-scalar $\phi$:
\begin{equation}
\frac{\delta S}{\delta\phi}=0\Longrightarrow\frac{\partial\mathcal{L}}{\partial\phi}-\partial_{\mu}\left(\frac{\partial\mathcal{L}}{\partial\left(\partial_{\mu}\phi\right)}\right)=0.
\end{equation}
Accordingly we get
\begin{eqnarray}
&&-\frac{2a_{0}^{2}\alpha\phi A\left(\chi,\phi\right)}{\chi}+48a_{0}\alpha^{2}\phi\Box\chi-\frac{8\alpha^{2}\phi\left(H_{2}\left(\chi\right)+96a_{0}\alpha^{2}\phi^{2}\right)}{A\left(\chi.\phi\right)\chi}\partial_{\rho}\chi\partial^{\rho}\chi
\nonumber
\\
&&-36a_{0}\alpha^{2}M_{T}^{2}\left(\chi\right)\Box\phi+\frac{\left(24\alpha^{2}\right)^{2}a_{0}M_{T}^{2}\left(\chi\right)\phi}{A\left(\chi,\phi\right)}\partial_{\rho}\phi\partial^{\rho}\phi
\nonumber
\\
&&-\frac{6a_{0}\alpha^{2}\left[H_{3}\left(\chi\right)-96\alpha^{2}M_{T}^{2}\left(\chi\right)\phi^{2}-128\alpha^{2}\chi\phi^{2}\right]}{A\left(\chi,\phi\right)\chi}\partial_{\rho}\chi\partial^{\rho}\phi=0,
\end{eqnarray}
where
\begin{eqnarray}
&&H_{2}\left(\chi\right)=48\chi^{2}M_{S}^{2}\left(\chi\right)-18a_{0}\left(12M_{T}^{2}\left(\chi\right)M_{S}^{2}\left(\chi\right)-16M_{S}^{2}\left(\chi\right)\chi\right.
\nonumber
\\
&&\left.\phantom{aaaaaaa}+M_{T}^{2}\left(\chi\right)\chi\right),
\\
&&H_{3}\left(\chi\right)=\left(3M_{T}^{2}\left(\chi\right)\right)^{2}\left(6M_{S}^{2}\left(\chi\right)+\frac{\chi}{2}\right).
\end{eqnarray}

\end{itemize}

Let us stress that to obtain the field equations we have not considered extra matter apart from the scalar and pseudo-scalar field since we are interested in vacuum solutions. Now, in order to shed some light on the Birkhoff theorem and the no-hair theorem in this theory we shall consider the following spherically symmetric and static four-dimensional metric
\begin{equation}
{\rm d}s^{2}=-\psi\left(r\right){\rm d}t^{2}+\frac{1}{\psi\left(r\right)}{\rm d}r^{2}+r^{2}\left({\rm d}\theta^{2}+{\rm sin}^{2}\theta\,{\rm d}\varphi^{2}\right).
\label{ansatz:simp}
\end{equation}
Although this is not the most general metric that meets the mentioned properties, the results that its study brings shall help us prove if the Birkhoff and no-hair theorems apply in the bi-scalar theory \eqref{eq:biscalar2}. The torsion scalar and pseudo-scalar shall also depend only on the radial component in order to maintain spherical symmetry and staticity. \\
Taking these precepts into account we can obtain the following expression by adding the $\left(t,t\right)$ and the $\left(\theta,\theta\right)$ Einstein Equations
\begin{equation}
r^{2}\psi''\left(r\right)=2\psi\left(r\right)-2.
\end{equation}
This equation has the solution
\begin{equation}
\psi\left(r\right)=1+\frac{C_{1}}{r}+C_{2}r^{2},
\label{dSS}
\end{equation}
which is the well-known de Sitter-Schwarzschild metric, with $C_1$ and $C_2$ arbitrary constants. This result is already telling us that the Birkoff's theorem does not apply in the bi-scalar theory, since we can find a spherically symmetric solution that is different from pure Schwarzschild. Moreover, we can check that the de Sitter term of the solution \eqref{dSS} is a direct consequence of having a potential different from zero. Indeed, if we impose that $\mathcal{U}\left(\chi,\phi\right)=0$, hence finding a relation between the two scalars, and insert \eqref{dSS} into the sum of the $\left(t,t\right)$ and $\left(r,r\right)$ Einstein Equations we find that $\chi\left(r\right)$, and consequently $\phi\left(r\right)$, both need to be a constant. Then, aplying this result into the $\left(t,t\right)$ Einstein equation we obtain that $C_2$ needs to be zero. This occurs because when the two scalars are constants the theory reduces to an $\mathring{R}+\mathring{R}^2$ theory, where the Birkhoff theorem holds \cite{delaCruzDombriz:2009et}.\\
Let us note that the potential term is also present when only one of the scalar modes propagates, and consequently, the metric \eqref{dSS} will also be a solution of the field equations. Hence, using the simplified ansatz \eqref{ansatz:simp} we have been able to prove that \bf{the only stable quadratic Poincar\'e Gauge theories that fulfill the Birkhoff theorem are the ones studied by Nieh and Rauch in the 1980s, \emph{i.e.} \eqref{Rauch:1} and \eqref{Rauch:2}}.\\

With respect to the no-hair theorem, we have already shown that if the potential is different from zero we can obtain a BH solution that breaks the theorem's conclusions. Therefore, in general, the no-hair theorem would not apply to the bi-scalar theory. Additionally, we have also seen that if we require the potential $\mathcal{U}$ to be zero then we can write one scalar in terms of the other one. If we take into account this result in the Lagrangian \eqref{eq:biscalar2}, we can clearly observe that such a Lagrangian will now describe a generalised Brans-Dicke theory, where we know that in the absence of a potential the no-hair theorem would apply \cite{Sotiriou:2011dz}.\\
Nevertheless, we still have one possibility to consider, which is the case of imposing that the BH solution needs to be asymptocally flat. This implies that the two scalars must behave in such a way that the potential tends to zero when $r$ goes to infinity. The search of this kind of solutions that may introduce hair apart from the cosmological one is beyond the scope of this thesis. Nevertheless, it is interesting to explore some aspects of these asymptotically flat solutions. In order to study this case we shall consider the most general static and spherically symmetric metric, namely\footnote{We have chosen the components of the metric such that $g_{rr}=-\frac{1}{g_{tt}}$, which implies that the angular part shall remain generic.}
\begin{equation}
{\rm d}s^{2}=-\psi\left(r\right){\rm d}t^{2}+\frac{1}{\psi\left(r\right)}{\rm d}r^{2}+\rho\left(r\right)^{2}\left({\rm d}\theta^{2}+{\rm sin}^{2}\theta\,{\rm d}\varphi^{2}\right).
\label{ansatz:gen}
\end{equation}
Then, by multiplying the $\left(t,t\right)$ Einstein Equation by $\psi\left(r\right)^2$ and subtracting the $\left(r,r\right)$ one we find
\begin{equation}
\frac{a_{0}}{8\chi^{2}}\mathcal{U}\left(\chi,\phi\right)=\frac{-1+\rho\left(r\right)\rho'\left(r\right)\psi\left(r\right)+\psi\left(r\right)\left(\rho'\left(r\right)^{2}+\rho\left(r\right)\rho''\left(r\right)\right)}{\rho\left(r\right)^{2}}.
\label{nohair:gen1}
\end{equation}
Also, adding the $\left(t,t\right)$ Einstein Equation multiplied by $\psi\left(r\right)$ to the $\left(\theta,\theta\right)$ one multiplied by $\rho\left(r\right)^2$ we obtain
\begin{equation}
\frac{-2+2\psi\left(r\right)\left(\rho'\left(r\right)^{2}+\rho\left(r\right)\rho''\left(r\right)\right)-\rho\left(r\right)\psi''\left(r\right)}{\rho\left(r\right)^{2}}=0.
\label{nohair:gen2}
\end{equation}
Now, multiplying \eqref{nohair:gen1} by 4 and substracting \eqref{nohair:gen1} we arrive at the following result
\begin{equation}
\tilde{R}=-\frac{a_{0}}{2\chi^{2}}\mathcal{U}\left(\chi,\phi\right),
\label{curvature:potential}
\end{equation}
which relates the scalar curvature with the potential of the bi-scalar theory. Here we can clearly see for the general case why a non-vanishing potential $\mathcal{U}$ at infinity makes that the solutions cannot be asymptotically flat, which is why one needs to impose that the potential tends to zero as the radial coordinate goes to infinity. \\
Moreover, from Eq. \eqref{curvature:potential} we can extract that all spherically symmetric BHs with null scalar curvature would not have hair induced by the bi-scalar theory. This is because if the left-hand-side of \eqref{curvature:potential} is zero we can express one of the scalars $\chi$ or $\phi$ in terms of the other one. We have already seen that the fact that the two scalars are related implies that the no-hair theorem holds, given that the spacetime is asymptotically flat, as it is the case.\\

Summarising, what we have obtained in this subsection are three very important results, in particular
\begin{enumerate}

\item The \bf{only} stable quadratic PG theories that fulfill the Birkhoff theorem are \eqref{Rauch:1} and \eqref{Rauch:2}. 

\item The no-hair theorem would hold for the two PG theories that describe the scalar and pseudo-scalar propagation separately, provided that the metric is asymptotically flat. 

\item In the bi-scalar theory the no-hair theorem would apply for spherically symmetric solutions with null scalar curvature.

\end{enumerate}
At this stage, one can also wonder what would happen if we consider PG theories that propagate ghostly degrees of freedom. Does the Birkhoff theorem hold then? Are the stability of the theory and the proof of the Birkhoff theorem related? We shall answer those questions in the following subsections.

\subsection{Instabilities and the Birkhoff theorem}

In this subsection we will explore the possible causal relation between the stability of the PG theories and the Birkhoff theorem. In particular, we shall study if there exists a logical connection between the two, \emph{i.e.} finding whether the consideration of a PG stable theory is a necessary and/or sufficient condition for the Birkhoff theorem to hold. From the analysis performed in the previous sections it is easy to elucidate this question by observing two particular scenarios. 

On the one hand, we know from Section \ref{2.3} that the general bi-scalar theory is stable under some parameter constraints, and also that the Birkhoff theorem does not hold due to the potential term. Therefore, we conclude that the stability of the theory is not a sufficient condition for the Birkhoff theorem to hold.

On the other hand, we have established that in the PG theory given by the Lagrangian \eqref{Rauch:1} the Birkhoff theorem is fulfilled for any choice of the parameters. Moreover, it is known that if $\frac{d^{2}}{dR^{2}}\left(\alpha R^{2}\right)=2\alpha<0$ then that theory suffers from a Dolgov-Kawasaki instability \cite{Dolgov:2003px,Seifert:2007fr}. Hence, the stability of the theory is not a necessary condition for the Birkhoff theorem to hold.

Accordingly, we can summarise the above discussion in the following simple logical inference
\begin{equation}
\rm{Stability\,\,conditions}\nLeftrightarrow\rm{Birkhoff\,\,theorem}
\end{equation}

Indeed, the fact that these two aspects are not related reveals a crucial statement: the fact that a theory is unstable does not mean that one cannot find particular stable solutions. This is why sometimes ghostly behaviour  in modified gravity theories goes unnoticed. In particular, regarding PG theories, there are many ``viable'' solutions that have been proposed in the literature which come from actions that propagate more than the two scalars, hence incurring in ghost instabilities. Therefore, such ``healthy'' solutions look like stable configurations, but the moment one performs perturbations up to a certain order these instabilities would be observed.

To illustrate the scenario described above in more detail we shall study an unstable theory and observe how the spherical symmetry of the spacetime does not allow us to see the ghosts. Such a statement shall be proved by showing that the Birkhoff theorem holds in the situations that we are going to consider. The calculations below are based on our work {\bf{\nameref{P2}}}\footnote{If the reader wants to explore the source of these subsections calculations, as presented in {\bf{\nameref{P2}}}, it would be noticed that in that work we regard the theory \eqref{unstable:PGlag} as stable. Later on, we proved that this is not true, as one can corroborate in Section \ref{2.3} of this Thesis. Of course, given this result, we could not include the results of {\bf{\nameref{P2}}} as a study of ``stable torsion theories''. Still, we strongly believe that it is a relevant study, since it can be used to illustrate how imposing a high symmetry of the solutions, \emph{e.g.} spherical symmetry or maximal symmetry (like in cosmological models), makes us blind to the ghosts, as long as we do not consider perturbation theory in those solutions. Moreover, it is a great example of the expertise acquired during the three years of PhD research, that allows us to treat this problem in a simpler and more efficient way.}.

Accordingly, let us consider the following theory
\begin{equation}
\label{unstable:PGlag}
\mathcal{L}=a_{0}\mathring{R}+a_{1}T_{\alpha\beta\gamma}T^{\alpha\beta\gamma}+a_{2}T_{\alpha\beta\gamma}T^{\beta\gamma\alpha}+a_{3}T_{\,\,\alpha\beta}^{\beta}T_{\gamma}^{\,\,\alpha\gamma}+2b_{1}\mathring{R}_{\left[\alpha\beta\right]}R^{\left[\alpha\beta\right]},
\end{equation}
Using the relation between a general connection with torsion and null non-metricity and the Levi-Civita connection \eqref{LCWeitzenbock}, and the expression of the contortion with respect to the torsion tensor \eqref{contorsion}, we find that \eqref{unstable:PGlag} can be rewritten as
\begin{equation}
\mathcal{L}=a_{0}\mathring{R}+T^{2}+T^{4}+\mathring{\nabla}T\mathring{\nabla}T+T^{2}\mathring{\nabla}T,
\label{unstable:PGlag2}
\end{equation}
where
\begin{equation}
T^{2}=\left(a_{1}+\frac{a_{0}}{4}\right)T_{\alpha\beta\gamma}T^{\alpha\beta\gamma}+\left(a_{2}-\frac{a_{0}}{2}\right)T_{\alpha\beta\gamma}T^{\beta\gamma\alpha}+\left(a_{3}-a_{0}\right)T_{\,\,\alpha\beta}^{\beta}T_{\gamma}^{\,\,\alpha\gamma},
\end{equation}
\begin{eqnarray}
&&T^{4}=\frac{b_{1}}{4}T_{\mu}\,^{\nu\rho}T^{\mu\sigma\lambda}T_{\sigma\lambda}\,^{\alpha}T_{\nu\rho\alpha}+\frac{b_{1}}{4}T^{\mu\sigma\lambda}T_{\sigma\lambda}\,^{\alpha}T_{\alpha}\,^{\nu\rho}T_{\nu\mu\rho}
\nonumber
\\
&&+b_{1}T_{\mu}T^{\mu\sigma\lambda}T_{\sigma}\,^{\nu\rho}T_{\nu\lambda\rho}+\frac{b_{1}}{2}T_{\mu}T_{\nu}T^{\mu\sigma\lambda}T^{\nu}\,_{\sigma\lambda},
\end{eqnarray}
\begin{eqnarray}
\mathring{\nabla}T\mathring{\nabla}T&=&-b_{1}\mathring{\nabla}_{\mu}T^{\nu}\mathring{\nabla}_{\nu}T^{\mu}+b_{1}\mathring{\nabla}_{\mu}T_{\nu}\mathring{\nabla}^{\mu}T^{\nu}+\frac{b_{1}}{2}\mathring{\nabla}_{\mu}T^{\mu\nu\rho}\mathring{\nabla}_{\sigma}T^{\sigma}\,_{\nu\rho}
\nonumber
\\
&&-2b_{1}\mathring{\nabla}^{\mu}T^{\nu}\mathring{\nabla}_{\rho}T^{\rho}\,_{\nu\mu},
\end{eqnarray}
\begin{eqnarray}
T^{2}\mathring{\nabla}T&=&b_{1}T^{\mu\nu\rho}T_{\nu\rho}\,^{\sigma}\left(2\mathring{\nabla}_{\left[\mu\right.}T_{\left.\sigma\right]}+\mathring{\nabla}_{\lambda}T^{\lambda}\,_{\mu\sigma}\right)
\nonumber
\\
&&+b_{1}T_{\mu}T^{\mu\nu\rho}\left(2\mathring{\nabla}_{\rho}T_{\nu}-\mathring{\nabla}_{\lambda}T^{\lambda}\,_{\nu\rho}\right).
\end{eqnarray}
In order to study the behaviour of this theory we first shall obtain the vacuum field equations. The variation with respect to the metric of the action with Lagrangian density \eqref{unstable:PGlag2} will provide us the Einstein Equations, and analogously the variation with respect to the torsion tensor will supply us with the so-called \emph{Cartan Equations}. Both of them are summarised in the following\footnote{Since the spacetime conventions are different in the Thesis with respect to {\bf{\nameref{P2}}}, the form of the field equations also differ. Also, the formalism presented here simplies the one used in {\bf{\nameref{P2}}}.}:
\begin{itemize}

\item {\bf{Cartan Equations}}:\\
Schematically, they can be written as 
\begin{equation}
\label{cartan1}
T+b_{1}T^{3}+b_{1}T\mathring{\nabla}T+b_{1}\mathring{\nabla}^{2}T=0,
\end{equation}
where the different terms are defined as 
\begin{equation}
T=\frac{1}{2}\left(a_{0}+4a_{1}\right)T^{\mu}\,_{\nu\rho}+\left(2a_{2}-a_{0}\right)T_{\left[\nu\rho\right]}\,^{\mu}+2\left(a_{0}-a_{3}\right)\delta_{\left[\nu\right.}^{\mu}T_{\left.\rho\right]},
\end{equation}
\begin{eqnarray}
T^3&=&-\frac{1}{2}T^{\mu\sigma\lambda}T_{\sigma\alpha\lambda}T_{\left[\nu\rho\right]}\,^{\alpha}+\frac{1}{2}T^{\sigma\mu}\,_{\left[\nu\right.}T_{\left.\rho\right]}\,^{\lambda\alpha}T_{\lambda\sigma\alpha}+\frac{1}{2}T_{\left[\nu\rho\right]}\,^{\alpha}T_{\alpha\sigma\lambda}T^{\sigma\mu\lambda}
\nonumber
\\
&&-\frac{1}{2}T^{\sigma\alpha}\,_{\left[\nu\right.}T^{\lambda\mu}\,_{\left.\rho\right]}T_{\lambda\sigma\alpha}-\delta_{\left[\nu\right.}^{\mu}T_{\left.\rho\right]}\,^{\sigma\lambda}T_{\sigma}\,^{\alpha\beta}T_{\alpha\lambda\beta}-T^{\mu}T_{\left[\nu\right.}\,^{\sigma\lambda}T_{\sigma\lambda\left|\rho\right]}
\nonumber
\\
&&+T^{\mu}T_{\sigma}T^{\sigma}\,_{\nu\rho}-T_{\left[\nu\rho\right]}\,^{\alpha}T^{\sigma\mu}\,_{\alpha}T_{\sigma}+T^{\sigma\mu}\,_{\left[\nu\right.}T^{\alpha}\,_{\left.\rho\right]\sigma}T_{\alpha}
\nonumber
\\
&&-\delta_{\left[\nu\right.}^{\mu}T_{\left.\rho\right]}\,^{\sigma\lambda}T^{\alpha}\,_{\sigma\lambda}T_{\alpha},
\end{eqnarray}
\begin{eqnarray}
T\mathring{\nabla}T&=&T^{\sigma\lambda}\,_{\left[\nu\right.}\mathring{\nabla}^{\mu}T_{\left.\rho\right]\sigma\lambda}-T_{\left[\nu\right.}\,^{\sigma\lambda}\mathring{\nabla}^{\mu}T_{\sigma\lambda\left|\rho\right]}+\mathring{\nabla}^{\mu}\left(T^{\sigma}T_{\sigma\nu\rho}\right)+T_{\left[\nu\rho\right]}\,^{\sigma}\mathring{\nabla}^{\mu}T_{\sigma}
\nonumber
\\
&&-2T^{\mu}\mathring{\nabla}_{\left[\nu\right.}T_{\left.\rho\right]}-T^{\sigma\mu}\,_{\left[\nu\right.}\mathring{\nabla}_{\left.\rho\right]}T_{\sigma}+\delta_{\left[\nu\right.}^{\mu}T^{\sigma\lambda\alpha}\mathring{\nabla}_{\sigma}T_{\lambda\alpha\left|\rho\right]}-T_{\left[\nu\rho\right]}\,^{\sigma}\mathring{\nabla}_{\sigma}T^{\mu}
\nonumber
\\
&&+T^{\sigma\mu}\,_{\left[\nu\right.}\mathring{\nabla}_{\sigma}T_{\left.\rho\right]}+T_{\left[\nu\rho\right]}\,^{\sigma}\mathring{\nabla}_{\lambda}T^{\lambda\mu}\,_{\sigma}-T^{\mu}\mathring{\nabla}_{\sigma}T^{\sigma}\,_{\nu\rho}-T^{\sigma\mu}\,_{\left[\nu\right.}\mathring{\nabla}_{\lambda}T^{\lambda}\,_{\left.\rho\right]\sigma}
\nonumber
\\
&&-2\delta_{\left[\nu\right.}^{\mu}T_{\left.\rho\right]}\,^{\sigma\lambda}\mathring{\nabla}_{\lambda}T_{\sigma}-2\delta_{\left[\nu\right.}^{\mu}T^{\sigma\lambda}\,_{\left.\rho\right]}\mathring{\nabla}_{\lambda}T_{\sigma}-\delta_{\left[\nu\right.}^{\mu}T^{\sigma\lambda\alpha}\mathring{\nabla}_{\alpha}T_{\left.\rho\right]\sigma\lambda }
\nonumber
\\
&&+\delta_{\left[\nu\right.}^{\mu}T^{\sigma}\mathring{\nabla}_{\alpha}T_{\sigma\left|\rho\right]}\,^{\alpha}+\delta_{\left[\nu\right.}^{\mu}T_{\left.\rho\right]}\,^{\sigma\lambda}\mathring{\nabla}_{\alpha}T^{\alpha}\,_{\sigma\lambda}+\delta_{\left[\nu\right.}^{\mu}T^{\sigma\lambda}\,_{\left.\rho\right]}\mathring{\nabla}_{\alpha}T^{\alpha}\,_{\sigma\lambda},
\end{eqnarray}
\begin{eqnarray}
\mathring{\nabla}^{2}T&=&-2\mathring{\nabla}^{\mu}\mathring{\nabla}_{\left[\nu\right.}T_{\left.\rho\right]}-\mathring{\nabla}^{\mu}\mathring{\nabla}_{\sigma}T^{\sigma}\,_{\nu\rho}-2\delta_{\left[\nu\right.}^{\mu}\mathring{\nabla}_{\sigma}\mathring{\nabla}_{\left.\rho\right]}T^{\sigma}+2\delta_{\left[\nu\right.}^{\mu}\mathring{\nabla}_{\sigma}\mathring{\nabla}_{\lambda}T^{\lambda\sigma}\,_{\left.\rho\right]}
\nonumber
\\
&&+2\delta_{\left[\nu\right.}^{\mu}\mathring{\nabla}_{\sigma}\mathring{\nabla}^{\sigma}T_{\left.\rho\right]}.
\end{eqnarray}

\item {\bf{Einstein Equations}}:

In order to make variations with respect to the metric we shall express the different terms in the action \eqref{unstable:PGlag2} as
\begin{eqnarray}
&&T^{2}=f_{1}\,^{\rho\sigma\beta\gamma}\,_{\mu\alpha}T^{\mu}\,_{\rho\sigma}T^{\alpha}\,_{\beta\gamma}=f_{1}\hat{T}^{2},
\nonumber
\\
&&T^{4}=f_{2}\,^{\nu\rho\lambda\gamma\beta_{1}\beta_{2}\beta_{3}\beta_{4}}\,_{\mu\sigma\alpha_{1}\alpha_{2}}T^{\mu}\,_{\nu\rho}T^{\sigma}\,_{\lambda\gamma}T^{\alpha_{1}}\,_{\beta_{1}\beta_{2}}T^{\alpha_{2}}\,_{\beta_{3}\beta_{4}}=f_{2}\hat{T}^{4},
\nonumber
\\
&&\mathring{\nabla}T\mathring{\nabla}T=f_{3}\,^{\mu\rho\sigma\lambda\beta\gamma}\,_{\nu\alpha}\mathring{\nabla}_{\mu}T^{\nu}\,_{\rho\sigma}\mathring{\nabla}_{\lambda}T^{\alpha}\,_{\beta\gamma}=f_{3}\mathring{\nabla}\hat{T}\mathring{\nabla}\hat{T},
\nonumber
\\
&&T^{2}\mathring{\nabla}T=f_{4}\,^{\nu\rho\lambda\alpha\gamma\beta_{2}\beta_{3}}\,_{\mu\sigma\beta_{1}}T^{\mu}\,_{\nu\rho}T^{\sigma}\,_{\lambda\alpha}\mathring{\nabla}_{\gamma}T^{\beta_{1}}\,_{\beta_{2}\beta_{3}}=f_{4}\hat{T}^{2}\mathring{\nabla}\hat{T},
\end{eqnarray}
where
\begin{eqnarray}
f_{1}\,^{\rho\sigma\beta\gamma}\,_{\mu\alpha}&=&\left(a_{1}+\frac{a_{0}}{4}\right)g_{\mu\alpha}g^{\rho\beta}g^{\sigma\gamma}+\left(a_{2}-\frac{a_{0}}{2}\right)\delta_{\mu}^{\gamma}\delta_{\alpha}^{\rho}g^{\sigma\beta}
\nonumber
\\
&&+\left(a_{3}-a_{0}\right)\delta_{\mu}^{\sigma}\delta_{\alpha}^{\gamma}g^{\rho\beta},
\end{eqnarray}
\begin{eqnarray}
f_{2}\,^{\nu\rho\lambda\gamma\beta_{1}\beta_{2}\beta_{3}\beta_{4}}\,_{\mu\sigma\alpha_{1}\alpha_{2}}&=&\frac{b_{1}}{4}\delta_{\alpha_{1}}^{\lambda}\delta_{\alpha_{2}}^{\nu}g_{\mu\sigma}g^{\rho\beta_{3}}g^{\gamma\beta_{1}}g^{\beta_{2}\beta_{4}}
\nonumber
\\
&&+\frac{b_{1}}{4}\delta_{\mu}^{\beta_{3}}\delta_{\sigma}^{\nu}\delta_{\alpha_{1}}^{\gamma}\delta_{\alpha_{2}}^{\beta_{1}}g^{\rho\lambda}g^{\beta_{2}\beta_{4}}
\nonumber
\\
&&+b_{1}\delta_{\mu}^{\rho}\delta_{\sigma}^{\nu}\delta_{\alpha_{1}}^{\lambda}\delta_{\alpha_{2}}^{\beta_{1}}g^{\gamma\beta_{3}}g^{\beta_{2}\beta_{4}}
\nonumber
\\
&&+\frac{b_{1}}{2}\delta_{\mu}^{\rho}\delta_{\sigma}^{\gamma}\delta_{\alpha_{1}}^{\nu}\delta_{\alpha_{2}}^{\lambda}g^{\beta_{1}\beta_{3}}g^{\beta_{2}\beta_{4}},
\end{eqnarray}
\begin{eqnarray}
f_{3}\,^{\mu\rho\sigma\lambda\beta\gamma}\,_{\nu\alpha}&=&-b_{1}\delta_{\nu}^{\sigma}\delta_{\alpha}^{\gamma}g^{\mu\beta}g^{\rho\lambda}+b_{1}\delta_{\nu}^{\sigma}\delta_{\alpha}^{\gamma}g^{\mu\lambda}g^{\rho\beta}+\frac{b_{1}}{2}\delta_{\nu}^{\mu}\delta_{\alpha}^{\lambda}g^{\rho\beta}g^{\sigma\gamma}
\nonumber
\\
&&-2b_{1}\delta_{\nu}^{\sigma}\delta_{\alpha}^{\lambda}g^{\mu\gamma}g^{\rho\beta},
\end{eqnarray}
\begin{eqnarray}
f_{4}\,^{\nu\rho\lambda\alpha\gamma\beta_{2}\beta_{3}}\,_{\mu\sigma\beta_{1}}&=&b_{1}\delta_{\mu}^{\gamma}\delta_{\sigma}^{\nu}\delta_{\beta_{1}}^{\beta_{3}}g^{\rho\lambda}g^{\alpha\beta_{2}}-b_{1}\delta_{\mu}^{\beta_{2}}\delta_{\sigma}^{\nu}\delta_{\beta_{1}}^{\beta_{3}}g^{\rho\lambda}g^{\alpha\gamma}
\nonumber
\\
&&+b_{1}\delta_{\mu}^{\beta_{2}}\delta_{\sigma}^{\nu}\delta_{\beta_{1}}^{\gamma}g^{\rho\lambda}g^{\alpha\beta_{3}}+2b_{1}\delta_{\mu}^{\rho}\delta_{\sigma}^{\nu}\delta_{\beta_{1}}^{\beta_{3}}g^{\lambda\beta_{2}}g^{\alpha\gamma}
\nonumber
\\
&&-b_{1}\delta_{\mu}^{\rho}\delta_{\sigma}^{\nu}\delta_{\beta_{1}}^{\gamma}g^{\lambda\beta_{2}}g^{\alpha\beta_{3}},
\end{eqnarray}
Then, the Einstein field equations will be schematically given by 
\begin{eqnarray}
&&\mathring{G}_{\mu\nu}+\frac{\partial f_{1}}{\partial g^{\mu\nu}}\hat{T}^{2}+\frac{\partial f_{2}}{\partial g^{\mu\nu}}\hat{T}^{4}+\frac{\partial f_{3}}{\partial g^{\mu\nu}}\mathring{\nabla}\hat{T}\mathring{\nabla}\hat{T}+\frac{\partial f_{4}}{\partial g^{\mu\nu}}\hat{T}^{2}\mathring{\nabla}\hat{T}
\nonumber
\\
&&-\frac{1}{2}g_{\mu\nu}\left(T^{2}+T^{4}+\mathring{\nabla}T\mathring{\nabla}T+T^{2}\mathring{\nabla}T\right)=0.
\end{eqnarray}
\end{itemize}

Now, in order to study the Birkhoff and no-hair theorems we shall consider the the most general spherically symmetric four-dimensional metric and torsion components. The metric will have the usual form~\cite{Wald:1984rg}
\begin{equation}
{\rm d}s^{2}=-\psi\left(\tilde{t},\,\rho\right){\rm d}\tilde{t}^{2}+\phi\left(\tilde{t},\,\rho\right){\rm d}\rho^{2}+\tilde{r}^{2}\left(\tilde{t},\,\rho\right)\left({\rm d}\theta^{2}+{\rm sin}^{2}\theta\,{\rm d}\varphi^{2}\right).
\label{metric}
\end{equation}
As widely known, under a suitable choice of coordinates this metric can be rewritten as 
\begin{equation}
{\rm d}s^{2}=-\psi\left(t,\,r\right){\rm d}t^{2}+\phi\left(t,\,r\right){\rm d}r^{2}+r^{2}\left({\rm d}\theta^{2}+{\rm sin}^{2}\theta\,{\rm d}\varphi^{2}\right).
\label{metric2}
\end{equation}

For the torsion components it is necessary to work out the constraints due to imposing form invariance under rotations. This is done using the well-known \emph{Killing equations}, as has been studied in~\cite{Sur:2013aia}. There, the authors obtained the non-zero components of the torsion field for different symmetry assumptions. In particular, the non-zero components for the torsion tensor in the spherically symmetric case are

\begin{equation}
\label{torsion}
\begin{cases}
T_{ttr}=a\left(t,r\right),\,\,\,\,T_{rtr}=b\left(t,r\right),\,\,\,\,T_{\theta_{i}t\theta_{i}}=c\left(t,r\right),\,\,\,\,T_{\theta_{i}r\theta_{i}}=d\left(t,r\right),\\
\,
\\
T_{t\theta_{i}\theta_{j}}=-T_{\theta_{i}t\theta_{j}}=\varepsilon_{ij}f\left(t,r\right),\,\,\,\,T_{r\theta_{i}\theta_{j}}=-T_{\theta_{i}r\theta_{j}}=\varepsilon_{ij}g\left(t,r\right),
\end{cases}
\end{equation}
where we have made the identification $\left\{ \theta_{1},\,\theta_{2}\right\} \equiv \left\{ \theta,\,\varphi\right\}$ and consequently $i,j=1,2$ with $i\neq j$ and $\varepsilon_{ab}$ is the Levi-Civita symbol with $a,b=1,2$. \\
A reader familiarised with the literature of the Birkhoff theorem in PG theories may realise that in all the previous studies, including {\bf{\nameref{P2}}}, there were eight independent components instead of six. This is due to the fact that it went unnoticed that the Killing equations impose that the components of the type $T_{t\theta_{i}\theta_{j}}$ and $T_{r\theta_{i}\theta_{j}}$ are totally antisymmetric. 

As we did in {\bf{\nameref{P2}}} we shall consider two physically relevant situations, in particular weak torsion and asymptotically flatness, where we will study if the Birkhoff and no-hair theorems apply.

\subsection{Weak torsion approximation}

In this subsection we shall deal with first-order perturbations on the torsion and see whether vacuum solutions different from Schwarzschild can be found. Such an approximation is physically motivated by the fact that experimental tests on torsion are compatible with this reasoning, since its effects are negligible when compared to the Riemannian curvature ones \cite{Lammerzahl:1997wk}. This would allow us to neglect the quartic terms in the torsion in the action \eqref{unstable:PGlag2}, and consequently the terms in the field equations that come from varying them.\\

\subsubsection{Determination of $f(t,r)$ and $g(t,r)$}
\label{Subsection-fg}

We will start with the $\left(\varphi,t,r\right)$ Cartan Equation
\begin{equation}
b_{1}f\left(t,r\right)=0,
\end{equation}
which clearly imposes that $f\left(t,r\right)$ must be null. Next we shall focus on the $\left(\varphi,r,\theta\right)$ Cartan Equation, namely
\begin{eqnarray}
&&2b_{1}\cot^{2}\theta\psi\left(t,r\right)^{2}\phi\left(t,r\right)g\left(t,r\right)+\Biggl\{ 2\psi\left(t,r\right)^{2}\left[2b_{1}+\left(2a_{1}+a_{2}\right)r^{2}\phi\left(t,r\right)\right]\Biggr.
\nonumber
\\
&&\Biggl.-b_{1}r\phi\left(t,r\right)\frac{\partial\psi\left(t,r\right)}{\partial t}+b_{1}r\psi\left(t,r\right)\frac{\partial\phi\left(t,r\right)}{\partial t}\Biggr\} g\left(t,r\right)
\nonumber
\\
&&+2b_{1}r\psi\left(t,r\right)\phi\left(t,r\right)\frac{\partial g\left(t,r\right)}{\partial t}=0.
\end{eqnarray}
It is easy to see that since neither the torsion nor the metric functions can depend on the $\theta$ coordinate this means that
\begin{equation}
g\left(t,r\right)=0.
\end{equation}

\subsubsection{Determination of $c(t,r)$ and $d(t,r)$}
\label{Subsection-cd}

Let us now explore the $\left(\theta,t,r\right)$ Cartan Equation
\begin{equation}
\left(\cos\left(2\theta\right)-7\right)c\left(t,r\right)+4r\left(\frac{\partial c\left(t,r\right)}{\partial r}-\frac{\partial d\left(t,r\right)}{\partial t}\right)=0.
\end{equation}
Since $c\left(t,r\right)$ does not depend on $\theta$ we have that 
\begin{equation}
c\left(t,r\right)=0.
\end{equation}
Also, the same equation imposes that $d\left(t,r\right)$ must be a function of $r$ only. We continue by looking at the $\left(t,t,r\right)$ Cartan Equation, which is given by
\begin{eqnarray}
8\left(2a_{1}-a_{2}+a_{3}\right)a\left(t,r\right)-F_{1}\left(t,r,\theta\right)d\left(r\right)=0,
\end{eqnarray}
where $F_1$ is a certain analytic function. Its explicit form is ommited for simplicity. Again, by the same arguments are before we obtain
\begin{equation}
d\left(r\right)=0.
\end{equation}

\subsubsection{Determination of $a(t,r)$ and $b(t,r)$}
\label{Subsection-ab}

At this time we just have two non-null torsion functions left, $a(t,r)$ and $b(t,r)$. To know their value we shall consider the system formed of the $\left(t,t,r\right)$ and $\left(r,t,r\right)$ Cartan Equations
\begin{equation}
\label{system:ab}
\begin{cases}
\begin{array}{c}
\left(2a_{1}-a_{2}+a_{3}\right)a\left(t,r\right)=0,\\
\,\\
\left[4b_{1}+\left(2a_{1}-a_{2}+a_{3}\right)r^{2}\phi\left(t,r\right)\right]b\left(t,r\right)=0,
\end{array} & \,\end{cases}
\end{equation}

There are two possibilities to solve this system
\begin{enumerate}

\item $2a_{1}-a_{2}+a_{3}=0$: This means that the first equation of the system \eqref{system:ab} holds. Also, the second one implies that 
\begin{equation}
b\left(t,r\right)=0.
\end{equation}
Now, we use these results in the $\left(\theta,r,\theta\right)$ Cartan Equation, that reads
\begin{equation}
\left(a_{3}-a_{0}\right)a\left(t,r\right)=0.
\label{eq:aa}
\end{equation}
Again, we have two ways this could be solved, namely:
\begin{itemize}

\item $a_{3}-a_{0}=0$: Then, all the Cartan Equations hold. Moreover, there is no influence of the remaining torsion function $a\left(t,r\right)$ in the Einstein Equations. Therefore the metric solution would be Schwarzschild and the Birkhoff theorem holds.

\item $a_{3}-a_{0}\ne 0$: In this case, Equation \eqref{eq:aa} imposes that $a\left(t,r\right)$. Hence, all the torsion functions are null and the Einstein Equations recover the GR form. Consequently, the Birkhoff theorem applies.

\end{itemize}

\item $2a_{1}-a_{2}+a_{3}\ne 0$: Then, from the first equation in \eqref{system:ab} we have
\begin{equation}
a\left(t,r\right)=0.
\end{equation}
Now,it is clear that the second equation of the system \eqref{system:ab} could be solved by having $\phi\left(t,r\right)=\frac{C}{r^2}$, and imposing a relation between $b_1$ and the $a_i$'s. Nevertheless, this choice is incompatible with the Einstein Field Equations, hence the only solution is $b\left(t,r\right)=0$. Then, since all the torsion functions are null the Birkhoff theorem holds.
\end{enumerate}

In Figure \ref{arbol:1} we present a tree of decision helping to clarify the reasoning developed to obtain $a(t,r)$ and $b(t,r)$.\\

\begin{figure}
\begin{center}

\begin{forest}
   [System of Eqs. \eqref{system:ab},draw={blue,thick},
      [{$2a_{1}-a_{2}+a_{3}=0$},draw,[{$b\left(t,r\right)=0$},draw,
      	[{$a_{3}-a_{0}=0$},draw,[{$a\left(t,r\right)\ne 0$, $b\left(t,r\right)=0$},draw={green,thick}]]
      	[{$a_{3}-a_{0}\ne 0$},draw,[{$a\left(t,r\right)=b\left(t,r\right)=0$},draw={green,thick}]]
        ]
	],
      [{$2a_{1}-a_{2}+a_{3}\ne 0$},draw,[{$a\left(t,r\right)=0$},draw,[Einstein equations[{$a\left(t,r\right)=b\left(t,r\right)=0$},draw={green,thick}]]]]
   ]
\end{forest}
\par\end{center}
\caption{Tree of decision representing the steps we have followed to obtain the values torsion functions $a(t,r)$ and $b(t,r)$. Also, the green colour represents if the Birkhoff theorem holds for that solution.}
\label{arbol:1}
\end{figure}
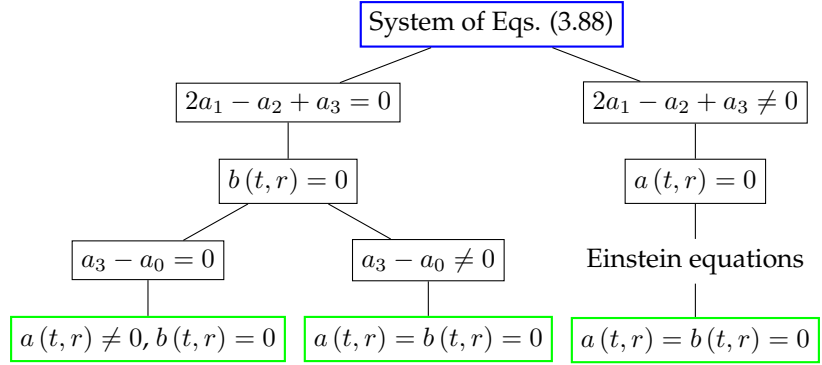

Therefore, in this subsection we have proved that the Birkhoff theorem holds for the considered theory in the weak torsion regime.

\subsection{Asymptotic flatness}

In this subsection we shall consider the assumption of asymptotic flatness and staticity, which is an usual condition when describing exterior spacetimes generated by astrophysical objects, as we have seen when studying the no-hair theorem. The asymptotic flatness condition allows us to impose boundary conditions on both the metric and the torsion functions. It is clear that under this assumption, one solution satisfying trivially both the Einstein and the Cartan Equations and fulfilling the condition above is torsionless Schwarzschild. Below we shall answer to the question if this is indeed the only asymptotically flat solution.

In order to prove this result, let us invoke the Existence and Uniqueness Theorem in the theory of differential equations~\cite{braun1983differential}. First, let us introduce the following definition
\begin{defn}
Let us consider the function $f(r,x)$, with $f:\,\,\mathbb{R}^{n+1}\longrightarrow\mathbb{R}^{n}$, $\left|r-r_{0}\right|\leq a$, $x\in D\subset\mathbb{R}^{n}$. Then, $f(r,x)$ satisfies the Lipschitz condition with respect to $x$ if in $\left[r_{0}-a,\,r_{0}+a\right]\times D$ one has
\begin{equation*}
\left\Vert f\left(r,x_{1}\right)-f\left(r,x_{2}\right)\right\Vert \leq L\left\Vert x_{1}-x_{2}\right\Vert ,
\end{equation*}
with $x_{1},x_{2}\in D$ and $L$ a constant known as the Lipschitz constant.
\end{defn}
That being so, the previous condition shall play an essential role in the next 
\begin{thm}
\label{uniqueness}
Let us consider the initial value problem
\begin{equation*}
\frac{{\rm d}x}{{\rm d}r}=f\left(r,x\right),\,\,\,x\left(r_{0}\right)=x_{0},
\end{equation*}
with $\left|r-r_{0}\right|\leq a$, $x\in D\subset\mathbb{R}^{n}$. $D=\left\{ x\,\,s.t.\,\,\left\Vert r-r_{0}\right\Vert \leq d\right\} $, where $a$ and $d$ are positive constants.\\
Then if the function $f$ satisfies the following conditions
\begin{enumerate}
\item $f\left(r,x\right)$ is continuous in $G=\left[r_{0}-a,\,r_{0}+a\right]\times D$.

\item $f\left(r,x\right)$ is Lipschitz continuous in $x$.
\end{enumerate}

Then the initial value problem has one and only one solution for $\left|r-r_{0}\right|\leq {\rm inf} \left\{a,\frac{d}{M}\right\}$, with
\begin{equation*}
M=\underset{G}{\rm sup}\left\Vert f\right\Vert .
\end{equation*}
\end{thm}

Now, let us apply the Theorem \ref{uniqueness} above to the system formed by the Cartan and Einstein Equations in the asymptotically flat case. Nevertheless, we shall not analyse the field equations directly, but the action from which they are derived. Indeed, by performing an inspection of the action, we find that only two of the six torsion functions contribute to the dynamics, namely $c\left(r\right)$ and $g\left(r\right)$, and only with first derivatives. This means that we can express these two functions in terms of the other four, using the Cartan Equations. Then, using this results in the remaining Cartan Equations we can find the following expressions
\begin{eqnarray}
\label{rel23}
&&c'\left(r\right)={\bf F}_{c}\left(c\left(r\right),\,r\right)\,\,;
\nonumber
\\
&&\,
\\
&&g'\left(r\right)={\bf F}_{g}\left(g\left(r\right),\,r\right)\,\,,
\nonumber
\end{eqnarray}
where ${\bf F}_{c}$ and ${\bf F}_{g}$ are Lipschitz continuous functions, since they are continuously differentiable (resorting to physical criteria).
Then, by use of the uniqueness Theorem \ref{uniqueness}, we can state that there only exists one solution for $c\left(r\right)$ and $g\left(r\right)$, and since we can write the rest of the torsion functions in terms of this one, this means that the whole systemof the Cartan Equations has only one solution (to be determined either by one initial or one boundary condition).

Moreover, since GR is recovered when the torsion is zero, we have that indeed one solution for the Cartan field equations would consist of having all the torsion functions equal to zero, a result which is obviously compatible with the asymptotic flatness assumption. Having null torsion implies that the Einstein Equations  would reduce to those in GR. Therefore, we are led to conclude that the only asymptotically flat and static solution is a torsionless Schwarzschild, and the Birkhoff theorem applies.\\
Moreover, the proof of Birkhoff theorem for these conditions allow us to prove also that the antisymetric Ricci term of the action does not introduce any new hair to spherically symmetric and static BHs.\\

We have therefore seen how having an unstable theory does not prevent us to find stable solutions, and to find situations where it seems like a healthy theory. Nevertheless, the moment that one considers perturbations up to a certain order around the Schwarzschild background, instabilities would appear. This implies that the \emph{Almost} Birkhoff theorem would not hold, since its proof is based on perturbations \cite{Goswami:2011ft}. 

\section{Chapter conclusions and outlook}

Within this chapter we have explored some of the interesting phenomenology of the PG theories of gravity, which is outlined in the following.

First, in section \ref{3.1} we have calculated how the fermionic particles move in spacetimes with torsion, at first order in the WKB approximation. Moreover, we have explicitly shown this non-geodesical behaviour in a particular BH solution.

In the next section we provide a new formulation of the singularity theorems so that they can predict the singularities of fermionic particles. We then prove that if the conditions for the appearance of BH/WHs of arbitrary co-dimension are met, then the fermionic trajectories would be singular, just as the geodesics.

Finally, in section \ref{3.3} we show that the only stable quadratic PG theories that fulfill the Birkhoff theorem are the ones studied by Nieh and Rauch in the 1980s. We also prove that the no-hair theorem applies for the most general stable quadratic Poincar\'e Gauge action, \emph{i.e.} the bi-scalar theory, if asymptotic flatness and constant scalar curvature are assumed. Moreover, we have seen how the Birkhoff and no-hair theorems are not related with the stability, and that indeed one can find that the BHs of GR can be solutions of unstable theories. Nevertheless, when performing perturbations up to certain order, those instabilities will start playing a role.\\

The previous findings can lead to new lines of research, such as
\begin{itemize}

\item The detection of torsion comparing fermionic and bosonic trajectories.

\item Using the new definition of singularity, based on the trajectories not having endpoints in the conformal infinity, to other modified theories of gravity.

\item Study the no-hair theorem in the general stable quadratic PG theory relaxing some assumptions. 

\end{itemize}

\renewcommand{\publ}{}


\chapter{Non-local extension of Poincar\'e Gauge gravity}

\label{4}

\PARstart{T}he theory of GR can be modified to incorporate the gauge structure of the Poincar\'e group, provided a torsion field is added, as we saw in Chapter \ref{2}. Nevertheless, both GR and PG gravity suffer from the short distance behaviour at a classical level, which is manifested explicitly in the appearance of BH and cosmological singularities. This is commonly known as the ultraviolet (UV) problem. Our aim in this chapter will be to construct an action which recovers PG theory of gravity in the infrared (IR), while ameliorating the UV singular behaviour of both metric and torsion fields. This extension would introduce the effect of infinite derivatives in the action, which results in a non-local theory.\\

This chapter is mainly based on results presented in {\bf{\nameref{P4}}} and {\bf{\nameref{P7}}}, and shall be divided as follows. In Section \ref{4.1} we shall review the infinite derivatives extensions of GR, and motivate their introduction. In Section \ref{4.2} we will provide the non-local extension of PG gravity and calculate the field equations at the linear limit. Finally, in Section \ref{4.3} we shall provide ghost- and singularity-free solutions of the theory in the linear regime.

\clearpage

\section{Infinite derivative gravity}
\label{4.1}

In String Theory there are several higher-derivative actions that contain infinite derivatives encoded in an exponential operator, \emph{e.g.} open string field theory \cite{Witten:1985cc}, p-adic theory \cite{Brekke:1988dg}, or strings on random lattices \cite{Douglas:1989ve}. Inspired by this kind of theories one can construct an UV extension of GR by introducing infinite derivatives in the Einstein-Hilbert action that contribute at the strong energy regime \cite{Biswas:2014tua}. The fact that this modification is based on infinite derivatives makes that even the sharpest of the singular behaviours, \emph{i.e.}, the delta ``function'', can be ameliorated, and hence one could potentially find singularity-free solutions.

Nevertheless, modifying GR in a consistent way without incurring in pathologies is quite a difficult task, as we have seen in Section \ref{2.3}. The main concern lies in the fact that the inclusion of infinite derivatives in the field equations may lead to the need of a set of infinite initial conditions in order to solve such equations. This would of course be a problem due to the following issues \cite{Barnaby:2010kx}:
\begin{itemize}

\item {\bf{Stability}}: If the equations of motion admit more than two initial data, \emph{i.e.}, the ones admitted in second order differential equations, then the extra degrees of freedom can be interpreted as physical excitations which carry wrong-sign kinetic energy. As we explained in Section \ref{2.3}, the classical theory would be plagued by Ostrogradski instabilities.

\item {\bf{Predictability}}: If the equations of motion require infinitely many initial data then, by a suitable choice of the infinite free parameters of the solution, it can be possible to construct nearly any time dependence over an arbitrarily long interval. Accordingly, the initial value problem would be completely bereft of predictivity.

\end{itemize}

Fortunately, that is not the case, neither of these two aspects are compromised when introducing infinite derivatives, provided that some constraints are fulfilled, as has been proven in several ocassions \cite{Barnaby:2007ve,Barnaby:2010kx,Calcagni:2018lyd}. In the following we shall summarise this fact with a simple example. In particular, let us introduce a scalar field action in Minkowski spacetime involving infinite derivatives:
\begin{equation}
\label{nonlocal:scalar}
S=\int {\rm d}^{4}x\left[\phi F\left(\Box\right)\phi-V\left(\phi\right)\right],
\end{equation}
where $F\left(\Box\right)$ is an entire analytic function of the d'Alembertian $\Box=\eta_{\mu\nu}\partial^{\mu}\partial^{\nu}$, of the form
\begin{equation}
\label{FFunction}
F\left(\Box\right)=\sum_{n=0}^{\infty}f_{n}\left(\frac{\Box}{M_{S}}\right)^{n},
\end{equation}
with $M_{S}$ being the mass defining the scale at which non-localities start to play a role, and the $f_{n}$'s being constants. \\
The field equation derived from \eqref{nonlocal:scalar} is
\begin{equation}
\label{nonlocal:scalarFE}
F\left(\Box\right)\phi=V'\left(\phi\right).
\end{equation} 
Now, using Weiertrass factorisation theorem we can write $F\left(\Box\right)$ as 
\begin{equation}
F\left(\Box\right)=\Gamma\left(\Box\right)\prod_{j=1}^{N}\left(\Box-m_{j}^{2}\right),
\end{equation}
where $\Gamma\left(\Box\right)^{-1}$ does not cointain any pole in the complex plane, and consequently it can be expressed as $\Gamma\left(\Box\right)={\rm e}^{-\gamma\left(\Box\right)}$, $\gamma\left(\Box\right)$ being an entire function, without losing generality. Then, with this decomposition it can be seen that Equation \eqref{nonlocal:scalarFE} describes $N$ physical states with masses $m_j$ \cite{Barnaby:2007ve}. In order to solve this equation we must find a particular solution and the general one for the associated homogeneous equation. 

On the one hand, to obtain the particular solution we expand the scalar field into Fourier modes as
\begin{equation}
\phi\left(t,\vec{x}\right)=\int\frac{{\rm d}^{3}k}{(2\pi)^{3/2}}{\rm e}^{i\vec{k}\cdot\vec{x}}\xi_{\vec{k}}\left(t\right),
\end{equation}
and plug it into Equation \eqref{nonlocal:scalarFE}, that now becomes
\begin{equation}
\label{nonlocal:scalarFourier}
F\left(-\partial_{t}^{2}-k^{2}\right)\xi_{\vec{k}}\left(t\right)=V_{\vec{k}}\left(t\right),
\end{equation}
where $k^2=\vec{k}\cdot\vec{k}$ as usual, and
\begin{equation}
V_{\vec{k}}\left(t\right)=\int\frac{{\rm d}^{3}x}{(2\pi)^{3/2}}{\rm e}^{i\vec{k}\cdot\vec{x}}V'(\phi\left(t,\vec{x}\right)).
\end{equation}
The equations of the form of \eqref{nonlocal:scalarFourier} are very well known in mathematical literature \cite{davis1936theory,carmichael1936linear,carleson1953infinite}, and the fact that they can be solvable without having to specify infinite initial conditions has been known since the 1930s. Nevertheless, this was unnoticed till Barnaby called the attention of theoretical physicists on this subject \cite{Barnaby:2007ve}. Following that reference, one can check that a particular solution of \eqref{nonlocal:scalarFourier} is given by
\begin{equation}
\phi_{\vec{k}}(t)=\frac{1}{2\pi i}\oint_{C}{\rm d}s\,{\rm e}^{st}\frac{\hat{V}_{\vec{k}}(s)}{F\left(-s^{2}-k^{2}\right)},
\end{equation}
where the $\hat{\,}$ means the Laplace transform.

For the solution of the homogeneous equation we can again resort to the existing literature, in particular the mentioned work by Barnaby \cite{Barnaby:2007ve}. The homogeneous part of Equation \eqref{nonlocal:scalarFourier} belongs to a known class of differential equations of the form 
\begin{equation}
f\left(\partial_{t}\right)\phi\left(t\right)=0.
\end{equation}
Then, if we assume that the solution admits a Laplace transform, we can rewrite the previous equation as
\begin{equation}
\frac{1}{2\pi i}\oint_{C}{\rm d}s\,{\rm e}^{st}f\left(s\right)\hat{\phi}\left(s\right)=0,
\label{eq:laplace}
\end{equation}
where $f\left(s\right)$ is the so-called \emph{generatrix}, which, as we have seen, it can be decomposed as
\begin{equation}
f(s)=\gamma(s)\prod_{i=1}^{M}\left(s-s_{i}\right)^{r_{i}},
\label{function:decom}
\end{equation} 
with $\gamma(s)$ different from zero everywhere. Therefore the function $f(s)$ has $M$ zeroes at the points $s=s_{i},$ the $i$ -th zero being of order $r_{i}$. The inverse of this function, $f(s)^{-1}$, is known as the \emph{resolvent generatrix}, which has simple poles at the points $s=s_{i},$ the $i$ -th pole being of order $r_{i}$. \\
Now, in order to solve \eqref{eq:laplace}, we need to ask ourselves which is the most general function $\hat{\phi}$ fulfilling such an equality. Using the Cauchy Integral Theorem \cite{burckel1980introduction}, we know that the equality \eqref{eq:laplace} holds if the integrand of such expression does not have any poles inside the region of integration. Consequently, taking into account \eqref{function:decom}, the solution $\hat{\phi}$ may have simple poles at the points $s=s_{i},$ the $i$ -th pole being of order $r_{i}$ or less. Therefore, the most general way to express $\hat{\phi}$ is
\begin{equation}
\hat{\phi}(s)=\frac{1}{\gamma(s)}\sum_{i=1}^{M}\sum_{j=1}^{r_{i}}\frac{C_{j}^{(i)}}{\left(s-s_{i}\right)^{j}}.
\end{equation}
It is clear that the solution has $N$ arbitrary coefficients $C_{j}^{(i)}$, where
\begin{equation}
N=\sum_{i=1}^{M}r_{i}.
\end{equation}
We can recover the solution $\phi$ in the configuration space by solving the integral of the Laplace transform, namely
\begin{equation}
\phi\left(t\right)=\frac{1}{2\pi i}\oint_{C}{\rm d}s\,{\rm e}^{st}\hat{\phi}\left(s\right)=\sum_{i=1}^{M}P_{i}(t){\rm e}^{s_{i}t},
\end{equation}
where each of the $P_{i}(t)$ are polynomials of order $r_i-1$
\begin{equation}
P_{i}(t)=\sum_{j=1}^{r_{i}} p_{j}^{(i)} t^{j-1}.
\end{equation}
Let us note that the $N$ coefficients $p_{j}^{(i)}$ are arbitrary and will serve to fix $N$ (and hence {\bf{finite}}) initial conditions $\phi^{(n)}(0)$ for $n=0, \cdots, N-1$.

Now, let us apply these results to the homogeneous equation associated to \eqref{nonlocal:scalarFourier}. In this case the generatrix function is given by
\begin{equation}
f(s)=F\left(-s^{2}-k^{2}\right)=\Gamma\left(-s^{2}-k^{2}\right)\prod_{j=1}^{N}\left(s+i\omega_{k}^{(j)}\right)\left(s-i\omega_{k}^{(j)}\right),
\end{equation}
where we have defined
\begin{equation}
\omega_{k}^{(i)}=\sqrt{k^{2}+m_{j}^{2}}.
\end{equation}
Then, by the analysis perfomed previously we know that since this generatrix has $2N$ poles of order one we expect the solutions to contain $2N$ free coefficients for each $k$-mode, two for each physical degree of freedom. Therefore, by choosing a correct generatrix, one can construct an infinite derivative action for a scalar field in Minkowski spacetime in such a way that it only propagates one degree of freedom with positive $m^2$, hence being a stable configuration. This is also possible when considering an arbitrary curved background, as was proven in \cite{Calcagni:2018lyd}. 

Moreover, this particular example allows us to show why we state that the introduction of infinite derivatives in the action makes the theory non-local, which is why we are using any of those two terms to refer to this kind of theories. As a matter of fact, we can rewrite the action \eqref{nonlocal:scalar} for the scalar field as \cite{Tomboulis:2015gfa}
\begin{equation}
S=\int {\rm d}^{4}x{\rm d}^{4}y\,\phi\left(x\right)K\left(x-y\right)\phi\left(y\right)-\int {\rm d}^{4}x\,V\left(\phi\right),
\end{equation}
where
\begin{equation}
K\left(x-y\right)=F\left(\Box\right)\delta^{\left(4\right)}\left(x-y\right).
\end{equation}
The operator $K\left(x-y\right)$ makes the dependence of the field variables at finite distances explicit, which acknowledges for the presence of a non-local nature.\\

With respect to an infinite derivative UV completion of GR, it has been possible to establish that the following action
\begin{equation}
\label{IDGmetric:action}
S=\int d^{4}x\sqrt{-g}\left[\mathring{R}+\mathring{R}F_{1}(\square)\mathring{R}+\mathring{R}_{\mu\nu}F_{2}(\square)\mathring{R}^{\mu\nu}+\mathring{R}_{\mu\nu\lambda\sigma}F_{3}(\square)\mathring{R}^{\mu\nu\lambda\sigma}\right],
\end{equation}
can be made free of extra degrees of freedom around Minkowski spacetime. This is done by expressing the non-local functions as exponentials of an entire function, which does not introduce any new complex poles, nor any new dynamical degrees of freedom \cite{Biswas:2011ar}. When exploring this for at the non-perturbative level one finds that there are 8 degrees of freedom, but it is not clear whether they are of ghostly nature or not. Even if they are stable modes, this signals that it may be a strong coupling issue, which can be interesting to explore in future research but it is beyond the scope of this thesis.

This kind of UV extensions of GR have been explored widely, and are known as infinite derivative theories of gravity (IDG). The most general action has been constructed around Minkowski spacetime~\cite{Biswas:2011ar}, and in de Sitter and anti-de Sitter~\cite{Biswas:2016etb}. The graviton propagator of such theories can be modified to avoid any ghosts around a Minkowski background. Therefore, such theories retain the original 2 dynamical degrees of freedom of GR, \emph{i.e.}, a transverse traceless graviton. Being infinite derivative theories, such an action introduces non-local gravitational interaction and has been argued to improve UV aspects of quantum nature of gravity~\cite{Tomboulis:1997gg,Modesto:2011kw}. As we have seen, despite having infinite derivatives, the Cauchy problem is well defined, hence the solutions are uniquely determined by finite initial conditions \cite{Calcagni:2018lyd}.

At a classical level, it has been shown that such IDG theories can yield a non-singular, static solution at the full non-linear level~\cite{Buoninfante:2018rlq}, can avoid ring singularities in a rotating metric at the linear level~\cite{Buoninfante:2018xif}, and also resolve charged source singularity at the linear level~\cite{Buoninfante:2018stt}. At a dynamical level such theories do not give rise to formation of a trapped surface~\cite{Frolov:2015bta,Frolov:2015bia,Frolov:2015usa}, and possibly even at the level of astrophysical masses there may not possess event horizon~\cite{Koshelev:2017bxd,Buoninfante:2019swn}. Exact solutions for IDG have been found in~\cite{Biswas:2005qr,Li:2015bqa,Buoninfante:2018rlq,Kilicarslan:2018yxd}, including static and time-dependent solutions.

As we have been stating and checking along this thesis, there is not a physically preferred affine structure for gravitational theories. Therefore, it is physically relevant to ask ourselves if it is possible to construct an UV extension of PG theories using the tools of usual metric IDG. This is exactly what we shall do in the following sections.

\section{The inclusion of torsion}
\label{4.2}

Motivated by the multiple studies mentioned above, various extensions of IDG have been made in the context of teleparallel gravity \cite{Koivisto:2018loq} and symmetric teleparallel gravity \cite{Conroy:2017yln}, as well as in what regards the extension of Poincar\'e gauge gravity that is well-behaved at the UV at a classical level ({\bf{\nameref{P4}}} and {\bf{\nameref{P7}}}).

In the standard IDG theories the connection is metric and symmetric, \emph{i.e.} the Levi-Civita one. Therefore, the linear action of IDG is built with the gravitational invariants and derivatives, considering only up to order $\mathcal{O}(h^{2})$, where $h$ is the linear perturbation around the Minkowski metric
\begin{equation}
g_{\mu\nu}=\eta_{\mu\nu}+h_{\mu\nu}.
\end{equation}
After substituting the linear expressions of the curvature tensors (Riemann, Ricci and curvature scalar) we can obtain the linearised action as first was shown in \cite{Biswas:2011ar}. With this in mind, we wish to generalise the expressions of the curvature tensor when we consider a non-symmetric connection. First, we must take into account that the torsion tensor is not geometrically related to the metric, therefore the conditions that are imposed in $h_{\mu\nu}$ are not sufficient to construct the linear action in connection. In order to tackle this issue, we will have to impose that the total connection must be of order $\mathcal{O}(h)$, \emph{i.e.} the same as the Levi-Civita one~\footnote{If this were not the case, we would have two options: either the contribution of the metric is of higher order than the torsion, hence recovering the usual IDG theory ~\cite{Biswas:2011ar}, or the torsion is of higher order than the metric. In the latter case we would have a somewhat similar action of the UV extension of teleparallel gravity~\cite{Koivisto:2018loq}.}. Then, by using the relation between the Levi-Civita, $\mathring{\Gamma}$, and the total connection with torsion and null non-metricity, $\widetilde{\Gamma}$\footnote{Let us note that to the remaining sections we shall be using the tilde $\,\widetilde{\,}\,$ to refer to the total connection, instead of just the plain $\Gamma$.} \eqref{LCWeitzenbock}, we can write
\begin{equation}
\widetilde{\Gamma}^{\rho}_{\,\,\mu\nu}=\mathring{\Gamma}^{\rho}_{\,\,\mu\nu}+K^{\rho}_{\,\,\mu\nu},
\end{equation}
where the contortion tensor $K$ must be of the same order as the metric perturbation. The latter may seem as a strong assumption, nevertheless, as it has been known in the literature, the current constraints on torsion suggest that its influence is very small compared to the purely metric gravitational effects \cite{Lammerzahl:1997wk,Kostelecky:2007kx}. Therefore, considering a higher order than the metric perturbation in the torsion sector would make no sense physically.

Thus, the way to generalise the IDG action will be to consider all the quadratic Lorentz invariant terms that can be constructed with the curvature tensors, the contortion, and infinite derivatives operators, as follows [{\bf{\nameref{P4}}}]
\begin{eqnarray}\label{GEQ}
S&=&\int {\rm d}^{4}x\sqrt{-g}\left[{\widetilde{R}}+\widetilde{R}_{\mu_{1}\nu_{1}\rho_{1}\sigma_{1}}\mathcal{O}_{\mu_{2}\nu_{2}\rho_{2}\sigma_{2}}^{\mu_{1}\nu_{1}\rho_{1}\sigma_{1}}\widetilde{R}^{\mu_{2}\nu_{2}\rho_{2}\sigma_{2}}+\widetilde{R}_{\mu_{1}\nu_{1}\rho_{1}\sigma_{1}}\mathcal{O}_{\mu_{2}\nu_{2}\rho_{2}}^{\mu_{1}\nu_{1}\rho_{1}\sigma_{1}}K^{\mu_{2}\nu_{2}\rho_{2}}\right.
\nonumber
\\
&&+\left.K_{\mu_{1}\nu_{1}\rho_{1}}\mathcal{O}_{\mu_{2}\nu_{2}\rho_{2}}^{\mu_{1}\nu_{1}\rho_{1}}K^{\mu_{2}\nu_{2}\rho_{2}}\right],
\end{eqnarray}
where $\mathcal{O}$ denote the possible differential operators containing covariant derivatives and the Minkowski metric $\eta_{\mu\nu}$, so also the contractions of the Riemann and contortion tensors are considered in the action. Moreover, the tilde $\,\widetilde{\,}\,$ represents the quantities calculated with respect to the total connection $\widetilde{\Gamma}$. We will expand the quadratic part of the previous expression to obtain the general form for the gravitational Lagrangian\footnote{Note that many terms in \eqref{lagrangian} are completely redundant in the linear regime, as we shall prove in the following. As a matter of fact one can see in \eqref{compac1}, \eqref{compac2}, and \eqref{compac3}, that the number of functions to describe the linear regime are significantly less.}
\begin{eqnarray}
\label{lagrangian}
\mathcal{L}_{q}&=&\widetilde{R}\widetilde{F}_{1}\left(\Box\right)\widetilde{R}+\widetilde{R}\widetilde{F}_{2}\left(\Box\right)\partial_{\mu}\partial_{\nu}\widetilde{R}^{\mu\nu}+\widetilde{R}_{\mu\nu}\widetilde{F}_{3}\left(\Box\right)\widetilde{R}^{\left(\mu\nu\right)}+\widetilde{R}_{\mu\nu}\widetilde{F}_{4}\left(\Box\right)\widetilde{R}^{\left[\mu\nu\right]}
\nonumber
\\
&+&\widetilde{R}_{\left(\mu\right.}^{\,\,\,\left.\nu\right)}\widetilde{F}_{5}\left(\Box\right)\partial_{\nu}\partial_{\lambda}\widetilde{R}^{\mu\lambda}+\widetilde{R}_{\left[\mu\right.}^{\,\,\,\left.\nu\right]}\widetilde{F}_{6}\left(\Box\right)\partial_{\nu}\partial_{\lambda}\widetilde{R}^{\mu\lambda}+\widetilde{R}_{\mu}^{\,\,\,\nu}\widetilde{F}_{7}\left(\Box\right)\partial_{\nu}\partial_{\lambda}\widetilde{R}^{\left(\mu\lambda\right)}
\nonumber
\\
&+&\widetilde{R}_{\mu}^{\,\,\,\nu}\widetilde{F}_{8}\left(\Box\right)\partial_{\nu}\partial_{\lambda}\widetilde{R}^{\left[\mu\lambda\right]}+\widetilde{R}^{\lambda\sigma}\widetilde{F}_{9}\left(\Box\right)\partial_{\mu}\partial_{\sigma}\partial_{\nu}\partial_{\lambda}\widetilde{R}^{\mu\nu}+\widetilde{R}_{\left(\mu\lambda\right)}\widetilde{F}_{10}\left(\Box\right)\partial_{\nu}\partial_{\sigma}\widetilde{R}^{\mu\nu\lambda\sigma}
\nonumber
\\
&+&\widetilde{R}_{\left[\mu\lambda\right]}\widetilde{F}_{11}\left(\Box\right)\partial_{\nu}\partial_{\sigma}\widetilde{R}^{\mu\nu\lambda\sigma}+\widetilde{R}_{\mu\lambda}\widetilde{F}_{12}\left(\Box\right)\partial_{\nu}\partial_{\sigma}\widetilde{R}^{\left(\mu\nu\right|\left.\lambda\sigma\right)}
\nonumber
\\
&+&\widetilde{R}_{\mu\lambda}\widetilde{F}_{13}\left(\Box\right)\partial_{\nu}\partial_{\sigma}\widetilde{R}^{\left[\mu\nu\right|\left.\lambda\sigma\right]}+\widetilde{R}_{\mu\nu\lambda\sigma}\widetilde{F}_{14}\left(\Box\right)\widetilde{R}^{\left(\mu\nu\right|\left.\lambda\sigma\right)}+\widetilde{R}_{\mu\nu\lambda\sigma}\widetilde{F}_{15}\left(\Box\right)\widetilde{R}^{\left[\mu\nu\right|\left.\lambda\sigma\right]}
\nonumber
\\
&+&\widetilde{R}_{\left(\rho\mu\right|\left.\nu\lambda\right)}\widetilde{F}_{16}\left(\Box\right)\partial^{\rho}\partial_{\sigma}\widetilde{R}^{\mu\nu\lambda\sigma}+\widetilde{R}_{\left[\rho\mu\right|\left.\nu\lambda\right]}\widetilde{F}_{17}\left(\Box\right)\partial^{\rho}\partial_{\sigma}\widetilde{R}^{\mu\nu\lambda\sigma}
\nonumber
\\
&+&\widetilde{R}_{\rho\mu\nu\lambda}\widetilde{F}_{18}\left(\Box\right)\partial^{\rho}\partial_{\sigma}\widetilde{R}^{\left(\mu\nu\right|\left.\lambda\sigma\right)}+\widetilde{R}_{\rho\mu\nu\lambda}\widetilde{F}_{19}\left(\Box\right)\partial^{\rho}\partial_{\sigma}\widetilde{R}^{\left[\mu\nu\right|\left.\lambda\sigma\right]}
\nonumber
\\
&+&\widetilde{R}_{\left(\mu\nu\right|\left.\rho\sigma\right)}\widetilde{F}_{20}\left(\Box\right)\partial^{\nu}\partial^{\sigma}\partial_{\alpha}\partial_{\beta}\widetilde{R}^{\mu\alpha\rho\beta}+\widetilde{R}_{\left[\mu\nu\right|\left.\rho\sigma\right]}\widetilde{F}_{21}\left(\Box\right)\partial^{\nu}\partial^{\sigma}\partial_{\alpha}\partial_{\beta}\widetilde{R}^{\mu\alpha\rho\beta}
\nonumber
\\
&+&\widetilde{R}_{\mu\nu\rho\sigma}\widetilde{F}_{22}\left(\Box\right)\partial^{\nu}\partial^{\sigma}\partial_{\alpha}\partial_{\beta}\widetilde{R}^{\left(\mu\alpha\right|\left.\rho\beta\right)}+\widetilde{R}_{\mu\nu\rho\sigma}\widetilde{F}_{23}\left(\Box\right)\partial^{\nu}\partial^{\sigma}\partial_{\alpha}\partial_{\beta}\widetilde{R}^{\left[\mu\alpha\right|\left.\rho\beta\right]}
\nonumber
\\
&+&K_{\mu\nu\rho}\widetilde{F}_{24}\left(\Box\right)K^{\mu\nu\rho}+K_{\mu\nu\rho}\widetilde{F}_{25}\left(\Box\right)K^{\mu\rho\nu}+K_{\mu\,\,\rho}^{\,\,\rho}\widetilde{F}_{26}\left(\Box\right)K_{\,\,\,\,\,\sigma}^{\mu\sigma}
\nonumber
\\
&+&K_{\,\,\nu\rho}^{\mu}\widetilde{F}_{27}\left(\Box\right)\partial_{\mu}\partial_{\sigma}K^{\sigma\nu\rho}+K_{\,\,\nu\rho}^{\mu}\widetilde{F}_{28}\left(\Box\right)\partial_{\mu}\partial_{\sigma}K^{\sigma\rho\nu}+K_{\mu\,\,\,\,\,\nu}^{\,\,\rho}\widetilde{F}_{29}\left(\Box\right)\partial_{\rho}\partial_{\sigma}K^{\mu\nu\sigma}
\nonumber
\\
&+&K_{\mu\,\,\,\,\,\nu}^{\,\,\rho}\widetilde{F}_{30}\left(\Box\right)\partial_{\rho}\partial_{\sigma}K^{\mu\sigma\nu}+K_{\,\,\,\,\,\rho}^{\mu\rho}\widetilde{F}_{31}\left(\Box\right)\partial_{\mu}\partial_{\nu}K_{\,\,\,\,\,\sigma}^{\nu\sigma}
\nonumber
\\
&+&K_{\mu}^{\,\,\nu\rho}\widetilde{F}_{32}\left(\Box\right)\partial_{\nu}\partial_{\rho}\partial_{\alpha}\partial_{\sigma}K^{\mu\alpha\sigma}+K_{\,\,\,\lambda\sigma}^{\lambda}\widetilde{F}_{33}\left(\Box\right)\partial_{\rho}\partial_{\nu}K^{\nu\rho\sigma}
\nonumber
\\
&+&\widetilde{R}_{\,\,\nu\rho\sigma}^{\mu}\widetilde{F}_{34}\left(\Box\right)\partial_{\mu}K^{\nu\rho\sigma}+\widetilde{R}_{\mu\nu\,\,\sigma}^{\,\,\,\,\,\,\rho}\widetilde{F}_{35}\left(\Box\right)\partial_{\rho}K^{\mu\nu\sigma}+\widetilde{R}_{\left(\rho\sigma\right)}\widetilde{F}_{36}\left(\Box\right)\partial_{\nu}K^{\nu\rho\sigma}
\nonumber
\\
&+&\widetilde{R}_{\left[\rho\sigma\right]}\widetilde{F}_{37}\left(\Box\right)\partial_{\nu}K^{\nu\rho\sigma}+\widetilde{R}_{\rho\sigma}\widetilde{F}_{38}\left(\Box\right)\partial_{\nu}K^{\rho\nu\sigma}+\widetilde{R}_{\left(\rho\sigma\right)}\widetilde{F}_{39}\left(\Box\right)\partial^{\sigma}K_{\,\,\,\,\,\mu}^{\rho\mu}
\nonumber
\\
&+&\widetilde{R}_{\left[\rho\sigma\right]}\widetilde{F}_{40}\left(\Box\right)\partial^{\sigma}K_{\,\,\,\,\,\mu}^{\rho\mu}+\widetilde{R}\widetilde{F}_{41}\left(\Box\right)\partial_{\rho}K_{\,\,\,\,\,\mu}^{\rho\mu}+\widetilde{R}_{\,\,\alpha\,\,\sigma}^{\mu\,\,\rho}\widetilde{F}_{42}\left(\Box\right)\partial_{\mu}\partial_{\rho}\partial_{\nu}K^{\nu\left(\alpha\sigma\right)}
\nonumber
\\
&+&\widetilde{R}_{\,\,\alpha\,\,\sigma}^{\mu\,\,\rho}\widetilde{F}_{43}\left(\Box\right)\partial_{\mu}\partial_{\rho}\partial_{\nu}K^{\nu\left[\alpha\sigma\right]}+\widetilde{R}_{\,\,\alpha\,\,\sigma}^{\mu\,\,\rho}\widetilde{F}_{44}\left(\Box\right)\partial_{\mu}\partial_{\rho}\partial_{\nu}K^{\alpha\nu\sigma}
\nonumber
\\
&+&\widetilde{R}_{\,\,\left.\sigma\right)}^{\left(\mu\right.}\widetilde{F}_{45}\left(\Box\right)\partial_{\mu}\partial_{\nu}\partial_{\alpha}K^{\sigma\nu\alpha}+\widetilde{R}_{\,\,\left.\sigma\right]}^{\left[\mu\right.}\widetilde{F}_{46}\left(\Box\right)\partial_{\mu}\partial_{\nu}\partial_{\alpha}K^{\sigma\nu\alpha}
\nonumber
\\
&+&\widetilde{R}_{\mu\nu\lambda\sigma}\widetilde{F}_{47}\left(\Box\right)\widetilde{R}^{\mu\lambda\nu\sigma},
\end{eqnarray}
where the $\widetilde{F}_{i}\left(\Box\right)$'s are functions of the d'Alembertian $\Box=\eta_{\mu\nu}\partial^{\mu}\partial^{\nu}$, which have the same form as \eqref{FFunction}, namely
\begin{equation}
\label{infinite}
\widetilde{F}_{i}\left(\Box\right)=\sum_{n=0}^{N}\widetilde{f}_{i,n}\left(\frac{\Box}{M_{S}}\right)^{n},
\end{equation}
where $M_{S}$ holds for the mass defining the scale at which non-localities starts to play a role. Also, in the previous expression $n$ can be a finite (finite higher-order derivatives theories), or infinite (IDG) number, as we will consider from now onwards, since finite derivatives will incur ghosts and other instabilities. In the final Section of this Chapter, we shall show how only considering an infinite number of derivatives in \eqref{infinite} one can avoid the ghosts appearance for the torsion sector, which extends the current results on the metric one \cite{Biswas:2011ar}.

Since one needs to recover the purely metric IDG action when the torsion is zero, there are some constraints in the form of the $\widetilde{F}$ functions. In order to obtain these relations, let us write the Lagrangian of the metric theory around a Minkowski background as presented in \cite{Biswas:2011ar}
\begin{eqnarray}
\label{metricidg}
\mathcal{L}_{\rm IDG}&=&\mathring{R}{F}_{1}\left(\Box\right)\mathring{R}+\mathring{R}{F}_{2}\left(\Box\right)\partial_{\mu}\partial_{\nu}\mathring{R}^{\mu\nu}+\mathring{R}_{\mu\nu}{F}_{3}\left(\Box\right)\mathring{R}^{\mu\nu}+\mathring{R}_{\mu}^{\,\,\,\nu}{F}_{4}\left(\Box\right)\partial_{\nu}\partial_{\lambda}\mathring{R}^{\mu\lambda}
\nonumber
\\
&+&\mathring{R}^{\lambda\sigma}F_{5}\left(\Box\right)\partial_{\mu}\partial_{\sigma}\partial_{\nu}\partial_{\lambda}\mathring{R}^{\mu\nu}+\mathring{R}_{\mu\lambda}F_{6}\left(\Box\right)\partial_{\nu}\partial_{\sigma}\mathring{R}^{\mu\nu\lambda\sigma}+{\mathring{R}}_{\mu\nu\lambda\sigma}{F}_{7}\left(\Box\right){\mathring{R}}^{\mu\nu\lambda\sigma}
\nonumber
\\
&+&{\mathring{R}}_{\rho\mu\nu\lambda}F_{8}\left(\Box\right)\partial^{\rho}\partial_{\sigma}{\mathring{R}}^{\mu\nu\lambda\sigma}+{\mathring{R}}_{\mu\nu\rho\sigma}F_{9}\left(\Box\right)\partial^{\nu}\partial^{\sigma}\partial_{\alpha}\partial_{\beta}{\mathring{R}}^{\mu\alpha\rho\beta},
\end{eqnarray}
and compare it with the Lagrangian in \eqref{lagrangian} in the limit when torsion goes to zero
\begin{eqnarray}
\label{nulltorsion}
\mathcal{L}_{q}\left(K_{\,\,\,\nu\sigma}^{\mu}\rightarrow0\right)&=&\mathring{R}\widetilde{F}_{1}\left(\Box\right)\mathring{R}+\mathring{R}\widetilde{F}_{2}\left(\Box\right)\partial_{\mu}\partial_{\nu}\mathring{R}^{\mu\nu}+\mathring{R}_{\mu\nu}\widetilde{F}_{3}\left(\Box\right)\mathring{R}^{\mu\nu}
\nonumber
\\
&+&\mathring{R}_{\mu}^{\,\,\,\nu}\left(\widetilde{F}_{5}\left(\Box\right)+\widetilde{F}_{7}\left(\Box\right)\right)\partial_{\nu}\partial_{\lambda}\mathring{R}^{\mu\lambda}+\mathring{R}^{\lambda\sigma}\widetilde{F}_{9}\left(\Box\right)\partial_{\mu}\partial_{\sigma}\partial_{\nu}\partial_{\lambda}\mathring{R}^{\mu\nu}
\nonumber
\\
&+&\mathring{R}_{\mu\lambda}\left(\widetilde{F}_{10}\left(\Box\right)+\widetilde{F}_{12}\left(\Box\right)\right)\partial_{\nu}\partial_{\sigma}\mathring{R}^{\mu\nu\lambda\sigma}
\nonumber
\\
&+&{\mathring{R}}_{\mu\nu\lambda\sigma}\left(\widetilde{F}_{14}\left(\Box\right)+\frac{\widetilde{F}_{47}\left(\Box\right)}{2}\right){\mathring{R}}^{\mu\nu\lambda\sigma}
\nonumber
\\
&+&{\mathring{R}}_{\rho\mu\nu\lambda}\left(\widetilde{F}_{16}\left(\Box\right)+\widetilde{F}_{18}\left(\Box\right)\right)\partial^{\rho}\partial_{\sigma}{\mathring{R}}^{\mu\nu\lambda\sigma}
\nonumber
\\
&+&{\mathring{R}}_{\mu\nu\rho\sigma}\left(\widetilde{F}_{20}\left(\Box\right)+\widetilde{F}_{22}\left(\Box\right)\right)\partial^{\nu}\partial^{\sigma}\partial_{\alpha}\partial_{\beta}{\mathring{R}}^{\mu\alpha\rho\beta}.
\end{eqnarray}
Then a straightforward comparison between Eqs.\eqref{metricidg} and \eqref{nulltorsion} makes it clear that the following relations need to hold
\begin{eqnarray}
&&\tilde{F}_{1}\left(\Box\right)=F_{1}\left(\Box\right),\;\tilde{F}_{2}\left(\Box\right)=F_{2}\left(\Box\right),\;\tilde{F}_{3}\left(\Box\right)=F_{3}\left(\Box\right),\;\tilde{F}_{5}\left(\Box\right)+\tilde{F}_{7}\left(\Box\right)=F_{4}\left(\Box\right),
\nonumber
\\
&&\,\nonumber
\\
&&\tilde{F}_{9}\left(\Box\right)=F_{5}\left(\Box\right),\,\tilde{F}_{10}\left(\Box\right)+\tilde{F}_{12}\left(\Box\right)=F_{6}\left(\Box\right),\;\tilde{F}_{14}\left(\Box\right)+\frac{\widetilde{F}_{47}\left(\Box\right)}{2}=F_{7}\left(\Box\right),
\\
&&\tilde{F}_{16}\left(\Box\right)+\tilde{F}_{18}\left(\Box\right)=F_{8}\left(\Box\right),\;\tilde{F}_{20}\left(\Box\right)+\tilde{F}_{22}\left(\Box\right)=F_{9}\left(\Box\right).\nonumber
\end{eqnarray}
In order to check which are the terms that are of order ${\cal O} (h^{2})$ in the Lagrangian \eqref{lagrangian}, and get rid of redundant terms, we still need to substitute the linearized expressions of the curvature tensors, namely
\begin{equation}
\label{riemann}
\tilde{R}_{\mu\nu\rho\lambda}=\partial_{\left[\nu\right.}\partial_{\rho}h_{\left.\lambda\mu\right]}-\partial_{\left[\nu\right.}\partial_{\lambda}h_{\left.\rho\mu\right]}+2\partial_{\left[\nu\right.}K_{\rho\left|\mu\right]\lambda},
\end{equation}
\begin{equation}
\label{ricci}
\tilde{R}_{\mu\nu}=\partial_{\sigma}\partial_{\left(\nu\right.}h_{\left.\mu\right)}^{\,\,\,\sigma}-\frac{1}{2}\left(\partial_{\mu}\partial_{\nu}h+\Box h_{\mu\nu}\right)-\partial_{\sigma}K_{\,\,\,\mu\nu}^{\sigma}+\partial_{\mu}K_{\,\,\,\sigma\nu}^{\sigma},
\end{equation}
\begin{equation}
\label{scalar}
\tilde{R}=\partial_{\mu}\partial_{\nu}h^{\mu\nu}-\Box h-2\partial_{\mu}K_{\,\,\,\,\,\,\,\nu}^{\mu\nu},
\end{equation}
We have computed each term appearing in the Lagrangian \eqref{lagrangian} separately. Explicit calculations can be found in Appendix \ref{ap:3}. Finally, using the expressions obtained and performing a further simplification we obtain the linearised action which can be split in metric, torsion and the mixed terms as follows
\begin{equation}
\label{laaccion}
S=-\int {\rm d}^{4}x\left(\mathcal{L}_{M}+\mathcal{L}_{MT}+\mathcal{L}_{T}\right)=S_{M}+S_{MT}+S_{T},
\end{equation}
where
\begin{eqnarray}
\label{compac1}
\mathcal{L}_{M}&=&\frac{1}{2}h_{\mu\nu}\Box a\left(\Box\right)h^{\mu\nu}+h_{\mu}^{\,\,\alpha}b\left(\Box\right)\partial_{\alpha}\partial_{\sigma}h^{\sigma\mu}+hc\left(\Box\right)\partial_{\mu}\partial_{\nu}h^{\mu\nu}+\frac{1}{2}h\Box d\left(\Box\right)h
\nonumber
\\
&&+h^{\lambda\sigma}\frac{f\left(\Box\right)}{\Box}\partial_{\sigma}\partial_{\lambda}\partial_{\mu}\partial_{\nu}h^{\mu\nu},
\end{eqnarray}
\begin{eqnarray}
\label{compac2}
\mathcal{L}_{MT}&=&h\Box u\left(\Box\right)\partial_{\rho}K_{\,\,\,\,\,\sigma}^{\rho\sigma}+h_{\mu\nu}v_{1}\left(\Box\right)\partial^{\mu}\partial^{\nu}\partial_{\rho}K_{\,\,\,\,\,\sigma}^{\rho\sigma}+h_{\mu\nu}v_{2}\left(\Box\right)\partial^{\nu}\partial_{\sigma}\partial_{\rho}K^{\mu\sigma\rho}
\nonumber
\\
&&+h_{\mu\nu}\Box w\left(\Box\right)\partial_{\rho}K^{\rho\mu\nu},
\end{eqnarray}
\begin{eqnarray}
\label{compac3}
\mathcal{L}_{T}&=&K^{\mu\sigma\lambda}p_{1}\left(\Box\right)K_{\mu\sigma\lambda}+K^{\mu\sigma\lambda}p_{2}\left(\Box\right)K_{\mu\lambda\sigma}+K_{\mu\,\,\rho}^{\,\,\rho}p_{3}\left(\Box\right)K_{\,\,\,\,\,\sigma}^{\mu\sigma}
\nonumber
\\
&+&K_{\,\,\nu\rho}^{\mu}q_{1}\left(\Box\right)\partial_{\mu}\partial_{\sigma}K^{\sigma\nu\rho}+K_{\,\,\nu\rho}^{\mu}q_{2}\left(\Box\right)\partial_{\mu}\partial_{\sigma}K^{\sigma\rho\nu}+K_{\mu\,\,\,\,\,\nu}^{\,\,\rho}q_{3}\left(\Box\right)\partial_{\rho}\partial_{\sigma}K^{\mu\nu\sigma}
\nonumber
\\
&+&K_{\mu\,\,\,\,\,\nu}^{\,\,\rho}q_{4}\left(\Box\right)\partial_{\rho}\partial_{\sigma}K^{\mu\sigma\nu}+K_{\,\,\,\,\,\rho}^{\mu\rho}q_{5}\left(\Box\right)\partial_{\mu}\partial_{\nu}K_{\,\,\,\,\,\sigma}^{\nu\sigma}+K_{\,\,\,\lambda\sigma}^{\lambda}q_{6}\left(\Box\right)\partial_{\mu}\partial_{\alpha}K^{\sigma\mu\alpha}
\nonumber
\\
&+&K_{\mu}^{\,\,\nu\rho}s\left(\Box\right)\partial_{\nu}\partial_{\rho}\partial_{\alpha}\partial_{\sigma}K^{\mu\alpha\sigma}.
\end{eqnarray}
In order to get a deeper insight about how the functions involved in Eqs.\eqref{compac1}, \eqref{compac2} and \eqref{compac3} are related with the $\tilde{F}_{i}\left(\Box\right)$'s in \eqref{lagrangian}, we refer to Appendix \ref{ap:4}.
At this stage, it is interesting to note that $\mathcal{L}_{M}$ in \eqref{compac1} possesses metric terms only and coincides with the Lagrangian of the non-torsion case~\cite{Biswas:2011ar}, as expected. On the other hand, $\mathcal{L}_{MT}$ in \eqref{compac2} contains the mixed terms between metric and torsion, whereas $\mathcal{L}_{T}$ contains only torsion expressions.  

It is also worth calculating the local limit of \eqref{laaccion} by taking $M_{S} \rightarrow \infty$, since it will allow us to know the conditions to be imposed in the non-local functions in order to recover a PG theory in the IR. For the detailed calculations we refer the reader to Appendix \ref{ap:5}. Here we will just summarise that the local limit of the theory is 
\begin{eqnarray}
\mathcal{L}_{{\rm GPG}}&=&\tilde{R}+b_{1}\tilde{R}^{2}+b_{2}\tilde{R}_{\mu\nu\rho\sigma}\tilde{R}^{\mu\nu\rho\sigma}+b_{3}\tilde{R}_{\mu\nu\rho\sigma}\tilde{R}^{\rho\sigma\mu\nu}+2\left(b_{1}-b_{2}-b_{3}\right)\tilde{R}_{\mu\nu\rho\sigma}\tilde{R}^{\mu\rho\nu\sigma}
\nonumber
\\
&&+b_{5}\tilde{R}_{\mu\nu}\tilde{R}^{\mu\nu}-\left(4b_{1}+b_{5}\right)\tilde{R}_{\mu\nu}\tilde{R}^{\nu\mu}+a_{1}K_{\mu\nu\rho}K^{\mu\nu\rho}+a_{2}K_{\mu\nu\rho}K^{\mu\rho\nu}
\nonumber
\\
&&+a_{3}K_{\nu\,\,\,\,\,\mu}^{\,\,\,\mu}K_{\,\,\,\,\,\,\rho}^{\nu\rho}+c_{1}K_{\,\,\nu\rho}^{\mu}\nabla_{\mu}\nabla_{\sigma}K^{\sigma\nu\rho}+c_{2}K_{\,\,\nu\rho}^{\mu}\nabla_{\mu}\nabla_{\sigma}K^{\sigma\rho\nu}
\nonumber
\\
&&+c_{3}K_{\mu\,\,\,\,\,\nu}^{\,\,\rho}\nabla_{\rho}\nabla_{\sigma}K^{\mu\nu\sigma}+c_{4}K_{\mu\,\,\,\,\,\nu}^{\,\,\rho}\nabla_{\rho}\nabla_{\sigma}K^{\mu\sigma\nu},
\end{eqnarray} 
given that the conditions in \eqref{localconditions} are met.

As we saw in Section \ref{2.3}, the fact that the terms of the form $\nabla_{\mu}K_{\,\,\,\nu\rho}^{\mu}\nabla_{\sigma}K^{\sigma\nu\rho}$ are part of the Lagrangian can contribute to make the vector modes present in the theory ghost-free in the IR limit. We shall prove that the two vector modes can be made ghost-free in the proposed non-local extension of PG gravity.

\subsection{Field equations}

Since the connection under consideration is different from the Levi-Civita one, and consequently the metric and the connections are {\it a priori} independent, we will have two set of equations, namely
\begin{itemize}
\item \textbf{Einstein Equations}: Variation of the action \eqref{laaccion}, with respect to the metric:
\begin{equation}
\frac{\delta_{g}S_{M}}{\delta g^{\mu\nu}}+\frac{\delta_{g}S_{MT}}{\delta g^{\mu\nu}}=0.
\end{equation}
\item \textbf{Cartan Equations}: Variation of the action \eqref{laaccion}, with respect to the contortion\footnote{Note that varying with respect to the contortion is equivalent to varying with respect to the torsion, since they are related by a linear expression.}
\begin{equation}
\frac{\delta_{K}S_{MT}}{\delta K_{\,\,\nu\rho}^{\mu}}+\frac{\delta_{K}S_{T}}{\delta K_{\,\,\nu\rho}^{\mu}}=0.
\end{equation}
\end{itemize}
It is interesting to note that $\frac{\delta_{g}S_{M}}{\delta g^{\mu\nu}}$ has already been calculated in \cite{Biswas:2011ar}, although, calculations involving such a term have been performed again as a consistency check. Let us sketch the calculations leading towards the field equations.

\subsubsection{Einstein Equations}

Performing variations with respect to the metric in $S_{M}$, we find
\begin{eqnarray}
\frac{\delta_{g}S_{M}}{\delta g^{\mu\nu}}&=&\Box a\left(\Box\right)h_{\mu\nu}+b\left(\Box\right)\partial_{\sigma}\partial_{\left(\nu\right.}h_{\left.\mu\right)}^{\,\,\,\sigma}+c\left(\Box\right)\left[\partial_{\mu}\partial_{\nu}h+\eta_{\mu\nu}\partial_{\rho}\partial_{\sigma}h^{\rho\sigma}\right]+\eta_{\mu\nu}\Box d\left(\Box\right)h
\nonumber
\\
&+&2\frac{f\left(\Box\right)}{\Box}\partial_{\mu}\partial_{\nu}\partial_{\rho}\partial_{\sigma}h^{\rho\sigma},
\end{eqnarray}
which is compatible with the results in Ref.~\cite{Biswas:2011ar}.
For $S_{MT}$, we have
\begin{eqnarray}
\frac{\delta_{g}S_{MT}}{\delta g^{\mu\nu}}&=&\eta_{\mu\nu}\Box u\left(\Box\right)\partial_{\rho}K_{\,\,\,\,\,\sigma}^{\rho\sigma}+v_{1}\left(\Box\right)\partial_{\mu}\partial_{\nu}\partial_{\rho}K_{\,\,\,\,\,\sigma}^{\rho\sigma}+v_{2}\left(\Box\right)\partial_{\sigma}\partial_{\rho}\partial_{\left(\nu\right.}K_{\left.\mu\right)}^{\,\,\,\sigma\rho}
\nonumber
\\
&+&\Box w\left(\Box\right)\partial_{\rho}K_{\,\,\left(\mu\nu\right)}^{\rho}.
\end{eqnarray}
Therefore, the resulting Einstein's equations are
\begin{eqnarray}
\label{einsteingen}
&&\Box a\left(\Box\right)h_{\mu\nu}+b\left(\Box\right)\partial_{\sigma}\partial_{\left(\nu\right.}h_{\left.\mu\right)}^{\,\,\,\sigma}+c\left(\Box\right)\left[\partial_{\mu}\partial_{\nu}h+\eta_{\mu\nu}\partial_{\rho}\partial_{\sigma}h^{\rho\sigma}\right]+\eta_{\mu\nu}\Box d\left(\Box\right)h
\nonumber
\\
&&+2\frac{f\left(\Box\right)}{\Box}\partial_{\mu}\partial_{\nu}\partial_{\rho}\partial_{\sigma}h^{\rho\sigma}+\eta_{\mu\nu}\Box u\left(\Box\right)\partial_{\rho}K_{\,\,\,\,\,\sigma}^{\rho\sigma}+v_{1}\left(\Box\right)\partial_{\mu}\partial_{\nu}\partial_{\rho}K_{\,\,\,\,\,\sigma}^{\rho\sigma}
\nonumber
\\
&&+v_{2}\left(\Box\right)\partial_{\sigma}\partial_{\rho}\partial_{\left(\nu\right.}K_{\left.\mu\right)}^{\,\,\,\sigma\rho}+\Box w\left(\Box\right)\partial_{\rho}K_{\,\,\left(\mu\nu\right)}^{\rho}=\tau_{\mu\nu},
\end{eqnarray}
where $\tau_{\mu\nu}={\delta S_{matter}}/{\delta g^{\mu\nu}}$ is the usual energy-momentum tensor for matter fields.
At this stage, we can resort to the conservation of the energy-momentum tensor, $\partial_{\mu}\tau^{\mu\nu}=0$, to find the following constraints on the functions involved in \eqref{einsteingen}
\begin{eqnarray}
\label{constraints}
&&a(\Box)+b(\Box)=0,~c(\Box)+d(\Box)=0,~b(\Box)+c(\Box)+f(\Box)=0\,,
\nonumber
\\
&&u(\Box)+v_{1}(\Box)=0\,,~~~~~~v_{2}(\Box)-w(\Box)=0\,.
\end{eqnarray}
We can also prove these constraints by looking at the explicit expression of the functions in Eq.\eqref{constraints} provided in the Appendix \ref{ap:4}. 


\subsubsection{Cartan Equations}

On the other hand, performing variations with respect to the contortion, we find
\begin{eqnarray}
\label{cartan1}
\frac{\delta S_{MT}}{\delta K_{\,\,\nu\rho}^{\mu}}&=&-\Box u\left(\Box\right)\partial_{\left[\mu\right.}\eta^{\left.\rho\right]\nu}h-v_{1}\left(\Box\right)\partial^{\alpha}\partial^{\beta}\partial_{\left[\mu\right.}\eta^{\left.\rho\right]\nu}h_{\alpha\beta}-v_{2}\left(\Box\right)\partial^{\beta}\partial^{\nu}\partial^{\left[\rho\right.}h_{\left.\mu\right]\beta}
\nonumber
\\
&&-\Box w\left(\Box\right)\partial_{\left[\mu\right.}h^{\left.\rho\right]\nu},
\end{eqnarray}
and
\begin{eqnarray}
\label{cartan2}
\frac{\delta S_{T}}{\delta K_{\,\,\nu\rho}^{\mu}}&=&2p_{1}\left(\Box\right)K_{\mu}^{\,\,\nu\rho}+2p_{2}\left(\Box\right)K_{\left[\mu\right.}^{\,\,\,\,\,\left.\rho\right]\nu}+2p_{3}\left(\Box\right)\eta^{\nu\left[\rho\right.}K_{\left.\mu\right]\,\,\,\,\,\sigma}^{\,\,\,\,\sigma}
\nonumber
\\
&&-2q_{1}\left(\Box\right)\partial_{\sigma}\partial_{\left[\mu\right.}K^{\left.\rho\right]\nu\sigma}+2q_{2}\left(\Box\right)\partial_{\sigma}\partial_{\left[\mu\right.}K^{\sigma\left|\rho\right]\nu}
\nonumber
\\
&&+q_{3}\left(\Box\right)\left(\partial^{\nu}\partial_{\sigma}K_{\left[\mu\right.}^{\,\,\,\left.\rho\right]\sigma}+\partial_{\sigma}\partial^{\left[\rho\right.}K_{\left.\mu\right]\,\,\,\,\,}^{\,\,\,\,\sigma\nu}\right)+2q_{4}\left(\Box\right)\partial^{\nu}\partial_{\sigma}K_{\mu}^{\,\,\,\sigma\rho}
\nonumber
\\
&&+2q_{5}\left(\Box\right)\eta^{\nu\left[\rho\right.}\partial_{\left.\mu\right]}\partial_{\lambda}K_{\,\,\,\,\,\sigma}^{\lambda\sigma}+q_{6}\left(\Box\right)\left(\partial_{\lambda}\partial_{\alpha}\eta_{\left[\mu\right.}^{\nu}K^{\left.\rho\right]\lambda\alpha}-\partial^{\nu}\partial^{\left[\rho\right.}K_{\left.\mu\right]\lambda}^{\,\,\,\,\,\,\lambda}\right)
\nonumber
\\
&&+2s\left(\Box\right)\partial^{\sigma}\partial^{\lambda}\partial^{\nu}\partial^{\left[\rho\right.}K_{\left.\mu\right]\sigma\lambda}.
\end{eqnarray}
This leads us to the Cartan Equations
\begin{eqnarray}
\label{cartangen}
&-&\Box u\left(\Box\right)\partial_{\left[\mu\right.}\eta^{\left.\rho\right]\nu}h-v_{1}\left(\Box\right)\partial^{\alpha}\partial^{\beta}\partial_{\left[\mu\right.}\eta^{\left.\rho\right]\nu}h_{\alpha\beta}-v_{2}\left(\Box\right)\partial^{\beta}\partial^{\nu}\partial^{\left[\rho\right.}h_{\left.\mu\right]\beta}
\nonumber
\\
&-&\Box w\left(\Box\right)\partial_{\left[\mu\right.}h^{\left.\rho\right]\nu}+2p_{1}\left(\Box\right)K_{\mu}^{\,\,\nu\rho}+2p_{2}\left(\Box\right)K_{\left[\mu\right.}^{\,\,\,\,\,\left.\rho\right]\nu}+2p_{3}\left(\Box\right)\eta^{\nu\left[\rho\right.}K_{\left.\mu\right]\,\,\,\,\,\sigma}^{\,\,\,\,\sigma}
\nonumber
\\
&-&2q_{1}\left(\Box\right)\partial_{\sigma}\partial_{\left[\mu\right.}K^{\left.\rho\right]\nu\sigma}+2q_{2}\left(\Box\right)\partial_{\sigma}\partial_{\left[\mu\right.}K^{\sigma\left|\rho\right]\nu}
\nonumber
\\
&+&q_{3}\left(\Box\right)\left(\partial^{\nu}\partial_{\sigma}K_{\left[\mu\right.}^{\,\,\,\left.\rho\right]\sigma}+\partial_{\sigma}\partial^{\left[\rho\right.}K_{\left.\mu\right]\,\,\,\,\,}^{\,\,\,\,\sigma\nu}\right)+2q_{4}\left(\Box\right)\partial^{\nu}\partial_{\sigma}K_{\mu}^{\,\,\,\sigma\rho}
\nonumber
\\
&+&2q_{5}\left(\Box\right)\eta^{\nu\left[\rho\right.}\partial_{\left.\mu\right]}\partial_{\lambda}K_{\,\,\,\,\,\sigma}^{\lambda\sigma}+q_{6}\left(\Box\right)\left(\partial_{\lambda}\partial_{\alpha}\eta_{\left[\mu\right.}^{\nu}K^{\left.\rho\right]\lambda\alpha}-\partial^{\nu}\partial^{\left[\rho\right.}K_{\left.\mu\right]\lambda}^{\,\,\,\,\,\,\lambda}\right)
\nonumber
\\
&+&2s\left(\Box\right)\partial^{\sigma}\partial^{\lambda}\partial^{\nu}\partial^{\left[\rho\right.}K_{\left.\mu\right]\sigma\lambda}=\Sigma_{\mu}^{\,\,\,\nu\rho},
\end{eqnarray}
where $\Sigma_{\mu}^{\,\,\,\nu\rho}={\delta S_{matter}}/{\delta K_{\,\,\,\nu\rho}^{\mu}}$. From these field equations \eqref{einsteingen} and \eqref{cartangen} exact solutions cannot be obtained easily. In order to solve them, in the following we shall decompose the contortion field $K_{\mu\nu\rho}$ into its three irreducible components.

\subsection{Torsion decomposition}

In four dimensions, the torsion field $T_{\mu\nu\rho}$, as well as the contortion field $K_{\mu\nu\rho}$ (since it is also a three rank tensor with two antisymmetric indices), can be decomposed into three irreducible Lorentz invariant terms \cite{Shapiro:2001rz}, as we saw at the end of Section \ref{2.2}. By abusing the language we will denote the two vectors and the tensor of contortion decomposition equal to those of the torsion decomposition.\\
This decomposition turns out to be very useful, thanks to the fact that the three terms in Eq. \eqref{decomposition} propagate different dynamical off-shell degrees of freedom. Hence, it is convenient to study each of them separately, compared to the whole torsion contribution at the same time. Also interaction with matter, more specifically with fermions, is only made via the axial vector\footnote{Note that the axial part of the torsion and the contortion are the same.} \cite{Shapiro:2001rz}. That is why the two remaining components are usually known as {\it inert torsion}. Under this decomposition we will study how the torsion related terms in the linearised Lagrangian in Eq.\eqref{laaccion} change, and how to rederive the corresponding field equations. 
Introducing \eqref{decomposition2} (in terms of the contortion), and the constrains of the functions obtained in \eqref{constraints}, in \eqref{compac2} we find that the mixed term of the Lagrangian becomes
\begin{eqnarray}
\label{mixedecom}
\mathcal{L}_{MT}&=&h\Box\left(u\left(\Box\right)+\frac{1}{3}v_{2}\left(\Box\right)\right)\partial_{\mu}T^{\mu}-h_{\mu\nu}\left(u\left(\Box\right)+\frac{1}{3}v_{2}\left(\Box\right)\right)\partial^{\mu}\partial^{\nu}\partial_{\rho}T^{\rho}
\nonumber
\\
&+&h_{\mu\nu}v_{2}\left(\Box\right)\partial^{\nu}\partial_{\rho}\partial_{\sigma}q^{\mu\rho\sigma}+h_{\mu\nu}\Box v_{2}\left(\Box\right)\partial_{\sigma}q^{\mu\nu\sigma}.
\end{eqnarray}
Now, integrating by parts and using the linearised expression for the Ricci scalar we find
\begin{eqnarray}
\label{mixedecom2}
\mathcal{L}_{MT}&=&-\mathring{R}\left(u\left(\Box\right)+\frac{1}{3}v_{2}\left(\Box\right)\right)\partial_{\mu}T^{\mu}+h_{\mu\nu}v_{2}\left(\Box\right)\partial^{\nu}\partial_{\rho}\partial_{\sigma}q^{\mu\rho\sigma}
\nonumber
\\
&+&h_{\mu\nu}\Box v_{2}\left(\Box\right)\partial_{\sigma}q^{\mu\nu\sigma}
\end{eqnarray}
The first term accounts for a non-minimal coupling of the trace vector with the curvature, which, as seen in Section \ref{2.3}, is known for producing ghostly degrees of freedom. Therefore, for stability reasons we impose $v_{2}\left(\Box\right)=-3u\left(\Box\right)$, finally obtaining
\begin{equation}
\label{mixedecom3}
\mathcal{L}_{MT}=-3h_{\mu\nu}u\left(\Box\right)\partial^{\nu}\partial_{\rho}\partial_{\sigma}q^{\mu\rho\sigma}-3h_{\mu\nu}\Box u\left(\Box\right)\partial_{\sigma}q^{\mu\nu\sigma}.
\end{equation}
In order to obtain the pure torsion part of the Lagrangian we substitute \eqref{decomposition2}, in terms of the contortion, into \eqref{compac3}
\begin{eqnarray}
\label{torsdecom}
\mathcal{L}_{T}&=&\frac{1}{6}S_{\mu}\left(p_{2}\left(\Box\right)-p_{1}\left(\Box\right)\right)S^{\mu}+\frac{1}{9}\partial_{\left[\mu\right.}S_{\left.\nu\right]}\left(q_{1}\left(\Box\right)-q_{2}\left(\Box\right)-q_{3}\left(\Box\right)+q_{4}\left(\Box\right)\right)\partial^{\left[\mu\right.}S^{\left.\nu\right]}
\nonumber
\\
&+&\frac{1}{3}T_{\mu}\left(2p_{1}\left(\Box\right)+p_{2}\left(\Box\right)+3p_{3}\left(\Box\right)+\frac{1}{2}s\left(\Box\right)\Box^{2}\right)T^{\mu}
\nonumber
\\
&-&\frac{2}{9}\partial_{\left[\mu\right.}T_{\left.\nu\right]}\left(q_{1}\left(\Box\right)+q_{3}\left(\Box\right)+2q_{4}\left(\Box\right)-3q_{6}\left(\Box\right)\right)\partial^{\left[\mu\right.}T^{\left.\nu\right]}
\nonumber
\\
&-&\frac{1}{9}\partial_{\mu}T^{\mu}\left(3q_{1}\left(\Box\right)+3q_{2}\left(\Box\right)+9q_{5}\left(\Box\right)-s\left(\Box\right)\Box\right)\partial_{\nu}T^{\nu}+q_{\mu\nu\rho}p_{1}\left(\Box\right)q^{\mu\nu\rho}
\nonumber
\\
&+&q_{\mu\nu\rho}p_{2}\left(\Box\right)q^{\mu\rho\nu}+q_{\,\,\nu\rho}^{\mu}q_{1}\left(\Box\right)\partial_{\mu}\partial_{\sigma}q^{\sigma\nu\rho}+q_{\,\,\nu\rho}^{\mu}q_{2}\left(\Box\right)\partial_{\mu}\partial_{\sigma}q^{\sigma\rho\nu}
\nonumber
\\
&+&q_{\mu\,\,\,\,\,\nu}^{\,\,\rho}q_{3}\left(\Box\right)\partial_{\rho}\partial_{\sigma}q^{\mu\nu\sigma}+q_{\mu\,\,\,\,\,\nu}^{\,\,\rho}q_{4}\left(\Box\right)\partial_{\rho}\partial_{\sigma}q^{\mu\sigma\nu}+q^{\mu\nu\rho}s\left(\Box\right)\partial_{\nu}\partial_{\rho}\partial_{\sigma}\partial_{\lambda}q_{\mu}^{\,\,\,\sigma\lambda}
\nonumber
\\
&+&\frac{1}{3}T_{\mu}\left(2q_{1}\left(\Box\right)+2q_{3}\left(\Box\right)+4q_{4}\left(\Box\right)-3q_{6}\left(\Box\right)+2s\left(\Box\right)\Box\right)\partial_{\nu}\partial_{\rho}q^{\mu\nu\rho}
\nonumber
\\
&+&\frac{1}{2}\varepsilon_{\mu\nu\rho\sigma}q^{\rho\lambda\sigma}q_{3}\left(\Box\right)\partial_{\lambda}\partial^{\nu}S^{\mu}
\end{eqnarray}
Now we can proceed to calculate the field equations under the torsion decomposition. Varying the complete decomposed Lagrangian formed of \eqref{mixedecom3}, \eqref{torsdecom}, and $\mathcal{L}_M$, with respect to the metric we find the Einstein Equations:
\begin{eqnarray}
\label{einsteindecomposed}
&&\Box a\left(\Box\right)h_{\mu\nu}+b\left(\Box\right)\partial_{\sigma}\partial_{\left(\nu\right.}h_{\left.\mu\right)}^{\,\,\,\sigma}+c\left(\Box\right)\left[\partial_{\mu}\partial_{\nu}h+\eta_{\mu\nu}\partial_{\rho}\partial_{\sigma}h^{\rho\sigma}\right]+\eta_{\mu\nu}\Box d\left(\Box\right)h
\nonumber
\\
&&+2\frac{f\left(\Box\right)}{\Box}\partial_{\mu}\partial_{\nu}\partial_{\rho}\partial_{\sigma}h^{\rho\sigma}-3u\left(\Box\right)\partial_{\sigma}\partial_{\rho}\partial_{\left(\nu\right.}q_{\left.\mu\right)}^{\,\,\,\sigma\rho}-3\Box u\left(\Box\right)\partial_{\rho}q_{\,\,\left(\mu\nu\right)}^{\rho}=\tau_{\mu\nu},
\end{eqnarray}
where we can see that the vectorial parts of the torsion tensor do not appear.\\
On the other hand, performing variations with respect to the three different invariants of the torsion we find the corresponding Cartan Equations:
\begin{itemize}
\item Variations with respect to the axial vector $S^{\mu}$
\begin{eqnarray}
\label{cartandecax}
&&\frac{1}{6}\left(p_{2}\left(\Box\right)-p_{1}\left(\Box\right)\right)S_{\mu}
\nonumber
\\
&&+\frac{1}{18}\left(q_{1}\left(\Box\right)-q_{2}\left(\Box\right)-q_{3}\left(\Box\right)+q_{4}\left(\Box\right)\right)\left(\partial_{\mu}\partial_{\nu}S^{\nu}-\Box S_{\mu}\right)
\nonumber
\\
&&+\frac{1}{2}\varepsilon_{\mu\nu\rho\sigma}q_{3}\left(\Box\right)\partial_{\lambda}\partial^{\nu}q^{\rho\lambda\sigma}=\frac{\delta\mathcal{L}_{matter}}{\delta S^{\mu}}.
\end{eqnarray}
\item Variations with respect to the trace vector $T^{\mu}$
\begin{eqnarray}
\label{cartandectra}
&&\frac{1}{3}\left(2p_{1}\left(\Box\right)+p_{2}\left(\Box\right)+3p_{3}\left(\Box\right)+\frac{1}{2}s\left(\Box\right)\Box^{2}\right)T_{\mu}
\nonumber
\\
&&-\frac{1}{9}\left(q_{1}\left(\Box\right)+q_{3}\left(\Box\right)+2q_{4}\left(\Box\right)-3q_{6}\left(\Box\right)\right)\left(\partial_{\mu}\partial_{\nu}T^{\nu}-\Box T_{\mu}\right)
\nonumber
\\
&&+\frac{1}{9}\left(3q_{1}\left(\Box\right)+3q_{2}\left(\Box\right)+9q_{5}\left(\Box\right)-s\left(\Box\right)\Box\right)\partial_{\mu}\partial_{\nu}T^{\nu}
\nonumber
\\
&&+\frac{1}{3}\left(2q_{1}\left(\Box\right)+2q_{3}\left(\Box\right)+4q_{4}\left(\Box\right)-3q_{6}\left(\Box\right)+2s\left(\Box\right)\Box\right)\partial_{\nu}\partial_{\rho}q_{\mu}^{\,\,\nu\rho}
\nonumber
\\
&&=\frac{\delta\mathcal{L}_{matter}}{\delta T^{\mu}}.
\end{eqnarray}
\item Variations with respect to the tensor part $q^{\mu\nu\rho}$
\begin{eqnarray}
&&p_{1}\left(\Box\right)q_{\mu\nu\rho}+p_{2}\left(\Box\right)q_{\left[\mu\rho\right]\nu}+q_{1}\left(\Box\right)\partial_{\left[\mu\right.}\partial_{\sigma}q_{\,\,\nu\left.\rho\right]}^{\sigma}+q_{2}\left(\Box\right)\partial_{\left[\mu\right.}\partial_{\sigma}q_{\,\,\left.\rho\right]\nu}^{\sigma}
\nonumber
\\
&&+q_{3}\left(\Box\right)\partial_{\sigma}\partial_{\left[\rho\right.}q_{\left.\mu\right]\,\,\,\,\,\nu}^{\,\,\,\sigma}+q_{4}\left(\Box\right)\partial_{\nu}\partial_{\sigma}q_{\mu\,\,\,\,\,\rho}^{\,\,\sigma}+s\left(\Box\right)\partial_{\nu}\partial_{\sigma}\partial_{\lambda}\partial_{\left[\rho\right.}q_{\left.\mu\right]}^{\,\,\,\sigma\lambda}
\nonumber
\\
&&+\frac{1}{3}\left(2q_{1}\left(\Box\right)+2q_{3}\left(\Box\right)+4q_{4}\left(\Box\right)-3q_{6}\left(\Box\right)+2s\left(\Box\right)\Box\right)\partial_{\nu}\partial_{\rho}T_{\mu}
\nonumber
\\
&&=\frac{\delta\mathcal{L}_{matter}}{\delta q^{\mu\nu\rho}}.
\end{eqnarray}
\end{itemize}
These decomposed equations will help us to find exact solutions of the theory, as we will see in the following section.

\section{Ghost and singularity free solutions}
\label{4.3}

In this Section we shall find solutions of the proposed UV extension of PG gravity, provided there exists a fermion as a source, and assuming that both axial and trace torsion are different from zero\footnote{The fact that the traceless tensor part of the torsion $q_{\,\,\nu\rho}^{\mu}$ is considered to be negligible is motivated by the fact that in a completely symmetric spacetime this component is identically zero \cite{Sur:2013aia}.}. For the usual IDG theory, solutions for this configuration were presented in \cite{Buoninfante:2018stt}. In order to render our case clearer, we have divided the calculations in the following two subsections. In the first one, we will solve Cartan equations to obtain the torsion tensor, while in the second one we will solve Einstein equations for the metric tensor.

\subsection{Cartan Equations} 
Let us write down the linearised Lagrangian decomposed into the two vector invariants, where the tensor component of the torsion has been set to zero. Thus,
\begin{eqnarray}
\label{vectoraction}
\mathcal{L}&=&\mathcal{L}_{M}+\frac{1}{6}S_{\mu}\left(p_{2}\left(\Box\right)-p_{1}\left(\Box\right)\right)S^{\mu}
\nonumber
\\
&&+\frac{1}{9}\partial_{\left[\mu\right.}S_{\left.\nu\right]}\left(q_{1}\left(\Box\right)-q_{2}\left(\Box\right)-q_{3}\left(\Box\right)+q_{4}\left(\Box\right)\right)\partial^{\left[\mu\right.}S^{\left.\nu\right]}
\nonumber
\\
&&+\frac{1}{3}T_{\mu}\left(2p_{1}\left(\Box\right)+p_{2}\left(\Box\right)+3p_{3}\left(\Box\right)+\frac{1}{2}s\left(\Box\right)\Box^{2}\right)T^{\mu}
\nonumber
\\
&&-\frac{2}{9}\partial_{\left[\mu\right.}T_{\left.\nu\right]}\left(q_{1}\left(\Box\right)+q_{3}\left(\Box\right)+2q_{4}\left(\Box\right)-3q_{6}\left(\Box\right)\right)\partial^{\left[\mu\right.}T^{\left.\nu\right]}
\nonumber
\\
&&-\frac{1}{9}\partial_{\mu}T^{\mu}\left(3q_{1}\left(\Box\right)+3q_{2}\left(\Box\right)+9q_{5}\left(\Box\right)-s\left(\Box\right)\Box\right)\partial_{\nu}T^{\nu},
\end{eqnarray}
where we have taken into account the constraints on the functions in \eqref{constraints} and the stability condition for the trace vector found in the previous section, namely $v_{2}\left(\Box\right)=-3u\left(\Box\right)$. Due to these conditions, there are no mixed terms between metric and torsion, so the Cartan and Einstein Equations would be decoupled. \\
Despite these constraints, the torsion part of the Lagrangian \eqref{vectoraction} is far from being stable, so before finding some solutions we need to explore under which form of the functions the theory does not have any pathologies.\\
By taking a closer look at \eqref{vectoraction} we realise that, as it is usual in metric IDG, we can make the combinations of the non-local functions to be described by an entire function. Such an entire function would not introduce any new poles in the propagators, so that we can use the same stability arguments as in the local theory. This means that
\begin{eqnarray}
&&p_{2}\left(\Box\right)-p_{1}\left(\Box\right)=C_{1}{\rm e}^{-\frac{\Box}{M_{S}^{2}}},
\nonumber
\\
&&q_{1}\left(\Box\right)-q_{2}\left(\Box\right)-q_{3}\left(\Box\right)+q_{4}\left(\Box\right)=C_{2}{\rm e}^{-\frac{\Box}{M_{S}^{2}}},
\nonumber
\\
&&2p_{1}\left(\Box\right)+p_{2}\left(\Box\right)+3p_{3}\left(\Box\right)+\frac{1}{2}s\left(\Box\right)\Box^{2}=C_{3}{\rm e}^{-\frac{\Box}{M_{S}^{2}}},
\\
&&q_{1}\left(\Box\right)+q_{3}\left(\Box\right)+2q_{4}\left(\Box\right)-3q_{6}\left(\Box\right)=C_{4}{\rm e}^{-\frac{\Box}{M_{S}^{2}}},
\nonumber
\\
&&3q_{1}\left(\Box\right)+3q_{2}\left(\Box\right)+9q_{5}\left(\Box\right)-s\left(\Box\right)\Box=C_{5}{\rm e}^{-\frac{\Box}{M_{S}^{2}}}, \nonumber
\end{eqnarray}
where the $C_{i}$ are constants and we have used the exponential as a paradigmatic example of an entire function.\\
This gives us the following Lagrangian
\begin{eqnarray}
\label{vectorentire}
\mathcal{L}&=&\mathcal{L}_{M}+\frac{1}{6}C_{1}\hat{S}_{\mu}\hat{S}^{\mu}+\frac{1}{9}C_{2}\partial_{\left[\mu\right.}\hat{S}_{\left.\nu\right]}\partial^{\left[\mu\right.}\hat{S}^{\left.\nu\right]}+\frac{1}{3}C_{3}\hat{T}_{\mu}\hat{T}^{\mu}-\frac{2}{9}C_{4}\partial_{\left[\mu\right.}\hat{T}_{\left.\nu\right]}\partial^{\left[\mu\right.}\hat{T}^{\left.\nu\right]}
\nonumber
\\
&&-\frac{1}{9}C_{5}\partial_{\mu}\hat{T}^{\mu}\partial_{\nu}\hat{T}^{\nu},
\end{eqnarray}
where $\hat{S}^{\mu}={\rm e}^{-\frac{\Box}{2M_{S}^{2}}}S^{\mu}$ and $\hat{T}^{\mu}={\rm e}^{-\frac{\Box}{2M_{S}^{2}}}T^{\mu}$. From the standard theory of vector fields we know that the last term introduces ghostly degrees of freedom, therefore we need to impose that $C_{5}=0$. Moreover, the kinetic terms of both vectors need to be positive, hence we also have the conditions $C_{2}>0$ and $C_{4}<0$.

At this time we know that our theory is absent of ghosts, and we are ready to find some possible solutions, that we will show can be singularity-free. We will study the solutions of the trace and axial vector separately in the following Subsections. This is indeed possible since parity breaking terms in the action are not considered, so there are no mixed trace-axial terms.

\subsubsection{Axial vector and the ring singularity}
\label{axialsec}
First, we will consider the Cartan Equations for the axial vector \eqref{cartandecax} with a fermionic source term
\begin{eqnarray}
\label{cartananti2}
\tilde{C}_{1}S_{\mu}+\tilde{C}_{2}\left(\partial_{\mu}\partial_{\nu}S^{\nu}-\Box S_{\mu}\right)={\rm e}^{\frac{\Box}{M_{S}^{2}}}B_{\mu},
\end{eqnarray}
where $\tilde{C}_{1}=\frac{1}{6}C_{1}$, $\tilde{C}_{2}=\frac{1}{18}C_{2}$, and $B_{\mu}=\frac{\delta\mathcal{L}_{fermion}}{\delta S^{\mu}}$ accounts for the internal spin of the fermion, which minimally couples to the axial vector \cite{Shapiro:2001rz}. Equation \eqref{cartananti2} describes a non-local Proca field in a Minkowski spacetime. Furthermore, this non-local aspect cannot be hidden by a redefinition of the field since there is a source term $B_{\mu}$.\\

In order to find a solution of Eq.\eqref{cartananti2} for the axial vector, we shall assume the transverse condition $\partial_{\mu}S^{\mu}=0$. Moreover, we need to provide a form of the $A_\mu$ function. Since we are trying to prove that in the UV extension of PG we can also avoid singularities, we will consider the ``most singular'' possible configuration, and see if we are able to ameliorate it. In this case, since fermions have an intrinsic spin, instead of having a Dirac-delta point source, we would need to consider a singular source endowed with angular momentum. Indeed, we would need a rotating singular Dirac-delta ring, where we shall fix the angular momentum to be in $z$ direction. We will use cartesian coordinates, in which the singular source can be expressed as 
\begin{equation}
\begin{cases}
\begin{array}{c}
B^{z}=A\delta\left(z\right)\delta\left(x^{2}+y^{2}-R^{2}\right),\\
\,\\
B^{\mu}=0\quad,\quad\mu=t,x,y,
\end{array} & \,\end{cases}
\label{source}
\end{equation}
where $A$ is a constant. Then, the homogeneous solution of the Equation \eqref{cartananti2} will be the local Proca solution, and will propagate three stable degrees of freedom. Due to the specific source \eqref{source}, the $z$ component of the axial vector will also have an additional non-local term, that will be given by the particular solution of \eqref{cartananti2}. In order to obtain it, we shall substitute this source  \eqref{source} into Equation \eqref{cartananti2}, and taking into account the gauge choice that we mentioned, we find that
\begin{equation}
\label{singular}
\left(\tilde{C}_{1}-\tilde{C}_{2}\Box\right){\rm e}^{-\frac{\Box}{M_{S}^{2}}}S^{z}=A\delta\left(z\right)\delta\left(x^{2}+y^{2}-R^{2}\right),
\end{equation}
where $R$ holds for the so-called \emph{Cartan radius} of a singular rotating ring, where effectively the singularity is located. We shall now calculate the Fourier transform $\mathcal{F}$ of the source, as follows
\begin{equation}
\mathcal{F}\left[\delta\left(z\right)\delta\left(x^{2}+y^{2}-R^{2}\right)\right]=\pi {\rm J}_{0}\left(-R\sqrt{k_{x}^{2}+k_{y}^{2}}\right),
\end{equation} 
where ${\rm J}_{0}$ represents the Bessel function of first kind ($n=0$). Thus, applying the Fourier transform to Eq.\eqref{singular} one obtains
\begin{eqnarray}
\label{calcu}
&&\mathcal{F}\left[\left(\tilde{C}_{1}-\tilde{C}_{2}\Box\right){\rm e}^{-\Box/M_{S}^{2}}S^{z}\left(\overrightarrow{x}\right)\right]=\mathcal{F}\left[A\delta\left(z\right)\delta\left(x^{2}+y^{2}-R^{2}\right)\right]\Rightarrow
\nonumber
\\
&&\left(\tilde{C}_{1}+\tilde{C}_{2}k^{2}\right){\rm e}^{k^{2}/M_{S}^{2}}S^{z}\left(\overrightarrow{k}\right)=\pi A{\rm J}_{0}\left(-R\sqrt{k_{x}^{2}+k_{y}^{2}}\right)\Rightarrow 
\nonumber
\\
&&S^{z}\left(\overrightarrow{k}\right)=\pi A\frac{{\rm e}^{-k^{2}/M_{S}^{2}}}{\tilde{C}_{1}+\tilde{C}_{2}k^{2}}{\rm J}_{0}\left(-R\sqrt{k_{x}^{2}+k_{y}^{2}}\right),
\end{eqnarray}
Then, performing the inverse of the transform of \eqref{calcu} we find that the particular solution of Eq. \eqref{singular} can be expressed as 
\begin{equation}
\label{sol_S}
S^{\mu}=\pi A^{\mu}\int\frac{{\rm d}^{3}k}{\left(2\pi\right)^{3}}\frac{{\rm e}^{-k^{2}/M_{S}^{2}}}{\tilde{C}_{1}+\tilde{C}_{2}k^{2}}\,{\rm J}_{0}\left(-R\sqrt{k_{x}^{2}+k_{y}^{2}}\right)\,{\rm e}^{i\left(k_{x}x+k_{y}y+k_{z}z\right)},
\end{equation}
where ${\rm d}^{3}k={\rm d}k_{x}{\rm d}k_{y}{\rm d}k_{z}$ and $k^{2}=k_{x}^{2}+k_{y}^{2}+k_{z}^{2}$. In order to see how the axial vector behaves at the singularity $r=R$, we can restrict the study of the integral in (\ref{sol_S}) to the $z=0$ plane, since we have assumed that the ring rotation axis lies along the $z$ direction. By using cylindrical coordinates, $k_{x}=\xi\cos\left(\varphi\right)$, $k_{y}=\xi\sin\left(\varphi\right)$, $k_{z}=k_{z}$, we obtain
\begin{eqnarray}
\label{integral}
&&S^{z}\left(r\right)=\pi A\int_{\xi=0}^{\xi=\infty}\int_{\varphi=0}^{\varphi=2\pi}\int_{k_{z}=0}^{k_{z}=\infty}\frac{\xi{\rm d}\varphi{\rm d}\xi{\rm d}k_{z}}{\left(2\pi\right)^{3}}\frac{{\rm e}^{-\left(\xi^{2}+k_{z}^{2}\right)/M_{S}^{2}}}{\tilde{C}_{1}+\tilde{C}_{2}\left(\xi^{2}+k_{z}^{2}\right)}
\nonumber
\\
&&\times{\rm J}_{0}\left(-R\xi\right){\rm e}^{i\xi x\cos\left(\varphi\right)}{\rm e}^{i\xi y\sin\left(\varphi\right)}=\frac{\pi A}{\left(2\pi\right)^{3}}\int_{\xi=0}^{\xi=\infty}\xi{\rm d}\xi{\rm J}_{0}\left(-R\xi\right)
\nonumber
\\
&&\times\left(\int_{k_{z}=0}^{k_{z}=\infty}{\rm d}k_{z}\frac{{\rm e}^{-\left(\xi^{2}+k_{z}^{2}\right)/M_{S}^{2}}}{\tilde{C}_{1}+\tilde{C}_{2}\left(\xi^{2}+k_{z}^{2}\right)}\right)\left(\int_{\varphi=0}^{\varphi=2\pi}{\rm d}\varphi\,{\rm e}^{i\xi x\cos\left(\varphi\right)}{\rm e}^{i\xi y\sin\left(\varphi\right)}\right)
\nonumber
\\
&&=\frac{A}{8\tilde{C}_{2}}{\rm e}^{\frac{\tilde{C}_{1}}{\tilde{C}_{2}M_{S}^{2}}}\int_{0}^{\infty}{\rm d}\xi\sqrt{\frac{\tilde{C}_{2}\xi^{2}}{\tilde{C}_{1}+\tilde{C}_{2}\xi^{2}}}{\rm J}_{0}\left(-R\xi\right){\rm J}_{0}\left(-\xi r\right)
\nonumber
\\
&&\times{\rm Erfc}\left(\sqrt{\frac{\tilde{C}_{1}+\tilde{C}_{2}\xi^{2}}{\tilde{C}_{2}M_{S}^{2}}}\right)\,,
\end{eqnarray}
where $r^2=x^2+y^2$ and ${\rm Erfc}$ is the complementary error function. In order to performed the previous derivations we have further assumed that $\tilde{C}_1$ and $\tilde{C}_2$ are of the same sign, so that the integral in $k_z$ could be solved. \\
Since finding the analytically closed form of \eqref{integral} is not possible, the integral can be solved numerically, as can be seen in Fig. \ref{fig:5}. There one can check that in the case of stable local Poincar\'e Gauge theories of gravity, in the limit $M_{S}\rightarrow \infty$, the singularity at $r=R$ is unavoidable. Nevertheless, we can state that \bf{within the infinite derivative theory of Poincar\'e gravity, the ring singularity can be smeared out}. Therefore, for IDG theories we conclude that the axial torsion is regular everywhere in presence of a Dirac-delta fermionic source with spin. This result is similar to the Kerr-like singularity which is cured in the infinite derivative metric theory of gravity~\cite{Buoninfante:2018xif}. Nevertheless, in this torsion infinite derivative theory of gravity, there is a crucial difference with respect to the purely metric one. In this case, since the Proca field is massive, \emph{i.e.} $\tilde{C}_1\ne 0$, the non-local effects are visible even when we are far away from the source, due to the factor ${\rm exp}\left(\frac{\tilde{C}_{1}}{\tilde{C}_{2}M_{S}^{2}}\right)$. This occurs if the mass of the Proca field, modulated by $\tilde{C}_1$, is of the same order, or higher, than the mass scale of non-locality $M_S$. Then, this effect can be avoided\footnote{The effect of the exponential term in Eq. \eqref{integral} shall be a problem, and would be advisable to avoid, if one wants to use this theory to resolve the singularity, and at the same time wants to obtain the same values as in the local theory when being away from the source.} if the mass of the Proca field is much smaller than the mass scale at which non-locality starts playing a role. 

Finally, since this particular solution that we need to add to the $z$ component of the axial vector is static, it does not contribute to propagate more than the three degrees of freedom of the local Proca theory, therefore \bf{the solution is also ghost-free}.

\begin{figure}
\begin{center}
\includegraphics[width=1\linewidth]{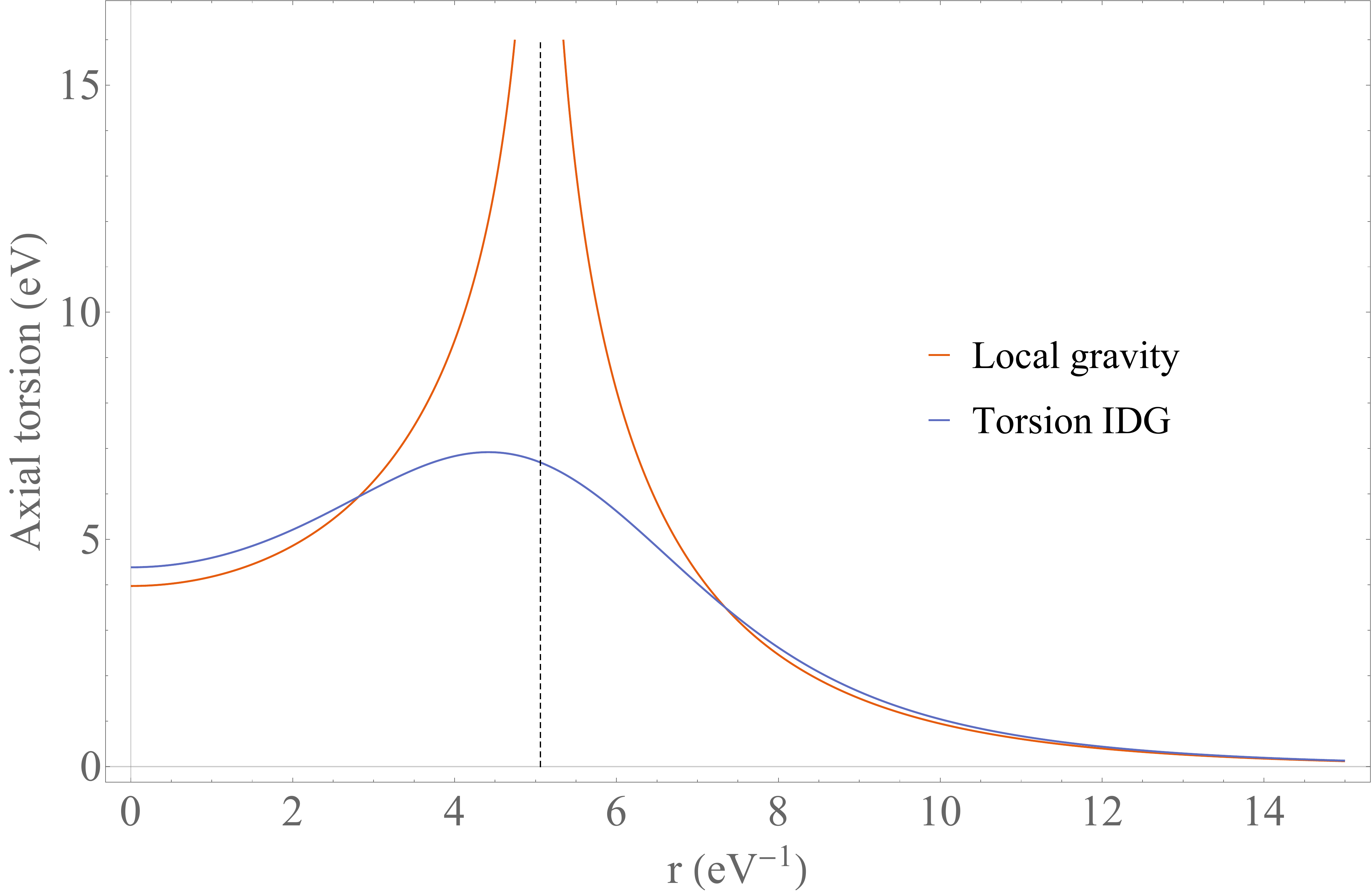}
\par\end{center}
\caption{Results of the numerical computation of \eqref{integral} for the case of local theories of gravity (limit when $M_{S}\rightarrow \infty$) and in the proposed IDG theory with torsion. We have chosen $A=800\,{\rm eV}$, $R=5.06\,{\rm eV}^{-1}$, $M_{S}=1\,{\rm eV}$, $\tilde{C}_1=0.1\,{\rm eV}^{2}$ and $\tilde{C}_2=1$.}
\label{fig:5}
\end{figure}


\subsubsection{Trace vector}
\label{tracesec}

Let us now explore the Cartan Equation for the trace vector \eqref{cartandectra}
\begin{eqnarray}
\label{procatrace}
\frac{1}{3}C_{3}T_{\mu}-\frac{1}{9}C_{4}\left(\partial_{\mu}\partial_{\nu}T^{\nu}-\Box T_{\mu}\right)=0.
\end{eqnarray}
We observe that this is just the local Proca Equation for a vector field. Therefore, it will have the same plane-wave solutions propagating three stable degrees of freedom.\\
Moreover, it is important to stress that in this case the kinetic term can have the same sign as the one of the axial vector, something that is not possible for quadratic PG theories, as we saw in section \ref{2.3}.

Now, with all the components for the torsion tensor calculated, we will solve Einstein's equations to obtain the corresponding metric $h_{\mu\nu}$.

\subsection{Einstein Equations solutions}

Let us recall that Einstein Equations for a fermionic source, where the tensor component of the torsion has been set to zero are given by \eqref{einsteindecomposed}:
\begin{eqnarray}
&&\Box a\left(\Box\right)h_{\mu\nu}+b\left(\Box\right)\partial_{\sigma}\partial_{\left(\nu\right.}h_{\left.\mu\right)}^{\,\,\,\sigma}+c\left(\Box\right)\left(\partial_{\mu}\partial_{\nu}h+\eta_{\mu\nu}\partial_{\rho}\partial_{\sigma}h^{\rho\sigma}\right)+\eta_{\mu\nu}\Box d\left(\Box\right)h
\nonumber
\\
&&+2\frac{f\left(\Box\right)}{\Box}\partial_{\mu}\partial_{\nu}\partial_{\rho}\partial_{\sigma}h^{\rho\sigma}=\tau_{\mu\nu},
\end{eqnarray}
where $\tau_{\mu\nu}=\eta_{\sigma\nu}F_{\mu\rho}F^{\sigma\rho}-\frac{1}{4}\eta_{\mu\nu}F_{\sigma\rho}F^{\sigma\rho}$, $F_{\mu\nu}$ being the electromagnetic tensor. It is clear that this equation is the same as in the pure metric case, since the torsion terms do not contribute.\\
Now, if we apply the constraints that we obtained from the energy-momentum conservation, and ghost-free conditions in the metric sector, see Eq.\eqref{constraints}, we are left with the following expression
\begin{eqnarray}
\label{einssimp}
{\rm {e}}^{-\Box/M_{S}^{2}}\left(\Box h_{\mu\nu}+\partial_{\mu}\partial_{\nu}h+\eta_{\mu\nu}\partial_{\rho}\partial_{\sigma}h^{\rho\sigma}-2\partial_{\sigma}\partial_{\left(\nu\right.}h_{\left.\mu\right)}^{\,\,\,\sigma}-\eta_{\mu\nu}\Box h\right)=\tau_{\mu\nu}.
\end{eqnarray}
It is interesting to note that this equation has already been studied in Ref.~\cite{Buoninfante:2018stt}, where a non-singular Reissner-Nordstr\"om solution were obtained for the same choice of the entire function in ghost free IDG, namely
\begin{equation}
{\rm d}s^{2}=-\left(1+2\Phi\left(r\right)\right){\rm d}t^{2}+\left(1-2\Psi\left(r\right)\right)\left({\rm d}r^{2}+r^{2}{\rm d}\Omega^{2}\right),
\end{equation}
where $\Phi\left(r\right)$ and $\Psi\left(r\right)$ take the following form \cite{Buoninfante:2018stt}
\begin{eqnarray}
\label{metricsol1}
&&\Phi\left(r\right)=-\frac{Gm}{r}\text{Erf}\left(\frac{M_{S}r}{2}\right)+\frac{GQ^{2}M_{S}}{2r}\text{F}\left(\frac{M_{S}r}{2}\right),
\\
&&\Psi\left(r\right)=-\frac{Gm}{r}\text{Erf}\left(\frac{M_{S}r}{2}\right)+\frac{GQ^{2}M_{S}}{4r}\text{F}\left(\frac{M_{S}r}{2}\right),
\label{metricsol2}
\end{eqnarray}
in which Erf$(x)$ is the error function and F$(x)$ the Dawson function.\
This solution is non-singular when $r\rightarrow 0$ and recasts a Reissner-Nordstr\"om  when $r\gg M_S^{-1}$.

\section{Chapter conclusions and outlook}

Within this chapter we have proposed a non-local extension of Poincar\'e Gauge gravity. For this purpose, first we have motivated the introduction of non-local terms into the action in order to ameliorate the singular behaviour at large energies. 

Then, in section \ref{4.2} we have constructed an Ultra-Violet extension of Poincar\'e Gauge Gravity and calculate the corresponding field equations. 

Finally, in the last section of this chapter we have found solutions of the theory at the linear level, and proved that they can be made ghost and singularity free by adjusting the theory parameters. Moreover, we find that if the mass of the axial vector mode is of the order of the mass-scale of non-locality, then the non-local effects can be observed macroscopically, which is something that is not possible in metric Infinite Derivative Gravity. \\

Based on the previous findings one could embark into new lines of research, such as the study at the non-linear limit of the proposed non-local theory and the search of new singularity and ghost-free solutions.

\renewcommand{\publ}{}


\chapter{Conclusions}

\PARstart{A}long this thesis we have studied some interesting aspects of Poincar\'e Gauge theories of gravity and proposed a non-local ultraviolet extension of them capable of potentially resolving some space-time singularities. Let us review the most important results that we have obtained throughout this work.

In Chapter \ref{2}, firstly we have explained some fundamentals of differential geometry, which are the base of any gravitational theory. We have also seen how the affine structure and the metric of the spacetime are not generally related. Consequently, there is no physical reason to impose a certain affine connection to the gravitational theory. Then we have reviewed the gauge procedure and constructed the quadratic  Lagrangian of Poincar\'e Gauge Gravity by requiring that the gravitational theory must be invariant under local Poincar\'e transformations. Finally, we have studied the stability of the quadratic Poincar\'e Gauge Lagrangian, which in principle propagates two massive scalar fields, two massive vectors fields, and two massive spin-2 fields. There have proven that only the two scalar degrees of freedom (one scalar and one pseudo-scalar) can propagate without introducing pathologies. In this regard, we have provided extensive details on the scalar, pseudo-scalar, and bi-scalar theories. Moreover, to conclude this Chapter we have suggested how to extend the quadratic Poincar\'e Gauge Lagrangian so that the two vector modes can propagate safely.

In Chapter \ref{3}, first we have explored how fermionic particles move in spacetimes endowed with a non-symmetric connection. We showed that the Dirac equation is modified with a coupling involving the totally antisymmetric part of the torsion tensor. Accordingly we have calculated the predicted non-geodesical behaviour at first order in the WKB approximation. Then, we have used this result in a particular black-hole solution of Poincar\'e Gauge gravity, and showed that there can be measurable differences between the trajectories of a fermion and a boson. Motivated by this fact, we have studied the singularity theorems in theories with torsion, to determine whether this non-geodesical behaviour could lead to the avoidance of singularities. Nevertheless, we have proven in Proposition \ref{prop:sin} that this would not possible provided that the conditions for the appearance of black holes are met. In the last section of this chapter, we have found that the only stable quadratic Poincar\'e Gauge theories that fulfill the Birkhoff theorem are the ones studied by Nieh and Rauch in the 1980s. We have also proved that, assuming asymptotic flatness and constant scalar curvature, the no-hair theorem applies for the most general stable quadratic Poincar\'e Gauge action. Moreover, we have seen how both Birkhoff and no-hair theorems are not related with the stability of the gravitational theory under consideration, and that indeed standard black-hole solutions present in General Relativity can also be solutions of unstable theories. Nevertheless, when performing perturbations up to a certain order, those instabilities will start playing a role.

In Chapter \ref{4}, we have motivated the introduction of non-local terms into the action in order to ameliorate the singular behaviour at large energies. Then, we have constructed one possible ultraviolet extension of Poincar\'e Gauge gravity. Finally, in the last section of this chapter we have found solutions of such a theory at the linear level, and proved that they are ghost- and singularity-free. Interestingly, we found that provided the mass of the axial vector mode is of the order of the mass-scale of non-locality, then the non-local effects can be observed macroscopically, which is something that is not possible in metric Infinite Derivative Gravity.

\subsubsection*{Open questions}

As it is customary in Science, we have answered some questions and established new concepts, while at the same time we have opened the box to future lines of research, which are summarised in the following:

\begin{itemize}

\item The construction of cosmological solutions of the bi-scalar model may be a worthwhile topic to explore, since the coupling of the pseudo-scalar with the fermions could lead to curious features. For instance, the effective mass of the neutrinos would change due to the torsion-spin coupling, hence affecting the large-scale structure formation \cite{Brookfield:2005td,Fukugita:1999as}.

\item Using the results in Section \ref{3.1} about the fermion dynamics in theories with torsion, one can explore the consequences of the torsion-spin coupling in the current quantum experiments, or propose new ones, in order to find better constraints for the torsion \cite{Lammerzahl:1997wk,Kostelecky:2007kx}.

\item It will be relevant to study possible black-hole solutions for the bi-scalar stable theory relaxing the assumptions of asymptotic flatness and constant scalar curvature, that we have made to study the no-hair theorem.

\item With respect to the non-local theories, it would be of great interest to elucidate the potential strong-coupling problem, mentioned in Section \ref{4.1}.

\item Finally, the study of the ultraviolet extension of Poincar\'e Gauge Gravity at the full non-linear regime, may bring us new solutions like regular black-holes or bouncing universes, that could be physically relevant to describe the current measures.

\end{itemize}

As always, it will be exciting to see where the future investigations would lead us to.


\renewcommand{\publ}{}


\appendix

\chapter{Acceleration components for an electron}

\label{ap:1}
Here we present the components of the acceleration of an explicitly. Such components have been calculated following the WKB approximation, in a Reissner-Nordstr\"om solution, as discussed in Subsection~\ref{3.1.2}.

\begin{eqnarray}
a^{t}&=&-\frac{\kappa\hbar}{2m_{esp}r^{2}\left(\frac{\kappa-2mr+r^{2}}{r^{2}}\right)^{3/2}}\left\{ \sqrt{\frac{\kappa-2mr+r^{2}}{r^{2}}}\sin(\alpha)\cos(\beta)r'(s) \right.
\nonumber
\\
&-&\theta'(s)\left[\sin(\alpha)\sin(\beta)\left(r-m\right)+\kappa r\cos(\alpha)\right]
\nonumber
\\
&+&\Biggl.\sin(\theta)\varphi'(s)\left[\cos(\alpha)\left(m-r\right)+\kappa r\sin(\alpha)\sin(\beta)\right]\Biggr\}
\end{eqnarray}

\begin{eqnarray}
a^{r}&=&-\frac{\hbar}{2m_{esp}r^{4}\left(\kappa-2mr+r^{2}\right)}\left\{ r\sqrt{\frac{\kappa-2mr+r^{2}}{r^{2}}}\left[\theta'(s)\left(\cos(\alpha)\left(2m^{2}r^{2}\right.\right.\right.\right.
\nonumber
\\
&-&\left.\left.mr^{3}-3m\kappa r+\kappa^{2}-\kappa^{2}r^{4}+\kappa r^{2}\right)+\kappa r^{3}\sin(\alpha)\sin(\beta)(m-r)\right)
\nonumber
\\
&+&+\sin(\theta)\varphi'(s)\left(\sin(\alpha)\sin(\beta)\left(-2m^{2}r^{2}+mr^{3}+3m\kappa r-\kappa^{2}+\kappa^{2}r^{4}-\kappa r^{2}\right)\right.
\nonumber
\\
&+&\Biggl.\left.\left.\kappa r^{3}\cos(\alpha)(m-r)\right)\right]+\kappa\sin(\alpha)\cos(\beta)\left(\kappa-2mr+r^{2}\right)^{2}t'(s)\Biggr\},
\end{eqnarray}

\begin{eqnarray}
a^{\theta}&=&-\frac{\hbar\sin(\theta)}{4m_{esp}r^{7}\left(\frac{\kappa-2mr+r^{2}}{r^{2}}\right)^{3/2}}\Biggl\{ -2\csc(\theta)r'(s)\left[\cos(\alpha)\left(2m^{2}r^{2}-mr^{3}-3m\kappa r\right.\right.\Biggr.
\nonumber
\\
&+&\left.\left.+\kappa^{2}-\kappa^{2}r^{4}+\kappa r^{2}\right)+\kappa r^{3}\sin(\alpha)\sin(\beta)(m-r)\right]
\nonumber
\\
&-&2r\left(-\kappa+2mr-r^{2}\right)\left[\sin(\alpha)\cos(\beta)(2mr-\kappa)\sqrt{\frac{\kappa-2mr+r^{2}}{r^{2}}}\varphi'(s)\right.
\nonumber
\\
&-&\Biggl.\Biggl.\kappa\csc(\theta)t'(s)\left(\sin(\alpha)\sin(\beta)(r-m)+\kappa r\cos(\alpha)\right)\Biggr]\Biggr\},
\end{eqnarray}

\begin{eqnarray}
a^{\varphi}&=&-\frac{\hbar\csc(\theta)}{4m_{esp}r^{7}\left(\frac{\kappa-2mr+r^{2}}{r^{2}}\right)^{3/2}}\Biggl\{ 2r'(s)\left[\sin(\alpha)\sin(\beta)\left(2m^{2}r^{2}-mr^{3}\right.\right.\Biggr.
\nonumber
\\
&-&\left.\left.3m\kappa r+\kappa^{2}-\kappa^{2}r^{4}+\kappa r^{2}\right)-\kappa r^{3}\cos(\alpha)(m-r)\right]
\nonumber
\\
&+&2r\left(\kappa-2mr+r^{2}\right)\left[\sin(\alpha)\cos(\beta)(\kappa-2mr)\sqrt{\frac{\kappa-2mr+r^{2}}{r^{2}}}\theta'(s)\right.
\nonumber
\\
&+&\Biggl.\Biggl.\kappa t'(s)\left(\cos(\alpha)(m-r)+\kappa r\sin(\alpha)\sin(\beta)\right)\Biggr]\Biggr\}
\end{eqnarray}

\section{Acceleration at low $\kappa$}
\label{ap:2}

Here we display the acceleration components at first order of the dimensionless parameter $\xi=\kappa/m^{2}$, as indicated in the Subsection~\ref{3.1.2}.

\begin{eqnarray}
a^{t}=-\frac{\xi m^{2}\hbar}{2\left(m_{esp}r(r-2m)\sqrt{1-\frac{2m}{r}}\right)}\left[\sin(\alpha)\cos(\beta)\sqrt{1-\frac{2m}{r}}r'(s)\right.
\nonumber
\\
+\Biggl.\left(m-r\right)\left(\sin(\alpha)\sin(\beta)\theta'(s)+\cos(\alpha)\sin(\theta)\varphi'(s)\right)\Biggr]+\mathcal{O}\left(\xi^{2}\right),
\end{eqnarray}

\begin{eqnarray}
a^{r}&=&\frac{m\hbar\sqrt{1-\frac{2m}{r}}}{2m_{esp}r^{2}}\left(\cos(\alpha)\theta'(s)-\sin(\alpha)\sin(\beta)\sin(\theta)\varphi'(s)\right)
\nonumber
\\
&-&\frac{\xi m^{2}\hbar}{4\left(m_{esp}r^{4}\sqrt{1-\frac{2m}{r}}\right)}\Biggl[\theta'(s)\left(2r^{2}\sin(\alpha)\sin(\beta)(m-r)+\cos(\alpha)(2r-5m)\right)\Biggr.
\nonumber
\\
&+&\sin(\theta)\varphi'(s)\left(2r^{2}\cos(\alpha)(m-r)+\sin(\alpha)\sin(\beta)(5m-2r)\right)
\nonumber
\\
&+&\left.2r\sin(\alpha)\cos(\beta)\sqrt{1-\frac{2m}{r}}(r-2m)t'(s)\right]+\mathcal{O}\left(\xi^{2}\right),
\end{eqnarray}

\begin{eqnarray}
a^{\theta}&=&-\frac{m\hbar}{2m_{esp}r^{4}}\left(\frac{\cos(\alpha)r'(s)}{\sqrt{1-\frac{2m}{r}}}+2r\sin(\alpha)\cos(\beta)\sin(\theta)\varphi'(s)\right)
\nonumber
\\
&+&\frac{m^{2}\hbar\xi}{4m_{esp}r^{5}(r-2m)\sqrt{1-\frac{2m}{r}}}\Biggl[r'(s)\left(2r^{2}\sin(\alpha)\sin(\beta)(m-r)+\cos(\alpha)(2r-3m)\right)\Biggr.
\nonumber
\\
&+&\left.r\sin(\alpha)(r-2m)\left(2\cos(\beta)\sin(\theta)\sqrt{1-\frac{2m}{r}}\varphi'(s)-2\sin(\beta)(m-r)t'(s)\right)\right]
\nonumber
\\
&+&\mathcal{O}\left(\xi^{2}\right),
\end{eqnarray}

\begin{eqnarray}
a^{\varphi}&=&\frac{m\hbar\sin(\alpha)\csc(\theta)}{2m_{esp}r^{4}}\left(\frac{\sin(\beta)r'(s)}{\sqrt{1-\frac{2m}{r}}}+2r\cos(\beta)\theta'(s)\right)
\nonumber
\\
&+&\frac{m^{2}\hbar\xi\csc(\theta)}{4m_{esp}r^{5}\sqrt{1-\frac{2m}{r}}(r-2m)}\Biggl[r'(s)\left(2r^{2}\cos(\alpha)(m-r)+\sin(\alpha)\sin(\beta)(3m-2r)\right)\Biggr.
\nonumber
\\
&+&\left.r(r-2m)\left(-2\sin(\alpha)\cos(\beta)\sqrt{1-\frac{2m}{r}}\theta'(s)-2\cos(\alpha)(m-r)t'(s)\right)\right]
\nonumber
\\
&+&\mathcal{O}\left(\xi^{2}\right).
\end{eqnarray}

\chapter{Components of the infinite derivative action}

\label{ap:3}

In this Appendix we give the different terms that appear in the linearised action \eqref{lagrangian}.

\begin{eqnarray}
\tilde{R}\tilde{F}_{1}\left(\Box\right)\tilde{R}&=&\tilde{F}_{1}\left(\Box\right)\left[h\Box^{2}h+h^{\rho\sigma}\partial_{\rho}\partial_{\sigma}\partial_{\mu}\partial_{\nu}h^{\mu\nu}-2h\Box\partial_{\mu}\partial_{\nu}h^{\mu\nu}-4h^{\mu\nu}\partial_{\mu}\partial_{\nu}\partial_{\rho}K_{\,\,\,\,\,\sigma}^{\rho\sigma}\right.
\nonumber
\\
&+&\left. 4h\Box\partial_{\rho}K_{\,\,\,\,\,\sigma}^{\rho\sigma}-4K_{\,\,\,\,\,\sigma}^{\rho\sigma}\partial_{\rho}\partial_{\mu}K_{\,\,\,\,\,\nu}^{\mu\nu}\right],
\end{eqnarray}

\begin{eqnarray}
\tilde{R}\tilde{F}_{2}\left(\Box\right)\partial_{\mu}\partial_{\nu}\tilde{R}^{\mu\nu}&=&\tilde{F}_{2}\left(\Box\right)\left[\frac{1}{2}h^{\rho\sigma}\Box\partial_{\rho}\partial_{\sigma}\partial_{\mu}\partial_{\nu}h^{\mu\nu}-h\Box^{2}\partial_{\mu}\partial_{\nu}h^{\mu\nu}+\frac{1}{2}h\Box^{3}h\right.
\nonumber
\\
&-&\left.h^{\mu\nu}\Box\partial_{\mu}\partial_{\nu}\partial_{\rho}K_{\,\,\,\,\,\sigma}^{\rho\sigma}-2K_{\,\,\,\,\,\sigma}^{\rho\sigma}\Box\partial_{\rho}\partial_{\nu}K_{\,\,\,\,\,\lambda}^{\nu\lambda}\right],
\end{eqnarray}

\begin{eqnarray}
&&\tilde{R}_{\mu\nu}\tilde{F}_{3}\left(\Box\right)\tilde{R}^{\left(\mu\nu\right)}=\tilde{F}_{3}\left(\Box\right)\left[\frac{1}{4}h\Box^{2}h+\frac{1}{4}h_{\mu\nu}\Box^{2}h^{\mu\nu}-\frac{1}{2}h_{\mu}^{\sigma}\partial_{\sigma}\partial_{\nu}h^{\mu\nu}-\frac{1}{2}h\Box\partial_{\mu}\partial_{\nu}h^{\mu\nu}\right.
\nonumber
\\
&&+\frac{1}{2}h^{\mu\nu}\partial_{\sigma}\partial_{\mu}\partial_{\nu}\partial_{\rho}h^{\rho\sigma}-\frac{1}{2}h_{\mu}^{\sigma}\partial_{\sigma}\partial_{\nu}\partial_{\rho}K^{\rho\nu\mu}-\frac{1}{2}h^{\nu\sigma}\partial_{\sigma}\partial_{\nu}\partial_{\mu}K_{\,\,\,\,\,\rho}^{\mu\rho}-\frac{1}{2}h_{\mu}^{\sigma}\partial_{\sigma}\Box K_{\,\,\,\,\,\rho}^{\mu\rho}
\nonumber
\\
&&+\frac{1}{2}h_{\mu\nu}\Box\partial_{\rho}K^{\rho\mu\nu}-K_{\,\,\mu\nu}^{\rho}\partial_{\rho}\partial_{\sigma}K^{\sigma\left(\mu\nu\right)}-K_{\,\,\mu\nu}^{\rho}\partial_{\rho}\partial^{\mu}K_{\,\,\,\,\,\lambda}^{\nu\lambda}-\frac{1}{2}K_{\,\,\,\,\,\lambda}^{\nu\lambda}\Box K_{\nu\,\,\,\rho}^{\,\,\rho}
\nonumber
\\
&&-\left.\frac{1}{2}K_{\,\,\,\,\,\rho}^{\mu\rho}\partial_{\mu}\partial_{\nu}K_{\,\,\,\,\,\lambda}^{\nu\lambda}\right],
\end{eqnarray}

\begin{eqnarray}
&&\tilde{R}_{\mu\nu}\tilde{F}_{4}\left(\Box\right)\tilde{R}^{\left[\mu\nu\right]}=\tilde{F}_{4}\left(\Box\right)\left[-K_{\,\,\mu\nu}^{\rho}\partial_{\rho}\partial_{\sigma}K^{\sigma\left[\mu\nu\right]}-K_{\,\,\mu\nu}^{\rho}\partial_{\rho}\partial^{\mu}K_{\,\,\,\,\,\lambda}^{\nu\lambda}-\frac{1}{2}K_{\,\,\,\,\,\lambda}^{\nu\lambda}\Box K_{\nu\,\,\,\rho}^{\,\,\rho}\right.
\nonumber
\\
&&+\left.\frac{1}{2}K_{\,\,\,\,\,\rho}^{\mu\rho}\partial_{\mu}\partial_{\nu}K_{\,\,\,\,\,\lambda}^{\nu\lambda}\right],
\end{eqnarray}

\begin{eqnarray}
&&\tilde{R}_{\left(\mu\right.}^{\,\,\,\left.\nu\right)}\tilde{F}_{5}\left(\Box\right)\partial_{\nu}\partial_{\lambda}\tilde{R}^{\mu\lambda}=\tilde{F}_{5}\left(\Box\right)\left[\frac{1}{4}h\Box^{3}h-\frac{1}{2}h\Box^{2}\partial_{\mu}\partial_{\nu}h^{\mu\nu}+\frac{1}{4}h^{\lambda\sigma}\Box\partial_{\sigma}\partial_{\lambda}\partial_{\mu}\partial_{\nu}h^{\mu\nu}\right.
\nonumber
\\
&&-\Bigl.h^{\nu\sigma}\Box\partial_{\sigma}\partial_{\nu}\partial_{\mu}K_{\,\,\,\,\,\rho}^{\mu\rho}+h\Box^{2}\partial_{\nu}K_{\,\,\,\,\,\lambda}^{\nu\lambda}-K_{\,\,\,\,\,\rho}^{\mu\rho}\Box\partial_{\mu}\partial_{\nu}K_{\,\,\,\,\,\lambda}^{\nu\lambda}\Bigr],
\end{eqnarray}

\begin{equation}
\tilde{R}_{\left[\mu\right.}^{\,\,\,\left.\nu\right]}\tilde{F}_{6}\left(\Box\right)\partial_{\nu}\partial_{\lambda}\tilde{R}^{\mu\lambda}=0,
\end{equation}

\begin{eqnarray}
&&\tilde{R}_{\mu}^{\,\,\,\nu}\tilde{F}_{7}\left(\Box\right)\partial_{\nu}\partial_{\lambda}\tilde{R}^{\left(\mu\lambda\right)}=\tilde{F}_{7}\left(\Box\right)\left[\frac{1}{4}h\Box^{3}h-\frac{1}{2}h\Box^{2}\partial_{\mu}\partial_{\nu}h^{\mu\nu}+\frac{1}{4}h^{\lambda\sigma}\Box\partial_{\sigma}\partial_{\lambda}\partial_{\mu}\partial_{\nu}h^{\mu\nu}\right.
\nonumber
\\
&&-\Bigl.h^{\nu\sigma}\Box\partial_{\sigma}\partial_{\nu}\partial_{\mu}K_{\,\,\,\,\,\rho}^{\mu\rho}+h\Box^{2}\partial_{\nu}K_{\,\,\,\,\,\lambda}^{\nu\lambda}-K_{\,\,\,\,\,\rho}^{\mu\rho}\Box\partial_{\mu}\partial_{\nu}K_{\,\,\,\,\,\lambda}^{\nu\lambda}\Bigr],
\end{eqnarray}

\begin{equation}
\tilde{R}_{\mu}^{\,\,\,\nu}\tilde{F}_{8}\left(\Box\right)\partial_{\nu}\partial_{\lambda}\tilde{R}^{\left[\mu\lambda\right]}=0,
\end{equation}

\begin{eqnarray}
&&\tilde{R}^{\lambda\sigma}\tilde{F}_{9}\left(\Box\right)\partial_{\mu}\partial_{\sigma}\partial_{\nu}\partial_{\lambda}\tilde{R}^{\mu\nu}=\tilde{F}_{9}\left(\Box\right)\left[\frac{1}{4}h\Box^{4}h-\frac{1}{2}h\Box^{3}\partial_{\mu}\partial_{\nu}h^{\mu\nu}\right.
\nonumber
\\
&&+\frac{1}{4}h^{\lambda\sigma}\Box^{2}\partial_{\sigma}\partial_{\lambda}\partial_{\mu}\partial_{\nu}h^{\mu\nu}-h^{\nu\sigma}\Box^{2}\partial_{\sigma}\partial_{\nu}\partial_{\mu}K_{\,\,\,\,\,\rho}^{\mu\rho}+h\Box^{3}\partial_{\nu}K_{\,\,\,\,\,\lambda}^{\nu\lambda}
\nonumber
\\
&&-\left.K_{\,\,\,\,\,\rho}^{\mu\rho}\Box^{2}\partial_{\mu}\partial_{\nu}K_{\,\,\,\,\,\lambda}^{\nu\lambda}\right],
\end{eqnarray}

\begin{eqnarray}
&&\tilde{R}_{\left(\mu\lambda\right)}\tilde{F}_{10}\left(\Box\right)\partial_{\nu}\partial_{\sigma}\tilde{R}^{\mu\nu\lambda\sigma}=\tilde{F}_{10}\left(\Box\right)\left[\frac{1}{4}h_{\mu\lambda}\Box^{3}h^{\mu\lambda}-\frac{1}{2}h_{\mu}^{\,\,\alpha}\Box^{2}\partial_{\alpha}\partial_{\sigma}h^{\sigma\mu}\right.
\nonumber
\\
&&+\frac{1}{4}h^{\lambda\sigma}\Box\partial_{\sigma}\partial_{\lambda}\partial_{\mu}\partial_{\nu}h^{\mu\nu}-h_{\mu\sigma}\Box^{2}\partial_{\lambda}K^{\mu\sigma\lambda}+h_{\mu}^{\,\,\alpha}\Box\partial_{\alpha}\partial_{\lambda}\partial_{\sigma}K^{\mu\sigma\lambda}
\nonumber
\\
&&\left.+K_{\alpha\left(\mu\lambda\right)}\Box\partial^{\alpha}\partial_{\sigma}K^{\lambda\mu\sigma}-\frac{1}{2}K_{\alpha\mu\lambda}\partial^{\alpha}\partial^{\mu}\partial_{\sigma}\partial_{\nu}K^{\lambda\nu\sigma}\right],
\end{eqnarray}

\begin{eqnarray}
&&\tilde{R}_{\left[\mu\lambda\right]}\tilde{F}_{11}\left(\Box\right)\partial_{\nu}\partial_{\sigma}\tilde{R}^{\mu\nu\lambda\sigma}=\tilde{F}_{11}\left(\Box\right)\Bigl[K_{\alpha\left[\mu\lambda\right]}\Box\partial^{\alpha}\partial_{\sigma}K^{\lambda\mu\sigma}\Bigr.
\nonumber
\\
&&-\left.\frac{1}{2}K_{\alpha\mu\lambda}\partial^{\alpha}\partial^{\mu}\partial_{\sigma}\partial_{\nu}K^{\lambda\nu\sigma}\right],
\end{eqnarray}

\begin{eqnarray}
&&\tilde{R}_{\mu\lambda}\tilde{F}_{12}\left(\Box\right)\partial_{\nu}\partial_{\sigma}\left(\tilde{R}^{\mu\nu\lambda\sigma}+\tilde{R}^{\lambda\sigma\mu\nu}\right)=\tilde{F}_{12}\left(\Box\right)\left[\frac{1}{2}h_{\mu\lambda}\Box^{3}h^{\mu\lambda}-h_{\mu}^{\,\,\alpha}\Box^{2}\partial_{\alpha}\partial_{\sigma}h^{\sigma\mu}\right.
\nonumber
\\
&&+\frac{1}{2}h^{\lambda\sigma}\Box\partial_{\sigma}\partial_{\lambda}\partial_{\mu}\partial_{\nu}h^{\mu\nu}+2h_{\mu}^{\,\,\alpha}\Box\partial_{\alpha}\partial_{\lambda}\partial_{\sigma}K^{\mu\sigma\lambda}-2h_{\mu\sigma}\Box^{2}\partial_{\lambda}K^{\mu\sigma\lambda}
\nonumber
\\
&&\Bigl.+2K_{\alpha\left(\mu\lambda\right)}\Box\partial^{\alpha}\partial_{\sigma}K^{\lambda\mu\sigma}-K_{\alpha\mu\lambda}\partial^{\alpha}\partial^{\mu}\partial_{\sigma}\partial_{\nu}K^{\lambda\nu\sigma}\Bigr],
\end{eqnarray}

\begin{eqnarray}
&&\tilde{R}_{\mu\lambda}\tilde{F}_{13}\left(\Box\right)\partial_{\nu}\partial_{\sigma}\left(\tilde{R}^{\mu\nu\lambda\sigma}-\tilde{R}^{\lambda\sigma\mu\nu}\right)=\tilde{F}_{13}\left(\Box\right)\left[2K_{\alpha\left[\mu\lambda\right]}\Box\partial^{\alpha}\partial_{\sigma}K^{\lambda\mu\sigma}\right.
\nonumber
\\
&&-\left.K_{\alpha\mu\lambda}\partial^{\alpha}\partial^{\mu}\partial_{\sigma}\partial_{\nu}K^{\lambda\nu\sigma}\right],
\end{eqnarray}

\begin{eqnarray}
&&\tilde{R}_{\mu\nu\lambda\sigma}\tilde{F}_{14}\left(\Box\right)\left(\tilde{R}^{\mu\nu\lambda\sigma}+\tilde{R}^{\lambda\sigma\mu\nu}\right)=\tilde{F}_{14}\left(\Box\right)\left[2h_{\mu\lambda}\Box^{2}h^{\mu\lambda}+2h^{\lambda\sigma}\partial_{\sigma}\partial_{\lambda}\partial_{\mu}\partial_{\nu}h^{\mu\nu}\right.
\nonumber
\\
&&-4h_{\mu}^{\,\,\alpha}\Box\partial_{\alpha}\partial_{\sigma}h^{\sigma\mu}+8h_{\sigma\mu}\Box\partial_{\nu}K^{\nu\mu\sigma}+8h_{\sigma\mu}\partial_{\nu}\partial_{\lambda}\partial^{\mu}K^{\sigma\lambda\nu}-2K^{\mu\sigma\lambda}\Box K_{\mu\sigma\lambda}
\nonumber
\\
&&\left.-4K^{\nu\sigma\lambda}\partial_{\nu}\partial^{\mu}K_{\mu\lambda\sigma}+2K^{\lambda\nu\mu}\partial_{\nu}\partial^{\sigma}K_{\lambda\sigma\mu}\right],
\end{eqnarray}

\begin{eqnarray}
&&\tilde{R}_{\mu\nu\lambda\sigma}\tilde{F}_{15}\left(\Box\right)\left(\tilde{R}^{\mu\nu\lambda\sigma}-\tilde{R}^{\lambda\sigma\mu\nu}\right)=\tilde{F}_{15}\left(\Box\right)\left[-2K^{\mu\sigma\lambda}\Box K_{\mu\sigma\lambda}+4K^{\nu\sigma\lambda}\partial_{\nu}\partial^{\mu}K_{\mu\lambda\sigma}\right.
\nonumber
\\
&&\left.+2K^{\lambda\nu\mu}\partial_{\nu}\partial^{\sigma}K_{\lambda\sigma\mu}\right],
\end{eqnarray}

\begin{eqnarray}
&&\left(\tilde{R}_{\rho\mu\nu\lambda}+\tilde{R}_{\nu\lambda\rho\mu}\right)\tilde{F}_{16}\left(\Box\right)\partial^{\rho}\partial_{\sigma}\tilde{R}^{\mu\nu\lambda\sigma}=\tilde{F}_{16}\left(\Box\right)\left[\frac{1}{2}h_{\mu\lambda}\Box^{3}h^{\mu\lambda}-h_{\mu}^{\,\,\alpha}\Box^{2}\partial_{\alpha}\partial_{\sigma}h^{\sigma\mu}\right.
\nonumber
\\
&&+\frac{1}{2}h^{\lambda\sigma}\Box\partial_{\sigma}\partial_{\lambda}\partial_{\mu}\partial_{\nu}h^{\mu\nu}+2h_{\sigma\mu}\Box^{2}\partial_{\nu}K^{\nu\mu\sigma}+2h_{\sigma\mu}\Box\partial_{\nu}\partial_{\lambda}\partial^{\mu}K^{\sigma\lambda\nu}
\nonumber
\\
&&\left.+2K_{\alpha\left(\mu\lambda\right)}\Box\partial^{\alpha}\partial_{\sigma}K^{\lambda\mu\sigma}-K_{\alpha\mu\lambda}\partial^{\alpha}\partial^{\mu}\partial_{\sigma}\partial_{\nu}K^{\lambda\nu\sigma}\right],
\end{eqnarray}

\begin{eqnarray}
&&\left(\tilde{R}_{\rho\mu\nu\lambda}-\tilde{R}_{\nu\lambda\rho\mu}\right)\tilde{F}_{17}\left(\Box\right)\partial^{\rho}\partial_{\sigma}\tilde{R}^{\mu\nu\lambda\sigma}=\tilde{F}_{17}\left(\Box\right)\left[-2K^{\mu\sigma\lambda}\Box\partial^{\rho}\partial_{\sigma}K_{\lambda\mu\rho}\right.
\nonumber
\\
&&-\left.2K^{\nu\sigma\lambda}\partial^{\mu}\partial^{\rho}\partial_{\sigma}\partial_{\lambda}K_{\nu\mu\rho}\right],
\end{eqnarray}

\begin{eqnarray}
&&\tilde{R}_{\rho\mu\nu\lambda}\tilde{F}_{18}\left(\Box\right)\partial^{\rho}\partial_{\sigma}\left(\tilde{R}^{\mu\nu\lambda\sigma}+\tilde{R}^{\lambda\sigma\mu\nu}\right)=\tilde{F}_{18}\left(\Box\right)\left[\frac{1}{2}h_{\mu\lambda}\Box^{3}h^{\mu\lambda}-h_{\mu}^{\,\,\alpha}\Box^{2}\partial_{\alpha}\partial_{\sigma}h^{\sigma\mu}\right.
\nonumber
\\
&&+\frac{1}{2}h^{\lambda\sigma}\Box\partial_{\sigma}\partial_{\lambda}\partial_{\mu}\partial_{\nu}h^{\mu\nu}+2h_{\sigma\mu}\Box^{2}\partial_{\nu}K^{\nu\mu\sigma}+2h_{\sigma\mu}\Box\partial_{\nu}\partial_{\lambda}\partial^{\mu}K^{\sigma\lambda\nu}
\nonumber
\\
&&\Bigl.+2K_{\alpha\mu\lambda}\Box\partial^{\alpha}\partial_{\sigma}K^{\lambda\mu\sigma}-2K_{\left[\nu\mu\right]\lambda}\Box\partial_{\sigma}\partial^{\lambda}K^{\mu\sigma\nu}+K^{\mu\sigma\lambda}\Box^{2}K_{\sigma\lambda\mu}\Bigr],
\end{eqnarray}

\begin{eqnarray}
&&\tilde{R}_{\rho\mu\nu\lambda}\tilde{F}_{19}\left(\Box\right)\partial^{\rho}\partial_{\sigma}\left(\tilde{R}^{\mu\nu\lambda\sigma}-\tilde{R}^{\lambda\sigma\mu\nu}\right)=\tilde{F}_{19}\left(\Box\right)\left[-2K_{\alpha\mu\lambda}\partial^{\alpha}\partial^{\mu}\partial_{\sigma}\partial_{\nu}K^{\lambda\nu\sigma}\right.
\nonumber
\\
&&\left.+2K_{\left[\nu\mu\right]\lambda}\Box\partial_{\sigma}\partial^{\lambda}K^{\mu\sigma\nu}+K_{\alpha\mu\lambda}\Box\partial^{\alpha}\partial_{\sigma}K^{\lambda\mu\sigma}-K^{\mu\sigma\lambda}\Box^{2}K_{\sigma\lambda\mu}\right],
\end{eqnarray}

\begin{eqnarray}
&&\left(\tilde{R}_{\mu\nu\rho\sigma}+\tilde{R}_{\rho\sigma\mu\nu}\right)\tilde{F}_{20}\left(\Box\right)\partial^{\nu}\partial^{\sigma}\partial_{\alpha}\partial_{\beta}\tilde{R}^{\mu\alpha\rho\beta}=\tilde{F}_{20}\left(\Box\right)\left[\frac{1}{2}h_{\mu\lambda}\Box^{4}h^{\mu\lambda}\right.
\nonumber
\\
&&-h_{\mu}^{\,\,\alpha}\Box^{3}\partial_{\alpha}\partial_{\sigma}h^{\sigma\mu}+\frac{1}{2}h^{\lambda\sigma}\Box^{2}\partial_{\sigma}\partial_{\lambda}\partial_{\mu}\partial_{\nu}h^{\mu\nu}+2h_{\mu}^{\,\,\alpha}\Box^{2}\partial_{\alpha}\partial_{\lambda}\partial_{\sigma}K^{\mu\sigma\lambda}
\nonumber
\\
&&\left.-2h_{\mu\lambda}\Box^{3}\partial_{\sigma}K^{\mu\lambda\sigma}+2K_{\alpha\left(\mu\lambda\right)}\Box\partial^{\alpha}\partial_{\sigma}K^{\lambda\mu\sigma}-K_{\alpha\mu\lambda}\partial^{\alpha}\partial^{\mu}\partial_{\sigma}\partial_{\nu}K^{\lambda\nu\sigma}\right],
\end{eqnarray}

\begin{eqnarray}
&&\left(\tilde{R}_{\mu\nu\rho\sigma}-\tilde{R}_{\rho\sigma\mu\nu}\right)\tilde{F}_{21}\left(\Box\right)\partial^{\nu}\partial^{\sigma}\partial_{\alpha}\partial_{\beta}\tilde{R}^{\mu\alpha\rho\beta}=\tilde{F}_{21}\left(\Box\right)\left[2K_{\alpha\left[\mu\lambda\right]}\Box\partial^{\alpha}\partial_{\sigma}K^{\mu\lambda\sigma}\right.
\nonumber
\\
&&-\left.K_{\alpha\mu\lambda}\partial^{\alpha}\partial^{\mu}\partial_{\sigma}\partial_{\nu}K^{\lambda\nu\sigma}\right],
\end{eqnarray}

\begin{eqnarray}
&&\tilde{R}_{\mu\nu\rho\sigma}\tilde{F}_{22}\left(\Box\right)\partial^{\nu}\partial^{\sigma}\partial_{\alpha}\partial_{\beta}\left(\tilde{R}^{\mu\alpha\rho\beta}+\tilde{R}^{\rho\beta\mu\alpha}\right)=\tilde{F}_{22}\left(\Box\right)\left[\frac{1}{2}h_{\mu\lambda}\Box^{4}h^{\mu\lambda}\right.
\nonumber
\\
&&-h_{\mu}^{\,\,\alpha}\Box^{3}\partial_{\alpha}\partial_{\sigma}h^{\sigma\mu}+\frac{1}{2}h^{\lambda\sigma}\Box^{2}\partial_{\sigma}\partial_{\lambda}\partial_{\mu}\partial_{\nu}h^{\mu\nu}+2h_{\mu}^{\,\,\alpha}\Box^{2}\partial_{\alpha}\partial_{\lambda}\partial_{\sigma}K^{\mu\sigma\lambda}
\nonumber
\\
&&\left.-2h_{\mu\lambda}\Box^{3}\partial_{\sigma}K^{\mu\lambda\sigma}+2K_{\alpha\left(\mu\lambda\right)}\Box\partial^{\alpha}\partial_{\sigma}K^{\lambda\mu\sigma}-K_{\alpha\mu\lambda}\partial^{\alpha}\partial^{\mu}\partial_{\sigma}\partial_{\nu}K^{\lambda\nu\sigma}\right],
\end{eqnarray}

\begin{eqnarray}
&&\tilde{R}_{\mu\nu\rho\sigma}\tilde{F}_{23}\left(\Box\right)\partial^{\nu}\partial^{\sigma}\partial_{\alpha}\partial_{\beta}\left(\tilde{R}^{\mu\alpha\rho\beta}-\tilde{R}^{\rho\beta\mu\alpha}\right)=\tilde{F}_{23}\left(\Box\right)\left[-2K_{\alpha\left[\mu\lambda\right]}\Box\partial^{\alpha}\partial_{\sigma}K^{\mu\lambda\sigma}\right.
\nonumber
\\
&&\left.-K_{\alpha\mu\lambda}\partial^{\alpha}\partial^{\mu}\partial_{\sigma}\partial_{\nu}K^{\lambda\nu\sigma}\right],
\end{eqnarray}

\begin{eqnarray}
&&\tilde{R}_{\mu\nu\rho\sigma}\tilde{F}_{34}\left(\Box\right)\partial^{\mu}K^{\nu\rho\sigma}=\tilde{F}_{34}\left(\Box\right)\left[-\frac{1}{2}h_{\sigma\mu}\partial_{\nu}\partial_{\lambda}\partial^{\mu}K^{\sigma\lambda\nu}-\frac{1}{2}h_{\sigma\mu}\Box\partial_{\nu}K^{\nu\mu\sigma}\right.
\nonumber
\\
&&+\biggl.K_{\nu\mu\lambda}\partial_{\sigma}\partial^{\lambda}K^{\mu\sigma\nu}-K^{\mu\sigma\lambda}\Box K_{\sigma\mu\lambda}\biggr],
\end{eqnarray}

\begin{eqnarray}
&&\tilde{R}_{\mu\nu\rho\sigma}\tilde{F}_{35}\left(\Box\right)\partial^{\rho}K^{\mu\nu\sigma}=\tilde{F}_{35}\left(\Box\right)\left[-\frac{1}{2}h_{\sigma\mu}\partial_{\nu}\partial_{\lambda}\partial^{\mu}K^{\sigma\lambda\nu}-\frac{1}{2}h_{\sigma\mu}\Box\partial_{\nu}K^{\nu\mu\sigma}\right.
\nonumber
\\
&&\biggl.-K_{\nu\mu\lambda}\partial_{\sigma}\partial^{\lambda}K^{\mu\sigma\nu}+K_{\lambda\mu\nu}\partial_{\sigma}\partial^{\lambda}K^{\sigma\mu\nu}\biggr],
\end{eqnarray}

\begin{eqnarray}
&&\tilde{R}_{\left(\rho\sigma\right)}\tilde{F}_{36}\left(\Box\right)\partial_{\nu}K^{\mu\nu\sigma}=\tilde{F}_{36}\left(\Box\right)\left[-\frac{1}{2}h_{\sigma\mu}\partial_{\nu}\partial_{\lambda}\partial^{\mu}K^{\sigma\lambda\nu}-\frac{1}{2}h_{\sigma\mu}\Box\partial_{\nu}K^{\nu\mu\sigma}\right.
\nonumber
\\
&&\left.+K_{\lambda\left(\mu\nu\right)}\partial_{\sigma}\partial^{\lambda}K^{\sigma\mu\nu}-\frac{1}{2}K_{\,\,\,\lambda\sigma}^{\lambda}\partial_{\rho}\partial_{\nu}K^{\nu\rho\sigma}\right],
\end{eqnarray}

\begin{eqnarray}
\tilde{R}_{\left[\rho\sigma\right]}\tilde{F}_{37}\left(\Box\right)\partial_{\nu}K^{\nu\rho\sigma}=\tilde{F}_{37}\left(\Box\right)\left[K_{\lambda\left[\mu\nu\right]}\partial_{\sigma}\partial^{\lambda}K^{\sigma\mu\nu}-\frac{1}{2}K_{\,\,\,\lambda\sigma}^{\lambda}\partial_{\rho}\partial_{\nu}K^{\nu\rho\sigma}\right],
\end{eqnarray}

\begin{eqnarray}
\tilde{R}_{\rho\sigma}\tilde{F}_{38}\left(\Box\right)\partial_{\nu}K^{\rho\nu\sigma}=\tilde{F}_{38}\left(\Box\right)\left[-K_{\nu\mu\lambda}\partial_{\sigma}\partial^{\lambda}K^{\mu\sigma\nu}-K_{\,\,\,\lambda\sigma}^{\lambda}\partial_{\rho}\partial_{\nu}K^{\nu\rho\sigma}\right],
\end{eqnarray}

\begin{eqnarray}
&&\tilde{R}_{\left(\rho\sigma\right)}\tilde{F}_{39}\left(\Box\right)\partial^{\sigma}K_{\,\,\,\,\,\,\mu}^{\rho\mu}=\tilde{F}_{39}\left(\Box\right)\left[\frac{1}{2}h_{\sigma\lambda}\partial^{\sigma}\partial^{\lambda}\partial_{\rho}K_{\,\,\,\,\,\,\mu}^{\rho\mu}-\frac{1}{2}h\Box\partial_{\rho}K_{\,\,\,\,\,\,\mu}^{\rho\mu}\right.
\nonumber
\\
&&\left.+\frac{1}{2}K_{\nu\mu\rho}\partial^{\nu}\partial^{\mu}K_{\,\,\,\,\,\,\mu}^{\rho\mu}-\frac{1}{2}K_{\,\,\,\lambda\sigma}^{\lambda}\partial^{\sigma}\partial_{\rho}K_{\,\,\,\,\,\,\mu}^{\rho\mu}-\frac{1}{2}K_{\,\,\,\lambda\rho}^{\lambda}\Box K_{\,\,\,\,\,\,\mu}^{\rho\mu}\right],
\end{eqnarray}

\begin{eqnarray}
&&\tilde{R}_{\left[\rho\sigma\right]}\tilde{F}_{40}\left(\Box\right)\partial^{\sigma}K_{\,\,\,\,\,\,\mu}^{\rho\mu}=\tilde{F}_{40}\left(\Box\right)\left[-\frac{1}{2}K_{\nu\mu\rho}\partial^{\nu}\partial^{\mu}K_{\,\,\,\,\,\,\mu}^{\rho\mu}-\frac{1}{2}K_{\,\,\,\lambda\sigma}^{\lambda}\partial^{\sigma}\partial_{\rho}K_{\,\,\,\,\,\,\mu}^{\rho\mu}\right.
\nonumber
\\
&&\left.+\frac{1}{2}K_{\,\,\,\lambda\rho}^{\lambda}\Box K_{\,\,\,\,\,\,\mu}^{\rho\mu}\right],
\end{eqnarray}

\begin{eqnarray}
&&\tilde{R}\tilde{F}_{41}\left(\Box\right)\partial_{\rho}K_{\,\,\,\,\,\,\mu}^{\rho\mu}=\tilde{F}_{41}\left(\Box\right)\left[h_{\sigma\lambda}\partial^{\sigma}\partial^{\lambda}\partial_{\rho}K_{\,\,\,\,\,\,\mu}^{\rho\mu}-h\Box\partial_{\rho}K_{\,\,\,\,\,\,\mu}^{\rho\mu}\right.
\nonumber
\\
&&\left.-2K_{\,\,\,\lambda\sigma}^{\lambda}\partial^{\sigma}\partial_{\rho}K_{\,\,\,\,\,\,\mu}^{\rho\mu}\right],
\end{eqnarray}

\begin{eqnarray}
&&\tilde{R}_{\mu\alpha\rho\sigma}\tilde{F}_{42}\left(\Box\right)\partial^{\mu}\partial^{\rho}\partial_{\nu}K^{\nu\left(\alpha\sigma\right)}=\tilde{F}_{42}\left(\Box\right)\left[-\frac{1}{2}h_{\sigma\mu}\Box\partial_{\nu}\partial_{\lambda}\partial^{\mu}K^{\sigma\lambda\nu}\right.
\nonumber
\\
&&\left.-\frac{1}{2}h_{\sigma\mu}\Box^{2}\partial_{\nu}K^{\nu\mu\sigma}-\frac{1}{2}K_{\nu\mu\lambda}\partial_{\alpha}\partial_{\rho}\partial^{\nu}\partial^{\mu}K^{\alpha\rho\lambda}+K_{\lambda\left(\mu\nu\right)}\Box\partial_{\sigma}\partial^{\lambda}K^{\sigma\mu\nu}\right],
\end{eqnarray}

\begin{eqnarray}
&&\tilde{R}_{\mu\alpha\rho\sigma}\tilde{F}_{43}\left(\Box\right)\partial^{\mu}\partial^{\rho}\partial_{\nu}K^{\nu\left[\alpha\sigma\right]}=\tilde{F}_{43}\left(\Box\right)\left[-\frac{1}{2}K_{\nu\mu\lambda}\partial_{\alpha}\partial_{\rho}\partial^{\nu}\partial^{\mu}K^{\alpha\rho\lambda}\right.
\nonumber
\\
&&\biggl.+K_{\lambda\left[\mu\nu\right]}\Box\partial_{\sigma}\partial^{\lambda}K^{\sigma\mu\nu}\biggr],
\end{eqnarray}

\begin{eqnarray}
&&\tilde{R}_{\mu\alpha\rho\sigma}\tilde{F}_{44}\left(\Box\right)\partial^{\mu}\partial^{\rho}\partial_{\nu}K^{\alpha\nu\sigma}=\tilde{F}_{44}\left(\Box\right)\left[-K_{\nu\mu\lambda}\partial_{\alpha}\partial_{\rho}\partial^{\nu}\partial^{\mu}K^{\alpha\rho\lambda}\right.
\nonumber
\\
&&\left.+K_{\lambda\mu\nu}\Box\partial_{\sigma}\partial^{\lambda}K^{\mu\sigma\nu}\right],
\end{eqnarray}

\begin{eqnarray}
&&\tilde{R}_{\left(\rho\sigma\right)}\tilde{F}_{45}\left(\Box\right)\partial^{\rho}\partial_{\nu}\partial_{\alpha}K^{\sigma\nu\alpha}=\tilde{F}_{45}\left(\Box\right)\left[-\frac{1}{2}K_{\nu\mu\lambda}\partial_{\alpha}\partial_{\rho}\partial^{\nu}\partial^{\mu}K^{\alpha\rho\lambda}\right.
\nonumber
\\
&&\left.-\frac{1}{2}K_{\,\,\,\lambda\sigma}^{\lambda}\Box\partial_{\mu}\partial_{\alpha}K^{\sigma\mu\alpha}\right],
\end{eqnarray}

\begin{eqnarray}
&&\tilde{R}_{\left(\rho\sigma\right)}\tilde{F}_{46}\left(\Box\right)\partial^{\rho}\partial_{\nu}\partial_{\alpha}K^{\sigma\nu\alpha}=\tilde{F}_{46}\left(\Box\right)\left[-\frac{1}{2}K_{\nu\mu\lambda}\partial_{\alpha}\partial_{\rho}\partial^{\nu}\partial^{\mu}K^{\alpha\rho\lambda}\right.
\nonumber
\\
&&\left.-\frac{1}{2}K_{\,\,\,\lambda\sigma}^{\lambda}\Box\partial_{\mu}\partial_{\alpha}K^{\sigma\mu\alpha}\right],
\end{eqnarray}

\begin{eqnarray}
\tilde{R}_{\mu\nu\lambda\sigma}\tilde{F}_{47}\left(\Box\right)\widetilde{R}^{\mu\lambda\nu\sigma}&=&\tilde{F}_{47}\left(\Box\right)\left[h_{\mu\lambda}\Box^{2}h^{\mu\lambda}+h^{\lambda\sigma}\partial_{\sigma}\partial_{\lambda}\partial_{\mu}\partial_{\nu}h^{\mu\nu}-2h_{\mu}^{\,\,\alpha}\Box\partial_{\alpha}\partial_{\sigma}h^{\sigma\mu}\right.
\nonumber
\\
&+&4h_{\sigma\mu}\Box\partial_{\nu}K^{\nu\mu\sigma}+4h_{\sigma\mu}\partial_{\nu}\partial_{\lambda}\partial^{\mu}K^{\sigma\lambda\nu}-K^{\mu\sigma\lambda}\Box K_{\mu\lambda\sigma}
\nonumber
\\
&-&\left.K^{\nu\lambda\sigma}\partial_{\nu}\partial^{\mu}K_{\mu\lambda\sigma}+2K^{\lambda\nu\mu}\partial_{\nu}\partial^{\sigma}K_{\lambda\mu\sigma}\right].
\end{eqnarray}

\chapter{Functions of the linearised action}
\label{ap:4}

In this Appendix one can find the explicit form of the functions that compose the linearised action \eqref{laaccion}.

\begin{eqnarray}
&&a\left(\Box\right)=1-\frac{1}{2}\tilde{F}_{3}\left(\Box\right)\Box-\frac{1}{2}\tilde{F}_{10}\left(\Box\right)\Box^{2}-\frac{1}{2}\tilde{F}_{12}\left(\Box\right)\Box^{2}-2\tilde{F}_{14}\left(\Box\right)\Box
\nonumber
\\
&&-\frac{1}{2}\tilde{F}_{16}\left(\Box\right)\Box^{2}-\frac{1}{2}\tilde{F}_{18}\left(\Box\right)\Box^{2}-\frac{1}{2}\tilde{F}_{20}\left(\Box\right)\Box^{3}-\frac{1}{2}\tilde{F}_{22}\left(\Box\right)\Box^{3}
\nonumber
\\
&&-\tilde{F}_{47}\left(\Box\right)\Box,
\end{eqnarray}

\begin{eqnarray}
&&b\left(\Box\right)=-1+\frac{1}{2}\tilde{F}_{3}\left(\Box\right)\Box+\frac{1}{2}\tilde{F}_{10}\left(\Box\right)\Box^{2}+\frac{1}{2}\tilde{F}_{12}\left(\Box\right)\Box^{2}+2\tilde{F}_{14}\left(\Box\right)\Box
\nonumber
\\
&&+\frac{1}{2}\tilde{F}_{16}\left(\Box\right)\Box^{2}+\frac{1}{2}\tilde{F}_{18}\left(\Box\right)\Box^{2}+\frac{1}{2}\tilde{F}_{20}\left(\Box\right)\Box^{3}+\frac{1}{2}\tilde{F}_{22}\left(\Box\right)\Box^{3}
\nonumber
\\
&&+\tilde{F}_{47}\left(\Box\right)\Box,
\end{eqnarray}

\begin{eqnarray}
&&c\left(\Box\right)=1+2\tilde{F}_{1}\left(\Box\right)\Box+\tilde{F}_{2}\left(\Box\right)\Box^{2}+\frac{1}{2}\tilde{F}_{3}\left(\Box\right)\Box+\frac{1}{2}\tilde{F}_{5}\left(\Box\right)\Box^{2}+\frac{1}{2}\tilde{F}_{7}\left(\Box\right)\Box^{2}
\nonumber
\\
&&+\frac{1}{2}\tilde{F}_{9}\left(\Box\right)\Box^{3},
\end{eqnarray}

\begin{eqnarray}
&&d\left(\Box\right)=-1-2\tilde{F}_{1}\left(\Box\right)\Box-\tilde{F}_{2}\left(\Box\right)\Box^{2}-\frac{1}{2}\tilde{F}_{3}\left(\Box\right)\Box-\frac{1}{2}\tilde{F}_{5}\left(\Box\right)\Box^{2}-\frac{1}{2}\tilde{F}_{7}\left(\Box\right)\Box^{2}
\nonumber
\\
&&-\frac{1}{2}\tilde{F}_{9}\left(\Box\right)\Box^{3},
\end{eqnarray}

\begin{eqnarray}
&&f\left(\Box\right)=-\tilde{F}_{1}\left(\Box\right)\Box-\frac{1}{2}\tilde{F}_{2}\left(\Box\right)\Box^{2}-\frac{1}{2}\tilde{F}_{3}\left(\Box\right)\Box-\frac{1}{4}\tilde{F}_{5}\left(\Box\right)\Box^{2}-\frac{1}{4}\tilde{F}_{7}\left(\Box\right)\Box^{2}
\nonumber
\\
&&-\frac{1}{4}\tilde{F}_{9}\left(\Box\right)\Box^{3}-\frac{1}{4}\tilde{F}_{10}\left(\Box\right)\Box^{2}-\frac{1}{4}\tilde{F}_{12}\left(\Box\right)\Box^{2}-\tilde{F}_{14}\left(\Box\right)\Box-\frac{1}{4}\tilde{F}_{16}\left(\Box\right)\Box^{2}
\nonumber
\\
&&-\frac{1}{4}\tilde{F}_{18}\left(\Box\right)\Box^{2}-\frac{1}{4}\tilde{F}_{20}\left(\Box\right)\Box^{3}-\frac{1}{4}\tilde{F}_{22}\left(\Box\right)\Box^{3}-\frac{1}{2}\tilde{F}_{47}\left(\Box\right)\Box,
\end{eqnarray}

\begin{equation}
u\left(\Box\right)=-4\tilde{F}_{1}\left(\Box\right)-\tilde{F}_{5}\left(\Box\right)\Box-\tilde{F}_{7}\left(\Box\right)\Box-\tilde{F}_{9}\left(\Box\right)\Box^{2}+\frac{1}{2}\tilde{F}_{39}\left(\Box\right)+\tilde{F}_{41}\left(\Box\right),
\end{equation}

\begin{equation}
v_{1}\left(\Box\right)=4\tilde{F}_{1}\left(\Box\right)+\tilde{F}_{5}\left(\Box\right)\Box+\tilde{F}_{7}\left(\Box\right)\Box+\tilde{F}_{9}\left(\Box\right)\Box^{2}-\frac{1}{2}\tilde{F}_{39}\left(\Box\right)-\tilde{F}_{41}\left(\Box\right),
\end{equation}

\begin{eqnarray}
&&v_{2}\left(\Box\right)=-\frac{1}{2}\tilde{F}_{3}\left(\Box\right)-\tilde{F}_{10}\left(\Box\right)\Box-\tilde{F}_{12}\left(\Box\right)\Box+\tilde{F}_{9}\left(\Box\right)\Box^{2}-4\tilde{F}_{14}\left(\Box\right)
\nonumber
\\
&&-\tilde{F}_{16}\left(\Box\right)\Box-\tilde{F}_{18}\left(\Box\right)\Box-\tilde{F}_{20}\left(\Box\right)\Box^{2}-\tilde{F}_{22}\left(\Box\right)\Box^{2}+\frac{1}{2}\tilde{F}_{34}\left(\Box\right)+\frac{1}{2}\tilde{F}_{35}\left(\Box\right)
\nonumber
\\
&&+\frac{1}{2}\tilde{F}_{36}\left(\Box\right)+\frac{1}{2}\tilde{F}_{42}\left(\Box\right)-2\tilde{F}_{47}\left(\Box\right),
\end{eqnarray}

\begin{eqnarray}
&&w\left(\Box\right)=-\frac{1}{2}\tilde{F}_{3}\left(\Box\right)-\tilde{F}_{10}\left(\Box\right)\Box-\tilde{F}_{12}\left(\Box\right)\Box+\tilde{F}_{9}\left(\Box\right)\Box^{2}-4\tilde{F}_{14}\left(\Box\right)
\nonumber
\\
&&-\tilde{F}_{16}\left(\Box\right)\Box-\tilde{F}_{18}\left(\Box\right)\Box-\tilde{F}_{20}\left(\Box\right)\Box^{2}-\tilde{F}_{22}\left(\Box\right)\Box^{2}+\frac{1}{2}\tilde{F}_{34}\left(\Box\right)+\frac{1}{2}\tilde{F}_{35}\left(\Box\right)
\nonumber
\\
&&+\frac{1}{2}\tilde{F}_{36}\left(\Box\right)+\frac{1}{2}\tilde{F}_{42}\left(\Box\right)-2\tilde{F}_{47}\left(\Box\right),
\end{eqnarray}

\begin{eqnarray}
&&q_{1}\left(\Box\right)=\frac{1}{2}\tilde{F}_{3}\left(\Box\right)+\frac{1}{2}\tilde{F}_{4}\left(\Box\right)+\frac{1}{2}\tilde{F}_{10}\left(\Box\right)\Box+\frac{1}{2}\tilde{F}_{11}\left(\Box\right)\Box+\frac{1}{2}\tilde{F}_{12}\left(\Box\right)\Box
\nonumber
\\
&&+\frac{1}{2}\tilde{F}_{13}\left(\Box\right)\Box+\frac{1}{2}\tilde{F}_{16}\left(\Box\right)\Box+\frac{1}{2}\tilde{F}_{18}\left(\Box\right)\Box+\frac{1}{2}\tilde{F}_{19}\left(\Box\right)\Box+\frac{1}{2}\tilde{F}_{20}\left(\Box\right)\Box
\nonumber
\\
&&+\frac{1}{2}\tilde{F}_{21}\left(\Box\right)\Box+\frac{1}{2}\tilde{F}_{22}\left(\Box\right)\Box+\frac{1}{2}\tilde{F}_{23}\left(\Box\right)\Box+\tilde{F}_{27}\left(\Box\right)-\frac{1}{2}\tilde{F}_{36}\left(\Box\right)-\frac{1}{2}\tilde{F}_{37}\left(\Box\right)
\nonumber
\\
&&-\frac{1}{2}\tilde{F}_{42}\left(\Box\right)\Box-\frac{1}{2}\tilde{F}_{43}\left(\Box\right)\Box-\tilde{F}_{47}\left(\Box\right),
\end{eqnarray}

\begin{eqnarray}
&&q_{2}\left(\Box\right)=\frac{1}{2}\tilde{F}_{3}\left(\Box\right)-\frac{1}{2}\tilde{F}_{4}\left(\Box\right)+\frac{1}{2}\tilde{F}_{10}\left(\Box\right)\Box-\frac{1}{2}\tilde{F}_{11}\left(\Box\right)\Box+\frac{1}{2}\tilde{F}_{12}\left(\Box\right)\Box
\nonumber
\\
&&-\frac{1}{2}\tilde{F}_{13}\left(\Box\right)\Box+2\tilde{F}_{14}\left(\Box\right)-2\tilde{F}_{15}\left(\Box\right)+\frac{1}{2}\tilde{F}_{16}\left(\Box\right)\Box+\frac{1}{2}\tilde{F}_{20}\left(\Box\right)\Box-\frac{1}{2}\tilde{F}_{21}\left(\Box\right)\Box
\nonumber
\\
&&+\frac{1}{2}\tilde{F}_{22}\left(\Box\right)\Box-\frac{1}{2}\tilde{F}_{23}\left(\Box\right)\Box+\tilde{F}_{28}\left(\Box\right)-\frac{1}{2}\tilde{F}_{36}\left(\Box\right)+\frac{1}{2}\tilde{F}_{37}\left(\Box\right)-\frac{1}{2}\tilde{F}_{42}\left(\Box\right)\Box
\nonumber
\\
&&+\frac{1}{2}\tilde{F}_{43}\left(\Box\right)\Box,
\end{eqnarray}

\begin{eqnarray}
&&q_{3}\left(\Box\right)=-\tilde{F}_{17}\left(\Box\right)\Box-\tilde{F}_{18}\left(\Box\right)\Box+\tilde{F}_{19}\left(\Box\right)\Box+\tilde{F}_{29}\left(\Box\right)+\tilde{F}_{34}\left(\Box\right)-\tilde{F}_{35}\left(\Box\right)
\nonumber
\\
&&-\tilde{F}_{38}\left(\Box\right)-\tilde{F}_{44}\left(\Box\right)\Box+2\tilde{F}_{47}\left(\Box\right),
\end{eqnarray}

\begin{equation}
q_{4}\left(\Box\right)=-\tilde{F}_{14}\left(\Box\right)-\tilde{F}_{15}\left(\Box\right)+\tilde{F}_{30}\left(\Box\right),
\end{equation}

\begin{eqnarray}
&&q_{5}\left(\Box\right)=4\tilde{F}_{1}\left(\Box\right)+2\tilde{F}_{2}\left(\Box\right)\Box+\frac{1}{2}\tilde{F}_{3}\left(\Box\right)-\frac{1}{2}\tilde{F}_{4}\left(\Box\right)+\tilde{F}_{5}\left(\Box\right)\Box+\tilde{F}_{7}\left(\Box\right)\Box
\nonumber
\\
&&+\tilde{F}_{9}\left(\Box\right)\Box^{2}+\tilde{F}_{31}\left(\Box\right)-\frac{1}{2}\tilde{F}_{39}\left(\Box\right)-\frac{1}{2}\tilde{F}_{40}\left(\Box\right)-2\tilde{F}_{41}\left(\Box\right),
\end{eqnarray}

\begin{eqnarray}
&&q_{6}\left(\Box\right)=\tilde{F}_{3}\left(\Box\right)+\tilde{F}_{4}\left(\Box\right)+\tilde{F}_{32}\left(\Box\right)+\frac{1}{2}\tilde{F}_{36}\left(\Box\right)+\frac{1}{2}\tilde{F}_{37}\left(\Box\right)-\tilde{F}_{38}\left(\Box\right)
\nonumber
\\
&&-\frac{1}{2}\tilde{F}_{39}\left(\Box\right)+\frac{1}{2}\tilde{F}_{40}\left(\Box\right)+\frac{1}{2}\tilde{F}_{45}\left(\Box\right)\Box+\frac{1}{2}\tilde{F}_{46}\left(\Box\right)\Box,
\end{eqnarray}

\begin{equation}
p_{1}\left(\Box\right)=\tilde{F}_{14}\left(\Box\right)\Box+\tilde{F}_{15}\left(\Box\right)\Box+\tilde{F}_{24}\left(\Box\right),
\end{equation}

\begin{equation}
p_{2}\left(\Box\right)=\frac{1}{2}\tilde{F}_{18}\left(\Box\right)\Box^{2}-\frac{1}{2}\tilde{F}_{19}\left(\Box\right)\Box^{2}+\tilde{F}_{25}\left(\Box\right)+\tilde{F}_{34}\left(\Box\right)-\tilde{F}_{47}\left(\Box\right)\Box,
\end{equation}

\begin{equation}
p_{3}\left(\Box\right)=\frac{1}{2}\tilde{F}_{3}\left(\Box\right)\Box+\frac{1}{2}\tilde{F}_{4}\left(\Box\right)\Box+\tilde{F}_{26}\left(\Box\right)-\frac{1}{2}\tilde{F}_{39}\left(\Box\right)\Box+\frac{1}{2}\tilde{F}_{40}\left(\Box\right)\Box,
\end{equation}

\begin{eqnarray}
&&s\left(\Box\right)=-\frac{1}{2}\tilde{F}_{10}\left(\Box\right)-\frac{1}{2}\tilde{F}_{11}\left(\Box\right)-\frac{1}{2}\tilde{F}_{12}\left(\Box\right)-\frac{1}{2}\tilde{F}_{13}\left(\Box\right)-\frac{1}{2}\tilde{F}_{16}\left(\Box\right)
\nonumber
\\
&&+\tilde{F}_{17}\left(\Box\right)-\frac{1}{2}\tilde{F}_{20}\left(\Box\right)-\frac{1}{2}\tilde{F}_{21}\left(\Box\right)-\frac{1}{2}\tilde{F}_{22}\left(\Box\right)-\frac{1}{2}\tilde{F}_{23}\left(\Box\right)+\tilde{F}_{33}\left(\Box\right)
\nonumber
\\
&&+\frac{1}{2}\tilde{F}_{42}\left(\Box\right)-\frac{1}{2}\tilde{F}_{43}\left(\Box\right)-\frac{1}{2}\tilde{F}_{44}\left(\Box\right)-\frac{1}{2}\tilde{F}_{45}\left(\Box\right)-\frac{1}{2}\tilde{F}_{46}\left(\Box\right).
\end{eqnarray}

\chapter{Poincar\'e Gauge gravity as the local limit}
\label{ap:5}

Here we will give more insight on how Poincar\'e Gauge Gravity can be recast as the local limit of our theory \eqref{laaccion}.

Poincar\'e Gauge Gravity is constructed by gauging the Poincar\'e group, that is formed of the homogeneous Lorentz group $SO(3,1)$ together with the spacetime translations. The field strength of the latter is the torsion field, while the Riemann curvature is associated to the homogeneous part \cite{Gauge}.
Inspired by Yang-Mills theories, the usual Lagrangian of this theory is built using quadratic terms in the field strengths, such as\footnote{Please note that we have used the contorsion tensor instead of the torsion one without losing any generality, since they are related by a linear expression.}
\begin{eqnarray}
\label{9lag}
&&\mathcal{L}_{\rm PG}=\tilde{R}+b_{1}\tilde{R}^{2}+b_{2}\tilde{R}_{\mu\nu\rho\sigma}\tilde{R}^{\mu\nu\rho\sigma}+b_{3}\tilde{R}_{\mu\nu\rho\sigma}\tilde{R}^{\rho\sigma\mu\nu}+b_{4}\tilde{R}_{\mu\nu\rho\sigma}\tilde{R}^{\mu\rho\nu\sigma}
\nonumber
\\
&&+b_{5}\tilde{R}_{\mu\nu}\tilde{R}^{\mu\nu}+b_{6}\tilde{R}_{\mu\nu}\tilde{R}^{\nu\mu}+a_{1}K_{\mu\nu\rho}K^{\mu\nu\rho}+a_{2}K_{\mu\nu\rho}K^{\mu\rho\nu}
\nonumber
\\
&&+a_{3}K_{\nu\,\,\,\,\,\mu}^{\,\,\,\mu}K_{\,\,\,\,\,\,\rho}^{\nu\rho},
\end{eqnarray}
which is usually known as the nine-parameter Lagrangian. Since in the torsion-free limit we want to recover the results of usual IDG, the local limit at zero torsion must be GR. This fact imposes the following constraints in the Lagrangian \eqref{9lag}
\begin{eqnarray}
\label{Bonnet}
b_{6}=-4b_{1}-b_{5}\;,\;b_{4}=2\left(b_{1}-b_{2}-b_{3}\right),
\end{eqnarray}
where we have used the topological character of the Gauss-Bonnet term.\\
From \eqref{9lag}, and taking into account \eqref{Bonnet}, one can calculate the linearised Lagrangian just by substituing the expressions of the curvature tensors (\ref{riemann},\ref{ricci},\ref{scalar}), obtaining
\begin{eqnarray}
\label{PGlinear}
&&\mathcal{L}_{{\rm PG}}^{{\rm linear}}=\frac{1}{2}h_{\mu\nu}\Box h^{\mu\nu}-h_{\mu}^{\,\,\alpha}\partial_{\alpha}\partial_{\sigma}h^{\sigma\mu}+h\partial_{\mu}\partial_{\nu}h^{\mu\nu}-\frac{1}{2}h\Box h-4b_{1}h\Box\partial_{\rho}K_{\,\,\,\,\,\sigma}^{\rho\sigma}
\nonumber
\\
&&+4b_{1}h_{\mu\nu}\partial^{\mu}\partial^{\nu}\partial_{\rho}K_{\,\,\,\,\,\sigma}^{\rho\sigma}-\left(6b_{1}+b_{5}\right)h_{\mu\nu}\partial^{\nu}\partial_{\sigma}\partial_{\rho}K^{\mu\sigma\rho}-\left(6b_{1}+b_{5}\right)h_{\mu\nu}\Box\partial_{\rho}K^{\rho\mu\nu}
\nonumber
\\
&&+K^{\mu\sigma\lambda}\left(a_{1}+2b_{2}\Box\right)K_{\mu\sigma\lambda}+K^{\mu\sigma\lambda}\left(a_{2}-2\left(b_{1}-b_{2}-b_{3}\right)\Box\right)K_{\mu\lambda\sigma}
\nonumber
\\
&&+K_{\mu\,\,\rho}^{\,\,\rho}\left(a_{3}+b_{5}\Box\right)K_{\,\,\,\,\,\sigma}^{\mu\sigma}+\left(b_{5}-2b_{1}+2b_{2}+2b_{3}\right)K_{\,\,\nu\rho}^{\mu}\partial_{\mu}\partial_{\sigma}K^{\sigma\nu\rho}
\nonumber
\\
&&+\left(-4b_{1}-b_{5}+4b_{3}\right)K_{\,\,\nu\rho}^{\mu}\partial_{\mu}\partial_{\sigma}K^{\sigma\rho\nu}+4\left(b_{1}-b_{2}-b_{3}\right)K_{\mu\,\,\,\,\,\nu}^{\,\,\rho}\partial_{\rho}\partial_{\sigma}K^{\mu\nu\sigma}
\nonumber
\\
&&-2b_{2}K_{\mu\,\,\,\,\,\nu}^{\,\,\rho}\partial_{\rho}\partial_{\sigma}K^{\mu\sigma\nu}+\left(4b_{1}+b_{3}\right)K_{\,\,\,\,\,\rho}^{\mu\rho}\partial_{\mu}\partial_{\nu}K_{\,\,\,\,\,\sigma}^{\nu\sigma}+2b_{5}K_{\,\,\,\lambda\sigma}^{\lambda}\partial_{\mu}\partial_{\alpha}K^{\sigma\mu\alpha}.
\end{eqnarray}
At the same time, the local limit of our theory can be expressed as follows, where the constraints \eqref{constraints} have been applied
\begin{eqnarray}
\label{linearlocal}
&&\mathcal{L}\left(M_{S}\rightarrow\infty\right)=\frac{1}{2}a\left(0\right)h_{\mu\nu}\Box h^{\mu\nu}-a\left(0\right)h_{\mu}^{\,\,\alpha}\partial_{\alpha}\partial_{\sigma}h^{\sigma\mu}+c\left(0\right)h\partial_{\mu}\partial_{\nu}h^{\mu\nu}
\nonumber
\\
&&-\frac{1}{2}c\left(0\right)h\Box h+\frac{a\left(0\right)-c\left(0\right)}{\Box}h^{\lambda\sigma}\partial_{\sigma}\partial_{\lambda}\partial_{\mu}\partial_{\nu}h^{\mu\nu}+u\left(0\right)h\Box\partial_{\rho}K_{\,\,\,\,\,\sigma}^{\rho\sigma}
\nonumber
\\
&&-u\left(0\right)h_{\mu\nu}\partial^{\mu}\partial^{\nu}\partial_{\rho}K_{\,\,\,\,\,\sigma}^{\rho\sigma}+v_{2}\left(0\right)h_{\mu\nu}\partial^{\nu}\partial_{\sigma}\partial_{\rho}K^{\mu\sigma\rho}+v_{2}\left(0\right)h_{\mu\nu}\Box\partial_{\rho}K^{\rho\mu\nu}
\nonumber
\\
&&+p_{1}\left(0\right)K^{\mu\sigma\lambda}K_{\mu\sigma\lambda}+p_{2}\left(0\right)K^{\mu\sigma\lambda}K_{\mu\lambda\sigma}+p_{3}\left(0\right)K_{\mu\,\,\rho}^{\,\,\rho}K_{\,\,\,\,\,\sigma}^{\mu\sigma}+q_{1}\left(0\right)K_{\,\,\nu\rho}^{\mu}\partial_{\mu}\partial_{\sigma}K^{\sigma\nu\rho}
\nonumber
\\
&&+q_{2}\left(0\right)K_{\,\,\nu\rho}^{\mu}\partial_{\mu}\partial_{\sigma}K^{\sigma\rho\nu}+q_{3}\left(0\right)K_{\mu\,\,\,\,\,\nu}^{\,\,\rho}\partial_{\rho}\partial_{\sigma}K^{\mu\nu\sigma}+q_{4}\left(0\right)K_{\mu\,\,\,\,\,\nu}^{\,\,\rho}\partial_{\rho}\partial_{\sigma}K^{\mu\sigma\nu}
\nonumber
\\
&&+q_{5}\left(0\right)K_{\,\,\,\,\,\rho}^{\mu\rho}\partial_{\mu}\partial_{\nu}K_{\,\,\,\,\,\sigma}^{\nu\sigma}+q_{6}\left(0\right)K_{\,\,\,\lambda\sigma}^{\lambda}\partial_{\mu}\partial_{\alpha}K^{\sigma\mu\alpha}+s\left(0\right)K_{\mu}^{\,\,\nu\rho}\partial_{\nu}\partial_{\rho}\partial_{\alpha}\partial_{\sigma}K^{\mu\alpha\sigma}.
\end{eqnarray} 
It is straightforward to realise that we have more free parameters in \eqref{linearlocal} than in \eqref{PGlinear}, which means that if we want \eqref{PGlinear} as the local limit, we will need to impose more constraints in the parameters present in \eqref{linearlocal}. The question now is if there exists a PG theory that can be recast as the local limit of our theory without compromising the independence of the parameters. The answer is affirmative, as can be seen in the following Lagrangian
\begin{eqnarray}
\label{genPG}
&&\mathcal{L}_{{\rm GPG}}=\tilde{R}+b_{1}\tilde{R}^{2}+b_{2}\tilde{R}_{\mu\nu\rho\sigma}\tilde{R}^{\mu\nu\rho\sigma}+b_{3}\tilde{R}_{\mu\nu\rho\sigma}\tilde{R}^{\rho\sigma\mu\nu}
\nonumber
\\
&&+2\left(b_{1}-b_{2}-b_{3}\right)\tilde{R}_{\mu\nu\rho\sigma}\tilde{R}^{\mu\rho\nu\sigma}+b_{5}\tilde{R}_{\mu\nu}\tilde{R}^{\mu\nu}-\left(4b_{1}+b_{5}\right)\tilde{R}_{\mu\nu}\tilde{R}^{\nu\mu}+a_{1}K_{\mu\nu\rho}K^{\mu\nu\rho}
\nonumber
\\
&&+a_{2}K_{\mu\nu\rho}K^{\mu\rho\nu}+a_{3}K_{\nu\,\,\,\,\,\mu}^{\,\,\,\mu}K_{\,\,\,\,\,\,\rho}^{\nu\rho}+c_{1}K_{\,\,\nu\rho}^{\mu}\nabla_{\mu}\nabla_{\sigma}K^{\sigma\nu\rho}+c_{2}K_{\,\,\nu\rho}^{\mu}\nabla_{\mu}\nabla_{\sigma}K^{\sigma\rho\nu}
\nonumber
\\
&&+c_{3}K_{\mu\,\,\,\,\,\nu}^{\,\,\rho}\nabla_{\rho}\nabla_{\sigma}K^{\mu\nu\sigma}+c_{4}K_{\mu\,\,\,\,\,\nu}^{\,\,\rho}\nabla_{\rho}\nabla_{\sigma}K^{\mu\sigma\nu},
\end{eqnarray}
which is Poincar\'e Gauge invariant and local. Its corresponding linearised expression is
\begin{eqnarray}
\label{genPGlin}
&&\mathcal{L}_{{\rm GPG}}^{{\rm linear}}=\frac{1}{2}h_{\mu\nu}\Box h^{\mu\nu}-h_{\mu}^{\,\,\alpha}\partial_{\alpha}\partial_{\sigma}h^{\sigma\mu}+h\partial_{\mu}\partial_{\nu}h^{\mu\nu}-\frac{1}{2}h\Box h-4b_{1}h\Box\partial_{\rho}K_{\,\,\,\,\,\sigma}^{\rho\sigma}
\nonumber
\\
&&+4b_{1}h_{\mu\nu}\partial^{\mu}\partial^{\nu}\partial_{\rho}K_{\,\,\,\,\,\sigma}^{\rho\sigma}-\left(6b_{1}+b_{5}\right)h_{\mu\nu}\partial^{\nu}\partial_{\sigma}\partial_{\rho}K^{\mu\sigma\rho}-\left(6b_{1}+b_{5}\right)h_{\mu\nu}\Box\partial_{\rho}K^{\rho\mu\nu}
\nonumber
\\
&&+K^{\mu\sigma\lambda}\left(a_{1}+2b_{2}\Box\right)K_{\mu\sigma\lambda}+K^{\mu\sigma\lambda}\left(a_{2}-2\left(b_{1}-b_{2}-b_{3}\right)\Box\right)K_{\mu\lambda\sigma}
\nonumber
\\
&&+K_{\mu\,\,\rho}^{\,\,\rho}\left(a_{3}+b_{5}\Box\right)K_{\,\,\,\,\,\sigma}^{\mu\sigma}+\left(b_{5}-2b_{1}+2b_{2}+2b_{3}+c_{1}\right)K_{\,\,\nu\rho}^{\mu}\partial_{\mu}\partial_{\sigma}K^{\sigma\nu\rho}
\nonumber
\\
&&+\left(-4b_{1}-b_{5}+4b_{3}+c_{2}\right)K_{\,\,\nu\rho}^{\mu}\partial_{\mu}\partial_{\sigma}K^{\sigma\rho\nu}+\left(4b_{1}-4b_{2}-4b_{3}+c_{3}\right)K_{\mu\,\,\,\,\,\nu}^{\,\,\rho}\partial_{\rho}\partial_{\sigma}K^{\mu\nu\sigma}
\nonumber
\\
&&-\left(2b_{2}-c_{4}\right)K_{\mu\,\,\,\,\,\nu}^{\,\,\rho}\partial_{\rho}\partial_{\sigma}K^{\mu\sigma\nu}+\left(4b_{1}+b_{3}\right)K_{\,\,\,\,\,\rho}^{\mu\rho}\partial_{\mu}\partial_{\nu}K_{\,\,\,\,\,\sigma}^{\nu\sigma}
\nonumber
\\
&&+2b_{5}K_{\,\,\,\lambda\sigma}^{\lambda}\partial_{\mu}\partial_{\alpha}K^{\sigma\mu\alpha}.
\end {eqnarray}
Therefore, one finds the following relations for the local limit of the functions involved in the linearised action \eqref{linearlocal} and the parameters in \eqref{genPGlin}
\begin{eqnarray}
\label{localconditions}
&&a\left(0\right)=1,\;c\left(0\right)=1,\;u\left(0\right)=-4b_{1},\;v_{2}\left(0\right)=-4\left(6b_{1}+b_{5}\right),\;p_{1}\left(0\right)=a_{1}+2b_{2}\Box,
\nonumber
\\
&&p_{2}\left(0\right)=a_{2}-2\left(b_{1}-b_{2}-b_{3}\right)\Box,\;p_{3}\left(0\right)=a_{3}+b_{5}\Box,
\nonumber
\\
&&q_{1}\left(0\right)=b_{5}-2b_{1}+2b_{2}+2b_{3}+c_{1},\;q_{2}\left(0\right)=-4b_{1}-b_{5}+4b_{3}+c_{2},
\nonumber
\\
&&q_{3}\left(0\right)=4b_{1}-4b_{2}-4b_{3}+c_{3},\;q_{4}\left(0\right)=-2b_{2}+c_{4},\;q_{5}\left(0\right)=4b_{1}+b_{3},
\nonumber
\\
&&q_{6}\left(0\right)=2b_{5},\;s\left(0\right)=0.
\end{eqnarray}
It can be observed that these limits do not impose new relations between the functions. \\
Hence, we have proved that if the previous limits \eqref{localconditions} apply, the local limit of our theory is a local PG theory, specifically the one described by the Lagrangian \eqref{genPG}.

\addcontentsline{toc}{chapter}{Bibliography}
\bibliographystyle{JHEP}
\bibliography{thesis_bibliography}

\providecommand{\href}[2]{#2}\begingroup\raggedright\begin{thebibliography}{100}

\bibitem{le1858theorie}
U.~J. Le~Verrier, \emph{Theorie du mouvement apparent du soleil},  in
  \emph{Annales de l'Observatoire de Paris}, vol.~4, 1858.

\bibitem{Leverrier:1910}
\emph{The discovery of neptune leverrier's letter to galle},
  \href{http://dx.doi.org/10.1038/085184c0}{\emph{Nature} {\bfseries 85} (1910)
  184–185}.

\bibitem{galle1846account}
J.~Galle, \emph{Account of the discovery of le verrier's planet neptune, at
  berlin, sept. 23, 1846}, {\emph{Monthly Notices of the Royal Astronomical
  Society} {\bfseries 7} (1846) 153}.

\bibitem{leverrier1859compte}
U.~Leverrier, \emph{Compte rendu de 1'academie des sciences (paris) 49, 379
  (1859); theorie du mouvement de mercure}, {\emph{Annales de} {\bfseries 1}
  (1859) }.

\bibitem{newcomb1882discussion}
S.~Newcomb, \emph{Discussion and results of observations on transits of mercury
  from 1677 to 1881}, {\emph{[United States. Nautical Almanac Office.
  Astronomical paper; v. 1, pt. 6 (1882)],[Washington: US Nautical Almanac
  Office, 1882], p. 363-487: charts; 29 cm.} {\bfseries 1} (1882) 363--487}.

\bibitem{newcomb1895elements}
S.~Newcomb, \emph{The elements of the four inner planets and the fundamental
  constants of astronomy}.
\newblock US Government Printing Office, 1895.

\bibitem{hall1894suggestion}
A.~Hall, \emph{A suggestion in the theory of mercury}, {\emph{The astronomical
  journal} {\bfseries 14} (1894) 49--51}.

\bibitem{tisserand1890mouvement}
F.~Tisserand, \emph{Sur le mouvement des plan{\`e}tes, en supposant
  l’attraction repr{\'e}sent{\'e}e par l’une des lois {\'e}lectrodynamiques
  de gauss ou de weber}, {\emph{Comptes Rendus de l'Acad{\'e}mie des Sciences
  (Paris)} {\bfseries 100} (1890) 313--5}.

\bibitem{zollner1872natur}
J.~K.~F. Z{\"o}llner, \emph{{\"U}ber die Natur der Cometen: Beitr{\"a}ge zur
  Geschichte und Theorie der Erkenntnis}.
\newblock W. Engelmann, 1872.

\bibitem{levy1890application}
M.~L{\'e}vy, \emph{Sur l'application des lois {\'e}lectrodynamiques au
  mouvement des plan{\`e}tes}, {\emph{Comptes Rendus de l'Acad{\'e}mie des
  Sciences (Paris)} {\bfseries 110} (1890) 545--51}.

\bibitem{warburton1946advance}
F.~Warburton, \emph{The advance of the perihelion of mercury}, {\emph{Physical
  Review} {\bfseries 70} (1946) 86}.

\bibitem{einstein1915allgemeinen}
A.~Einstein, \emph{Zur allgemeinen relativit{\"a}tstheorie}.
\newblock Akademie der Wissenschaften, in Kommission bei W. de Gruyter, 1915.

\bibitem{Einstein:1915ca}
A.~Einstein, \emph{{The Field Equations of Gravitation}}, {\emph{Sitzungsber.
  Preuss. Akad. Wiss. Berlin (Math. Phys. )} {\bfseries 1915} (1915) 844--847}.

\bibitem{Einstein:1916vd}
A.~Einstein, \emph{{The Foundation of the General Theory of Relativity}},
  \href{http://dx.doi.org/10.1002/andp.200590044}{\emph{Annalen Phys.}
  {\bfseries 49} (1916) 769--822}.

\bibitem{Einstein:1915bz}
A.~Einstein, \emph{{Explanation of the Perihelion Motion of Mercury from the
  General Theory of Relativity}}, {\emph{Sitzungsber. Preuss. Akad. Wiss.
  Berlin (Math. Phys. )} {\bfseries 1915} (1915) 831--839}.

\bibitem{ROSEVEARE1979165}
N.~Roseveare, \emph{Leverrier to einstein: A review of the mercury problem},
  \href{http://dx.doi.org/https://doi.org/10.1016/0083-6656(79)90003-5}{\emph{Vistas
  in Astronomy} {\bfseries 23} (1979) 165 -- 171}.

\bibitem{Riess:1998cb}
{\scshape Supernova Search Team} collaboration, A.~G. Riess et~al.,
  \emph{Observational evidence from supernovae for an accelerating universe and
  a cosmological constant},
  \href{http://dx.doi.org/10.1086/300499}{\emph{Astron. J.} {\bfseries 116}
  (1998) 1009--1038}.

\bibitem{Eisenstein:2005su}
{\scshape SDSS} collaboration, D.~J. Eisenstein et~al., \emph{{Detection of the
  Baryon Acoustic Peak in the Large-Scale Correlation Function of SDSS Luminous
  Red Galaxies}}, \href{http://dx.doi.org/10.1086/466512}{\emph{Astrophys. J.}
  {\bfseries 633} (2005) 560--574},
  [\href{https://arxiv.org/abs/astro-ph/0501171}{{\ttfamily
  astro-ph/0501171}}].

\bibitem{Boughn:2003yz}
S.~Boughn and R.~Crittenden, \emph{{A Correlation of the cosmic microwave sky
  with large scale structure}},
  \href{http://dx.doi.org/10.1038/nature02139}{\emph{Nature} {\bfseries 427}
  (2004) 45--47}, [\href{https://arxiv.org/abs/astro-ph/0305001}{{\ttfamily
  astro-ph/0305001}}].

\bibitem{Sahni:2004ai}
V.~Sahni, \emph{{Dark matter and dark energy}},
  \href{http://dx.doi.org/10.1007/b99562,
  10.1007/978-3-540-31535-3_5}{\emph{Lect. Notes Phys.} {\bfseries 653} (2004)
  141--180}, [\href{https://arxiv.org/abs/astro-ph/0403324}{{\ttfamily
  astro-ph/0403324}}].

\bibitem{Peebles:2002gy}
P.~J.~E. Peebles and B.~Ratra, \emph{{The Cosmological constant and dark
  energy}}, \href{http://dx.doi.org/10.1103/RevModPhys.75.559}{\emph{Rev. Mod.
  Phys.} {\bfseries 75} (2003) 559--606},
  [\href{https://arxiv.org/abs/astro-ph/0207347}{{\ttfamily
  astro-ph/0207347}}].

\bibitem{weinberg2008cosmology}
S.~Weinberg, \emph{Cosmology}.
\newblock Oxford university press, 2008.

\bibitem{Weinberg:1988cp}
S.~Weinberg, \emph{{The Cosmological Constant Problem}},
  \href{http://dx.doi.org/10.1103/RevModPhys.61.1}{\emph{Rev. Mod. Phys.}
  {\bfseries 61} (1989) 1--23}.

\bibitem{Martin:2012bt}
J.~Martin, \emph{{Everything You Always Wanted To Know About The Cosmological
  Constant Problem (But Were Afraid To Ask)}},
  \href{http://dx.doi.org/10.1016/j.crhy.2012.04.008}{\emph{Comptes Rendus
  Physique} {\bfseries 13} (2012) 566--665},
  [\href{https://arxiv.org/abs/1205.3365}{{\ttfamily 1205.3365}}].

\bibitem{Bertone:2004pz}
G.~Bertone, D.~Hooper and J.~Silk, \emph{{Particle dark matter: Evidence,
  candidates and constraints}},
  \href{http://dx.doi.org/10.1016/j.physrep.2004.08.031}{\emph{Phys. Rept.}
  {\bfseries 405} (2005) 279--390},
  [\href{https://arxiv.org/abs/hep-ph/0404175}{{\ttfamily hep-ph/0404175}}].

\bibitem{Ackermann:2013yva}
{\scshape Fermi-LAT} collaboration, M.~Ackermann et~al., \emph{{Dark Matter
  Constraints from Observations of 25 Milky Way Satellite Galaxies with the
  Fermi Large Area Telescope}},
  \href{http://dx.doi.org/10.1103/PhysRevD.89.042001}{\emph{Phys. Rev. D}
  {\bfseries 89} (2014) 042001},
  [\href{https://arxiv.org/abs/1310.0828}{{\ttfamily 1310.0828}}].

\bibitem{Khachatryan:2014rra}
{\scshape CMS} collaboration, V.~Khachatryan et~al., \emph{{Search for dark
  matter, extra dimensions, and unparticles in monojet events in proton--proton
  collisions at $\sqrt{s} = 8$ TeV}},
  \href{http://dx.doi.org/10.1140/epjc/s10052-015-3451-4}{\emph{Eur. Phys. J.
  C} {\bfseries 75} (2015) 235},
  [\href{https://arxiv.org/abs/1408.3583}{{\ttfamily 1408.3583}}].

\bibitem{Buckley:2017ijx}
M.~R. Buckley and A.~H.~G. Peter, \emph{{Gravitational probes of dark matter
  physics}}, \href{http://dx.doi.org/10.1016/j.physrep.2018.07.003}{\emph{Phys.
  Rept.} {\bfseries 761} (2018) 1--60},
  [\href{https://arxiv.org/abs/1712.06615}{{\ttfamily 1712.06615}}].

\bibitem{Riess:2020sih}
A.~G. Riess, \emph{{The Expansion of the Universe is Faster than Expected}},
  \href{http://dx.doi.org/10.1038/s42254-019-0137-0}{\emph{Nature Rev. Phys.}
  {\bfseries 2} (2019) 10--12},
  [\href{https://arxiv.org/abs/2001.03624}{{\ttfamily 2001.03624}}].

\bibitem{Senovilla:2018aav}
J.~M.~M. Senovilla, \emph{Singularity theorems and their consequences},
  \href{http://dx.doi.org/10.1023/A:1018801101244}{\emph{Gen.Rel.Grav.}
  {\bfseries 30} (1998) 701},
  [\href{https://arxiv.org/abs/1801.04912}{{\ttfamily 1801.04912}}].

\bibitem{Sotiriou:2008rp}
T.~P. Sotiriou and V.~Faraoni, \emph{{f(R) Theories Of Gravity}},
  \href{http://dx.doi.org/10.1103/RevModPhys.82.451}{\emph{Rev. Mod. Phys.}
  {\bfseries 82} (2010) 451--497},
  [\href{https://arxiv.org/abs/0805.1726}{{\ttfamily 0805.1726}}].

\bibitem{DeFelice:2010aj}
A.~De~Felice and S.~Tsujikawa, \emph{{f(R) theories}},
  \href{http://dx.doi.org/10.12942/lrr-2010-3}{\emph{Living Rev. Rel.}
  {\bfseries 13} (2010) 3}, [\href{https://arxiv.org/abs/1002.4928}{{\ttfamily
  1002.4928}}].

\bibitem{Capozziello:2011et}
S.~Capozziello and M.~De~Laurentis, \emph{{Extended Theories of Gravity}},
  \href{http://dx.doi.org/10.1016/j.physrep.2011.09.003}{\emph{Phys. Rept.}
  {\bfseries 509} (2011) 167--321},
  [\href{https://arxiv.org/abs/1108.6266}{{\ttfamily 1108.6266}}].

\bibitem{delaCruzDombriz:2006fj}
A.~de~la Cruz-Dombriz and A.~Dobado, \emph{{A f(R) gravity without cosmological
  constant}}, \href{http://dx.doi.org/10.1103/PhysRevD.74.087501}{\emph{Phys.
  Rev. D} {\bfseries 74} (2006) 087501},
  [\href{https://arxiv.org/abs/gr-qc/0607118}{{\ttfamily gr-qc/0607118}}].

\bibitem{Amendola:2006we}
L.~Amendola, R.~Gannouji, D.~Polarski and S.~Tsujikawa, \emph{{Conditions for
  the cosmological viability of f(R) dark energy models}},
  \href{http://dx.doi.org/10.1103/PhysRevD.75.083504}{\emph{Phys. Rev. D}
  {\bfseries 75} (2007) 083504},
  [\href{https://arxiv.org/abs/gr-qc/0612180}{{\ttfamily gr-qc/0612180}}].

\bibitem{Cognola:2007zu}
G.~Cognola, E.~Elizalde, S.~Nojiri, S.~Odintsov, L.~Sebastiani and S.~Zerbini,
  \emph{{A Class of viable modified f(R) gravities describing inflation and the
  onset of accelerated expansion}},
  \href{http://dx.doi.org/10.1103/PhysRevD.77.046009}{\emph{Phys. Rev. D}
  {\bfseries 77} (2008) 046009},
  [\href{https://arxiv.org/abs/0712.4017}{{\ttfamily 0712.4017}}].

\bibitem{Hu:2007nk}
W.~Hu and I.~Sawicki, \emph{{Models of f(R) Cosmic Acceleration that Evade
  Solar-System Tests}},
  \href{http://dx.doi.org/10.1103/PhysRevD.76.064004}{\emph{Phys. Rev. D}
  {\bfseries 76} (2007) 064004},
  [\href{https://arxiv.org/abs/0705.1158}{{\ttfamily 0705.1158}}].

\bibitem{Barrow:1988xh}
J.~D. Barrow and S.~Cotsakis, \emph{{Inflation and the Conformal Structure of
  Higher Order Gravity Theories}},
  \href{http://dx.doi.org/10.1016/0370-2693(88)90110-4}{\emph{Phys. Lett. B}
  {\bfseries 214} (1988) 515--518}.

\bibitem{Nunez:2004ji}
A.~Nunez and S.~Solganik, \emph{{The Content of f(R) gravity}},
  \href{https://arxiv.org/abs/hep-th/0403159}{{\ttfamily hep-th/0403159}}.

\bibitem{Lovelock:1971yv}
D.~Lovelock, \emph{{The Einstein tensor and its generalizations}},
  \href{http://dx.doi.org/10.1063/1.1665613}{\emph{J. Math. Phys.} {\bfseries
  12} (1971) 498--501}.

\bibitem{Capozziello:2007ec}
S.~Capozziello and M.~Francaviglia, \emph{{Extended Theories of Gravity and
  their Cosmological and Astrophysical Applications}},
  \href{http://dx.doi.org/10.1007/s10714-007-0551-y}{\emph{Gen. Rel. Grav.}
  {\bfseries 40} (2008) 357--420},
  [\href{https://arxiv.org/abs/0706.1146}{{\ttfamily 0706.1146}}].

\bibitem{Clifton:2011jh}
T.~Clifton, P.~G. Ferreira, A.~Padilla and C.~Skordis, \emph{{Modified Gravity
  and Cosmology}},
  \href{http://dx.doi.org/10.1016/j.physrep.2012.01.001}{\emph{Phys. Rept.}
  {\bfseries 513} (2012) 1--189},
  [\href{https://arxiv.org/abs/1106.2476}{{\ttfamily 1106.2476}}].

\bibitem{Berti:2015itd}
E.~Berti et~al., \emph{{Testing General Relativity with Present and Future
  Astrophysical Observations}},
  \href{http://dx.doi.org/10.1088/0264-9381/32/24/243001}{\emph{Class. Quant.
  Grav.} {\bfseries 32} (2015) 243001},
  [\href{https://arxiv.org/abs/1501.07274}{{\ttfamily 1501.07274}}].

\bibitem{baojiu2019modified}
L.~Baojiu and K.~Kazuya, \emph{Modified Gravity: Progresses And Outlook Of
  Theories, Numerical Techniques And Observational Tests}.
\newblock World Scientific, 2019.

\bibitem{Heisenberg:2018vsk}
L.~Heisenberg, \emph{{A systematic approach to generalisations of General
  Relativity and their cosmological implications}},
  \href{http://dx.doi.org/10.1016/j.physrep.2018.11.006}{\emph{Phys.\ Rept.}
  {\bfseries 796} (2019) 1--113},
  [\href{https://arxiv.org/abs/1807.01725}{{\ttfamily 1807.01725}}].

\bibitem{Horndeski:1974wa}
G.~W. Horndeski, \emph{{Second-order scalar-tensor field equations in a
  four-dimensional space}},
  \href{http://dx.doi.org/10.1007/BF01807638}{\emph{Int. J. Theor. Phys.}
  {\bfseries 10} (1974) 363--384}.

\bibitem{Zumalacarregui:2013pma}
M.~Zumalacárregui and J.~García-Bellido, \emph{{Transforming gravity: from
  derivative couplings to matter to second-order scalar-tensor theories beyond
  the Horndeski Lagrangian}},
  \href{http://dx.doi.org/10.1103/PhysRevD.89.064046}{\emph{Phys. Rev. D}
  {\bfseries 89} (2014) 064046},
  [\href{https://arxiv.org/abs/1308.4685}{{\ttfamily 1308.4685}}].

\bibitem{Gleyzes:2014dya}
J.~Gleyzes, D.~Langlois, F.~Piazza and F.~Vernizzi, \emph{{Healthy theories
  beyond Horndeski}},
  \href{http://dx.doi.org/10.1103/PhysRevLett.114.211101}{\emph{Phys. Rev.
  Lett.} {\bfseries 114} (2015) 211101},
  [\href{https://arxiv.org/abs/1404.6495}{{\ttfamily 1404.6495}}].

\bibitem{Gleyzes:2014qga}
J.~Gleyzes, D.~Langlois, F.~Piazza and F.~Vernizzi, \emph{{Exploring
  gravitational theories beyond Horndeski}},
  \href{http://dx.doi.org/10.1088/1475-7516/2015/02/018}{\emph{JCAP} {\bfseries
  02} (2015) 018}, [\href{https://arxiv.org/abs/1408.1952}{{\ttfamily
  1408.1952}}].

\bibitem{Horndeski:2016bku}
G.~W. Horndeski, \emph{{Lagrange Multipliers and Third Order Scalar-Tensor
  Field Theories}},  \href{https://arxiv.org/abs/1608.03212}{{\ttfamily
  1608.03212}}.

\bibitem{Langlois:2015cwa}
D.~Langlois and K.~Noui, \emph{{Degenerate higher derivative theories beyond
  Horndeski: evading the Ostrogradski instability}},
  \href{http://dx.doi.org/10.1088/1475-7516/2016/02/034}{\emph{JCAP} {\bfseries
  02} (2016) 034}, [\href{https://arxiv.org/abs/1510.06930}{{\ttfamily
  1510.06930}}].

\bibitem{Achour:2016rkg}
J.~Ben~Achour, D.~Langlois and K.~Noui, \emph{{Degenerate higher order
  scalar-tensor theories beyond Horndeski and disformal transformations}},
  \href{http://dx.doi.org/10.1103/PhysRevD.93.124005}{\emph{Phys. Rev. D}
  {\bfseries 93} (2016) 124005},
  [\href{https://arxiv.org/abs/1602.08398}{{\ttfamily 1602.08398}}].

\bibitem{Bergmann:1968ve}
P.~G. Bergmann, \emph{{Comments on the scalar tensor theory}},
  \href{http://dx.doi.org/10.1007/BF00668828}{\emph{Int. J. Theor. Phys.}
  {\bfseries 1} (1968) 25--36}.

\bibitem{Nicolis:2008in}
A.~Nicolis, R.~Rattazzi and E.~Trincherini, \emph{{The Galileon as a local
  modification of gravity}},
  \href{http://dx.doi.org/10.1103/PhysRevD.79.064036}{\emph{Phys. Rev. D}
  {\bfseries 79} (2009) 064036},
  [\href{https://arxiv.org/abs/0811.2197}{{\ttfamily 0811.2197}}].

\bibitem{Charmousis:2011bf}
C.~Charmousis, E.~J. Copeland, A.~Padilla and P.~M. Saffin, \emph{{General
  second order scalar-tensor theory, self tuning, and the Fab Four}},
  \href{http://dx.doi.org/10.1103/PhysRevLett.108.051101}{\emph{Phys. Rev.
  Lett.} {\bfseries 108} (2012) 051101},
  [\href{https://arxiv.org/abs/1106.2000}{{\ttfamily 1106.2000}}].

\bibitem{Horndeski:1976gi}
G.~Horndeski, \emph{{Conservation of Charge and the Einstein-Maxwell Field
  Equations}}, \href{http://dx.doi.org/10.1063/1.522837}{\emph{J. Math. Phys.}
  {\bfseries 17} (1976) 1980--1987}.

\bibitem{Jimenez:2013qsa}
J.~Beltran~Jimenez, R.~Durrer, L.~Heisenberg and M.~Thorsrud, \emph{{Stability
  of Horndeski vector-tensor interactions}},
  \href{http://dx.doi.org/10.1088/1475-7516/2013/10/064}{\emph{JCAP} {\bfseries
  10} (2013) 064}, [\href{https://arxiv.org/abs/1308.1867}{{\ttfamily
  1308.1867}}].

\bibitem{Heisenberg:2014rta}
L.~Heisenberg, \emph{{Generalization of the Proca Action}},
  \href{http://dx.doi.org/10.1088/1475-7516/2014/05/015}{\emph{JCAP} {\bfseries
  05} (2014) 015}, [\href{https://arxiv.org/abs/1402.7026}{{\ttfamily
  1402.7026}}].

\bibitem{Jimenez:2016isa}
J.~Beltran~Jimenez and L.~Heisenberg, \emph{{Derivative self-interactions for a
  massive vector field}},
  \href{http://dx.doi.org/10.1016/j.physletb.2016.04.017}{\emph{Phys. Lett. B}
  {\bfseries 757} (2016) 405--411},
  [\href{https://arxiv.org/abs/1602.03410}{{\ttfamily 1602.03410}}].

\bibitem{Heisenberg:2016eld}
L.~Heisenberg, R.~Kase and S.~Tsujikawa, \emph{{Beyond generalized Proca
  theories}},
  \href{http://dx.doi.org/10.1016/j.physletb.2016.07.052}{\emph{Phys. Lett. B}
  {\bfseries 760} (2016) 617--626},
  [\href{https://arxiv.org/abs/1605.05565}{{\ttfamily 1605.05565}}].

\bibitem{Jacobson:2008aj}
T.~Jacobson, \emph{{Einstein-aether gravity: A Status report}},
  \href{http://dx.doi.org/10.22323/1.043.0020}{\emph{PoS} {\bfseries QG-PH}
  (2007) 020}, [\href{https://arxiv.org/abs/0801.1547}{{\ttfamily 0801.1547}}].

\bibitem{Sotiriou:2010wn}
T.~P. Sotiriou, \emph{{Horava-Lifshitz gravity: a status report}},
  \href{http://dx.doi.org/10.1088/1742-6596/283/1/012034}{\emph{J. Phys. Conf.
  Ser.} {\bfseries 283} (2011) 012034},
  [\href{https://arxiv.org/abs/1010.3218}{{\ttfamily 1010.3218}}].

\bibitem{deRham:2010ik}
C.~de~Rham and G.~Gabadadze, \emph{{Generalization of the Fierz-Pauli Action}},
  \href{http://dx.doi.org/10.1103/PhysRevD.82.044020}{\emph{Phys. Rev. D}
  {\bfseries 82} (2010) 044020},
  [\href{https://arxiv.org/abs/1007.0443}{{\ttfamily 1007.0443}}].

\bibitem{rosen1978bimetric}
N.~Rosen, \emph{Bimetric gravitation theory on a cosmological basis},
  {\emph{General Relativity and Gravitation} {\bfseries 9} (1978) 339--351}.

\bibitem{Moffat:2005si}
J.~Moffat, \emph{{Scalar-tensor-vector gravity theory}},
  \href{http://dx.doi.org/10.1088/1475-7516/2006/03/004}{\emph{JCAP} {\bfseries
  03} (2006) 004}, [\href{https://arxiv.org/abs/gr-qc/0506021}{{\ttfamily
  gr-qc/0506021}}].

\bibitem{Biswas:2005qr}
T.~Biswas, A.~Mazumdar and W.~Siegel, \emph{{Bouncing universes in
  string-inspired gravity}},
  \href{http://dx.doi.org/10.1088/1475-7516/2006/03/009}{\emph{JCAP} {\bfseries
  0603} (2006) 009}, [\href{https://arxiv.org/abs/hep-th/0508194}{{\ttfamily
  hep-th/0508194}}].

\bibitem{Biswas:2011ar}
T.~Biswas, E.~Gerwick, T.~Koivisto and A.~Mazumdar, \emph{{Towards singularity
  and ghost free theories of gravity}},
  \href{http://dx.doi.org/10.1103/PhysRevLett.108.031101}{\emph{Phys. Rev.
  Lett.} {\bfseries 108} (2012) 031101},
  [\href{https://arxiv.org/abs/1110.5249}{{\ttfamily 1110.5249}}].

\bibitem{Biswas:2010zk}
T.~Biswas, T.~Koivisto and A.~Mazumdar, \emph{{Towards a resolution of the
  cosmological singularity in non-local higher derivative theories of
  gravity}}, \href{http://dx.doi.org/10.1088/1475-7516/2010/11/008}{\emph{JCAP}
  {\bfseries 11} (2010) 008},
  [\href{https://arxiv.org/abs/1005.0590}{{\ttfamily 1005.0590}}].

\bibitem{Biswas:2012bp}
T.~Biswas, A.~S. Koshelev, A.~Mazumdar and S.~Y. Vernov, \emph{{Stable bounce
  and inflation in non-local higher derivative cosmology}},
  \href{http://dx.doi.org/10.1088/1475-7516/2012/08/024}{\emph{JCAP} {\bfseries
  08} (2012) 024}, [\href{https://arxiv.org/abs/1206.6374}{{\ttfamily
  1206.6374}}].

\bibitem{Koshelev:2018rau}
A.~S. Koshelev, J.~Marto and A.~Mazumdar, \emph{{Towards resolution of
  anisotropic cosmological singularity in infinite derivative gravity}},
  \href{http://dx.doi.org/10.1088/1475-7516/2019/02/020}{\emph{JCAP} {\bfseries
  02} (2019) 020}, [\href{https://arxiv.org/abs/1803.07072}{{\ttfamily
  1803.07072}}].

\bibitem{Frolov:2015bta}
V.~P. Frolov, \emph{{Mass-gap for black hole formation in higher derivative and
  ghost free gravity}},
  \href{http://dx.doi.org/10.1103/PhysRevLett.115.051102}{\emph{Phys. Rev.
  Lett.} {\bfseries 115} (2015) 051102},
  [\href{https://arxiv.org/abs/1505.00492}{{\ttfamily 1505.00492}}].

\bibitem{Buoninfante:2019swn}
L.~Buoninfante and A.~Mazumdar, \emph{{Nonlocal star as a blackhole mimicker}},
  \href{http://dx.doi.org/10.1103/PhysRevD.100.024031}{\emph{Phys. Rev.}
  {\bfseries D100} (2019) 024031},
  [\href{https://arxiv.org/abs/1903.01542}{{\ttfamily 1903.01542}}].

\bibitem{TheLIGOScientific:2017qsa}
{\scshape LIGO Scientific, Virgo} collaboration, B.~Abbott et~al.,
  \emph{{GW170817: Observation of Gravitational Waves from a Binary Neutron
  Star Inspiral}},
  \href{http://dx.doi.org/10.1103/PhysRevLett.119.161101}{\emph{Phys. Rev.
  Lett.} {\bfseries 119} (2017) 161101},
  [\href{https://arxiv.org/abs/1710.05832}{{\ttfamily 1710.05832}}].

\bibitem{GBM:2017lvd}
B.~Abbott et~al., \emph{{Multi-messenger Observations of a Binary Neutron Star
  Merger}}, \href{http://dx.doi.org/10.3847/2041-8213/aa91c9}{\emph{Astrophys.
  J. Lett.} {\bfseries 848} (2017) L12},
  [\href{https://arxiv.org/abs/1710.05833}{{\ttfamily 1710.05833}}].

\bibitem{Ezquiaga:2017ekz}
J.~M. Ezquiaga and M.~Zumalacárregui, \emph{{Dark Energy After GW170817: Dead
  Ends and the Road Ahead}},
  \href{http://dx.doi.org/10.1103/PhysRevLett.119.251304}{\emph{Phys. Rev.
  Lett.} {\bfseries 119} (2017) 251304},
  [\href{https://arxiv.org/abs/1710.05901}{{\ttfamily 1710.05901}}].

\bibitem{Baker:2017hug}
T.~Baker, E.~Bellini, P.~Ferreira, M.~Lagos, J.~Noller and I.~Sawicki,
  \emph{{Strong constraints on cosmological gravity from GW170817 and GRB
  170817A}},
  \href{http://dx.doi.org/10.1103/PhysRevLett.119.251301}{\emph{Phys. Rev.
  Lett.} {\bfseries 119} (2017) 251301},
  [\href{https://arxiv.org/abs/1710.06394}{{\ttfamily 1710.06394}}].

\bibitem{Creminelli:2017sry}
P.~Creminelli and F.~Vernizzi, \emph{{Dark Energy after GW170817 and
  GRB170817A}},
  \href{http://dx.doi.org/10.1103/PhysRevLett.119.251302}{\emph{Phys. Rev.
  Lett.} {\bfseries 119} (2017) 251302},
  [\href{https://arxiv.org/abs/1710.05877}{{\ttfamily 1710.05877}}].

\bibitem{Sakstein:2017xjx}
J.~Sakstein and B.~Jain, \emph{{Implications of the Neutron Star Merger
  GW170817 for Cosmological Scalar-Tensor Theories}},
  \href{http://dx.doi.org/10.1103/PhysRevLett.119.251303}{\emph{Phys. Rev.
  Lett.} {\bfseries 119} (2017) 251303},
  [\href{https://arxiv.org/abs/1710.05893}{{\ttfamily 1710.05893}}].

\bibitem{Wald:1984rg}
R.~M. Wald, \emph{{General Relativity}}.
\newblock Chicago Univ. Pr., Chicago, USA, 1984,
  \href{http://dx.doi.org/10.7208/chicago/9780226870373.001.0001}{10.7208/chicago/9780226870373.001.0001}.

\bibitem{Will:2005va}
C.~M. Will, \emph{{The Confrontation between general relativity and
  experiment}}, \href{http://dx.doi.org/10.12942/lrr-2006-3}{\emph{Living Rev.
  Rel.} {\bfseries 9} (2006) 3},
  [\href{https://arxiv.org/abs/gr-qc/0510072}{{\ttfamily gr-qc/0510072}}].

\bibitem{Abbott:2016blz}
{\scshape LIGO Scientific, Virgo} collaboration, B.~P. Abbott et~al.,
  \emph{{Observation of Gravitational Waves from a Binary Black Hole Merger}},
  \href{http://dx.doi.org/10.1103/PhysRevLett.116.061102}{\emph{Phys. Rev.
  Lett.} {\bfseries 116} (2016) 061102},
  [\href{https://arxiv.org/abs/1602.03837}{{\ttfamily 1602.03837}}].

\bibitem{Hehl:1994ue}
F.~W. Hehl, J.~D. McCrea, E.~W. Mielke and Y.~Ne'eman, \emph{{Metric affine
  gauge theory of gravity: Field equations, Noether identities, world spinors,
  and breaking of dilation invariance}},
  \href{http://dx.doi.org/10.1016/0370-1573(94)00111-F}{\emph{Phys. Rept.}
  {\bfseries 258} (1995) 1--171},
  [\href{https://arxiv.org/abs/gr-qc/9402012}{{\ttfamily gr-qc/9402012}}].

\bibitem{sciama1962analogy}
D.~W. Sciama, \emph{On the analogy between charge and spin in general
  relativity}, {\emph{Recent developments in general relativity} (1962) 415}.

\bibitem{Kibble:1961ba}
T.~W.~B. Kibble, \emph{{Lorentz invariance and the gravitational field}},
  \href{http://dx.doi.org/10.1063/1.1703702}{\emph{J. Math. Phys.} {\bfseries
  2} (1961) 212--221}.

\bibitem{Gauge}
M.~Blagojevi{\'c}, F.~W. Hehl and T.~W.~B. Kibble, \emph{Gauge Theories of
  Gravitation}.
\newblock Imperial College Press, 2013,
  \href{http://dx.doi.org/10.1142/p781}{10.1142/p781}.

\bibitem{Gauge2}
V.~N. Ponomarev, A.~Barvinsky and Y.~Obukhov, \emph{Gauge approach and
  quantization methods in gravity theory}.
\newblock Nuclear Safety Institute of the Russian Academy of Sciences, Nauka,
  2017.

\bibitem{Hawking:1973uf}
S.~W. Hawking and G.~F.~R. Ellis, \emph{{The Large Scale Structure of
  Space-Time}}.
\newblock Cambridge Monographs on Mathematical Physics. Cambridge University
  Press, 2011,
  \href{http://dx.doi.org/10.1017/CBO9780511524646}{10.1017/CBO9780511524646}.

\bibitem{sanchez2003introduccion}
M.~S{\'a}nchez and J.~Flores, \emph{Introducci{\'o}n a la geometr{\i}a
  diferencial de variedades}, {\emph{Universidad de Granada: Granada} (2003) }.

\bibitem{kobayashi1963foundations}
S.~Kobayashi and K.~Nomizu, \emph{Foundations of differential geometry},
  vol.~1.
\newblock New York, London, 1963.

\bibitem{sylvester1852xix}
J.~J. Sylvester, \emph{A demonstration of the theorem that every homogeneous
  quadratic polynomial is reducible by real orthogonal substitutions to the
  form of a sum of positive and negative squares}, {\emph{The London,
  Edinburgh, and Dublin Philosophical Magazine and Journal of Science}
  {\bfseries 4} (1852) 138--142}.

\bibitem{o1983semi}
B.~O'neill, \emph{Semi-Riemannian geometry with applications to relativity}.
\newblock Academic press, 1983.

\bibitem{BeltranJimenez:2019tjy}
J.~B. Jiménez, L.~Heisenberg and T.~S. Koivisto, \emph{{The Geometrical
  Trinity of Gravity}},
  \href{http://dx.doi.org/10.3390/universe5070173}{\emph{Universe} {\bfseries
  5} (2019) 173}, [\href{https://arxiv.org/abs/1903.06830}{{\ttfamily
  1903.06830}}].

\bibitem{hessenberg1917vektorielle}
G.~Hessenberg, \emph{Vektorielle begr{\"u}ndung der differentialgeometrie.},
  {\emph{Mathematische Annalen} {\bfseries 78} (1917) 187--217}.

\bibitem{DeAndrade:2000sf}
V.~C. De~Andrade, L.~C.~T. Guillen and J.~G. Pereira, \emph{{Teleparallel
  gravity: An Overview}},  in \emph{{Recent developments in theoretical and
  experimental general relativity, gravitation and relativistic field theories.
  Proceedings, 9th Marcel Grossmann Meeting, MG'9, Rome, Italy, July 2-8, 2000.
  Pts. A-C}}, 2000, \href{https://arxiv.org/abs/gr-qc/0011087}{{\ttfamily
  gr-qc/0011087}}.

\bibitem{Aldrovandi:2013wha}
R.~Aldrovandi and J.~G. Pereira, \emph{{Teleparallel Gravity}}, vol.~173.
\newblock Springer, Dordrecht, 2013,
  \href{http://dx.doi.org/10.1007/978-94-007-5143-9}{10.1007/978-94-007-5143-9}.

\bibitem{cartan1922generalisation}
{\'E}.~Cartan, \emph{Sur une g{\'e}n{\'e}ralisation de la notion de courbure de
  riemann et les espaces {\`a} torsion}, {\emph{Comptes Rendus, Ac. Sc. Paris}
  {\bfseries 174} (1922) 593--595}.

\bibitem{Weitzenbock:1923efa}
R.~Weitzenb{\"o}ck, \emph{{Invariantentheorie}}.
\newblock P. Noordhoff, Groningen, 1923.

\bibitem{cartan1979letters}
E.~Cartan and A.~Einstein, \emph{Letters on absolute parallelism: 1929-1932}.
\newblock Princeton University Press, 1979.

\bibitem{Nester:1998mp}
J.~M. Nester and H.-J. Yo, \emph{{Symmetric teleparallel general relativity}},
  {\emph{Chin. J. Phys.} {\bfseries 37} (1999) 113},
  [\href{https://arxiv.org/abs/gr-qc/9809049}{{\ttfamily gr-qc/9809049}}].

\bibitem{BeltranJimenez:2017tkd}
J.~Beltrán~Jiménez, L.~Heisenberg and T.~Koivisto, \emph{{Coincident General
  Relativity}}, \href{http://dx.doi.org/10.1103/PhysRevD.98.044048}{\emph{Phys.
  Rev.} {\bfseries D98} (2018) 044048},
  [\href{https://arxiv.org/abs/1710.03116}{{\ttfamily 1710.03116}}].

\bibitem{Jimenez:2019ghw}
J.~B. Jiménez, L.~Heisenberg, D.~Iosifidis, A.~Jiménez-Cano and T.~S.
  Koivisto, \emph{{General Teleparallel Quadratic Gravity}},
  \href{https://arxiv.org/abs/1909.09045}{{\ttfamily 1909.09045}}.

\bibitem{Mills:1989wj}
R.~Mills, \emph{{Gauge Fields}},
  \href{http://dx.doi.org/10.1119/1.15984}{\emph{Am.\ J.\ Phys.} {\bfseries 57}
  (1989) 493--507}.

\bibitem{dick1981emmy}
A.~Dick and H.~Weyl, \emph{Emmy Noether, 1882-1935}.
\newblock Springer, 1981.

\bibitem{Noether:1918zz}
E.~Noether, \emph{{Invariant Variation Problems}},
  \href{http://dx.doi.org/10.1080/00411457108231446}{\emph{Gott.\ Nachr.}
  {\bfseries 1918} (1918) 235--257},
  [\href{https://arxiv.org/abs/physics/0503066}{{\ttfamily physics/0503066}}].

\bibitem{Weyl:1918pdp}
H.~Weyl, \emph{{Reine Infinitesimalgeometrie}},
  \href{http://dx.doi.org/10.1007/BF01199420}{\emph{Math.\ Z.} {\bfseries 2}
  (1918) 384--411}.

\bibitem{Yang:1954ek}
C.-N. Yang and R.~L. Mills, \emph{{Conservation of Isotopic Spin and Isotopic
  Gauge Invariance}},
  \href{http://dx.doi.org/10.1103/PhysRev.96.191}{\emph{Phys.\ Rev.} {\bfseries
  96} (1954) 191--195}.

\bibitem{Weinberg:1979pi}
S.~Weinberg, \emph{{Conceptual Foundations of the Unified Theory of Weak and
  Electromagnetic Interactions}},
  \href{http://dx.doi.org/10.1103/RevModPhys.52.515}{\emph{Rev.\ Mod.\ Phys.}
  {\bfseries 52} (1980) 515--523}.

\bibitem{Gross:1973id}
D.~J. Gross and F.~Wilczek, \emph{{Ultraviolet Behavior of Nonabelian Gauge
  Theories}}, \href{http://dx.doi.org/10.1103/PhysRevLett.30.1343}{\emph{Phys.\
  Rev.\ Lett.} {\bfseries 30} (1973) 1343--1346}.

\bibitem{Tanabashi:2018oca}
{\scshape Particle Data Group} collaboration, M.~Tanabashi et~al.,
  \emph{{Review of Particle Physics}},
  \href{http://dx.doi.org/10.1103/PhysRevD.98.030001}{\emph{Phys.\ Rev.\ D}
  {\bfseries 98} (2018) 030001}.

\bibitem{Feynman:1996kb}
R.~Feynman, \emph{{Feynman lectures on gravitation}}.
\newblock 12, 1996.

\bibitem{Aldrovandi:1996ke}
R.~Aldrovandi and J.~Pereira, \emph{{An Introduction to geometrical physics}}.
\newblock 12, 1996.

\bibitem{Hehl:2019csx}
F.~W. Hehl and Y.~N. Obukhov, \emph{{Conservation of energy-momentum of matter
  as the basis for the gauge theory of gravitation}},
  \href{https://arxiv.org/abs/1909.01791}{{\ttfamily 1909.01791}}.

\bibitem{Colella:1975dq}
R.~Colella, A.~Overhauser and S.~Werner, \emph{{Observation of gravitationally
  induced quantum interference}},
  \href{http://dx.doi.org/10.1103/PhysRevLett.34.1472}{\emph{Phys.\ Rev.\
  Lett.} {\bfseries 34} (1975) 1472--1474}.

\bibitem{Kasevich:1991zz}
M.~Kasevich and S.~Chu, \emph{{Atomic interferometry using stimulated Raman
  transitions}},
  \href{http://dx.doi.org/10.1103/PhysRevLett.67.181}{\emph{Phys.\ Rev.\ Lett.}
  {\bfseries 67} (1991) 181--184}.

\bibitem{Asenbaum:2017rwf}
P.~Asenbaum, C.~Overstreet, T.~Kovachy, D.~D. Brown, J.~M. Hogan and M.~A.
  Kasevich, \emph{{Phase Shift in an Atom Interferometer due to Spacetime
  Curvature across its Wave Function}},
  \href{http://dx.doi.org/10.1103/PhysRevLett.118.183602}{\emph{Phys.\ Rev.\
  Lett.} {\bfseries 118} (2017) 183602},
  [\href{https://arxiv.org/abs/1610.03832}{{\ttfamily 1610.03832}}].

\bibitem{Overstreet:2017gdp}
C.~Overstreet, P.~Asenbaum, T.~Kovachy, R.~Notermans, J.~M. Hogan and M.~A.
  Kasevich, \emph{{Effective inertial frame in an atom interferometric test of
  the equivalence principle}},
  \href{http://dx.doi.org/10.1103/PhysRevLett.120.183604}{\emph{Phys.\ Rev.\
  Lett.} {\bfseries 120} (2018) 183604},
  [\href{https://arxiv.org/abs/1711.09986}{{\ttfamily 1711.09986}}].

\bibitem{Wigner:1939cj}
E.~P. Wigner, \emph{{On Unitary Representations of the Inhomogeneous Lorentz
  Group}}, \href{http://dx.doi.org/10.2307/1968551}{\emph{Annals Math.}
  {\bfseries 40} (1939) 149--204}.

\bibitem{Hayashi:1968hc}
K.~Hayashi, \emph{{Gauge theories of massive and massless tensor fields}},
  \href{http://dx.doi.org/10.1143/PTP.39.494}{\emph{Prog.\ Theor.\ Phys.}
  {\bfseries 39} (1968) 494--515}.

\bibitem{Hehl:1976kj}
F.~Hehl, P.~Von Der~Heyde, G.~Kerlick and J.~Nester, \emph{{General Relativity
  with Spin and Torsion: Foundations and Prospects}},
  \href{http://dx.doi.org/10.1103/RevModPhys.48.393}{\emph{Rev.\ Mod.\ Phys.}
  {\bfseries 48} (1976) 393--416}.

\bibitem{Blagojevic:2003cg}
M.~Blagojevic, \emph{{Three lectures on Poincare gauge theory}}, {\emph{SFIN A}
  {\bfseries 1} (2003) 147--172},
  [\href{https://arxiv.org/abs/gr-qc/0302040}{{\ttfamily gr-qc/0302040}}].

\bibitem{Shapiro:2001rz}
I.~L. Shapiro, \emph{{Physical aspects of the space-time torsion}},
  \href{http://dx.doi.org/10.1016/S0370-1573(01)00030-8}{\emph{Phys. Rept.}
  {\bfseries 357} (2002) 113},
  [\href{https://arxiv.org/abs/hep-th/0103093}{{\ttfamily hep-th/0103093}}].

\bibitem{Hayashi:1979wj}
K.~Hayashi and T.~Shirafuji, \emph{{Gravity from Poincare Gauge Theory of the
  Fundamental Particles. 1. Linear and Quadratic Lagrangians}},
  \href{http://dx.doi.org/10.1143/PTP.64.866}{\emph{Prog. Theor. Phys.}
  {\bfseries 64} (1980) 866}.

\bibitem{Hayashi:1980qp}
K.~Hayashi and T.~Shirafuji, \emph{{Gravity From Poincare Gauge Theory of the
  Fundamental Particles. 4. Mass and Energy of Particle Spectrum}},
  \href{http://dx.doi.org/10.1143/PTP.64.2222}{\emph{Prog.\ Theor.\ Phys.}
  {\bfseries 64} (1980) 2222}.

\bibitem{Sezgin:1979zf}
E.~Sezgin and P.~van Nieuwenhuizen, \emph{{New Ghost Free Gravity Lagrangians
  with Propagating Torsion}},
  \href{http://dx.doi.org/10.1103/PhysRevD.21.3269}{\emph{Phys. Rev.}
  {\bfseries D21} (1980) 3269}.

\bibitem{Sezgin:1981xs}
E.~Sezgin, \emph{{Class of Ghost Free Gravity Lagrangians With Massive or
  Massless Propagating Torsion}},
  \href{http://dx.doi.org/10.1103/PhysRevD.24.1677}{\emph{Phys. Rev.}
  {\bfseries D24} (1981) 1677--1680}.

\bibitem{Yo:1999ex}
H.-j. Yo and J.~M. Nester, \emph{{Hamiltonian analysis of Poincare gauge theory
  scalar modes}}, \href{http://dx.doi.org/10.1142/S021827189900033X}{\emph{Int.
  J. Mod. Phys.} {\bfseries D8} (1999) 459--479},
  [\href{https://arxiv.org/abs/gr-qc/9902032}{{\ttfamily gr-qc/9902032}}].

\bibitem{Yo:2001sy}
H.-J. Yo and J.~M. Nester, \emph{{Hamiltonian analysis of Poincare gauge
  theory: Higher spin modes}},
  \href{http://dx.doi.org/10.1142/S0218271802001998}{\emph{Int. J. Mod. Phys.}
  {\bfseries D11} (2002) 747--780},
  [\href{https://arxiv.org/abs/gr-qc/0112030}{{\ttfamily gr-qc/0112030}}].

\bibitem{Blagojevic:2018dpz}
M.~Blagojevi{\'c} and B.~Cvetkovi{\'c}, \emph{{General Poincar{\'e} gauge
  theory: Hamiltonian structure and particle spectrum}},
  \href{http://dx.doi.org/10.1103/PhysRevD.98.024014}{\emph{Phys. Rev.}
  {\bfseries D98} (2018) 024014},
  [\href{https://arxiv.org/abs/1804.05556}{{\ttfamily 1804.05556}}].

\bibitem{Sbisa:2014pzo}
F.~Sbisà, \emph{{Classical and quantum ghosts}},
  \href{http://dx.doi.org/10.1088/0143-0807/36/1/015009}{\emph{Eur.\ J.\ Phys.}
  {\bfseries 36} (2015) 015009},
  [\href{https://arxiv.org/abs/1406.4550}{{\ttfamily 1406.4550}}].

\bibitem{Ostrogradsky:1850fid}
M.~Ostrogradsky, \emph{{Mémoires sur les équations différentielles,
  relatives au problème des isopérimètres}}, {\emph{Mem. Acad. St.
  Petersbourg} {\bfseries 6} (1850) 385--517}.

\bibitem{Woodard:2006nt}
R.~P. Woodard, \emph{{Avoiding dark energy with 1/r modifications of gravity}},
  \href{http://dx.doi.org/10.1007/978-3-540-71013-4\_14}{\emph{Lect. Notes
  Phys.} {\bfseries 720} (2007) 403--433},
  [\href{https://arxiv.org/abs/astro-ph/0601672}{{\ttfamily
  astro-ph/0601672}}].

\bibitem{Boulware:1973my}
D.~Boulware and S.~Deser, \emph{{Can gravitation have a finite range?}},
  \href{http://dx.doi.org/10.1103/PhysRevD.6.3368}{\emph{Phys. Rev. D}
  {\bfseries 6} (1972) 3368--3382}.

\bibitem{Hassan:2011hr}
S.~Hassan and R.~A. Rosen, \emph{{Resolving the Ghost Problem in non-Linear
  Massive Gravity}},
  \href{http://dx.doi.org/10.1103/PhysRevLett.108.041101}{\emph{Phys. Rev.
  Lett.} {\bfseries 108} (2012) 041101},
  [\href{https://arxiv.org/abs/1106.3344}{{\ttfamily 1106.3344}}].

\bibitem{Charmousis:2009tc}
C.~Charmousis, G.~Niz, A.~Padilla and P.~M. Saffin, \emph{{Strong coupling in
  Horava gravity}},
  \href{http://dx.doi.org/10.1088/1126-6708/2009/08/070}{\emph{JHEP} {\bfseries
  08} (2009) 070}, [\href{https://arxiv.org/abs/0905.2579}{{\ttfamily
  0905.2579}}].

\bibitem{Deffayet:2005ys}
C.~Deffayet and J.-W. Rombouts, \emph{{Ghosts, strong coupling and accidental
  symmetries in massive gravity}},
  \href{http://dx.doi.org/10.1103/PhysRevD.72.044003}{\emph{Phys. Rev. D}
  {\bfseries 72} (2005) 044003},
  [\href{https://arxiv.org/abs/gr-qc/0505134}{{\ttfamily gr-qc/0505134}}].

\bibitem{Jimenez:2020ofm}
J.~Beltr\'an~Jim\'enez, A.~Golovnev, T.~Koivisto and H.~Veerm\"ae,
  \emph{{Minkowski space in $f(T)$ gravity}},
  \href{https://arxiv.org/abs/2004.07536}{{\ttfamily 2004.07536}}.

\bibitem{Luty:2003vm}
M.~A. Luty, M.~Porrati and R.~Rattazzi, \emph{{Strong interactions and
  stability in the DGP model}},
  \href{http://dx.doi.org/10.1088/1126-6708/2003/09/029}{\emph{JHEP} {\bfseries
  09} (2003) 029}, [\href{https://arxiv.org/abs/hep-th/0303116}{{\ttfamily
  hep-th/0303116}}].

\bibitem{Nicolis:2004qq}
A.~Nicolis and R.~Rattazzi, \emph{{Classical and quantum consistency of the DGP
  model}}, \href{http://dx.doi.org/10.1088/1126-6708/2004/06/059}{\emph{JHEP}
  {\bfseries 06} (2004) 059},
  [\href{https://arxiv.org/abs/hep-th/0404159}{{\ttfamily hep-th/0404159}}].

\bibitem{deRham:2012ew}
C.~de~Rham, G.~Gabadadze, L.~Heisenberg and D.~Pirtskhalava,
  \emph{{Nonrenormalization and naturalness in a class of scalar-tensor
  theories}}, \href{http://dx.doi.org/10.1103/PhysRevD.87.085017}{\emph{Phys.
  Rev. D} {\bfseries 87} (2013) 085017},
  [\href{https://arxiv.org/abs/1212.4128}{{\ttfamily 1212.4128}}].

\bibitem{Koyama:2013paa}
K.~Koyama, G.~Niz and G.~Tasinato, \emph{{Effective theory for the Vainshtein
  mechanism from the Horndeski action}},
  \href{http://dx.doi.org/10.1103/PhysRevD.88.021502}{\emph{Phys. Rev. D}
  {\bfseries 88} (2013) 021502},
  [\href{https://arxiv.org/abs/1305.0279}{{\ttfamily 1305.0279}}].

\bibitem{Brouzakis:2013lla}
N.~Brouzakis, A.~Codello, N.~Tetradis and O.~Zanusso, \emph{{Quantum
  corrections in Galileon theories}},
  \href{http://dx.doi.org/10.1103/PhysRevD.89.125017}{\emph{Phys. Rev. D}
  {\bfseries 89} (2014) 125017},
  [\href{https://arxiv.org/abs/1310.0187}{{\ttfamily 1310.0187}}].

\bibitem{Heisenberg:2014raa}
L.~Heisenberg, \emph{{Quantum Corrections in Galileons from Matter Loops}},
  \href{http://dx.doi.org/10.1103/PhysRevD.90.064005}{\emph{Phys. Rev. D}
  {\bfseries 90} (2014) 064005},
  [\href{https://arxiv.org/abs/1408.0267}{{\ttfamily 1408.0267}}].

\bibitem{Goon:2016ihr}
G.~Goon, K.~Hinterbichler, A.~Joyce and M.~Trodden, \emph{{Aspects of Galileon
  Non-Renormalization}},
  \href{http://dx.doi.org/10.1007/JHEP11(2016)100}{\emph{JHEP} {\bfseries 11}
  (2016) 100}, [\href{https://arxiv.org/abs/1606.02295}{{\ttfamily
  1606.02295}}].

\bibitem{Heisenberg:2020cyi}
L.~Heisenberg, J.~Noller and J.~Zosso, \emph{{Horndeski under the quantum
  loupe}}, \href{http://dx.doi.org/10.1088/1475-7516/2020/10/010}{\emph{JCAP}
  {\bfseries 10} (2020) 010},
  [\href{https://arxiv.org/abs/2004.11655}{{\ttfamily 2004.11655}}].

\bibitem{deRham:2013qqa}
C.~de~Rham, L.~Heisenberg and R.~H. Ribeiro, \emph{{Quantum Corrections in
  Massive Gravity}},
  \href{http://dx.doi.org/10.1103/PhysRevD.88.084058}{\emph{Phys. Rev. D}
  {\bfseries 88} (2013) 084058},
  [\href{https://arxiv.org/abs/1307.7169}{{\ttfamily 1307.7169}}].

\bibitem{Heisenberg:2020jtr}
L.~Heisenberg and J.~Zosso, \emph{{Quantum Stability of Generalized Proca
  Theories}},  \href{https://arxiv.org/abs/2005.01639}{{\ttfamily 2005.01639}}.

\bibitem{Jimenez:2015fva}
J.~Beltran~Jimenez and T.~S. Koivisto, \emph{{Spacetimes with vector
  distortion: Inflation from generalised Weyl geometry}},
  \href{http://dx.doi.org/10.1016/j.physletb.2016.03.047}{\emph{Phys. Lett.}
  {\bfseries B756} (2016) 400--404},
  [\href{https://arxiv.org/abs/1509.02476}{{\ttfamily 1509.02476}}].

\bibitem{Jimenez:2016opp}
J.~Beltran~Jimenez, L.~Heisenberg and T.~S. Koivisto, \emph{{Cosmology for
  quadratic gravity in generalized Weyl geometry}},
  \href{http://dx.doi.org/10.1088/1475-7516/2016/04/046}{\emph{JCAP} {\bfseries
  1604} (2016) 046}, [\href{https://arxiv.org/abs/1602.07287}{{\ttfamily
  1602.07287}}].

\bibitem{Himmetoglu:2008zp}
B.~Himmetoglu, C.~R. Contaldi and M.~Peloso, \emph{{Instability of anisotropic
  cosmological solutions supported by vector fields}},
  \href{http://dx.doi.org/10.1103/PhysRevLett.102.111301}{\emph{Phys. Rev.
  Lett.} {\bfseries 102} (2009) 111301},
  [\href{https://arxiv.org/abs/0809.2779}{{\ttfamily 0809.2779}}].

\bibitem{Jimenez:2008sq}
J.~Beltran~Jimenez and A.~L. Maroto, \emph{{Viability of vector-tensor theories
  of gravity}},
  \href{http://dx.doi.org/10.1088/1475-7516/2009/02/025}{\emph{JCAP} {\bfseries
  0902} (2009) 025}, [\href{https://arxiv.org/abs/0811.0784}{{\ttfamily
  0811.0784}}].

\bibitem{ArmendarizPicon:2009ai}
C.~Armendariz-Picon and A.~Diez-Tejedor, \emph{{Aether Unleashed}},
  \href{http://dx.doi.org/10.1088/1475-7516/2009/12/018}{\emph{JCAP} {\bfseries
  0912} (2009) 018}, [\href{https://arxiv.org/abs/0904.0809}{{\ttfamily
  0904.0809}}].

\bibitem{Himmetoglu:2009qi}
B.~Himmetoglu, C.~R. Contaldi and M.~Peloso, \emph{{Ghost instabilities of
  cosmological models with vector fields nonminimally coupled to the
  curvature}}, \href{http://dx.doi.org/10.1103/PhysRevD.80.123530}{\emph{Phys.
  Rev.} {\bfseries D80} (2009) 123530},
  [\href{https://arxiv.org/abs/0909.3524}{{\ttfamily 0909.3524}}].

\bibitem{Jimenez:2019hpl}
J.~Beltr\'an~Jim\'enez, C.~de~Rham and L.~Heisenberg, \emph{{Generalized Proca
  and its Constraint Algebra}},
  \href{http://dx.doi.org/10.1016/j.physletb.2020.135244}{\emph{Phys. Lett. B}
  {\bfseries 802} (2020) 135244},
  [\href{https://arxiv.org/abs/1906.04805}{{\ttfamily 1906.04805}}].

\bibitem{Motohashi:2016ftl}
H.~Motohashi, K.~Noui, T.~Suyama, M.~Yamaguchi and D.~Langlois, \emph{{Healthy
  degenerate theories with higher derivatives}},
  \href{http://dx.doi.org/10.1088/1475-7516/2016/07/033}{\emph{JCAP} {\bfseries
  07} (2016) 033}, [\href{https://arxiv.org/abs/1603.09355}{{\ttfamily
  1603.09355}}].

\bibitem{Klein:2016aiq}
R.~Klein and D.~Roest, \emph{{Exorcising the Ostrogradsky ghost in coupled
  systems}}, \href{http://dx.doi.org/10.1007/JHEP07(2016)130}{\emph{JHEP}
  {\bfseries 07} (2016) 130},
  [\href{https://arxiv.org/abs/1604.01719}{{\ttfamily 1604.01719}}].

\bibitem{deRham:2016wji}
C.~de~Rham and A.~Matas, \emph{{Ostrogradsky in Theories with Multiple
  Fields}}, \href{http://dx.doi.org/10.1088/1475-7516/2016/06/041}{\emph{JCAP}
  {\bfseries 06} (2016) 041},
  [\href{https://arxiv.org/abs/1604.08638}{{\ttfamily 1604.08638}}].

\bibitem{Hecht:1996np}
R.~D. Hecht, J.~M. Nester and V.~V. Zhytnikov, \emph{{Some Poincare gauge
  theory Lagrangians with well posed initial value problems}},
  \href{http://dx.doi.org/10.1016/0375-9601(96)00622-6}{\emph{Phys. Lett.}
  {\bfseries A222} (1996) 37--42}.

\bibitem{Ozkan:2015iva}
M.~Ozkan, Y.~Pang and S.~Tsujikawa, \emph{{Planck constraints on inflation in
  auxiliary vector modified $f(R)$ theories}},
  \href{http://dx.doi.org/10.1103/PhysRevD.92.023530}{\emph{Phys. Rev.}
  {\bfseries D92} (2015) 023530},
  [\href{https://arxiv.org/abs/1502.06341}{{\ttfamily 1502.06341}}].

\bibitem{Olmo:2011uz}
G.~J. Olmo, \emph{{Palatini Approach to Modified Gravity: f(R) Theories and
  Beyond}}, \href{http://dx.doi.org/10.1142/S0218271811018925}{\emph{Int. J.
  Mod. Phys.} {\bfseries D20} (2011) 413--462},
  [\href{https://arxiv.org/abs/1101.3864}{{\ttfamily 1101.3864}}].

\bibitem{Obukhov:1982zn}
{\relax Yu}.~N. Obukhov, \emph{{Conformal invariance and space-time torsion}},
  \href{http://dx.doi.org/10.1016/0375-9601(82)90037-8}{\emph{Phys. Lett.}
  {\bfseries A90} (1982) 13--16}.

\bibitem{HelayelNeto:1999tm}
J.~A. Helayel-Neto, A.~Penna-Firme and I.~L. Shapiro, \emph{{Conformal
  symmetry, anomaly and effective action for metric-scalar gravity with
  torsion}}, \href{http://dx.doi.org/10.1016/S0370-2693(00)00342-7}{\emph{Phys.
  Lett.} {\bfseries B479} (2000) 411--420},
  [\href{https://arxiv.org/abs/gr-qc/9907081}{{\ttfamily gr-qc/9907081}}].

\bibitem{Germani:2009iq}
C.~Germani and A.~Kehagias, \emph{{P-nflation: generating cosmic Inflation with
  p-forms}}, \href{http://dx.doi.org/10.1088/1475-7516/2009/03/028}{\emph{JCAP}
  {\bfseries 0903} (2009) 028},
  [\href{https://arxiv.org/abs/0902.3667}{{\ttfamily 0902.3667}}].

\bibitem{Koivisto:2009sd}
T.~S. Koivisto, D.~F. Mota and C.~Pitrou, \emph{{Inflation from N-Forms and its
  stability}},
  \href{http://dx.doi.org/10.1088/1126-6708/2009/09/092}{\emph{JHEP} {\bfseries
  09} (2009) 092}, [\href{https://arxiv.org/abs/0903.4158}{{\ttfamily
  0903.4158}}].

\bibitem{Koivisto:2009ew}
T.~S. Koivisto and N.~J. Nunes, \emph{{Three-form cosmology}},
  \href{http://dx.doi.org/10.1016/j.physletb.2010.01.051}{\emph{Phys. Lett.}
  {\bfseries B685} (2010) 105--109},
  [\href{https://arxiv.org/abs/0907.3883}{{\ttfamily 0907.3883}}].

\bibitem{Holst:1995pc}
S.~Holst, \emph{{Barbero's Hamiltonian derived from a generalized
  Hilbert-Palatini action}},
  \href{http://dx.doi.org/10.1103/PhysRevD.53.5966}{\emph{Phys. Rev.}
  {\bfseries D53} (1996) 5966--5969},
  [\href{https://arxiv.org/abs/gr-qc/9511026}{{\ttfamily gr-qc/9511026}}].

\bibitem{Hojman:1980kv}
R.~Hojman, C.~Mukku and W.~Sayed, \emph{{Parity violation in metric torsion
  theories of gravitation}},
  \href{http://dx.doi.org/10.1103/PhysRevD.22.1915}{\emph{Phys. Rev. D}
  {\bfseries 22} (1980) 1915--1921}.

\bibitem{Taveras:2008yf}
V.~Taveras and N.~Yunes, \emph{{The Barbero-Immirzi Parameter as a Scalar
  Field: K-Inflation from Loop Quantum Gravity?}},
  \href{http://dx.doi.org/10.1103/PhysRevD.78.064070}{\emph{Phys. Rev.}
  {\bfseries D78} (2008) 064070},
  [\href{https://arxiv.org/abs/0807.2652}{{\ttfamily 0807.2652}}].

\bibitem{Calcagni:2009xz}
G.~Calcagni and S.~Mercuri, \emph{{The Barbero-Immirzi field in canonical
  formalism of pure gravity}},
  \href{http://dx.doi.org/10.1103/PhysRevD.79.084004}{\emph{Phys. Rev.}
  {\bfseries D79} (2009) 084004},
  [\href{https://arxiv.org/abs/0902.0957}{{\ttfamily 0902.0957}}].

\bibitem{Kobayashi:2011nu}
T.~Kobayashi, M.~Yamaguchi and J.~Yokoyama, \emph{{Generalized G-inflation:
  Inflation with the most general second-order field equations}},
  \href{http://dx.doi.org/10.1143/PTP.126.511}{\emph{Prog. Theor. Phys.}
  {\bfseries 126} (2011) 511--529},
  [\href{https://arxiv.org/abs/1105.5723}{{\ttfamily 1105.5723}}].

\bibitem{Turner:1983he}
M.~S. Turner, \emph{{Coherent Scalar Field Oscillations in an Expanding
  Universe}}, \href{http://dx.doi.org/10.1103/PhysRevD.28.1243}{\emph{Phys.
  Rev.} {\bfseries D28} (1983) 1243}.

\bibitem{PhysRevLett.64.1084}
W.~H. Press, B.~S. Ryden and D.~N. Spergel, \emph{Single mechanism for
  generating large-scale structure and providing dark missing matter},
  \href{http://dx.doi.org/10.1103/PhysRevLett.64.1084}{\emph{Phys. Rev. Lett.}
  {\bfseries 64} (Mar, 1990) 1084--1087}.

\bibitem{Cembranos:2015oya}
J.~A.~R. Cembranos, A.~L. Maroto and S.~J. N{\'u}{\~n}ez~Jare{\~n}o,
  \emph{{Cosmological perturbations in coherent oscillating scalar field
  models}}, \href{http://dx.doi.org/10.1007/JHEP03(2016)013}{\emph{JHEP}
  {\bfseries 03} (2016) 013},
  [\href{https://arxiv.org/abs/1509.08819}{{\ttfamily 1509.08819}}].

\bibitem{Hui:2016ltb}
L.~Hui, J.~P. Ostriker, S.~Tremaine and E.~Witten, \emph{{Ultralight scalars as
  cosmological dark matter}},
  \href{http://dx.doi.org/10.1103/PhysRevD.95.043541}{\emph{Phys. Rev.}
  {\bfseries D95} (2017) 043541},
  [\href{https://arxiv.org/abs/1610.08297}{{\ttfamily 1610.08297}}].

\bibitem{Marsh:2015xka}
D.~J.~E. Marsh, \emph{{Axion Cosmology}},
  \href{http://dx.doi.org/10.1016/j.physrep.2016.06.005}{\emph{Phys. Rept.}
  {\bfseries 643} (2016) 1--79},
  [\href{https://arxiv.org/abs/1510.07633}{{\ttfamily 1510.07633}}].

\bibitem{Hu:2000ke}
W.~Hu, R.~Barkana and A.~Gruzinov, \emph{{Cold and fuzzy dark matter}},
  \href{http://dx.doi.org/10.1103/PhysRevLett.85.1158}{\emph{Phys. Rev. Lett.}
  {\bfseries 85} (2000) 1158--1161},
  [\href{https://arxiv.org/abs/astro-ph/0003365}{{\ttfamily
  astro-ph/0003365}}].

\bibitem{Cembranos:2008gj}
J.~A.~R. Cembranos, \emph{{Dark Matter from R2-gravity}},
  \href{http://dx.doi.org/10.1103/PhysRevLett.102.141301}{\emph{Phys. Rev.
  Lett.} {\bfseries 102} (2009) 141301},
  [\href{https://arxiv.org/abs/0809.1653}{{\ttfamily 0809.1653}}].

\bibitem{Hammond:2002rm}
R.~T. Hammond, \emph{{Torsion gravity}},
  \href{http://dx.doi.org/10.1088/0034-4885/65/5/201}{\emph{Rept. Prog. Phys.}
  {\bfseries 65} (2002) 599--649}.

\bibitem{GellMann:1960np}
M.~Gell-Mann and M.~Levy, \emph{{The axial vector current in beta decay}},
  \href{http://dx.doi.org/10.1007/BF02859738}{\emph{Nuovo Cim.} {\bfseries 16}
  (1960) 705}.

\bibitem{Stewart:1973ux}
J.~Stewart and P.~Hajicek, \emph{Can spin avert singularities?},
  \href{http://dx.doi.org/10.1038/244096a0}{\emph{Nature} {\bfseries 244}
  (1973) 96}.

\bibitem{Trautman:1973wy}
A.~Trautman, \emph{{Spin and torsion may avert gravitational singularities}},
  {\emph{Nature} {\bfseries 242} (1973) 7--8}.

\bibitem{Cembranos:2016xqx}
J.~A.~R. Cembranos, J.~Gigante~Valcarcel and F.~J. Maldonado~Torralba,
  \emph{{Singularities and n-dimensional black holes in torsion theories}},
  \href{http://dx.doi.org/10.1088/1475-7516/2017/04/021}{\emph{JCAP} {\bfseries
  1704} (2017) 021}, [\href{https://arxiv.org/abs/1609.07814}{{\ttfamily
  1609.07814}}].

\bibitem{Cembranos:2019mcb}
J.~A.~R. Cembranos, J.~G. Valcarcel and F.~J. Maldonado~Torralba,
  \emph{{Non-Geodesic Incompleteness in Poincar{\'e} Gauge Gravity}},
  \href{http://dx.doi.org/10.3390/e21030280}{\emph{Entropy} {\bfseries 21}
  (2019) 280}, [\href{https://arxiv.org/abs/1901.09899}{{\ttfamily
  1901.09899}}].

\bibitem{delaCruz-Dombriz:2018aal}
A.~de~la Cruz-Dombriz, F.~J. Maldonado~Torralba and A.~Mazumdar,
  \emph{{Nonsingular and ghost-free infinite derivative gravity with torsion}},
  \href{http://dx.doi.org/10.1103/PhysRevD.99.104021}{\emph{Phys. Rev.}
  {\bfseries D99} (2019) 104021},
  [\href{https://arxiv.org/abs/1812.04037}{{\ttfamily 1812.04037}}].

\bibitem{Neville:1979fk}
D.~E. Neville, \emph{{Birkhoff Theorems for $R+R^2$ Gravity Theories With
  Torsion}}, \href{http://dx.doi.org/10.1103/PhysRevD.21.2770}{\emph{Phys.
  Rev.} {\bfseries D21} (1980) 2770}.

\bibitem{Rauch:1981tva}
R.~Rauch and H.~T. Nieh, \emph{{Birkhoff's Theorem for General {Riemann-Cartan}
  Type $R+R^2$ Theories of Gravity}},
  \href{http://dx.doi.org/10.1103/PhysRevD.24.2029}{\emph{Phys. Rev.}
  {\bfseries D24} (1981) 2029}.

\bibitem{delaCruz-Dombriz:2018vzn}
A.~de~la Cruz-Dombriz and F.~J. Maldonado~Torralba, \emph{{Birkhoff's theorem
  for stable torsion theories}},
  \href{http://dx.doi.org/10.1088/1475-7516/2019/03/002}{\emph{JCAP} {\bfseries
  1903} (2019) 002}, [\href{https://arxiv.org/abs/1811.11021}{{\ttfamily
  1811.11021}}].

\bibitem{Bakler:1984cq}
P.~Bakler and F.~W. Hehl, \emph{{A charged TauB - But metric with torsion: A
  new axially symmetric solutions of the Poincar\'e gauge field theory}},
  \href{http://dx.doi.org/10.1016/0375-9601(84)90627-3}{\emph{Phys. Lett.}
  {\bfseries A100} (1984) 392--396}.

\bibitem{Obukhov:1987tz}
{\relax Yu}.~N. Obukhov, V.~N. Ponomarev and V.~V. Zhytnikov, \emph{{Quadratic
  Poincare Gauge Theory of Gravity: A Comparison With the General Relativity
  Theory}}, \href{http://dx.doi.org/10.1007/BF00763457}{\emph{Gen. Rel. Grav.}
  {\bfseries 21} (1989) 1107--1142}.

\bibitem{Blagojevic:2015zma}
M.~Blagojevi{\'c} and B.~Cvetkovi{\'c}, \emph{{Conformally flat black holes in
  Poincar{\'e} gauge theory}},
  \href{http://dx.doi.org/10.1103/PhysRevD.93.044018}{\emph{Phys. Rev.}
  {\bfseries D93} (2016) 044018},
  [\href{https://arxiv.org/abs/1510.00069}{{\ttfamily 1510.00069}}].

\bibitem{Cembranos:2016gdt}
J.~A.~R. Cembranos and J.~G. Valcarcel, \emph{{New torsion black hole solutions
  in Poincar{\'e} gauge theory}},
  \href{http://dx.doi.org/10.1088/1475-7516/2017/01/014}{\emph{JCAP} {\bfseries
  1701} (2017) 014}, [\href{https://arxiv.org/abs/1608.00062}{{\ttfamily
  1608.00062}}].

\bibitem{Obukhov:2019fti}
Y.~N. Obukhov, \emph{{Exact Solutions in Poincar{\'e} Gauge Gravity Theory}},
  \href{http://dx.doi.org/10.3390/universe5050127}{\emph{Universe} {\bfseries
  5} (2019) 127}, [\href{https://arxiv.org/abs/1905.11906}{{\ttfamily
  1905.11906}}].

\bibitem{Ziaie:2019dmq}
A.~H. Ziaie, \emph{{Wormholes in Poincar\'e gauge theory of gravity}},
  \href{https://arxiv.org/abs/1910.01904}{{\ttfamily 1910.01904}}.

\bibitem{Kerlick:1975tr}
G.~D. Kerlick, \emph{{Cosmology and Particle Pair Production via Gravitational
  Spin Spin Interaction in the Einstein-Cartan-Sciama-Kibble Theory of
  Gravity}}, \href{http://dx.doi.org/10.1103/PhysRevD.12.3004}{\emph{Phys.
  Rev.} {\bfseries D12} (1975) 3004--3006}.

\bibitem{Yo:2006qs}
H.-J. Yo and J.~M. Nester, \emph{{Dynamic Scalar Torsion and an Oscillating
  Universe}}, \href{http://dx.doi.org/10.1142/S0217732307025303}{\emph{Mod.
  Phys. Lett.} {\bfseries A22} (2007) 2057--2069},
  [\href{https://arxiv.org/abs/astro-ph/0612738}{{\ttfamily
  astro-ph/0612738}}].

\bibitem{Shie:2008ms}
K.-F. Shie, J.~M. Nester and H.-J. Yo, \emph{{Torsion Cosmology and the
  Accelerating Universe}},
  \href{http://dx.doi.org/10.1103/PhysRevD.78.023522}{\emph{Phys. Rev.}
  {\bfseries D78} (2008) 023522},
  [\href{https://arxiv.org/abs/0805.3834}{{\ttfamily 0805.3834}}].

\bibitem{Chen:2009at}
H.~Chen, F.-H. Ho, J.~M. Nester, C.-H. Wang and H.-J. Yo, \emph{{Cosmological
  dynamics with propagating Lorentz connection modes of spin zero}},
  \href{http://dx.doi.org/10.1088/1475-7516/2009/10/027}{\emph{JCAP} {\bfseries
  0910} (2009) 027}, [\href{https://arxiv.org/abs/0908.3323}{{\ttfamily
  0908.3323}}].

\bibitem{Baekler:2010fr}
P.~Baekler, F.~W. Hehl and J.~M. Nester, \emph{{Poincare gauge theory of
  gravity: Friedman cosmology with even and odd parity modes. Analytic part}},
  \href{http://dx.doi.org/10.1103/PhysRevD.83.024001}{\emph{Phys. Rev.}
  {\bfseries D83} (2011) 024001},
  [\href{https://arxiv.org/abs/1009.5112}{{\ttfamily 1009.5112}}].

\bibitem{Ho:2015ulu}
F.-H. Ho, H.~Chen, J.~M. Nester and H.-J. Yo, \emph{{General Poincaré Gauge
  Theory Cosmology}}, \href{http://dx.doi.org/10.6122/CJP.20151014}{\emph{Chin.
  J. Phys.} {\bfseries 53} (2015) 110109},
  [\href{https://arxiv.org/abs/1512.01202}{{\ttfamily 1512.01202}}].

\bibitem{Hehl:2013qga}
F.~W. Hehl, Y.~N. Obukhov and D.~Puetzfeld, \emph{{On Poincar\'e gauge theory
  of gravity, its equations of motion, and Gravity Probe B}},
  \href{http://dx.doi.org/10.1016/j.physleta.2013.04.055}{\emph{Phys. Lett.}
  {\bfseries A377} (2013) 1775--1781},
  [\href{https://arxiv.org/abs/1304.2769}{{\ttfamily 1304.2769}}].

\bibitem{Cembranos:2018ipn}
J.~A.~R. Cembranos, J.~G. Valcarcel and F.~J. Maldonado~Torralba,
  \emph{{Fermion dynamics in torsion theories}},
  \href{http://dx.doi.org/10.1088/1475-7516/2019/04/039}{\emph{JCAP} {\bfseries
  1904} (2019) 039}, [\href{https://arxiv.org/abs/1805.09577}{{\ttfamily
  1805.09577}}].

\bibitem{Chern:1992um}
D.-C. Chern, J.~M. Nester and H.-J. Yo, \emph{{Positive energy test of Poincare
  gauge theory}}, \href{http://dx.doi.org/10.1142/S0217751X92000879}{\emph{Int.
  J. Mod. Phys.} {\bfseries A7} (1992) 1993--2003}.

\bibitem{Vasilev:2017twr}
T.~B. Vasilev, J.~A.~R. Cembranos, J.~G. Valcarcel and P.~Mart{\'\i}n-Moruno,
  \emph{{Stability in quadratic torsion theories}},
  \href{http://dx.doi.org/10.1140/epjc/s10052-017-5331-6}{\emph{Eur. Phys. J.}
  {\bfseries C77} (2017) 755},
  [\href{https://arxiv.org/abs/1706.07080}{{\ttfamily 1706.07080}}].

\bibitem{Jimenez:2019qjc}
J.~B. Jiménez and F.~J. Maldonado~Torralba, \emph{{Revisiting the stability of
  quadratic Poincaré gauge gravity}},
  \href{http://dx.doi.org/10.1140/epjc/s10052-020-8163-8}{\emph{Eur. Phys. J.
  C} {\bfseries 80} (2020) 611},
  [\href{https://arxiv.org/abs/1910.07506}{{\ttfamily 1910.07506}}].

\bibitem{Ponomarev:1971zz}
V.~Ponomariev, \emph{Observable effects of torsion in space-time.},
  {\emph{BAPSS} {\bfseries 19} (1971) 545--550}.

\bibitem{Kleinert:1998cz}
H.~Kleinert, \emph{{Universality principle for orbital angular momentum and
  spin in gravity with torsion}},
  \href{http://dx.doi.org/10.1023/A:1001990604209}{\emph{Gen. Rel. Grav.}
  {\bfseries 32} (2000) 1271--1280},
  [\href{https://arxiv.org/abs/gr-qc/9807021}{{\ttfamily gr-qc/9807021}}].

\bibitem{Mao:2006bb}
Y.~Mao, M.~Tegmark, A.~H. Guth and S.~Cabi, \emph{{Constraining Torsion with
  Gravity Probe B}},
  \href{http://dx.doi.org/10.1103/PhysRevD.76.104029}{\emph{Phys. Rev. D}
  {\bfseries 76} (2007) 104029},
  [\href{https://arxiv.org/abs/gr-qc/0608121}{{\ttfamily gr-qc/0608121}}].

\bibitem{HEHL1971225}
F.~Hehl, \emph{How does one measure torsion of space-time?},
  \href{http://dx.doi.org/https://doi.org/10.1016/0375-9601(71)90433-6}{\emph{Physics
  Letters A} {\bfseries 36} (1971) 225 -- 226}.

\bibitem{Audretsch:1981xn}
J.~Audretsch, \emph{Dirac electron in space-times with torsion: Spinor
  propagation, spin precession, and nongeodesic orbits},
  \href{http://dx.doi.org/10.1103/PhysRevD.24.1470}{\emph{Phys.Rev.D}
  {\bfseries 24} (1981) 1470--1477}.

\bibitem{bergmann2012cosmology}
P.~G. Bergmann and V.~De~Sabbata, \emph{Cosmology and Gravitation: Spin,
  Torsion, Rotation, and Supergravity}, vol.~58.
\newblock Springer Science \& Business Media, 2012.

\bibitem{nomura1991spinning}
K.~Nomura, T.~Shirafuji and K.~Hayashi, \emph{Spinning test particles in
  spacetime with torsion}, {\emph{Progress of theoretical physics} {\bfseries
  86} (1991) 1239--1258}.

\bibitem{Cembranos:2017pcs}
J.~A. Cembranos and J.~Gigante~Valcarcel, \emph{{Extended Reissner--Nordström
  solutions sourced by dynamical torsion}},
  \href{http://dx.doi.org/10.1016/j.physletb.2018.01.081}{\emph{Phys. Lett. B}
  {\bfseries 779} (2018) 143--150},
  [\href{https://arxiv.org/abs/1708.00374}{{\ttfamily 1708.00374}}].

\bibitem{Alsing:2009px}
P.~Alsing, J.~Stephenson, G.J. and P.~Kilian, \emph{{Spin-induced non-geodesic
  motion, gyroscopic precession, Wigner rotation and EPR correlations of
  massive spin-1/2 particles in a gravitational field}},
  \href{https://arxiv.org/abs/0902.1396}{{\ttfamily 0902.1396}}.

\bibitem{Pirani:1956tn}
F.~Pirani, \emph{On the physical significance of the riemann tensor},
  \href{http://dx.doi.org/10.1007/s10714-009-0787-9}{\emph{Acta Phys.Polon.}
  {\bfseries 15} (1956) 389--405}.

\bibitem{Penrose:1964wq}
R.~Penrose, \emph{Gravitational collapse and space-time singularities},
  \href{http://dx.doi.org/10.1103/PhysRevLett.14.57}{\emph{Phys.Rev.Lett.}
  {\bfseries 14} (1965) 57--59}.

\bibitem{Frauendiener:2000mk}
J.~Frauendiener, \emph{Conformal infinity}, {\emph{Living Rev.Rel.} {\bfseries
  3} (2000) 4}.

\bibitem{Galloway:2010tx}
G.~J. Galloway and J.~M. Senovilla, \emph{{Singularity theorems based on
  trapped submanifolds of arbitrary co-dimension}},
  \href{http://dx.doi.org/10.1088/0264-9381/27/15/152002}{\emph{Class. Quant.
  Grav.} {\bfseries 27} (2010) 152002},
  [\href{https://arxiv.org/abs/1005.1249}{{\ttfamily 1005.1249}}].

\bibitem{Poplawski:2012ab}
N.~J. Poplawski, \emph{Big bounce from spin and torsion},
  \href{http://dx.doi.org/10.1007/s10714-011-1323-2}{\emph{Gen.Rel.Grav.}
  {\bfseries 44} (2012) 1007--1014},
  [\href{https://arxiv.org/abs/1105.6127}{{\ttfamily 1105.6127}}].

\bibitem{Lucat:2015rla}
S.~Lucat and T.~Prokopec, \emph{Cosmological singularities and bounce in
  cartan-einstein theory},
  \href{http://dx.doi.org/10.1088/1475-7516/2017/10/047}{\emph{JCAP} {\bfseries
  10} (2017) 047}, [\href{https://arxiv.org/abs/1512.06074}{{\ttfamily
  1512.06074}}].

\bibitem{birkhoff1923boundary}
G.~D. Birkhoff and R.~E. Langer, \emph{The boundary problems and developments
  associated with a system of ordinary linear differential equations of the
  first order},  in \emph{Proceedings of the American Academy of Arts and
  Sciences}, vol.~58, pp.~51--128, JSTOR, 1923.

\bibitem{Barack:2018yly}
L.~Barack et~al., \emph{{Black holes, gravitational waves and fundamental
  physics: a roadmap}},
  \href{http://dx.doi.org/10.1088/1361-6382/ab0587}{\emph{Class. Quant. Grav.}
  {\bfseries 36} (2019) 143001},
  [\href{https://arxiv.org/abs/1806.05195}{{\ttfamily 1806.05195}}].

\bibitem{Clifton:2006ug}
T.~Clifton, \emph{{Spherically Symmetric Solutions to Fourth-Order Theories of
  Gravity}}, \href{http://dx.doi.org/10.1088/0264-9381/23/24/015}{\emph{Class.
  Quant. Grav.} {\bfseries 23} (2006) 7445},
  [\href{https://arxiv.org/abs/gr-qc/0607096}{{\ttfamily gr-qc/0607096}}].

\bibitem{Nzioki:2009av}
A.~M. Nzioki, S.~Carloni, R.~Goswami and P.~K. Dunsby, \emph{{A New framework
  for studying spherically symmetric static solutions in f(R) gravity}},
  \href{http://dx.doi.org/10.1103/PhysRevD.81.084028}{\emph{Phys. Rev. D}
  {\bfseries 81} (2010) 084028},
  [\href{https://arxiv.org/abs/0908.3333}{{\ttfamily 0908.3333}}].

\bibitem{delaCruzDombriz:2009et}
A.~de~la Cruz-Dombriz, A.~Dobado and A.~Maroto, \emph{{Black Holes in f(R)
  theories}}, \href{http://dx.doi.org/10.1103/PhysRevD.80.124011}{\emph{Phys.
  Rev. D} {\bfseries 80} (2009) 124011},
  [\href{https://arxiv.org/abs/0907.3872}{{\ttfamily 0907.3872}}].

\bibitem{Ferraro:2011ks}
R.~Ferraro and F.~Fiorini, \emph{{Spherically symmetric static spacetimes in
  vacuum f(T) gravity}},
  \href{http://dx.doi.org/10.1103/PhysRevD.84.083518}{\emph{Phys. Rev. D}
  {\bfseries 84} (2011) 083518},
  [\href{https://arxiv.org/abs/1109.4209}{{\ttfamily 1109.4209}}].

\bibitem{Hui:2012qt}
L.~Hui and A.~Nicolis, \emph{{No-Hair Theorem for the Galileon}},
  \href{http://dx.doi.org/10.1103/PhysRevLett.110.241104}{\emph{Phys. Rev.
  Lett.} {\bfseries 110} (2013) 241104},
  [\href{https://arxiv.org/abs/1202.1296}{{\ttfamily 1202.1296}}].

\bibitem{Sotiriou:2013qea}
T.~P. Sotiriou and S.-Y. Zhou, \emph{{Black hole hair in generalized
  scalar-tensor gravity}},
  \href{http://dx.doi.org/10.1103/PhysRevLett.112.251102}{\emph{Phys. Rev.
  Lett.} {\bfseries 112} (2014) 251102},
  [\href{https://arxiv.org/abs/1312.3622}{{\ttfamily 1312.3622}}].

\bibitem{Ramaswamy:1979zz}
S.~Ramaswamy and P.~B. Yasskin, \emph{{Birkhoff theorem for an R+R2 theory of
  gravity with torsion}},
  \href{http://dx.doi.org/10.1103/PhysRevD.19.2264}{\emph{Phys. Rev. D}
  {\bfseries 19} (1979) 2264--2267}.

\bibitem{Rauch:1980qj}
R.~Rauch, J.~C. Shaw and H.~Nieh, \emph{{Birkhoff's Theorem for Ghost Free,
  Tachyon Free $R+R^2$ Theories With Torsion}},
  \href{http://dx.doi.org/10.1007/BF00756268}{\emph{Gen. Rel. Grav.} {\bfseries
  14} (1982) 331}.

\bibitem{Bakler:1980is}
P.~Bakler and P.~B. Yasskin, \emph{{All Torsion Free Spherical Vacuum Solutions
  of the Quadratic Poincare Gauge Theory of Gravity}},
  \href{http://dx.doi.org/10.1007/BF00760237}{\emph{Gen. Rel. Grav.} {\bfseries
  16} (1984) 1135}.

\bibitem{Rauch:1981ua}
R.~T. Rauch, \emph{{Asymptotic flatness, reflection symmetry, and Birkhoff's
  theorem for $R+R^{2}$ actions containing quadratic torsion terms}},
  \href{http://dx.doi.org/10.1103/PhysRevD.25.577}{\emph{Phys. Rev. D}
  {\bfseries 25} (1982) 577}.

\bibitem{1950SRToh..34..160N}
H.~{Nariai}, \emph{{On some static solutions of Einstein's gravitational field
  equations in a spherically symmetric case}}, {\emph{Sci. Rep. Tohoku Univ.
  Eighth Ser.} {\bfseries 34} (Jan., 1950) 160}.

\bibitem{Misner:1974qy}
C.~W. Misner, K.~Thorne and J.~Wheeler, \emph{{Gravitation}}.
\newblock W. H. Freeman, San Francisco, 1973.

\bibitem{Sotiriou:2011dz}
T.~P. Sotiriou and V.~Faraoni, \emph{{Black holes in scalar-tensor gravity}},
  \href{http://dx.doi.org/10.1103/PhysRevLett.108.081103}{\emph{Phys. Rev.
  Lett.} {\bfseries 108} (2012) 081103},
  [\href{https://arxiv.org/abs/1109.6324}{{\ttfamily 1109.6324}}].

\bibitem{Dolgov:2003px}
A.~Dolgov and M.~Kawasaki, \emph{{Can modified gravity explain accelerated
  cosmic expansion?}},
  \href{http://dx.doi.org/10.1016/j.physletb.2003.08.039}{\emph{Phys. Lett. B}
  {\bfseries 573} (2003) 1--4},
  [\href{https://arxiv.org/abs/astro-ph/0307285}{{\ttfamily
  astro-ph/0307285}}].

\bibitem{Seifert:2007fr}
M.~D. Seifert, \emph{{Stability of spherically symmetric solutions in modified
  theories of gravity}},
  \href{http://dx.doi.org/10.1103/PhysRevD.76.064002}{\emph{Phys. Rev. D}
  {\bfseries 76} (2007) 064002},
  [\href{https://arxiv.org/abs/gr-qc/0703060}{{\ttfamily gr-qc/0703060}}].

\bibitem{Sur:2013aia}
S.~Sur and A.~S. Bhatia, \emph{{Constraining Torsion in Maximally symmetric
  (sub)spaces}},
  \href{http://dx.doi.org/10.1088/0264-9381/31/2/025020}{\emph{Class. Quant.
  Grav.} {\bfseries 31} (2014) 025020},
  [\href{https://arxiv.org/abs/1306.0394}{{\ttfamily 1306.0394}}].

\bibitem{Lammerzahl:1997wk}
C.~Lammerzahl, \emph{{Constraints on space-time torsion from Hughes-Drever
  experiments}},
  \href{http://dx.doi.org/10.1016/S0375-9601(97)00127-8}{\emph{Phys. Lett.}
  {\bfseries A228} (1997) 223},
  [\href{https://arxiv.org/abs/gr-qc/9704047}{{\ttfamily gr-qc/9704047}}].

\bibitem{braun1983differential}
M.~Braun and M.~Golubitsky, \emph{Differential equations and their
  applications}, vol.~1.
\newblock Springer, 1983.

\bibitem{Goswami:2011ft}
R.~Goswami and G.~F.~R. Ellis, \emph{{Almost Birkhoff Theorem in General
  Relativity}}, \href{http://dx.doi.org/10.1007/s10714-011-1172-z}{\emph{Gen.
  Rel. Grav.} {\bfseries 43} (2011) 2157--2170},
  [\href{https://arxiv.org/abs/1101.4520}{{\ttfamily 1101.4520}}].

\bibitem{Witten:1985cc}
E.~Witten, \emph{{Noncommutative Geometry and String Field Theory}},
  \href{http://dx.doi.org/10.1016/0550-3213(86)90155-0}{\emph{Nucl. Phys. B}
  {\bfseries 268} (1986) 253--294}.

\bibitem{Brekke:1988dg}
L.~Brekke, P.~G. Freund, M.~Olson and E.~Witten, \emph{{Nonarchimedean String
  Dynamics}}, \href{http://dx.doi.org/10.1016/0550-3213(88)90207-6}{\emph{Nucl.
  Phys. B} {\bfseries 302} (1988) 365--402}.

\bibitem{Douglas:1989ve}
M.~R. Douglas and S.~H. Shenker, \emph{{Strings in Less Than One-Dimension}},
  \href{http://dx.doi.org/10.1016/0550-3213(90)90522-F}{\emph{Nucl. Phys. B}
  {\bfseries 335} (1990) 635}.

\bibitem{Biswas:2014tua}
T.~Biswas and S.~Talaganis, \emph{{String-Inspired Infinite-Derivative Theories
  of Gravity: A Brief Overview}},
  \href{http://dx.doi.org/10.1142/S021773231540009X}{\emph{Mod. Phys. Lett.}
  {\bfseries A30} (2015) 1540009},
  [\href{https://arxiv.org/abs/1412.4256}{{\ttfamily 1412.4256}}].

\bibitem{Barnaby:2010kx}
N.~Barnaby, \emph{{A New Formulation of the Initial Value Problem for Nonlocal
  Theories}},
  \href{http://dx.doi.org/10.1016/j.nuclphysb.2010.11.016}{\emph{Nucl. Phys. B}
  {\bfseries 845} (2011) 1--29},
  [\href{https://arxiv.org/abs/1005.2945}{{\ttfamily 1005.2945}}].

\bibitem{Barnaby:2007ve}
N.~Barnaby and N.~Kamran, \emph{{Dynamics with infinitely many derivatives: The
  Initial value problem}},
  \href{http://dx.doi.org/10.1088/1126-6708/2008/02/008}{\emph{JHEP} {\bfseries
  02} (2008) 008}, [\href{https://arxiv.org/abs/0709.3968}{{\ttfamily
  0709.3968}}].

\bibitem{Calcagni:2018lyd}
G.~Calcagni, L.~Modesto and G.~Nardelli, \emph{{Initial conditions and degrees
  of freedom of non-local gravity}},
  \href{http://dx.doi.org/10.1007/JHEP05(2018)087,
  10.1007/JHEP05(2019)095}{\emph{JHEP} {\bfseries 05} (2018) 087},
  [\href{https://arxiv.org/abs/1803.00561}{{\ttfamily 1803.00561}}].

\bibitem{davis1936theory}
H.~T. Davis, \emph{The theory of linear operators from the standpoint of
  differential equations of infinite order}, {\emph{Bloomington, Ind} (1936) }.

\bibitem{carmichael1936linear}
R.~Carmichael, \emph{Linear differential equations of infinite order},
  {\emph{Bulletin of the American Mathematical Society} {\bfseries 42} (1936)
  193--218}.

\bibitem{carleson1953infinite}
L.~Carleson, \emph{On infinite differential equations with constant
  coefficients. i}, {\emph{Mathematica Scandinavica} (1953) 31--38}.

\bibitem{burckel1980introduction}
R.~B. Burckel, \emph{An introduction to classical complex analysis}.
\newblock Academic Press, 1980.

\bibitem{Tomboulis:2015gfa}
E.~T. Tomboulis, \emph{{Nonlocal and quasilocal field theories}},
  \href{http://dx.doi.org/10.1103/PhysRevD.92.125037}{\emph{Phys. Rev.}
  {\bfseries D92} (2015) 125037},
  [\href{https://arxiv.org/abs/1507.00981}{{\ttfamily 1507.00981}}].

\bibitem{Biswas:2016etb}
T.~Biswas, A.~S. Koshelev and A.~Mazumdar, \emph{{Gravitational theories with
  stable (anti-)de Sitter backgrounds}},
  \href{http://dx.doi.org/10.1007/978-3-319-31299-6_5}{\emph{Fundam. Theor.
  Phys.} {\bfseries 183} (2016) 97--114},
  [\href{https://arxiv.org/abs/1602.08475}{{\ttfamily 1602.08475}}].

\bibitem{Tomboulis:1997gg}
E.~T. Tomboulis, \emph{{Superrenormalizable gauge and gravitational theories}},
   \href{https://arxiv.org/abs/hep-th/9702146}{{\ttfamily hep-th/9702146}}.

\bibitem{Modesto:2011kw}
L.~Modesto, \emph{{Super-renormalizable Quantum Gravity}},
  \href{http://dx.doi.org/10.1103/PhysRevD.86.044005}{\emph{Phys. Rev.}
  {\bfseries D86} (2012) 044005},
  [\href{https://arxiv.org/abs/1107.2403}{{\ttfamily 1107.2403}}].

\bibitem{Buoninfante:2018rlq}
L.~Buoninfante, A.~S. Koshelev, G.~Lambiase, J.~Marto and A.~Mazumdar,
  \emph{{Conformally-flat, non-singular static metric in infinite derivative
  gravity}}, \href{http://dx.doi.org/10.1088/1475-7516/2018/06/014}{\emph{JCAP}
  {\bfseries 1806} (2018) 014},
  [\href{https://arxiv.org/abs/1804.08195}{{\ttfamily 1804.08195}}].

\bibitem{Buoninfante:2018xif}
L.~Buoninfante, A.~S. Cornell, G.~Harmsen, A.~S. Koshelev, G.~Lambiase,
  J.~Marto et~al., \emph{{Towards nonsingular rotating compact object in
  ghost-free infinite derivative gravity}},
  \href{http://dx.doi.org/10.1103/PhysRevD.98.084041}{\emph{Phys. Rev.}
  {\bfseries D98} (2018) 084041},
  [\href{https://arxiv.org/abs/1807.08896}{{\ttfamily 1807.08896}}].

\bibitem{Buoninfante:2018stt}
L.~Buoninfante, G.~Harmsen, S.~Maheshwari and A.~Mazumdar, \emph{{Nonsingular
  metric for an electrically charged point-source in ghost-free infinite
  derivative gravity}},
  \href{http://dx.doi.org/10.1103/PhysRevD.98.084009}{\emph{Phys. Rev.}
  {\bfseries D98} (2018) 084009},
  [\href{https://arxiv.org/abs/1804.09624}{{\ttfamily 1804.09624}}].

\bibitem{Frolov:2015bia}
V.~P. Frolov, A.~Zelnikov and T.~de~Paula~Netto, \emph{{Spherical collapse of
  small masses in the ghost-free gravity}},
  \href{http://dx.doi.org/10.1007/JHEP06(2015)107}{\emph{JHEP} {\bfseries 06}
  (2015) 107}, [\href{https://arxiv.org/abs/1504.00412}{{\ttfamily
  1504.00412}}].

\bibitem{Frolov:2015usa}
V.~P. Frolov and A.~Zelnikov, \emph{{Head-on collision of ultrarelativistic
  particles in ghost-free theories of gravity}},
  \href{http://dx.doi.org/10.1103/PhysRevD.93.064048}{\emph{Phys. Rev.}
  {\bfseries D93} (2016) 064048},
  [\href{https://arxiv.org/abs/1509.03336}{{\ttfamily 1509.03336}}].

\bibitem{Koshelev:2017bxd}
A.~S. Koshelev and A.~Mazumdar, \emph{{Do massive compact objects without event
  horizon exist in infinite derivative gravity?}},
  \href{http://dx.doi.org/10.1103/PhysRevD.96.084069}{\emph{Phys. Rev.}
  {\bfseries D96} (2017) 084069},
  [\href{https://arxiv.org/abs/1707.00273}{{\ttfamily 1707.00273}}].

\bibitem{Li:2015bqa}
Y.-D. Li, L.~Modesto and L.~Rachwał, \emph{{Exact solutions and spacetime
  singularities in nonlocal gravity}},
  \href{http://dx.doi.org/10.1007/JHEP12(2015)173}{\emph{JHEP} {\bfseries 12}
  (2015) 173}, [\href{https://arxiv.org/abs/1506.08619}{{\ttfamily
  1506.08619}}].

\bibitem{Kilicarslan:2018yxd}
E.~Kilicarslan, \emph{{Weak Field Limit of Infinite Derivative Gravity}},
  \href{http://dx.doi.org/10.1103/PhysRevD.98.064048}{\emph{Phys. Rev.}
  {\bfseries D98} (2018) 064048},
  [\href{https://arxiv.org/abs/1808.00266}{{\ttfamily 1808.00266}}].

\bibitem{Koivisto:2018loq}
T.~Koivisto and G.~Tsimperis, \emph{{The spectrum of teleparallel gravity}},
  \href{https://arxiv.org/abs/1810.11847}{{\ttfamily 1810.11847}}.

\bibitem{Conroy:2017yln}
A.~Conroy and T.~Koivisto, \emph{{The spectrum of symmetric teleparallel
  gravity}}, \href{http://dx.doi.org/10.1140/epjc/s10052-018-6410-z}{\emph{Eur.
  Phys. J.} {\bfseries C78} (2018) 923},
  [\href{https://arxiv.org/abs/1710.05708}{{\ttfamily 1710.05708}}].

\bibitem{Kostelecky:2007kx}
V.~A. Kostelecky, N.~Russell and J.~Tasson, \emph{{New Constraints on Torsion
  from Lorentz Violation}},
  \href{http://dx.doi.org/10.1103/PhysRevLett.100.111102}{\emph{Phys. Rev.
  Lett.} {\bfseries 100} (2008) 111102},
  [\href{https://arxiv.org/abs/0712.4393}{{\ttfamily 0712.4393}}].

\bibitem{Brookfield:2005td}
A.~Brookfield, C.~van~de Bruck, D.~Mota and D.~Tocchini-Valentini,
  \emph{{Cosmology with massive neutrinos coupled to dark energy}},
  \href{http://dx.doi.org/10.1103/PhysRevLett.96.061301}{\emph{Phys. Rev.
  Lett.} {\bfseries 96} (2006) 061301},
  [\href{https://arxiv.org/abs/astro-ph/0503349}{{\ttfamily
  astro-ph/0503349}}].

\bibitem{Fukugita:1999as}
M.~Fukugita, G.-C. Liu and N.~Sugiyama, \emph{{Limits on neutrino mass from
  cosmic structure formation}},
  \href{http://dx.doi.org/10.1103/PhysRevLett.84.1082}{\emph{Phys. Rev. Lett.}
  {\bfseries 84} (2000) 1082--1085},
  [\href{https://arxiv.org/abs/hep-ph/9908450}{{\ttfamily hep-ph/9908450}}].

\end{thebibliography}\endgroup

\end{document}